\documentclass[aps,pra,11pt,amsmath,onecolumn,tightenlines,notitlepage,floatfix,superscriptaddress,longbibliography,eqsecnum]{revtex4-1}
\usepackage[dvipsnames]{xcolor}
\usepackage[colorlinks=true,hypertexnames=false,linkcolor=NavyBlue,urlcolor=NavyBlue,citecolor=NavyBlue,breaklinks]{hyperref}
\usepackage{amsmath,amsfonts,amssymb,amsthm}
\usepackage{mathtools}
\usepackage{thmtools}
\renewcommand\thmcontinues[1]{\textbf{continued}}
\usepackage{graphicx}
\usepackage{physics}
\usepackage{newtxtext,inconsolata}
\usepackage{enumitem}
\usepackage{longtable}

\newtheorem{theorem}{Theorem}[section]
\newtheorem{corollary}{Corollary}[section]
\newtheorem{proposition}{Proposition}[section]
\newtheorem{lemma}{Lemma}[section]
\newtheorem{example}{Example}[section]

\newcommand{\intall}{\int_{-\infty}^{\infty}}
\newcommand{\avg}[1]{\langle#1\rangle}

\newcommand{\Avg}[1]{\left\langle#1\right\rangle}

\newcommand{\bk}[1]{\qty(#1)}
\newcommand{\Bk}[1]{\qty[#1]}
\newcommand{\BK}[1]{\qty{#1}}

\newcommand{\mc}[1]{\mathcal #1}
\newcommand{\mb}[1]{\mathbb #1}

\newcommand{\insing}{e}
\newcommand{\outsing}{o}
\renewcommand{\norm}[1]{\Vert #1 \Vert}
\newcommand{\Norm}[1]{\left\Vert #1 \right\Vert}

\newcommand{\Real}{\mathfrak R}
\newcommand{\imag}{\mathfrak I}
\newcommand{\conj}{\mathfrak C}
\newcommand{\effinf}{\nabla\beta}
\newcommand{\effinfout}{\nabla_\out\beta}
\newcommand{\sdiff}{\Delta\beta}
\newcommand{\sdiffout}{\Delta_\out\beta}
\newcommand{\out}{\mathrm{out}}
\newcommand{\noise}{\mathrm{noise}}
\newcommand{\info}{\operatorname{info}}
\newcommand{\score}{\operatorname{score}}
\newcommand{\push}{\mathrm{push}}
\newcommand{\pull}{\mathrm{pull}}
\newcommand{\imap}{\operatorname{Imap}}
\newcommand{\IMAP}{\operatorname{IMAP}}
\newcommand{\Imat}{\operatorname{imat}}

\newcommand{\bnd}{\operatorname{bnd}}
\newcommand{\Bnd}{\operatorname{Bnd}}
\newcommand{\Amb}{\mc A}

\newcommand{\chan}{F}
\newcommand{\scatter}{S}
\newcommand{\Amp}{\operatorname{Amp}}
\newcommand{\inj}{\iota}
\newcommand{\jordan}{\star}
\newcommand{\covarspace}{\mc V}
\newcommand{\dual}{*}
\newcommand{\adj}{\dagger}
\newcommand{\el}{\omega}
\newcommand{\elout}{\eta}
\newcommand{\sspace}{\Omega}
\newcommand{\error}{\mathrm{error}}
\newcommand{\fock}{\Gamma}

\newcommand{\cconj}[1]{\overline{#1}}
\newcommand{\mean}{\tilde}
\DeclareMathOperator{\complete}{complete}

\DeclareMathOperator*{\argmin}{arg\,min}
\DeclareMathOperator{\spn}{span}

\DeclareMathOperator{\sinc}{sinc}

\DeclareMathOperator{\expect}{\mathbb E_\theta}

\DeclareMathOperator{\range}{range}

\DeclareMathOperator{\kernel}{kernel}
\DeclareMathOperator{\dom}{domain}
\DeclareMathOperator{\codom}{codomain}

\newcommand{\cgraphic}[2]{\centerline{\includegraphics[width=#1\textwidth]{#2}}}
\newcommand{\fig}[3]{
\begin{figure}[htbp!]
\cgraphic{#1}{#2}
\caption{\label{#2}#3}
\end{figure}
}

\begin{document}

\title{Unified theory of classical and quantum semiparametric efficiency}
\author{Mankei Tsang}
\email{mankei@nus.edu.sg}
\homepage{https://blog.nus.edu.sg/mankei/}
\affiliation{Department of Electrical and Computer Engineering, National University of Singapore, 4 Engineering Drive 3, Singapore 117583}

\affiliation{Department of Physics, National University of Singapore, 2 Science Drive 3, Singapore 117551}

\date{\today}


\begin{abstract}
  In classical and quantum statistics, high-dimensional unknown
  parameters are abundant and it is often prudent to make minimal
  assumptions about them using so-called semiparametric models.  To
  attack a wide range of semiparametric problems in one broad stroke,
  we present a unified treatment of statistical efficiency for
  classical and quantum semiparametric models, generalizing the
  Cram\'er--Rao and Helstrom bounds beyond finite-dimensional
  parameters. We introduce the fundamental concepts in abstract and
  geometric terms before applying them to many examples, covering
  general classical and quantum models as well as the paradigmatic
  special cases of Gaussian and Poisson fields. We give an in-depth
  treatment of channels in the semiparametric efficiency theory and
  advocate the use of the singular value decomposition to elucidate
  the statistical effects of channels. To demonstrate the utility of
  the formalism, we apply it to coherent and incoherent optical
  imaging problems, assuming an arbitrary field or intensity on the
  object plane without parametric assumptions. Our formalism enables
  us to compute classical and quantum limits to coherent and
  incoherent imaging resolution in statistical terms. For
  subdiffraction incoherent imaging, we demonstrate that spatial-mode
  demultiplexing can be far superior to direct imaging in estimating
  generalized Fourier coefficients and come closer to the quantum
  limits. We envision our theory becoming an essential tool for both
  classical and quantum statistics with useful applications to sensing
  and imaging, whenever minimal assumptions about a high-dimensional
  parameter should be made.
\end{abstract}

\maketitle


\section{Introduction}
An experiment often involves a large number of unknowns.  In the
language of statistics, all the unknowns are modeled as one parameter
in a high-dimensional parameter space.  To extract information about
the high-dimensional parameter from noisy observations, it is often
prudent to design the data processing based on a semiparametric model,
which generally refers to an infinite-dimensional parameter space
\cite{vaart,bickel}. Despite the vastness of the space, it is possible
to generalize the Cram\'er--Rao bound to a bound for semiparametric
models known as the semiparametric efficiency bound in classical
statistics \cite{vaart,bickel}.  By setting a lower limit on the
errors of a large class of estimators, the bound defines a fundamental
notion of statistical efficiency. In the same way the Cram\'er--Rao
bound has become a popular benchmark for many finite-dimensional
problems, the semiparametric efficiency bound has found widespread
applications in diverse areas, including biostatistics, econometrics,
and imaging
\cite{vaart,bickel,tsiatis,laan,kennedy24,newey90,spade_prr}.

In physics, infinite-dimensional parameter spaces are just as
common. An important example is optical imaging, where the unknown
parameter is, in general, the optical field or intensity on the object
plane \cite{villiers}. For physical objects, it is possible to
introduce an additional layer of physical processing by devices and
measurements before the data processing by a classical computer. For
example, the optical field arriving on the image plane of an imaging
system may be processed by linear or nonlinear optics before it is
measured by photodetectors.  To extract information about a
high-dimensional parameter from the physical objects in sensing and
imaging, it may again be prudent to design the processing, including
the physical layer, based on a semiparametric model of the
objects. When quantum effects, such as photon shot noise, are
important, the model of the objects should be quantum as well. An
efficiency bound for quantum semiparametric models was invented only
recently \cite{semi_prx}. Similar to the popular quantum Cram\'er--Rao
bound invented by Helstrom \cite{helstrom}, the semiparametric version
sets the fundamental limit on the errors of a large class of
processing. The bound has been applied to the problem of
semiparametric moment estimation in incoherent imaging
\cite{qlmoment_pra2,tan23}, proving the approximate optimality of a
measurement known as spatial-mode demultiplexing (SPADE)
\cite{tnl,review_cp,lvovsky26,spade_njp,spade_pra,zhou19,qlmoment_pra,qlmoment_pra2,tan23,tan23a}.

There should be many more areas and applications---both classical and
quantum---that can utilize the semiparametric efficiency theory. This
belief is the motivation for this work, which offers a unified
treatment of both the established classical theory \cite{vaart,bickel}
and the recently developed quantum version
\cite{semi_prx,qlmoment_pra2,tan23}.

We begin by introducing the fundamental concepts in abstract
mathematical terms in Sec.~\ref{sec_concepts}.  To elucidate the
abstract theory, we adopt a geometric perspective
\cite{rao45,amari,amari_app} with many pictures to explain the key
concepts, instead of solely relying on functional analysis or linear
algebra. The concepts are then given concrete statistical and physical
meanings through a multitude of examples introduced in
Sec.~\ref{sec_bnd}. The examples cover general classical and quantum
models as well as the paradigmatic special cases of Gaussian and
Poisson fields. This separation of the abstract mathematics and
geometry from the concrete applications allows us to achieve both
rigor and economy in our presentation. Most technical results in this
paper are known for classical models \cite{vaart,bickel}, but our
approach reveals their wider applicability. The hope is that the
fundamental theory can be applied to many more problems and need not
be repeated in future research, as was done too often in earlier
papers for special cases, such as the treatment of the Poisson and
quantum models in Refs.~\cite{spade_prr,semi_prx}. To demonstrate this
point, Sec.~\ref{sec_exa} presents many basic classical and quantum
semiparametric problems that can be solved by similar techniques.

Another emphasis of this work, as presented in
Secs.~\ref{sec_chan}--\ref{sec_chan_exa2}, is an in-depth treatment of
channels in semiparametric efficiency theory. As physical objects
often travel through many channels before reaching their destinations
in sensing, imaging, communication, and computation systems, the
importance of channels in classical and quantum information science
can hardly be overstated. Its treatment in classical semiparametric
statistics is not as prominent, however, and often confined to
specific examples \cite{bickel,vaart,trabs15}.  Here, we study how the
efficiency bounds and the related quantities transform under the
effects of channels, again starting from abstract terms before working
out concrete classical and quantum examples.

To characterize and visualize the channel effects, we advocate the
singular-value decomposition (SVD) of the linearized maps modeling a
channel. Engineers will recognize the SVD as a generalization of
Fourier analysis in signal processing; it is, of course, a fundamental
tool of superresolution research in particular
\cite{villiers,brady,xudong}. A key novel contribution of this work is
to integrate SVD with the semiparametric efficiency theory, where the
vector spaces are defined by statistics. As we will see, the geometric
picture of the SVD is intuitive, pairs well with the geometric picture
of the efficiency theory, and paves the way for numerical analysis of
semiparametric problems, in the same way Fourier analysis is useful
and vital in signal processing.

For Gaussian or Poisson fields \cite{bogachev,kingman} and their
quantum generalizations \cite{holevo_info,poisson_quantum}, we
demonstrate that the theory can be significantly simplified.  An
important application of these models is optical imaging, as studied
in Secs.~\ref{sec_coh}--\ref{sec_confocal} as well as
Appendix~\ref{app_qcoh}.  For the imaging of spatially coherent
sources, the Gaussian models are appropriate
\cite{villiers,yuen78,shapiro79,kolobov_fabre,beskrovnyy05,pinel12,taylor16,treps}
and used in Sec.~\ref{sec_coh} and Appendix~\ref{app_qcoh}. For the
imaging of incoherent sources at optical frequencies, on the other
hand, the Poisson models are appropriate
\cite{goodman_stat,zmuidzinas03,pawley,spade_prr,tnl,review_cp,lvovsky26,poisson_quantum}
and used in Secs.~\ref{sec_direct}--\ref{sec_confocal}. By presenting
the statistically motivated SVDs for the imaging systems, we can
express the fundamental classical and quantum limits in a terminology
familiar to opticians, or at least those well versed in
superresolution research \cite{villiers,brady,xudong}. As shown in
Appendix~\ref{app_num}, the numerical method we use to compute the
SVDs and the efficiency bounds is nontrivial, as care must be taken to
perform the decomposition with respect to the correct vector spaces.

As the Cram\'er--Rao bound has become a standard benchmark for
fluorescence microscopy in recent years
\cite{deschout,chao16,diezmann17}, we envision the semiparametric
efficiency bounds to be natural generalizations that should grow in
importance for imaging applications.

We will be brief on the coherent case and elaborate on the incoherent
case, because the quantum bounds for the latter case, as shown in
Sec.~\ref{sec_incoh_q}, reveal the more remarkable physics that
conventional imaging is far from optimal for subdiffraction objects,
even though the source is thermal and thus classical. We show further
in Sec.~\ref{sec_spade} that SPADE can come closer to the quantum
bounds and offer substantial enhancements over conventional
imaging. Such a result is perhaps unsurprising after all the recent
work on the subject
\cite{tnl,review_cp,lvovsky26,spade_njp,spade_pra,zhou19,qlmoment_pra,qlmoment_pra2,tan23,tan23a,superosc_ieee},
but the phenomenon represents a significant outcome from the research
of quantum limits that has real experimental and practical
implications \cite{lvovsky26,kim25,wallis26} and is worth showcasing
in the new light of the general formalism.

Sec.~\ref{sec_confocal}
briefly discusses a model of incoherent imaging with structured
illuminations, leaving a detailed analysis to future work.
Sec.~\ref{sec_ext} discusses other potential extensions of the
presented theory, while Sec.~\ref{sec_con} is the conclusion. More
background material is presented in the
appendices. Appendix~\ref{app_note} explains our basic notations and
provides a table of symbols in
Table~\ref{tab_symbols}. Appendix~\ref{app_fock} reviews the concepts
of bosonic Fock space and Gaussian states.  Appendix~\ref{app_error}
reviews the meaning of the efficiency bound in classical
statistics. Appendix~\ref{app_pinv}--\ref{app_num} describe some
important Hilbert-space concepts used in the main
text. Appendix~\ref{app_qcoh} studies a quantum model of coherent
imaging, following Sec.~\ref{sec_coh}. Appendix~\ref{app_proofs}
collects the proofs not presented in the main text or
elsewhere. Appendix~\ref{app_foot} contains footnotes about minor
issues that arise from the main text.

\section{\label{sec_concepts}Fundamental concepts}

\subsection{\label{sec_tangent}Parameter space, tangent space, parameter of interest}
In a statistical problem, the unknown variable, commonly denoted as
$\theta$, is called the parameter, and the set of all possible values
of $\theta$, denoted as $\Theta$, is called the parameter space. Let
us focus on one $\theta \in \Theta$ and call it the true parameter.
Throughout this paper, only a neighborhood around the true $\theta$
needs to be considered. Let $[t_1,t_2] \subset \mb R$ be an interval
containing $0$ and consider a curve $\phi: [t_1,t_2]\to \Theta$ in the
parameter space that passes through the true $\theta$ at
$\phi(0) = \theta$. In the context of statistics, $\phi$ is called a
parametric submodel. Define the tangent vector $\dot\phi$ associated
with the curve at $\theta$ as the directional derivative \cite{lee}
\begin{align}
\dot\phi h &= \left.\pdv{h(\phi(t))}{t}\right|_{t=0}
\end{align}
for any differentiable function $h:\Theta \to \codom h$. Denote an
appropriate set of such submodels that all pass through the same true
$\theta$ as $\Phi(\theta)$ and the set of all their tangent vectors as
the parameter tangent space $\mc T_\theta$. We assume that
$\mc T_\theta$ is a real vector space.  We do not assume any other
structure for $\Theta$ and $\mc T_\theta$ for now; the dimension of
$\mc T_\theta$, in particular, may be finite or infinite.

A direct estimation of $\theta$ is often undesirable. For example, if
$\Theta$ is infinite-dimensional, then an estimation of $\theta$ is
simply infeasible as it would require an infinite set of numbers to be
computed and reported, or a component of $\theta$ may be a nuisance
parameter that is unimportant. This issue motivates us to define a
parameter of interest $\beta:\Theta \to \mb R^q$ with finite dimension
$q$ to summarize the important features of $\theta$, such that one
estimates $\beta(\theta)$ rather than $\theta$. $\beta$ is also called
an estimand in statistics \cite{lehmann_casella}.  We further simplify
our formalism by focusing on the scalar case ($q = 1$); a
multidimensional $\beta$ can be treated by applying our formalism to
the components of $\beta$ one by one, or following the straightforward
extension in Ref.~\cite{semi_prx}. Assume that $\beta$ is
differentiable for all submodels $\phi \in \Phi(\theta)$, such that
\begin{align}
\dot\phi(\beta) &= \left.\pdv{\beta(\phi(t))}{t}\right|_{t=0}.
\end{align}
Then we can linearize $\beta$ by defining a linear functional
$\dd\beta:\mc T_\theta \to \mb R$ of tangent vectors, called the
differential of $\beta$, as
\begin{align}
\dd\beta(\dot\phi) &= \dot\phi(\beta).
\label{differential}
\end{align}
Denote the set of all linear functionals of $\mc T_\theta$, also
called tangent covectors, as $\mc T_\theta^\dual$; then
$\dd\beta \in \mc T_\theta^\dual$. Physically, $\dd\beta(\dot\phi)$
measures the sensitivity of $\beta(\theta)$ to a perturbation of the
parameter $\theta$ along the tangent vector $\dot\phi$.

\subsection{\label{sec_bnd}Submodel efficiency bound, examples}
Given a submodel $\phi \in \Phi(\theta)$, define the submodel
efficiency bound at $\theta$ as
\begin{align}
\bnd_\theta(\dot\phi) &= \frac{[\dd\beta(\dot\phi)]^2}
{\info_\theta(\dot\phi,\dot\phi)},
\label{Csub}
\end{align}
where $\info_\theta:\mc T_\theta\times\mc T_\theta \to \mb R$ is a
symmetric and positive-semidefinite bilinear form that we call the
information semimetric, or simply the information for short. We call
it a semimetric because $\info_\theta(\dot\phi,\dot\phi) = 0$ need not
imply that $\dot\phi = 0$. We call $\info_\theta(\dot\phi,\dot\phi)$
the information of the submodel $\phi$ and assume that $\Phi(\theta)$
contains only submodels with finite information. We will not consider
submodels with both $\dd\beta(\dot\phi) = 0$ and
$\info_\theta(\dot\phi,\dot\phi) = 0$, so that
$\bnd_\theta(\dot\phi) = 0/0$ will never occur, although it is
possible that $\dd\beta(\dot\phi) \neq 0$ and
$\info_\theta(\dot\phi,\dot\phi) = 0$, leading to
$\bnd_\theta(\dot\phi) = \infty$. Notice that
$\info_\theta(\dot\phi,\dot\chi)$ and $\bnd_\theta(\dot\phi)$ depend
on the submodels only in terms of their tangent vectors
$\dot\phi \in \mc T_\theta$. The information semimetric hence provides
an operational definition of the length of a tangent vector as well as
the angle between any pair of vectors. Some examples are in order.

\begin{example}[name=Classical,label=exa_classical]
  Let $P_\theta$ be a probability measure for a random element $X$ in
  the sample space $\sspace$ and some $\sigma$-algebra.  Assume that a
  dominating measure $\mu$ exists for all $\theta$ such that the
  probability density is $f(\theta,\el)= (\dv*{P_\theta}{\mu})(\el)$.
  Given a set of densities $\{f(\theta,\cdot):\theta \in \Theta\}$,
  define the so-called score function $\score(\dot\phi,\el)$ as
\begin{align}
\score(\dot\phi,\el) &= \dot\phi\Bk{\ln f(\theta,\el)} = 
\left.\pdv{}{t} \ln f(\phi(t),\el)\right|_{t=0}.
\label{score}
\end{align}
Then the Fisher information is 
\begin{align}
\info_\theta(\dot\phi,\dot\chi) &= \expect\Bk{\score(\dot\phi)\score(\dot\chi)}
= \int \score(\dot\phi,\el)
\score(\dot\chi,\el) f(\theta,\el) \dd\mu(\el),
\label{fish}
\end{align}
where $\expect$ denotes the expected value given $\theta$, and
$\bnd_\theta(\dot\phi)$ is the Cram\'er--Rao bound for the submodel at
$t = 0$ with a scalar unknown parameter $t$ and a parameter of
interest $\beta(\phi(t))$ \cite{lehmann_casella}. In particular, if
$\sspace$ is countable, we can take $f(\theta,\el)$ to be a
probability mass function, and the information becomes
\begin{align}
\info_\theta(\dot\phi,\dot\chi) &= \sum_\el  
\score(\dot\phi,\el) \score(\dot\chi,\el) f(\theta,\el).
\label{fish_countable}
\end{align}
\end{example}

\begin{example}[name=Poisson random field,label=exa_poisson]
  Let $X$ be a Poisson field, which is a random point measure with the
  points in some space $\sspace$. Its properties are completely
  specified by its mean measure on $\sspace$ \cite{kingman,cinlar}.
  Assume that the mean measure $\mean X_\theta$ depends on $\theta$,
  and the density of $\mean X_\theta$ with respect to a common measure
  $\mu$ exists and is given by
  $f(\theta,\el) = (\dv*{\mean X_\theta}{\mu})(\el)$. Unlike
  probability measures and densities, $\mean X_\theta$ and
  $f(\theta,\cdot)$ need not be normalized.  Write a linear functional
  of $X$ as
\begin{align}
L(X) &= \int L'(\el) \dd X(\omega)
\end{align}
in terms of a function $L':\sspace\to\mb R$. Then its expected value
becomes
\begin{align}
\expect(L) &= \int L'(\el) f(\theta,\el) \dd\mu(\el),
\end{align}
and the covariance between two linear functionals becomes
\begin{align}
\expect(LK)-\expect(L)\expect(K)
&= \int L'(\el) K'(\el) f(\theta,\el) \dd\mu(\el).
\end{align}
Given a set of densities $\{f(\theta,\cdot):\theta \in \Theta\}$,
define the score function as Eq.~(\ref{score}).  Then the Fisher
information has the same formula as Eq.~(\ref{fish})
\cite{takahashi90}. In particular, if $\sspace$ is countable, $X$ is a
set of independent Poisson random variables and we can set
$f(\theta,\el) = \mean X_\theta(\el)$, leading to
Eq.~(\ref{fish_countable}).
\end{example}

\begin{example}[name=Gaussian random field,label=exa_gauss]
  Let $Z \in\mc B$ be a zero-mean Gaussian field in a real separable
  Banach space $\mc B$, $\mu$ be its probability measure,
  $\mc B^\dual$ be the dual of $\mc B$, and $\mc B^{\dual\dual}$ be
  the dual of $\mc B^\dual$ \cite{bogachev_gauss}.  Each
  $L \in \mc B^\dual$ gives a zero-mean Gaussian random variable
  $L(Z)$.  The covariance map
  $\Sigma:\mc B^\dual \to \mc B^{\dual\dual}$ is defined such that
\begin{align}
(\Sigma L)(K) = \int L(\el)K(\el)\dd\mu(\el)
\end{align}
is the covariance between $L,K \in \mc B^\dual$.  Let
$\mc V \subseteq \mc B^\dual$ be the finite-variance subset of
$\mc B^\dual$ ($(\Sigma L)(L) < \infty$). Define a semi-inner product
for $\mc V$ as
\begin{align}
\Avg{L,K}_{\mc V} &= \bk{\Sigma L}(K).
\label{semi_gauss}
\end{align}
Now consider a field $X = f + Z\in \mc B$ that obeys a shifted
Gaussian measure $\mu_f(\cdot)= \mu(\cdot-f)$ with mean vector
$f\in \mc B$.  Then $\int L(\el) \dd\mu_f(\el) = L(f)$ for any
$L \in \mc B^\dual$.  $\mu_f$ and $\mu$ are equivalent if
$f \in \Sigma \mc V = \{\Sigma L: L \in \mc V\}$ (treating $f$ as
an element of $\mc B^{\dual\dual}$), such that the Radon--Nikodym
derivative is defined. Given a set of mean vectors
$\{f(\theta):\theta \in \Theta\}\subseteq \Sigma\mc V$ for $\mu_f$,
the Fisher information becomes \cite{sekine95}
\begin{align}
\info_\theta(\dot\phi,\dot\chi) &=
\Avg{\score(\dot\phi),\score(\dot\chi)}_{\mc V},
\label{info_gauss}
\end{align}
where $\score(\dot\phi)$ is any solution to
\begin{align}
\Sigma \score(\dot\phi) &= \dot\phi\Bk{f(\theta)} = 
\left.\pdv{f(\phi(t))}{t}\right|_{t=0}.
\label{score_gauss}
\end{align}
In the finite-dimensional case with
$\mc B = \mb R^p = \mc B^\dual = \mc B^{\dual\dual}$, $X$ is a vector
of $p$ Gaussian random variables with mean vector
$f(\theta) \in \mb R^p$ and $p\times p$ covariance matrix $\Sigma$.
Assume that all $\mb R^p$ vectors are column vectors. Each
$A \in \mc B^\dual$ can be expressed as $A(X) = A^\top X$.  Assume,
for simplicity, that $\Sigma$ is positive-definite ($\Sigma > 0$).
The Fisher information becomes \cite{kay}
\begin{align}
\info_\theta(\dot\phi,\dot\chi) &= 
\score(\dot\phi)^\top \Sigma \score(\dot\chi)
= (\dot\phi f)^\top \Sigma^{-1}(\dot\chi f ).
\label{info_gauss_finite}
\end{align}
\end{example}

\begin{example}[name=Quantum,label=exa_quantum]
  Let $\mc H$ be a complex separable Hilbert space. Denote the adjoint
  of an operator as the superscript $\dual$ and the trace of an
  operator as $\trace$.  An operator $\rho$ on $\mc H$ is called a
  density operator if it is self-adjoint ($\rho^\dual = \rho$),
  positive-semidefinite ($\rho \ge 0$), and unit-trace
  ($\trace \rho = 1$). A density operator models the state of a
  quantum object, just as a probability measure models the law of a
  random element. We also call self-adjoint operators observables when
  they act as generalizations of random variables.

  Let $\BK{f(\theta):\theta \in \Theta}$ be a set of density operators
  on $\mc H$ that model a quantum object. Define $\score(\dot\phi)$ as any
  self-adjoint solution to the Lyapunov equation
\begin{align}
\dot\phi\Bk{f(\theta)} &= f(\theta) \jordan \score(\dot\phi),
&
A \jordan B &= \frac{AB+BA}{2},
\label{lyap}
\end{align}
where $\jordan$ denotes the Jordan product between two operators.
$\score(\dot\phi)$ is called a quantum score operator
\cite{petz,semi_prx}, a symmetric logarithmic derivative, or the
e-representation of the tangent vector $\dot\phi$ \cite{hayashi}. The
Helstrom information, a quantum version of the Fisher information, is
given by
\begin{align}
\info_\theta(\dot\phi,\dot\chi) &= 
\trace \BK{ f(\theta) \Bk{\score(\dot\phi) \jordan \score(\dot\chi)}},
\label{helstrom}
\end{align}
and $\bnd_\theta(\dot\phi)$ is a quantum Cram\'er--Rao bound for the
submodel \cite{helstrom}, generalizing the classical bound.
\end{example}

\begin{example}[name=Quantum Poisson states,label=exa_qpoisson]
  Let $\mc H$ be the complex Hilbert space for the modes of a quantum
  field. In general, the Hilbert space for the quantum field is a
  many-body Fock space constructed by regarding $\mc H$ as the
  one-body space \cite{parth_qsc}, but a simpler type of states called
  Poisson states have been found useful in the study of weak thermal
  fields \cite{poisson_quantum}. A Poisson state models a collection
  of $L$ independent quantum bodies, such as photons.  $L$ is a
  Poisson random variable with mean $N \ge 0$, while the state of each
  body is modeled by a density operator $\rho$ on $\mc H$. A Poisson
  state is completely specified by the so-called intensity operator
  $N \rho$ on $\mc H$, which has all the properties of a density
  operator except that the trace need not be normalized.

  Given a set of intensity operators $\{f(\theta):\theta \in \Theta\}$,
  the Helstrom information obeys the same formulas as
  Eqs.~(\ref{lyap})--(\ref{helstrom}) \cite{poisson_quantum}.
\end{example}

\begin{example}[name=Quantum Gaussian states,label=exa_qgauss]
  For a bosonic quantum field with complex mode Hilbert space $\mc H$,
  a special type of states called Gaussian states on the bosonic Fock
  space $\fock(\mc H)$ have been studied extensively in quantum
  information science \cite{holevo_info}. Appendix~\ref{app_fock}
  reviews some basic results about $\fock(\mc H)$ and Gaussian states,
  following Refs.~\cite{parth_qsc,holevo_info}.

  To conform with the classical treatment in Example~\ref{exa_gauss},
  where all the spaces are real, we construct a real Hilbert space
  $\mc K$ called the phase space from the complex $\mc H$, as defined
  in Appendix~\ref{app_fock}. The real inner product of $\mc K$ is
  denoted as $\avg{\cdot,\cdot}_{\mc K}$. Each phase-space vector
  $A \in \mc K$ represents a mode of the field. Define $X(A)$ as a
  quadrature observable of the mode. A zero-mean Gaussian state $\mu$
  is specified by its covariance map $\Sigma:\mc K \to \mc K$, such
  that
\begin{align}
\trace\Bk{\mu X(A)} &= 0 \quad \forall A \in \mc K,
&
\trace\BK{\mu\Bk{X(A)\jordan X(B)}} &= \Avg{A,\Sigma B}_{\mc K}.
\end{align}
Define $\mc V$ as a vector space with the semi-inner product
\begin{align}
\Avg{A,B}_{\mc V} &= \Avg{A,\Sigma B}_{\mc K}.
\end{align}
Assume that $\Sigma$ is self-adjoint, positive-definite, bounded, and
invertible. For example, $\Sigma = I/2$ when $\mu$ is the vacuum
state. Since $\Sigma$ is invertible, $\mc V$ is isomorphic to $\mc K$
and also a Hilbert space.

Now consider a shifted Gaussian state given by
\begin{align}
\rho(f) &= W(f) \mu W(f)^\dual,
\end{align}
where $W(f)$ is the phase-space displacement operator, as defined in
Appendix~\ref{app_fock}, and $f \in \mc K$ is the displacement vector.
$f$ determines the mean of each $X(A)$ by
\begin{align}
\trace\Bk{\rho(f) X(A)} &= \Avg{f,A}_{\mc K}.
\end{align}
Given a set of mean vectors
$\{f(\theta):\theta \in \Theta\} \subseteq \mc K$ for $\rho(f)$, the
Helstrom information satisfies Eqs.~(\ref{info_gauss}) and
(\ref{score_gauss}), at least when $\mc H$ is finite-dimensional.  To
be specific, let $\mc H = \mb C^q$ and $\mc K = \mb R^{2q}$. Let
$X = (X_1,\dots,X_{2q})^\top$ be a column vector of quadrature
observables, where $(X_1,\dots,X_q)$ are the position observables for
mode $1$ to mode $q$ and $(X_{q+1},\dots,X_{2q})$ are the momentum
observables for mode $1$ to mode $q$.  The mean column vector
$f \in \mc V = \mb R^{2q}$ and the $2q\times 2q$ covariance matrix
$\Sigma$ become
\begin{align}
f_j &= \trace\bk{\rho X_j},
\label{fj}
\\
\Sigma_{jk} &= \trace\BK{\rho\Bk{\bk{X_j-f_j} \jordan \bk{X_k - f_k}}}.
\label{Sjk}
\end{align}
Then the Helstrom information has the same formula as
Eq.~(\ref{info_gauss_finite}) \cite{holevo_aspect,monras13}.
 \end{example}

\begin{example}[name=Finite-dimensional parameter space,label=exa_finite]
  Let $R \subseteq \Theta$ be a neighborhood of $\theta$ and
  $f:R \to \mb R^p$ be a chart, such that
  $f(\theta) = \mqty(f_1(\theta) & \dots &f_p(\theta))^\top$ are the
  coordinates of $\theta$.  Then the column vector
\begin{align}
\score(\dot\phi) &= \dot\phi(f) = \left.\pdv{f(\phi(t))}{t}\right|_{t=0} \in \mb R^p
\end{align}
is a Euclidean representation of $\dot\phi$. It follows that
\begin{align}
\dd\beta(\dot\phi) &= \score(\dot\phi)^\top \partial\beta,
\\
\info_\theta(\dot\phi,\dot\chi) &= 
\score(\dot\phi)^\top \Imat(\theta) \score(\dot\chi),
\end{align}
where $\partial\beta$ is a column vector, $\partial$ is defined as
\begin{align}
\bk{\partial h}_j &= \partial_j h = \pdv{}{x_j} h(f^{-1}(x)),
\end{align}
and 
\begin{align}
\Imat_{jk}(\theta) &= \info_\theta(\partial_j,\partial_k)
\label{Jmatrix}
\end{align}
is the information matrix. We stress that this result holds for any
$\info_\theta$, including the Fisher and Helstrom versions.
\end{example}

\subsection{\label{sec_score}Score tangent space}
In each of the preceding examples, observe that the information can be
expressed as
\begin{align}
\info_\theta(\dot\phi,\dot\chi) &= \Avg{\score(\dot\phi),\score(\dot\chi)}_{f(\theta)}
\end{align}
in terms of a semi-inner product between two elements
$\score(\dot\phi)$ and $\score(\dot\chi)$ of some real vector space
$\mc V$, and $\score:\mc T_\theta \to \mc V$ is a linear map.  The
information of a submodel is always expressed as
\begin{align}
\info_\theta(\dot\phi,\dot\phi) &= \norm{\score(\dot\phi)}_{f(\theta)}^2
\end{align}
in terms of the seminorm
\begin{align}
\norm{A}_{f(\theta)} &= \sqrt{\Avg{A,A}_{f(\theta)}}.
\end{align}
Observe also that $\score(\dot\phi)$ and the semi-inner product depend
on $\theta$ in terms of a function $f(\theta)$ that we call the
parametrization function, which determines a key statistic of the
model.
\begin{enumerate}
\item In Examples~\ref{exa_classical} and \ref{exa_poisson},
  $f(\theta) = f(\theta,\cdot)$ is a density function,
  $\score(\dot\phi) = \score(\dot\phi,\cdot)$ is a score function, and
  the semi-inner product between two functions is
\begin{align}
\Avg{A,B}_{f(\theta)} &= \int A(\el) B(\el) f(\theta,\el)\dd\mu(\el).
\label{inner_classical}
\end{align}

\item In Examples~\ref{exa_gauss} and \ref{exa_qgauss}, $f(\theta)$ is
  the mean vector of a classical or quantum Gaussian object,
  $\score(\dot\phi)$ is in a vector space $\mc V$ of finite-variance
  Gaussian variables, and the semi-inner product is the covariance
\begin{align}
\Avg{A,B}_{f(\theta)} &= \Avg{A,B}_{\mc V}.
\label{semi_gauss2}
\end{align}

\item In Examples~\ref{exa_quantum} and \ref{exa_qpoisson},
  $f(\theta)$ is a density operator or an intensity operator,
  $\score(\dot\phi)$ is a score operator, and the semi-inner product is
\begin{align}
\Avg{A,B}_{f(\theta)} &= \trace\Bk{f(\theta) \bk{A \jordan B}}.
\label{inner_quantum}
\end{align}

\item In Example~\ref{exa_finite}, $f(\theta)$ are the coordinates of
  $\theta$, $\score(\dot\phi)$ is a Euclidean representation of $\dot\phi$,
  and the semi-inner product is
\begin{align}
\Avg{A,B}_{f(\theta)} &= A^\top \Imat(\theta) B.
\end{align}
\end{enumerate}

We call $\Avg{\cdot,\cdot}_{f(\theta)}$ a semi-inner product (also
called a pre-inner product) and $\norm{\cdot}_{f(\theta)}$ a seminorm
because $\norm{A}_{f(\theta)} = 0$ need not imply that $A = 0$. It
will be mathematically fruitful, however, to construct from $\mc V$ an
abstract Hilbert space, where $\Avg{\cdot,\cdot}_{f(\theta)}$ becomes
a bona-fide inner product. The procedure is as follows \cite{conway}.
\begin{enumerate}
\item Let 
\begin{align}
\mc N(f) = \BK{A \in \mc V: \norm{A}_{f} = 0}
\end{align}
be the set of zero-seminorm elements of $\mc V$.  Group the elements
of $\mc V$ into equivalence classes
\begin{align}
A +\mc N(f) &= \BK{A + B: B \in \mc N(f)}.
\end{align}
Define each equivalence class as an element of a quotient space
denoted as $\mc V/\mc N(f)$.

\item For any $u,v \in \mc V/\mc N(f)$, define
\begin{align}
\Avg{u,v}_{f} &= \Avg{A_u,A_v}_{f},
&
\norm{u}_{f} &= \norm{A_u}_{f},
\end{align}
where $A_u$ is any element of the equivalence class $u$. Then
$\Avg{\cdot,\cdot}_{f}$ becomes an inner product, $\norm{\cdot}_{f}$
becomes a norm, and $\mc V/\mc N(f)$ is an inner-product space.  In
particular, $\mc N(f)$ becomes the unique zero element with
$\norm{\mc N}_f = 0$.

\item The completion of $\mc V/\mc N(f)$ by including all the limit
  points of $\mc V/\mc N(f)$ with respect to $\norm{\cdot}_{f}$ is a
  Hilbert space, denoted as $\Amb(f)$. We write the completion as
\begin{align}
\Amb(f) &= \complete\Bk{\mc V/\mc N(f)}.
\label{complete}
\end{align}
Since $\mc V/\mc N(f)$ is dense in $\Amb(f)$, we can also write
\begin{align}
\overline{\mc V/\mc N(f)} &= \Amb(f),
\end{align}
where the overline denotes the closure of the subset.
\end{enumerate}
We now make a few definitions.
\begin{enumerate}
\item We call $\mc V$ in each example the ambient vector space and the
  $\Amb\bk{f(\theta)}$ in Eq.~(\ref{complete}) the ambient Hilbert
  space.

\item Define a linear map $f_\push:\mc T_\theta \to \Amb\bk{f(\theta)}$
  that pushes each tangent vector $\dot\phi \in \mc T_\theta$ forward
  to the ambient Hilbert space $\Amb\bk{f(\theta)}$ as
\begin{align}
f_\push \dot\phi &= \score(\dot\phi) + \mc N\bk{f(\theta)} 
\in \mc V/\mc N\bk{f(\theta)}.
\end{align}
We call $f_\push$ the pushforward. 

\item The range of the pushforward is
\begin{align}
\range f_\push &= f_\push\mc T_\theta
\subseteq \mc V/\mc N\bk{f(\theta)}.
\end{align}
Define the closure of the range as the score tangent space
$\mc T\bk{f(\theta)}$, that is,
\begin{align}
\mc T\bk{f(\theta)} &= \overline{\range f_\push} \subseteq \Amb\bk{f(\theta)}.
\label{stat_space}
\end{align}
The completeness of $\Amb\bk{f(\theta)}$ ensures that the closure
stays in $\Amb\bk{f(\theta)}$, and the closure of the vector subspace
$\range f_\push$ ensures that $\mc T\bk{f(\theta)}$ is also a Hilbert
space. We call each element of $\mc T\bk{f(\theta)}$ a score.
\end{enumerate}
Because $f$, $f_\push$, and the inner product of $\mc T\bk{f(\theta)}$
all vary by the problem, the separate definitions of $\mc T_\theta$
and $\mc T(f(\theta))$, albeit cumbersome, are necessary to bring all
the problems under one general formalism. Fig.~\ref{pushforward}
illustrates the geometric concepts discussed so far.

\fig{0.7}{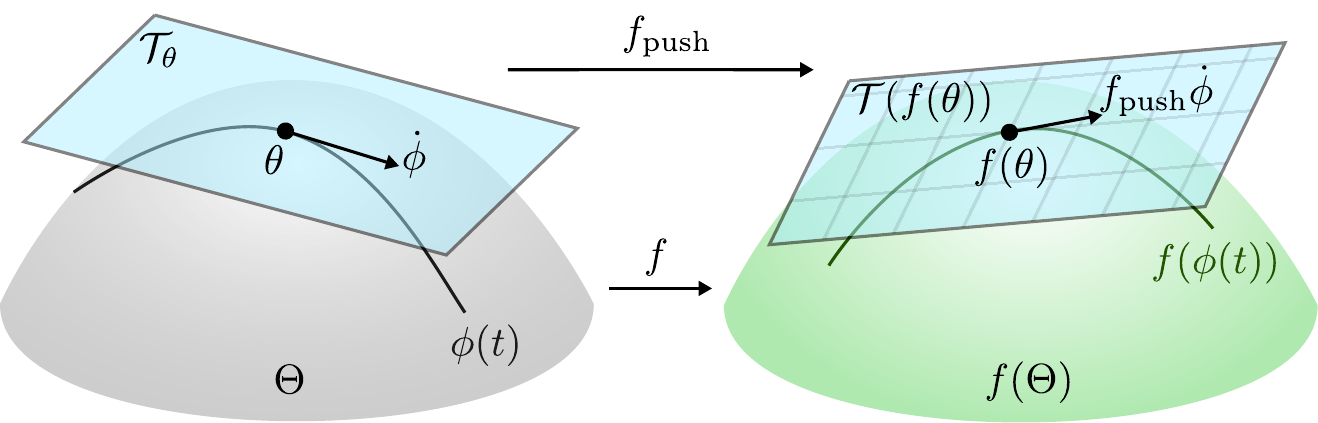}{Picture the parameter space $\Theta$ as a
  manifold.  A curve $\phi(t)$ in $\Theta$ passes through the true
  parameter $\theta$, and its tangent vector there is denoted as
  $\dot\phi$. The tangent space $\mc T_\theta$ is the set of tangent
  vectors of all curves.  The parametrization function $f$ maps the
  parameter to a key property of the statistical model, so that each
  curve $\phi(t)$ is mapped to a curve $f(\phi(t))$ in the statistical
  manifold $f(\Theta)$ on the right.  The pushforward $f_\push$ is the
  map of each tangent vector $\dot\phi$ to a score $f_\push\dot\phi$
  in the score tangent space $\mc T\bk{f(\theta)}$; the score can be
  pictured as a tangent vector of $f(\phi(t))$. The grid on
  $\mc T\bk{f(\theta)}$ represents its inner product, which determines
  the length of a tangent vector as well as the angle between any pair
  of vectors.}

In what follows, we will be less precise about the distinction between
$\mc V$ and $\mc V/\mc N(f)$, often abbreviating an equivalence class
$A + \mc N(f) \in \mc V/\mc N(f)$ as $A$. 

\subsection{\label{sec_Bnd}Semiparametric efficiency bound}
The Hilbert-space nature of the score tangent space
$\mc T\bk{f(\theta)}$ offers many convenient tools for analysis
\cite{reed_simon,debnath}, so we would now like to express the
efficiency bounds in terms of the space. We first impose the condition
\begin{align}
\dd\beta(\dot\phi) = 0  \quad \forall \dot\phi: f_\push\dot\phi = 0
\label{diff_cond}
\end{align}
on the differential $\dd\beta$ of the parameter of interest
$\beta(\theta)$ defined in
Sec.~\ref{sec_tangent}. Eq.~(\ref{diff_cond}) holds if and only if
$\dd\beta(\dot\phi) = \dd\beta(\dot\chi)$ for any
$\dot\phi,\dot\chi \in \mc T_\theta$ that satisfy
$f_\push\dot\phi = f_\push\dot\chi$. We can then define a linear
functional $\sdiff:\range f_\push \to \mb R$ by
\begin{align}
\sdiff(f_\push\dot\phi) &= \dd\beta(\dot\phi) \quad \forall \dot\phi \in \mc T_\theta.
\label{sdiff}
\end{align}
Eq.~(\ref{diff_cond}) ensures that
$\sdiff(f_\push\dot\phi) = \sdiff(f_\push\dot\chi)$ if
$f_\push\dot\phi = f_\push\dot\chi$, so that $\sdiff$ is well defined.
We call $\sdiff$ the score differential of $\beta$. Since we omit
submodels with both $\dd\beta(\dot\phi) = 0$ and
$\info_\theta(\dot\phi,\dot\phi) = 0$, Eq.~(\ref{diff_cond}) implies
that only submodels with
$\info_\theta(\dot\phi,\dot\phi) =
\norm{f_\push\dot\phi}_{f(\theta)}^2 > 0$ and thus
$f_\push\dot\phi \neq 0$ will be considered. We take the submodel
bounds $\bnd_\theta(\dot\phi)$ given by Eq.~(\ref{Csub}) for all the
considered submodels and define the supremum as the semiparametric
efficiency bound for the full model, denoted as $\Bnd(\theta)$.  If
Eq.~(\ref{diff_cond}) does not hold, set $\Bnd(\theta) = \infty$,
since $\bnd_\theta(\dot\phi) = \infty$ for the submodel(s) that
violate Eq.~(\ref{diff_cond}).  Ignoring this extreme case, we find
\begin{align}
\Bnd(\theta) &= \sup_{\phi \in \Phi(\theta):f_\push\dot\phi \neq 0} \bnd_\theta(\dot\phi)
= \sup_{u \in \range f_\push:u \neq 0} 
\frac{[\sdiff(u)]^2}{\norm{u}_{f(\theta)}^2} = \norm{\sdiff}^2,
\label{C}
\end{align}
where $\norm{\cdot}$ is a fundamental property of a linear map called
the operator norm.

If $\norm{\sdiff} < \infty$, the linear functional is called bounded
and it can be extended uniquely to a bounded functional on the closure
of the domain $\overline{\range f_\push} = \mc T(f(\theta))$
\cite{debnath}.  Denote the set of all bounded linear functionals of
$\mc T\bk{f(\theta)}$ as $\mc T^\dual\bk{f(\theta)}$. The pushforward
$f_\push$ induces a pullback
$f_\pull:\mc T^\dual\bk{f(\theta)} \to \mc T_\theta^\dual$ defined by
\begin{align}
\bk{f_\pull L}(\dot\phi) &= L(f_\push\dot\phi) 
\quad \forall \dot\phi \in \mc T_\theta,
L \in \mc T^\dual\bk{f(\theta)}.
\label{pullback}
\end{align}
We can then summarize the requirements of Eq.~(\ref{sdiff}) and
$\norm{\sdiff} < \infty$ by making the pullback of the score
differential $\sdiff$ coincide with the original differential
$\dd\beta$:
\begin{align}
f_\pull\sdiff &= \dd\beta.
\label{pullback_sdiff}
\end{align}
Fig.~\ref{covector} illustrates the concepts of tangent covectors,
differentials, and pullback geometrically.  If
Eq.~(\ref{pullback_sdiff}) has no solution for $\sdiff$ given
$\dd\beta$, $\Bnd(\theta) = \infty$. If
$\sdiff \in \mc T^\dual\bk{f(\theta)}$ exists, the parameter of
interest $\beta(\theta)$ is called pathwise differentiable in
classical statistics \cite{bickel,vaart}.

\fig{0.7}{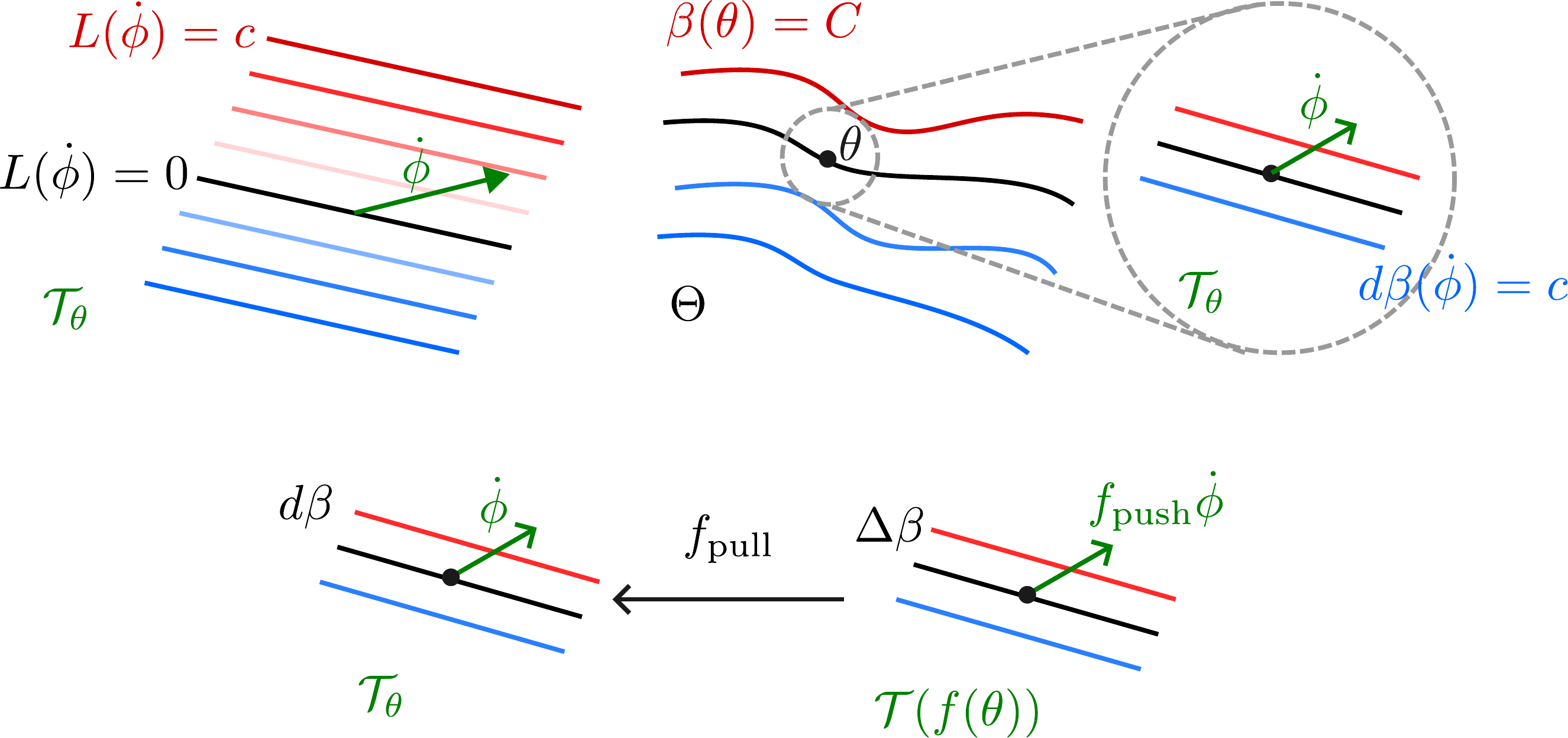}{Top left: a tangent covector
  $L \in \mc T_\theta^\dual$ can be pictured as a meter that measures
  a tangent vector $\dot\phi \in \mc T_\theta$ via the flat and
  parallel level sets $L(\dot\phi) = c$, $c \in \mb R$. Top right: the
  differential $\dd\beta$ can be pictured as level sets on
  $\mc T_\theta$ that are flat approximations of the global level sets
  of $\beta(\theta)$ on $\Theta$.  Bottom: a covector $\sdiff$ on
  $\mc T\bk{f(\theta)}$ can be pulled back to become a covector
  $(f_\pull \sdiff)(\dot\phi) = \sdiff(f_\push\dot\phi)$ on
  $\mc T_\theta$. Eq.~(\ref{pullback_sdiff}) requires the resulting
  covector to match the original differential $\dd\beta$.}

By the Riesz representation theorem,
$\norm{\sdiff} < \infty$ if and only if there exists a
score $\effinf \in \mc T\bk{f(\theta)}$ such that
\begin{align}
  \dd\beta(\dot\phi) &= \sdiff\bk{f_\push\dot\phi} 
= \Avg{\effinf,f_\push\dot\phi}_{f(\theta)}
\quad
\forall \dot\phi \in \mc T_\theta.
\label{riesz}
\end{align}
The symbol $\effinf$ is motivated by the resemblance of
Eq.~(\ref{riesz}) to the definition of the gradient in Euclidean
calculus. If $\effinf$ exists, it is the unique element in
$\mc T\bk{f(\theta)}$ that satisfies Eq.~(\ref{riesz}) and its norm
coincides with the operator norm of $\sdiff$, that is,
\begin{align}
\norm{\effinf}_{f(\theta)} = \norm{\sdiff}.
\end{align}
We have therefore just proved the following theorem.
\begin{theorem}
\label{thm_effinf}
$\Bnd(\theta) = \norm{\sdiff}^2 < \infty$ if and only if
the score differential $\sdiff$, satisfying Eq.~(\ref{sdiff}), has the
Riesz representation $\effinf \in \mc T\bk{f(\theta)}$ that satisfies
Eq.~(\ref{riesz}).  Then the efficiency bound becomes
\begin{align}
\Bnd(\theta) &= \norm{\effinf}_{f(\theta)}^2.
\end{align}
\end{theorem}
\begin{proof}
  Use the Riesz representation theorem.
\end{proof}
Fig.~\ref{covector_norm} illustrates the concepts of the operator norm
and the Riesz representation geometrically. We call $\effinf$ the
efficient influence, following classical statistics where $\effinf$ is
called the efficient influence function \cite{bickel,vaart}. A
parametric submodel that gives a score in the same direction as
$\effinf$ is called a least favorable submodel, since it maximizes the
uncertainty about $\beta$, as quantified by $\Bnd(\theta)$, among all
submodels.

\fig{0.5}{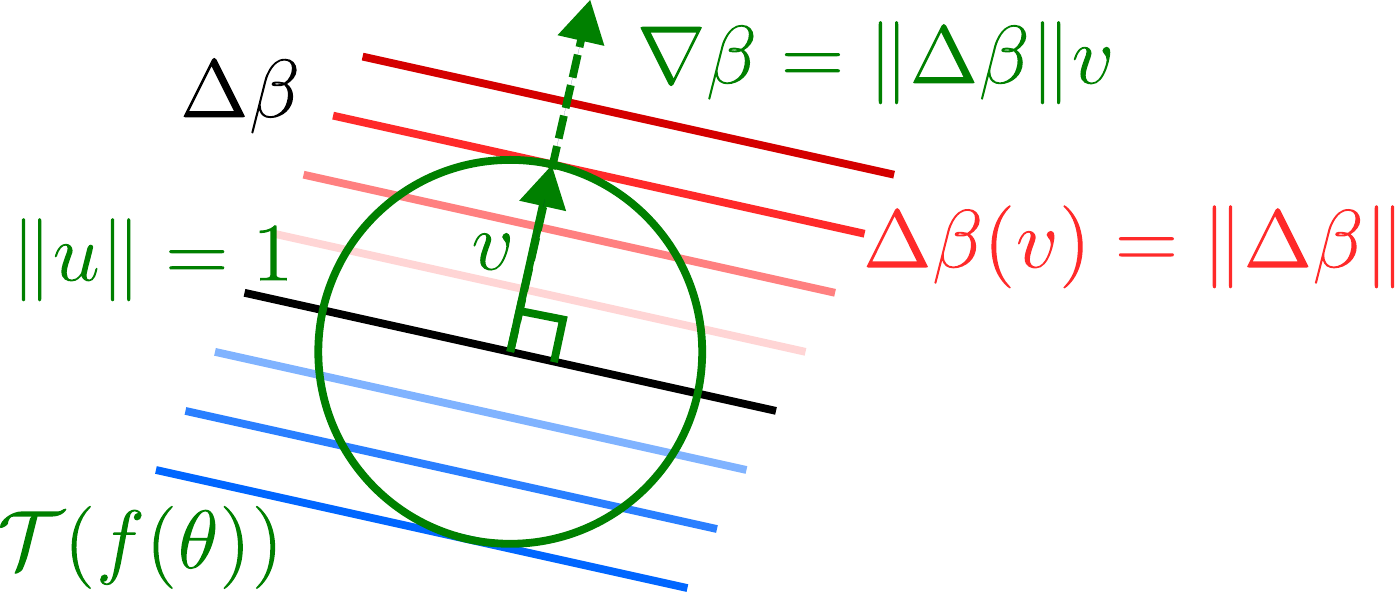}{Picture the covector $\sdiff$ as level sets
  on the score tangent space $\mc T\bk{f(\theta)}$. Consider the unit
  vectors on the unit sphere $\norm{u}_{f(\theta)} = 1$. The highest
  level $\sdiff(v)$ attained by a unit vector $v$ is the operator norm
  $\norm{\sdiff}$. The Riesz representation $\effinf$ of
  $\sdiff$ is this $v$ times $\norm{\sdiff}$. $\effinf$ is
  orthogonal to the $\sdiff(u) = 0$ level set because
  $\avg{\effinf,u}_{f(\theta)} = \sdiff(u) = 0$.}

The supremum form of the efficiency bound given by Eq.~(\ref{C}) has
been extensively studied in classical and quantum statistics and
proved to be a lower error bound for large classes of processing
methods
\cite{bickel,semi_prx,vaart,gross20,qlmoment_pra2}---Appendix~\ref{app_error}
offers a quick review of the classical case. The most well known
consequence of Eq.~(\ref{C}) may be the bound for a finite-dimensional
parameter space. Following Example~\ref{exa_finite}, the supremum in
Eq.~(\ref{C}) leads to
\begin{align}
\Bnd(\theta) &= \sup_{u \in \mb R^p:u^\top\partial\beta \neq 0} 
\frac{(u^\top\partial\beta)^2}{u^\top \Imat u}
= \bk{\partial\beta}^\top \Imat^{-1} \partial\beta
\label{C_matrix}
\end{align}
in terms of the inverse of the information matrix $\Imat$.
The case where $\Imat$ may not be invertible will be treated in
Sec.~\ref{sec_finite}.

When there are $N$ independent classical or quantum objects that
depend on the same $\theta$, the information is well known to be
additive, in the sense that the information for the $N$ objects is
\begin{align}
\info_\theta &= \sum_j \info_\theta^{(j)},
\label{info_add}
\end{align}
where $\info_\theta^{(j)}$ is the information for the $j$th object. In
particular, if the outputs are independent and identically distributed
(IID) according to a probability density or a density operator,
$\info_\theta^{(j)} = \info_\theta^{(1)}$, leading to
\begin{align}
\info_\theta &= N \info_\theta^{(1)}.
\end{align}
The efficiency bound is then simply 
\begin{align}
\Bnd(\theta) &= \frac{\Bnd^{(1)}(\theta)}{N}
\label{Bnd_N}
\end{align}
in terms of the one-copy bound $\Bnd^{(1)}$. The bound for $N$
independent but not identically distributed objects will be treated in
Sec.~\ref{sec_broadcast}.

In the quantum case (Example~\ref{exa_quantum}), Ref.~\cite{semi_prx}
calls $\Bnd(\theta)$ the generalized Helstrom bound.  When
$\beta:\Theta \to \mb R^q$ is multidimensional ($q > 1$), there is a
slightly tighter quantum bound called the Holevo--Nagaoka bound
\cite{holevo_aspect,nagaoka89,hayashi24}, although it is at most twice
the generalized Helstrom bound; see
Refs.~\cite{semi_prx,yang19,demkowicz20,tsang26} for further details.

The efficient influence $\effinf$, beyond determining the efficiency
bound, is an important quantity in the design of optimal estimators or
measurements. To wit:
\begin{enumerate}
\item In the classical case (Example~\ref{exa_classical}), $\effinf$
  is a zero-mean random variable. Given $N$ observations
  $(X_1,\dots,X_N)$ and an estimator
  $\check\beta^{(N)}(X_1,\dots,X_N)$, such as the maximum-likelihood
  estimator, $\sqrt{N}(\check\beta^{(N)}-\beta)$ may converge to a
  random variable that can be compared with $\effinf$ to determine if
  the estimator can approach the efficiency bound; see
  Appendix~\ref{app_error} for further details. There also exist
  alternative estimators, such as the one-step estimator
  \cite{kennedy24} and the targeted maximum-likelihood estimator
  \cite{laan}, that rely on the efficient influence
  $(\effinf)(\check\theta)$ evaluated at a preliminary estimate
  $\check\theta$ of $\theta$ in their algorithms to achieve asymptotic
  efficiency.

\item In the quantum case (Example~\ref{exa_quantum}), $\effinf$ is a
  zero-mean observable and $\norm{\effinf}_{f(\theta)}^2$ is its
  quantum variance. To achieve the Helstrom bound asymptotically for
  $N$ objects, an adaptive method measures the observable
  $(\effinf)(\check\theta_n)$ of the $n$th object sequentially, where
  $\check\theta_n$ is an estimate of $\theta$ based on the previous
  data \cite{fujiwara06}. The method to approach the Holevo--Nagaoka
  bound is similar \cite{yang19,demkowicz20,tsang26}.

\end{enumerate}
Thus, both $\effinf$ and $\Bnd$ are important quantities for an
estimation problem. The next section further demonstrates their
significance through examples.

\section{\label{sec_exa}Nonparametric models}
As shown in both classical and quantum statistics
\cite{bickel,semi_prx}, certain problems possess trivial score tangent
spaces $\mc T\bk{f(\theta)}$ that substantially simplify the
calculation of the efficiency bounds. This section presents some basic
examples where we make minimal assumptions about the
parametrizations; we call all of them nonparametric models. We will
make the abbreviation
\begin{align}
f(\theta) = \theta,
\end{align}
so that $\mc T(\theta) = \mc T\bk{f(\theta)}$ denotes the score
tangent space (not to be confused with $\mc T_\theta$).

\subsection{\label{sec_classical}
Nonparametric classical model, following Example~\ref{exa_classical}}
Let $f(\theta,\cdot) = \theta(\cdot)$ be a probability density for a
random element $X$. Since any probability density is normalized as
$\expect(1) = \int \theta(\el) \dd\mu(\el) = 1$, any score function
$\score(\dot\phi,\el)$ must have zero mean, that is,
\begin{align}
\expect\Bk{\score(\dot\phi)} &= \Avg{\score(\dot\phi),1}_\theta = 0,
\end{align}
where the semi-inner product is given by Eq.~(\ref{inner_classical}).
Define the ambient vector space $\mc V$ as the space of all zero-mean
functions with finite seminorm $\norm{\cdot}_\theta$:
\begin{align}
\mc V &= \BK{A: \norm{A}_\theta < \infty \textrm{ and }\Avg{A,1}_\theta = 0},
\end{align}
which correspond to zero-mean finite-variance random variables. Each
element of $\mc V/\mc N(\theta)$ is an equivalence class of random
variables that are equal almost everywhere $P_\theta$
\cite{parth_qm}. We then set
$\Amb(\theta) = \complete[\mc V/\mc N(\theta)]$.

Let the parameter space be
\begin{align}
\Theta = \BK{\textrm{any probability density with respect to $\mu$}}.
\end{align}
It is well known that the score tangent space $\mc T(\theta)$ is as
large as it can be, namely \cite{bickel},
\begin{align}
\mc T(\theta) &= \Amb(\theta).
\label{maximal}
\end{align}
We call such a tangent space maximal (it is called full-dimensional in
Ref.~\cite{semi_prx}). The key step of the proof takes an arbitrary
zero-mean function $A \in \mc V$ and constructs a parametric submodel
of $\Theta$ that passes through the true $\theta$ at $t = 0$, say,
\begin{align}
\phi_t(\el) &= \frac{\{1 + \tanh[A(\el) t]\} \theta(\el)}
{\int (\textrm{numerator}) \dd\mu}.
\end{align}
Functions other than $1 + \tanh$ may be used, but $1+\tanh$ makes it
obvious that the denominator is finite. The score function is
\begin{align}
\score(\dot\phi,\el)
&= \dot\phi\Bk{\ln \theta(\el)} = \left.\pdv{}{t} \ln \phi_t(\el)\right|_{t=0}
= A(\el),
\end{align}
implying that
\begin{align}
A + \mc N(\theta) = \score(\dot\phi) + \mc N(\theta) = f_\push\dot\phi \in \mc T(\theta)
\end{align}
and thus $\mc V/\mc N(\theta) \subseteq \mc T(\theta)$.  Since
$\mc T(\theta)$ is closed,
$\overline{\mc V/\mc N(\theta)} = \Amb(\theta) \subseteq \mc
T(\theta)$. Combining this relation with Eq.~(\ref{stat_space}), we
obtain $\mc T(\theta) = \Amb(\theta)$.

We note in passing that the maximality of $\mc T(\theta)$ depends not
only on the parameter space $\Theta$ but also on our choice of the
ambient vector space $\mc V$. If one chooses $\mc V$ to be the set of
score functions $\{\score(\dot\phi): \dot\phi \in \mc T_\theta\}$,
then $\mc T(\theta)$ is always maximal. In the other extreme, if
$\mc V$ is not restricted to have zero mean, then $\mc T(\theta)$ can
never be maximal for Example~\ref{exa_classical}.

Let the parameter of interest be the expected value of a random
variable $b(X)$, namely,
\begin{align}
\beta(\theta) &= \expect(b) = \int b(\el) \theta(\el) \dd\mu(\el).
\label{linear_beta}
\end{align}
Then, for any submodel $\phi$,
\begin{align}
\dd\beta(\dot\phi) &= 
\dot\phi(\beta) = \left.\pdv{}{t} \int b(\el) \phi_t(\el) \dd\mu(\el)\right|_{t=0}
= \int b(\el) \score(\dot\phi,\el) \theta(\el) \dd\mu(\el)
= \Avg{b,f_\push\dot\phi}_\theta.
\end{align}
It follows that, given $\mc T(\theta) = \Amb(\theta)$,
\begin{align}
\effinf = b - \Avg{b,1}_\theta \in \Amb(\theta) =  \mc T(\theta)
\end{align}
satisfies the definition of the efficient influence given by
Eq.~(\ref{riesz}). The efficiency bound is hence
\begin{align}
\Bnd(\theta) &= \norm{b-\Avg{b,1}_\theta}_\theta^2
= \expect\BK{\Bk{b-\beta(\theta)}^2}.
\end{align}
Notice that this bound is equal to the variance of $b$, so the estimator
\begin{align}
\check\beta(X) = b(X)
\end{align}
in terms of the random element $X$ is unbiased
($\expect(\check\beta) = \beta(\theta)$) and has a variance achieving
the bound. Such an estimator is called efficient. In general, when an
estimator $\check\beta(X)$ is efficient for one observation, its sample
mean
\begin{align}
\check\beta^{(N)}\bk{X_1,\dots,X_N} &= \frac{1}{N}\sum_{j=1}^N \check\beta(X_j)
\end{align}
for $N$ IID copies $(X_1,\dots,X_N)$ of $X$ is also efficient, since
it is unbiased and has a variance equal to the one-copy variance
divided by $N$, which achieves the $N$-copy bound given by
Eq.~(\ref{Bnd_N}).

\subsection{\label{sec_poisson}Nonparametric classical Poisson model,
  following Example~\ref{exa_poisson}}
Let $f(\theta,\cdot) = \theta(\cdot)$ be the mean density of a Poisson
field $X$.  It need not be normalized, so we take $\mc V$ to be the
space of all functions with finite seminorm
$\norm{\cdot}_{\theta} < \infty$, that is,
\begin{align}
\mc V &= \BK{A: \norm{A}_\theta < \infty},
\end{align}
with the semi-inner product given by Eq.~(\ref{inner_classical}).
Define $\Amb(\theta) = \complete[\mc V/\mc N(\theta)]$. If we take the
parameter space to be
\begin{align}
\Theta = \BK{\textrm{any density with respect to $\mu$}},
\end{align}
then it is straightforward to prove that the score tangent space
is maximal---take any function $A \in \mc V$ and construct the
submodel
\begin{align}
\phi_t(\el) &= \BK{1 + \tanh[A(\el) t]} \theta(\el),
\end{align}
which leads to $\Amb(\theta) \subseteq \mc T(\theta)$ and
$\mc T(\theta) = \Amb(\theta)$ by the same argument as that in
Sec.~\ref{sec_classical}.

Let the parameter of interest be a linear functional $b$ of the mean
density $\theta$; it is mathematically given by
Eq.~(\ref{linear_beta}). By the same argument as that in
Sec.~\ref{sec_classical}, $b$ satisfies the definition of the
efficient influence given by Eq.~(\ref{riesz}).  Hence,
\begin{align}
\effinf &= b,
&
\Bnd(\theta) &= \norm{b}_\theta^2 = \int b(\el)^2 \dd\mean X_\theta(\el).
\end{align}
It follows that the linear estimator
\begin{align}
\check\beta(X) &= \int b(\el) \dd X(\el)
\end{align}
is unbiased and efficient.

\subsection{\label{sec_quantum}Nonparametric quantum model, 
following Example~\ref{exa_quantum}}
Let $\theta$ be a density operator on a complex Hilbert space $\mc
H$. The quantum case is similar to the classical case in
Sec.~\ref{sec_classical}.  Assume the semi-inner product given by
Eq.~(\ref{inner_quantum}). Since all density operators have unit trace,
the score operators must have zero mean, namely,
\begin{align}
\trace\Bk{\theta \score(\dot\phi)} = \Avg{\score(\dot\phi),I}_\theta = 0.
\end{align}
Let the ambient vector space $\mc V$ be the set of all self-adjoint
operators that are bounded when applied to $\mc H$ and have zero
mean. For an infinite-dimensional $\mc H$, $\mc V/\mc N(\theta)$ turns
out to be incomplete; its completion includes operators that are
symmetric and unbounded when applied to $\mc H$
\cite{holevo_aspect}. The resulting ambient Hilbert space
$\Amb(\theta) = \complete[\mc V/\mc N(\theta)]$ comprises zero-mean
observables with finite quantum variance $\norm{A}_\theta^2 < \infty$.

Assume the nonparametric model
\begin{align}
\Theta &= \BK{\textrm{any density operator on $\mc H$}}.
\end{align}
Then $\mc T(\theta)$ can be proved to be maximal by assuming the
submodel
\begin{align}
\phi(t) &= \frac{\Bk{I + \tanh(A t/2)} \theta \Bk{I + \tanh(A t/2)} }
{\trace(\textrm{numerator})}
\end{align}
for any $A \in \mc V$ and following an argument similar to the
classical case in Sec.~\ref{sec_classical} \cite{semi_prx}.

Let the parameter of interest be the expected value of an observable
$b$, namely,
\begin{align}
\beta(\theta) &= \trace\bk{b\theta}.
\label{linear_beta2}
\end{align}
Then, similar to the classical case, we find \cite{semi_prx}
\begin{align}
\effinf &= b - \Avg{b,I}_\theta I,
&
\Bnd(\theta) &= \norm{b-\Avg{b,I}_\theta I}_\theta^2 = \trace\bk{ b^2\theta}
-\Bk{\trace(b\theta)}^2,
\end{align}
which is the quantum variance of $b$. If one measures $b$ in the sense
of von Neumann, the expected value of the outcome coincides with
$\trace(b\theta)$ and the variance coincides with the quantum
variance. We call an observable with an unbiased expected value and a
quantum-optimal variance an (operator-valued) efficient
estimator---$b$ is an efficient estimator for this example.

\subsection{\label{sec_qpoisson}Nonparametric quantum Poisson model,
following Example~\ref{exa_qpoisson}}
Let $\theta$ be the intensity operator of a Poisson state. Since
$\theta$ need not be normalized, take $\mc V$ to be the space of
all bounded self-adjoint operators and define
$\Amb(\theta) = \complete[\mc V/\mc N(\theta)]$, assuming the
semi-inner product given by Eq.~(\ref{inner_quantum}). With the
nonparametric model
\begin{align}
\Theta &= \BK{\textrm{any intensity operator on $\mc H$}},
\end{align}
we can again prove that $\mc T(\theta)$ is maximal by assuming the
submodel
\begin{align}
\phi(t) &= \Bk{I + \tanh(A t/2)} \theta \Bk{I + \tanh(A t/2)}
\end{align}
for any $A \in \mc V$ and using a similar argument as before. For a
parameter of interest that is a linear functional of $\theta$ given by
Eq.~(\ref{linear_beta2}), we find
\begin{align}
\effinf &= b,
&
\Bnd(\theta) &= \norm{b}_\theta^2 = \trace\bk{b^2\theta}.
\end{align}
To achieve this bound, find a decomposition of $b$ given by
\begin{align}
b &= \int b'(\el) \dd E(\el)
\label{pvm}
\end{align}
in terms of a projection-valued measure $E$ on $\mc H$ and measure
each of the $L$ bodies in the Poisson object according to $E$. The
outcomes form a Poisson field $X$ with mean measure
$\mean X_\theta(\cdot) = \trace[E(\cdot)\theta]$.  The estimator
$\check\beta(X) = \int b'(\el)\dd X(\el)$ is then unbiased and
efficient. In optics, such a measurement can be implemented by linear
optics and photon counting \cite{review_cp,spade_prr}.

\subsection{\label{sec_gauss}Nonparametric Gaussian shift models,
  following Examples~\ref{exa_gauss} and \ref{exa_qgauss}}
Let $\theta$ be the mean vector of a Gaussian field $X$ or a quantum
Gaussian state. The score vector spaces are denoted as $\mc V$ in both
examples, and we take them to be the ambient vector spaces.  If the
mean vector is assumed to be arbitrary, we can prove that
$\mc T(\theta)$ is maximal by taking any $A \in \mc V$ and
constructing the submodel
\begin{align}
\phi(t) &= \theta + \Sigma A t.
\end{align}
The score given by Eq.~(\ref{score_gauss}) becomes
\begin{align}
\Sigma \score(\dot\phi) &= 
\dot\phi(\theta) 
= \left.\pdv{\phi(t)}{t}\right|_{t=0} = \Sigma A,
&
\score(\dot\phi) &= A,
\end{align}
leading to $\mc T(\theta) = \Amb(\theta)$ by a similar argument as
before.

For the classical Example~\ref{exa_gauss}, let the parameter of
interest be a linear functional of the mean $\theta$, given by
\begin{align}
\beta(\theta) &= b(\theta),
\quad
b \in \covarspace.
\label{beta_gauss}
\end{align}
Then 
\begin{align}
\dd\beta(\dot\phi) = \dot\phi(\beta) = b\bk{\dot\phi(\theta)}
= b\Bk{\Sigma\score(\dot\phi)}
= \Avg{b,\score(\dot\phi)}_{\mc V}
= \Avg{b,f_\push\dot\phi}_{\theta},
\label{diff_gauss}
\end{align}
where we have used Eqs.~(\ref{score_gauss}) and (\ref{semi_gauss}).
It follows that $b \in \Amb(\theta) = \mc T(\theta)$ satisfies the
definition of the efficient influence given by Eq.~(\ref{riesz}),
leading to
\begin{align}
\effinf &= b,
&
\Bnd(\theta) &= \norm{\effinf}_{\mc V}^2 = \bk{\Sigma b}(b).
\label{effinf_gauss}
\end{align}
The linear estimator $\check\beta(X) = b(X)$ in terms of the Gaussian
field $X$ is then unbiased and efficient.

Similarly, for the quantum Example~\ref{exa_qgauss}, assume
\begin{align}
\beta(\theta) &= \Avg{b,\theta}_{\mc K},
\end{align}
which is the mean quadrature of a mode specified by the phase-space
vector $b \in \mc K$. Then
\begin{align}
\dd\beta(\dot\phi) = \dot\phi(\beta) = \Avg{b,\dot\phi(\theta)}_{\mc K}
= \Avg{b,\score(\dot\phi)}_{\mc V}
= \Avg{b,f_\push\dot\phi}_\theta,
\end{align}
and the efficient influence is also $b$. The bound becomes
\begin{align}
\Bnd(\theta) &= \norm{\effinf}_{\mc V}^2 = \Avg{b,\Sigma b}_{\mc K}.
\end{align}
The quadrature $X(b)$ is hence an unbiased and efficient estimator.

In the finite-dimensional case with $\mc B = \mb R^p$, assuming
\begin{align}
\beta(\theta) = b^\top \theta
\end{align}
for some column vector $b \in \mb R^p$, we obtain
\begin{align}
\effinf &= b,
&
\Bnd(\theta) &= \bk{\effinf}^\top \Sigma \effinf
= b^\top \Sigma b,
\end{align}
which is the variance of the random variable $b^\top X$ in the
classical case or the variance of the quadrature $b^\top X$ in the
quantum case. This result can also be obtained from the
finite-dimensional bound in Eq.~(\ref{C_matrix}), since
$\Imat = \Sigma^{-1}$ for Gaussian shift models \cite{kay,monras13},
but the infinite-dimensional case is far less trivial to prove.

\section{\label{sec_chan}Channels}

\subsection{\label{sec_chan_exa}Examples}
A channel is a system that takes a classical or quantum object as the
input and produces another object as the output. We first motivate
this section by presenting some examples of channels based on
Examples~\ref{exa_classical}--\ref{exa_qgauss}. In each of the
following examples, the output parametrization function $g(\theta)$ is
always related to the input parametrization function $f(\theta)$ by
the relation
\begin{align}
g(\theta) &= \bk{\chan \circ f}(\theta) = \chan\bk{f(\theta)},
\label{markov_op}
\end{align}
where $\circ$ denotes the composition and $\chan$ is determined by the
channel; we call $\chan$ the channel map.

\begin{example}[name=Classical channel,label=exa_cc]
  Let the input be the random element $X$ with probability density
  $f(\theta,\cdot)$ in Example~\ref{exa_classical}. For each outcome
  $\el \in \sspace$ of $X$, generate an output random element
  $Y \in \sspace_\out$ based on the probability density
  $\chan(\elout|\el)$ with respect to a common measure $\nu$.  Let
  $Q_\theta$ be the marginal probability measure of $Y$ and
  $g(\theta,\cdot) = (\dv*{Q_\theta}{\nu})(\cdot)$ be the probability
  density. Then the channel can be modeled as
\begin{align}
g(\theta,\elout) &= \int \chan(\elout|\el) f(\theta,\el) \dd\mu(\el),
\label{markov}
\end{align}
which is an example of Eq.~(\ref{markov_op}) if $\chan$ there is
regarded as a Markov map. The output score function becomes
\begin{align}
\score_\out(\dot\phi,\elout) = \dot\phi\Bk{\ln g(\theta,\elout)}
= \frac{1}{g(\theta,\elout)} \int \chan(\elout|\el) f(\theta,\el) 
\score(\dot\phi,\el) \dd\mu(\el),
\label{score_out}
\end{align}
which coincides with the conditional expectation of the input score
function $\score(\dot\phi,\cdot)$ by Bayes' rule.
\end{example}

\begin{example}[name=Classical Poisson channel,label=exa_pp]
  Let the input be the Poisson field $X$ with mean density
  $f(\theta,\cdot)$ in Example~\ref{exa_poisson}. For each point
  $\el \in \sspace$ of the point measure $X$, generate an output point
  $\elout \in \sspace_\out$ based on the probability density
  $\chan(\elout|\el)$.  Then the resulting point measure $Y$ is also a
  Poisson field \cite{kingman}. Let the mean measure of $Y$ be
  $\mean Y_\theta$; the mean density
  $g(\theta,\cdot) = (\dv*{\mean Y_\theta}{\nu})(\cdot)$ obeys
  Eq.~(\ref{markov}). The output score function has the same formula
  as Eq.~(\ref{score_out}).
\end{example}

\begin{example}[name=Classical Gaussian channel,label=exa_gg]
  Let the input be the Gaussian field $X \in \mc B$ with mean vector
  $f(\theta) \in \Sigma \mc V$ in Example~\ref{exa_gauss}.  Let
\begin{align}
Y = \chan X + Z
\label{io}
\end{align}
be the output of a linear map $\chan:\mc B \to \mc B_\out$ on $X$
followed by the addition of an independent zero-mean Gaussian noise
$Z \in \mc B_\out$.  Then $Y$ is also a Gaussian field with a
covariance map
$\Sigma_\out:\mc B_\out^\dual \to \mc B_\out^{\dual\dual}$. Let
$\mc V_\out$ be the finite-variance subset of $\mc B_\out^\dual$ with
the semi-inner product
\begin{align}
\Avg{\cdot,\cdot}_{\mc V_\out} = (\Sigma_\out \cdot)(\cdot).
\end{align}
Assume that $\chan$ can be restricted to be a bounded map
$\chan:\Sigma \mc V \to \Sigma_\out\mc V_\out$, so that the output
mean vector in $\Sigma_\out\mc V_\out$ becomes
\begin{align}
g(\theta) &= \chan f(\theta).
\label{gmean_out}
\end{align}
Following Eq.~(\ref{score_gauss}), the output score vector in
$\mc V_\out$ becomes
\begin{align}
\Sigma_\out \score_\out(\dot\phi) &= 
\dot\phi\Bk{g(\theta)} = \chan\dot\phi\Bk{f(\theta)}
= \chan \Sigma \score(\dot\phi).
\label{pushforward_gauss}
\end{align}
In the finite-dimensional case with $\mc B= \mb R^p$ and
$\mc B_\out = \mb R^q$, $\chan$ is a $q\times p$ matrix, while the
output covariance matrix becomes
\begin{align}
\Sigma_{\out}&= 
\chan \Sigma \chan^\top + \Sigma_\noise,
\label{gcovar_out}
\end{align}
where $\Sigma_\noise$ is the covariance matrix of the noise field $Z$.
Assume $\Sigma_\out > 0$ for simplicity.  The inner product of
$\mc V_\out$ becomes
\begin{align}
\Avg{A,B}_{\mc V_\out} &= A^\top \Sigma_\out B.
\label{inner_gauss_finite2}
\end{align}
\end{example}

\begin{example}[name=Quantum channel,label=exa_qq]
  Let the input be a quantum object modeled by the density operator
  $f(\theta)$ in Example~\ref{exa_quantum}.  If the output is another
  quantum object with density operator $g(\theta)$ on complex Hilbert
  space $\mc H_\out$, a quantum channel is commonly modeled as
\begin{align}
g(\theta) &= \chan f(\theta)
\label{tpcp}
\end{align}
in terms of a trace-preserving completely positive map $\chan$
\cite{bratteli} that we call a Markov map. Following Eq.~(\ref{lyap}),
the output score operator $\score_\out(\dot\phi)$ satisfies
\begin{align}
g(\theta) \jordan \score_\out(\dot\phi)
&= \chan \Bk{f(\theta) \jordan \score(\dot\phi)}
\label{gce}
\end{align}
in terms of each input score operator $\score(\dot\phi)$ with respect
to $f(\theta)$.  A solution $\score_\out(\dot\phi)$ to this Lyapunov
equation is called a generalized conditional expectation of
$\score(\dot\phi)$ \cite{hayashi,gce_pra,gce2}.
\end{example}

\begin{example}[name=Quantum Poisson channel,label=exa_qpp]
  Let the input be a many-body object in a Poisson state modeled by
  the intensity operator $f(\theta)$ in Example~\ref{exa_qpoisson}.
  If each body in the object undergoes a quantum channel described by
  a Markov map $\chan$, then the output is also a Poisson state with
  the intensity operator $g(\theta)$ given by Eq.~(\ref{tpcp})
  \cite{poisson_quantum}.  The output score operator is determined by
  Eq.~(\ref{gce}).
\end{example}

\begin{example}[name=Quantum Gaussian channel,label=exa_qgg]
  Let the input be a bosonic field in a Gaussian state with mean
  phase-space vector $f \in \mc K$ and covariance map
  $\Sigma:\mc K \to \mc K$ in Example~\ref{exa_qgauss}. The output of
  a Gaussian channel, as described in Appendix~\ref{app_fock}, is
  another bosonic field in a Gaussian state $\rho_\out$ with mean
  vector $g \in \mc K_\out$ and covariance map
  $\Sigma_\out:\mc K_\out \to \mc K_\out$, satisfying
\begin{align}
g &= \chan f,
\label{chan_gauss}
\\
\Sigma_\out &= \chan \Sigma \chan^\dual + \Sigma_\noise
\label{chan_gauss_cov}
\end{align}
for some linear map $\chan:\mc K\to\mc K_\out$, its adjoint
$\chan^\dual:\mc K_\out\to \mc K$, and noise covariance map
$\Sigma_\noise :\mc K_\out \to \mc K_\out$.
Eq.~(\ref{chan_gauss_cov}) becomes Eq.~(\ref{gcovar_out}) in the
finite-dimensional case \cite[Eq.~(12.138)]{holevo_info}. Assuming
that $\Sigma_\out$ is invertible, the inner product of the Hilbert
space $\mc V_\out$ for the output scores becomes
\begin{align}
\Avg{A,B}_{\mc V_\out} &= \Avg{A,\Sigma_\out B}_{\mc K_\out}.
\end{align}
Given a set of input mean vectors
$\{f(\theta):\theta \in \Theta\} \subseteq \mc K$, the output score
vector satisfies Eq.~(\ref{pushforward_gauss}).
\end{example}

\subsection{\label{sec_chan_push}Channel pushforward}
Countless more examples can be constructed, as one may further mix and
match the examples in Sec.~\ref{sec_bnd}. In all the cases,
Eq.~(\ref{markov_op}) models the relations between the input and
output parametrization functions via a channel map $\chan$. Let
$\mc V_\out$ be a judiciously chosen ambient vector space for the
output scores and
\begin{align}
\Amb\bk{g(\theta)} = \complete\Bk{\mc V_\out/\mc N(g)}
\end{align}
be the ambient Hilbert space at the output. For any channel map
$\chan$, even if it is nonlinear, the pushforward
$g_\push:\mc T_\theta \to \Amb\bk{g(\theta)}$ follows the covariant
chain rule
\begin{align}
g_\push &= \chan_\push f_\push
\label{cov}
\end{align}
in terms of a linear pushforward
\begin{align}
\chan_\push: \range f_\push \to \Amb\bk{g(\theta)},
\end{align}
as long as the chain rule holds for $\dot\phi$ in the computation of
the output scores. The output score tangent space becomes
\begin{align}
\mc T\bk{g(\theta)} &= \overline{\range g_\push}
= \overline{\chan_\push \range f_\push}.
\label{Tout}
\end{align}
We call $\chan_\push$ the channel pushforward, which pushes each input
score in $\range f_\push \subseteq \mc T\bk{f(\theta)}$ to become an
output score in $\mc T\bk{g(\theta)}$.  Fig.~\ref{channel_pushforward}
illustrates the concepts of $\chan_\push$ and $\mc T(g)$
geometrically. $\chan_\push$ can be considered as a linearized version
of the global channel map $\chan$, with statistically motivated input
and output vector spaces.

\fig{0.7}{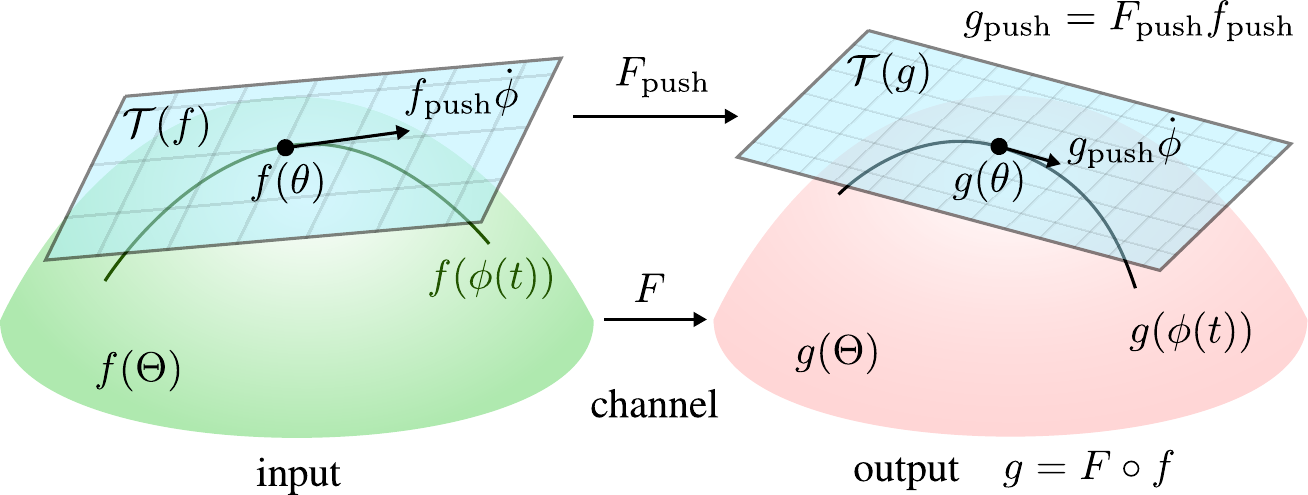}{ The channel map $\chan$ relates the
  input parametrization function $f(\theta)$ to the output
  parametrization function $g(\theta)$ by $g = \chan \circ f$. This
  map induces a linear channel pushforward map $\chan_\push$ that
  relates each input score $f_\push \dot\phi$ in the input score
  tangent space $\mc T(f)$ to the output score $g_\push\dot\phi$ in
  the output score tangent space $\mc T(g)$ by
  $g_\push = \chan_\push f_\push$.}

Because the domain and codomain of $\chan_\push$ are both subsets of
Hilbert spaces, a lot more can be said about $\chan_\push$, such as
its operator norm and its adjoint. We will assume that the operator
norm of $\chan_\push$, defined as
\begin{align}
\norm{\chan_\push}
&= \sup_{u \in \range f_\push: u \neq 0}\frac{\norm{\chan_\push u}_g}{\norm{u}_f},
\end{align}
is finite, so that $\chan_\push$ can be uniquely extended to become a
bounded map on the domain
$\overline{\range f_\push} = \mc T\bk{f(\theta)}$. For example,
channels with no access to $\theta$ can satisfy this assumption
because of information monotonicity, as discussed later in
Sec.~\ref{sec_mono}. The range of $\chan_\push$ becomes
\begin{align}
\range \chan_\push &= \chan_\push \mc T\bk{f(\theta)}.
\end{align}
We then obtain the following lemma regarding the output score tangent
space $\mc T\bk{g(\theta)}$, similar to Eq.~(\ref{stat_space}). For
brevity, we will often omit any mention of the argument $\theta$
hereafter.
\begin{lemma}
\label{lem_h}
\begin{align}
\mc T(g)  &= \overline{\range \chan_\push}.
\end{align}
\end{lemma}
The proof is delegated to Appendix~\ref{app_proofs}.  In what follows,
if a proof is not given immediately after a numbered proposition and
no reference is provided, it will also be delegated to
Appendix~\ref{app_proofs}.

\subsection{\label{sec_chan_pull}Channel pullback}
Recall the definitions of the score differential, the pullback, and
the efficiency bound in Sec.~\ref{sec_Bnd}.  Let
$\sdiffout:\range g_\push \to \mb R$ be the score differential at the
output. For the efficiency bound $\Bnd_\out$ at the output to be
finite, $\sdiffout$ needs to be in $\mc T^\dual(g)$ and satisfy
\begin{align}
g_\pull \sdiffout &= \dd\beta
\end{align}
in terms of the original differential $\dd\beta \in \mc T^\dual_\theta$ and
the pullback $g_\pull:\mc T^\dual(g) \to \mc T_\theta^\dual$ defined by
Eq.~(\ref{pullback}), now in terms of $g_\push$.  By the definition of the
pullbacks and the covariant chain rule given by Eq.~(\ref{cov}) for
the pushforwards, we have
\begin{align}
\bk{g_\pull L} (\dot\phi) &= L(g_\push\dot\phi)
= L(\chan_\push f_\push \dot\phi) 
= \bk{\chan_\pull L}(f_\push\dot\phi) = \bk{f_\pull \chan_\pull L}(\dot\phi),
\end{align}
meaning that the pullbacks must satisfy the contravariant chain rule
\begin{align}
g_\pull &= f_\pull \chan_\pull,
\label{contra}
\end{align}
where $\chan_\pull:\mc T^\dual(g) \to \mc T^\dual(f)$ is defined by
\begin{align}
\bk{\chan_\pull L} (u)  &= L(\chan_\push u) \quad
\forall L \in \mc T^\dual(g), u \in \mc T(f).
\label{h_pullback}
\end{align}
With the input score differential satisfying
Eq.~(\ref{pullback_sdiff}), the input and output score differentials
should then satisfy
\begin{align}
\chan_\pull \sdiffout &= \sdiff,
\label{sdiff_channel}
\end{align}
assuming that $\sdiff \in \mc T^\dual(f)$ exists at the input to begin
with. If $\sdiff$ is not in the range of $\chan_\pull$, then $\sdiffout$
does not exist, and the efficiency bound $\Bnd_\out$ becomes infinite. On
the other hand, if $\sdiff$ is in the range of $\chan_\pull$, then
$\sdiffout \in \mc T^\dual(g)$ can exist.

It will be much more convenient to work with the Riesz representations
of the score differentials. $\sdiff$ is in the range of $\chan_\pull$ and
a $\sdiffout \in \mc T^\dual(g)$ exists if and only if the latter has
the Riesz representation $\effinfout \in \mc T(g)$ that obeys
\begin{align}
\Avg{\effinfout,\chan_\push u}_g 
&=
\sdiffout(\chan_\push u)  
=
\bk{\chan_\pull \sdiffout}(u)  
=  \sdiff(u) = \Avg{\effinf, u}_f
\quad
\forall u \in \mc T(f).
\label{riesz_chan}
\end{align}
$\effinfout$ is the efficient influence at the output and $\effinf$ is
the efficient influence at the input. To solve for $\effinfout$, the
norm of which determines the efficiency bound $\Bnd_\out$, consider
the adjoint $\chan_\push^\adj:\Amb(g) \to \mc T(f)$ of
$\chan_\push:\mc T(f) \to \Amb(g)$, defined by
\begin{align}
\Avg{\chan_\push u,v}_{g} &= \Avg{u,\chan_\push^\adj v}_{f}
\quad
\forall u \in \mc T(f),v \in \Amb(g).
\label{adj}
\end{align}
(We reserve the superscript $\adj$ for the adjoint in terms of the
inner products of $\Amb(f)$ and $\Amb(g)$.)  For brevity, we write
\begin{align}
\chan_\push^\adj &= \chan^\pull.
\label{chan_pull}
\end{align}
Eq.~(\ref{riesz_chan}) is then equivalent to
\begin{align}
\chan^\pull \effinfout &= \effinf.
\label{pullback_effinf}
\end{align}
Fig.~\ref{channel_pullback} illustrates the concepts of $\chan_\pull$
and $\chan^\pull$ geometrically.  We will hereafter focus on
Eq.~(\ref{pullback_effinf}) in terms of the efficient influences
instead of Eq.~(\ref{sdiff_channel}). We call $\chan^\pull$ the
channel pullback, which is a representation of $\chan_\pull$, and we
will not use the latter again.

\fig{0.6}{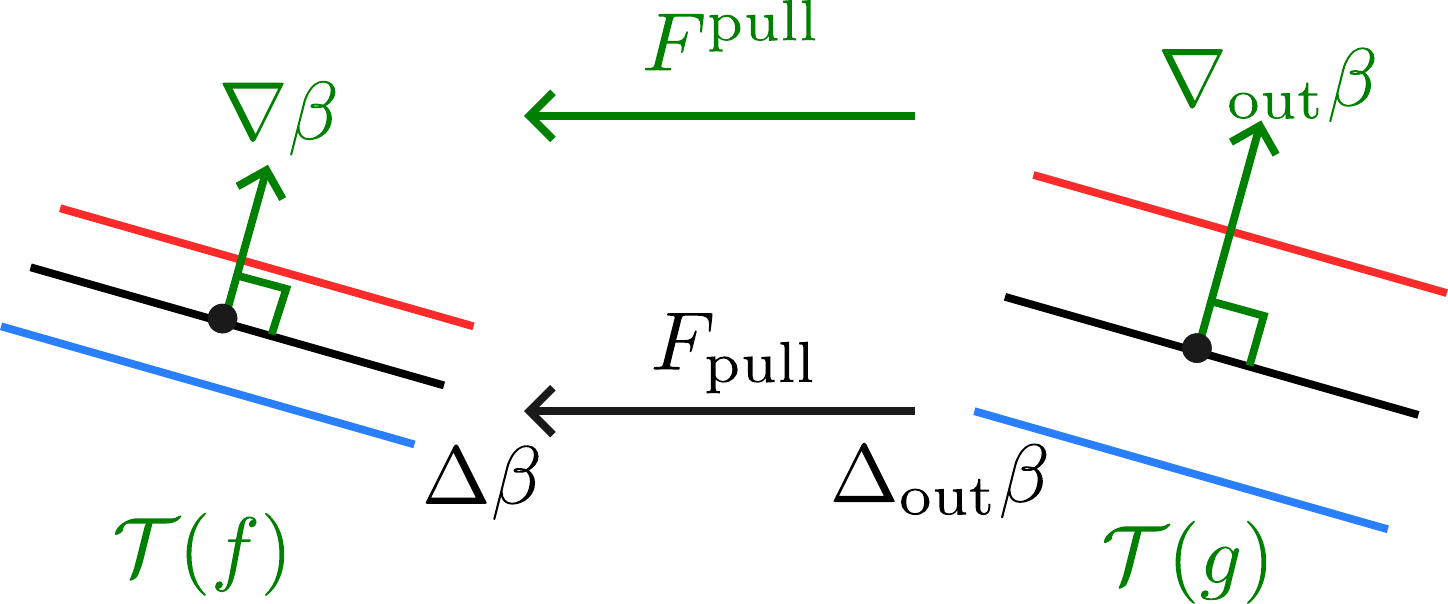}{The score differentials $\sdiff$ and
  $\sdiff_\out$ are depicted as level sets on the score tangent spaces
  $\mc T(f)$ and $\mc T(g)$, respectively. Eq.~(\ref{sdiff_channel})
  requires $\chan_\pull$ applied to the output score differential
  $\sdiff_\out$ to match the input score differential $\sdiff$.  The
  efficient influences $\effinf$ and $\effinfout$ are the Riesz
  representations of the score differentials and depicted as vectors
  orthogonal to the level sets.  Eq.~(\ref{pullback_effinf}) requires
  the channel pullback $\chan^\pull$ of the output efficient
  influence $\effinfout$ to match the input efficient influence
  $\effinf$.}

\subsection{\label{sec_Bndout}Output efficiency bound}
We have seen from Sec.~\ref{sec_exa} that the input efficient
influence $\effinf$ can often be solved.  With a channel,
Eq.~(\ref{pullback_effinf}) needs to be inverted to solve for
$\effinfout$. Beware, however, that the domain of $\chan^\pull$ is
$\Amb(g)$ and the solution to Eq.~(\ref{pullback_effinf}) in the
desired tangent space $\mc T(g)$ for $\effinfout$ is
nontrivial. Finding $\effinfout$ from this equation is the subject of
the following theorem.
\begin{theorem}
\label{thm_effinfout}
Let $\effinf \in \mc T(f)$ be the efficient influence at the input,
$\chan_\push:\mc T(f) \to \Amb(g)$ be the channel pushforward, and
$\chan^\pull = \chan_\push^\adj:\Amb(g) \to \mc T(f)$ be the channel
pullback. If $\effinf \notin \range \chan^\pull$, then
$\Bnd_\out = \infty$.  Otherwise, if
\begin{align}
\effinf \in \range \chan^\pull,
\label{effinf_range}
\end{align}
then the efficient influence $\effinfout \in \mc T(g)$ at the output
becomes
\begin{align}
\effinfout &= \chan^{\pull -} \effinf,
\label{effinf_pinv}
\end{align}
where $\chan^{\pull -}$ is the pseudo-inverse of $\chan^\pull$. The
efficiency bound at the output becomes
\begin{align}
\Bnd_\out &= \norm{\effinfout}_g^2 = \norm{\chan^{\pull -} \effinf}_g^2.
\label{Cout}
\end{align}
\end{theorem}
The pseudo-inverse is reviewed in Appendix~\ref{app_pinv}.

\subsection{\label{sec_imap}Information map}
If $\effinfout$ exists, then we know that
$\effinfout \in \mc T(g) = \overline{\range \chan_\push}$.  If we
further make the benign assumption
$\effinfout \in \range \chan_\push$, we can obtain an alternative form
of the bound that proves to be convenient later in
Sec.~\ref{sec_broadcast}. The bound is in terms of the map
\begin{align}
\imap &= \chan^\pull \chan_\push:\mc T(f) \to \mc T(f)
\label{imap}
\end{align}
called the information map. Before presenting the bound, we summarize
all the important conditions concerning $\effinf$ and $\effinfout$ in
the following lemma.
\begin{lemma}
\label{lem_effinf}
\begin{align}
&\quad \effinf \in \range \imap = \chan^\pull \range \chan_\push
\label{effinf_J}
\\
&\iff  \effinfout \in \range \chan_\push 
\label{effinf_range2}
\\
&\implies 
\effinfout \in \overline{\range \chan_\push} = \mc T(g)
\label{effinf_T}
\\
&\iff \effinf \in \range \chan^\pull.
\label{effinf_range3}
\end{align}
If $\range \chan_\push = \overline{\range \chan_\push}$, then all the
conditions are equivalent.
\end{lemma}

\begin{corollary}\label{cor_info}
  Assume Eq.~(\ref{effinf_J}) or equivalently
  Eq.~(\ref{effinf_range2}). Then
\begin{align}
\Bnd_\out &=\Avg{\effinf,\imap^-\effinf}_f,
\end{align}
where $\imap$ is the information map defined by
Eq.~(\ref{imap}).
\end{corollary}
As we have seen from Theorem~\ref{thm_effinfout} and
Corollary~\ref{cor_info}, the channel pushforward
$\chan_\push:\mc T(f) \to \Amb(g)$, the channel pullback
$\chan^\pull:\Amb(g) \to \mc T(f)$, and the information map
$\imap = \chan^\pull\chan_\push:\mc T(f) \to \mc T(f)$ are the three
key maps for determining the output efficient influence $\effinfout$
and the output efficiency bound $\Bnd_\out$. As the three maps are all
defined on the score tangent space $\mc T(f)$ at one point
$f(\theta)$, we call them \emph{local} channel maps, to be contrasted
with the \emph{global} channel map $\chan$ from which the local maps
are derived. The rest of this paper is devoted to examining the
properties of the three local maps and deriving their more concrete
forms for many channels, so that $\effinfout$ and $\Bnd_\out$ may be
evaluated more easily.

\subsection{\label{sec_finite}Finite-dimensional parameter space, following
Example~\ref{exa_finite}}
This section demonstrates a simple and well known consequence of
Theorem~\ref{thm_effinfout} and Corollary~\ref{cor_info} by assuming a
finite-dimensional parameter space and generalizing
Eq.~(\ref{C_matrix}).  Let
\begin{align}
f(\theta) = \theta \in \Theta \subseteq \mb R^p,
\end{align}
where $\Theta$ is an open subset of $\mb R^p$. Now each curve
$\phi(t)$ is in $\mb R^p$. Let each input score vector be
\begin{align}
\score(\dot\phi) &= \dot\phi(\theta) = \left.\pdv{\phi(t)}{t}\right|_{t=0}
\in \mc V = \mb R^p,
\end{align}
which is the tangent vector of $\phi(t)$ in Euclidean calculus. Let
the inner product of $\mc V$ be the Euclidean version
\begin{align}
\Avg{A,B}_f &= A^\top B.
\end{align}
We can then assume
\begin{align}
\mb R^p = \mc V = \mc V/\mc N(f) = \Amb(f) = \mc T(f).
\label{V_euclid}
\end{align}
We use the Euclidean inner product for the input in this section
because it is mathematically simple; it is not motivated by
statistics. There is no need to worry about equivalence classes
(because $\mc N(f) = \{0\}$), completion (because finite-dimensional
inner-product spaces are always complete), and closures (because a
finite-dimensional vector subspace is always closed). The last
equality in Eq.~(\ref{V_euclid}) assumes a maximal tangent space
$\mc T(f) = \Amb(f)$, so that the tangent vector can be any vector in
$\mb R^p$. We can now write
\begin{align}
f_\push\dot\phi &= \score(\dot\phi) = \dot\phi(\theta) \in \mb R^p.
\end{align}
The input efficient influence $\effinf \in \mc T(f) = \mb R^p$ satisfies
\begin{align}
\dd\beta(\dot\phi)
&= \dot\phi(\beta) = \dot\phi(\theta)^\top \partial\beta
= \Avg{\effinf,f_\push\dot\phi}_f =\dot\phi(\theta)^\top \effinf
\end{align}
in terms of the Euclidean gradient $\partial\beta$. For this equation to hold for any
$\dot\phi(\theta) \in \mc T(f) = \mb R^p$, we must have
\begin{align}
\effinf &= \partial\beta.
\label{effinf_kid}
\end{align}
Let us focus on the general quantum case, following
Example~\ref{exa_quantum}, and take the output density operator as
\begin{align}
g(\theta) &= \chan \circ f(\theta) = \chan(\theta).
\end{align}
This formula means that the channel map $\chan$ gives the output
density operator $g$ for each input $f(\theta) = \theta$. The output
score operator obeys
\begin{align}
\score_\out(\dot\phi) \jordan g &= \dot\phi(g) = 
\sum_j \dot\phi(\theta_j) \partial_j g(\theta).
\end{align}
We then obtain the formula
\begin{align}
\chan_{\push} u &= \sum_j u_j\score_\out(\partial_j)
\label{push_finite}
\end{align}
for the channel pushforward in terms of operators
$\score_\out(\partial_j)$ that satisfy
\begin{align}
\score_\out(\partial_j) \jordan g &= \partial_jg.
\label{scorej}
\end{align}
Even though $\chan$ may be nonlinear with respect to $\theta$,
$\chan_\push$ is always linear. The channel pullback
$\chan^\pull = \chan_\push^\adj$ defined by Eq.~(\ref{adj}), on the
other hand, becomes
\begin{align}
\bk{\chan^\pull u}_j &= \trace\bk{u \partial_j g} = \Avg{\score_\out(\partial_j),u}_g ,
\label{pull_finite}
\end{align}
and the information map $\imap$ defined by Eq.~(\ref{imap}) becomes
\begin{align}
\bk{\imap A}_j &= \bk{\chan^\pull \chan_\push A}_j
= \sum_k \Avg{\score_\out(\partial_j),\score_\out(\partial_k)}_g A_k
= \sum_k \imap_{jk} A_k,
\\
\imap_{jk} &= \Avg{\score_\out(\partial_j),\score_\out(\partial_k)}_g = \Imat_{jk},
\label{imap_matrix}
\end{align}
where $\Imat:\mb R^p \to \mb R^p$ coincides with the Helstrom version
of the information matrix. This equality motivates the name
information map for $\imap$.

We can now present a well known formula for the efficiency bound when
$\Imat$ may not be invertible, as shown earlier in
Ref.~\cite{stoica01} in the classical case and
Refs.~\cite{semi_prx,goldberg21,kwon25} in the quantum case.
\begin{corollary}
\label{cor_finite}
Assume that the parameter space $\Theta$ is an open subset of
$\mb R^p$.  Let $\partial\beta \in \mb R^p$ be the Euclidean gradient
as a column vector and $\Imat:\mb R^p \to \mb R^p$ be the information
matrix. If $\partial\beta \notin \range\Imat$, then
$\Bnd_\out = \infty$. Otherwise,
\begin{align}
\effinfout &= \score^\top \Imat^+ \partial\beta,
\label{effinfout_finite}
\\
\Bnd_\out &= \bk{\partial\beta}^\top \Imat^{+} \partial\beta,
\label{Bnd_finite}
\end{align}
where
$\score = \mqty(\score_\out(\partial_1) &\dots &
\score_\out(\partial_p) )^\top$ is a column vector of self-adjoint
operators that obey Eq.~(\ref{scorej}) and the superscript $+$ denotes
the Moore--Penrose inverse \cite{benisrael}.
\end{corollary}
The classical case follows the same argument and has the same result
essentially, except that $g$ is a probability density,
$\score_\out(\partial_j)$ is a score function, and $\Imat$ is the
Fisher information matrix.

\subsection{\label{sec_svd}Singular-value decomposition (SVD)}
To analyze and visualize the behavior of the channel pushforward
$\chan_\push$, it is illuminating to consider its SVD, assuming that
the map is compact. We write the SVD as
\begin{align}
\chan_\push &= \sum_j s_j \outsing_j \Avg{\insing_j,\cdot}_f,
\label{chan_push_SVD}
\end{align}
where $\{\insing_j\}$ is an orthonormal subset of $\mc T(f)$ called
the input singular vectors (also called the right singular vectors),
$\{\outsing_j\}$ is an orthonormal subset of $\Amb(g)$ called the
output singular vectors (also called the left singular vectors), and
$\{s_j\}$, all strictly positive, are called the singular values
\cite{reed_simon}.  The set of singular values is sorted in
descending order $s_0 \ge s_1 \ge \dots$ by convention and
required to be either a finite set or a sequence converging to $0$.

With the SVD, a compact pushforward $\chan_\push u$ can be viewed as a
transformation that
\begin{enumerate}
\item decomposes the input score $u$ into the series
  $\sum_j \insing_j \avg{\insing_j,u}_f$ in terms of the input
  singular vectors $\{\insing_j\}$ while nulling the component of $u$
  in $\kernel \chan_\push$,
\item maps each component $\insing_j$ to an output
singular vector $\outsing_j$, and
\item scales each output by a ``transmission coefficient'' $s_j$.
\end{enumerate}
In other words, engineers can regard a compact map as a generalized
low-pass filter and the SVD as a generalization of Fourier analysis.
As illustrated by Fig.~\ref{SVD0}, we can visualize this behavior as a
transformation of the unit sphere in $\mc T(f)$ to an ellipsoid in
$\Amb(g)$ with axes along $\{\outsing_j\}$, each with length
$s_j$. The longest length is the highest transmission coefficient and
coincides with the operator norm, that is,
\begin{align}
\norm{\chan_\push}  &= \max_j s_j = s_0.
\end{align}
For example, if the input score coincides with an input singular
vector $\insing_j$, then the output is simply
\begin{align}
\chan_\push \insing_j &= s_j \outsing_j.
\end{align}
Since $\norm{\insing_j}_f = 1$ and $\norm{\outsing_j}_g = 1$,
$\norm{\chan_\push \insing_j}_g^2/\norm{\insing_j}_f^2 = s_j^2$ is the
information transmission coefficient. Appendix~\ref{app_sv} presents
more facts about singular values.

\fig{0.48}{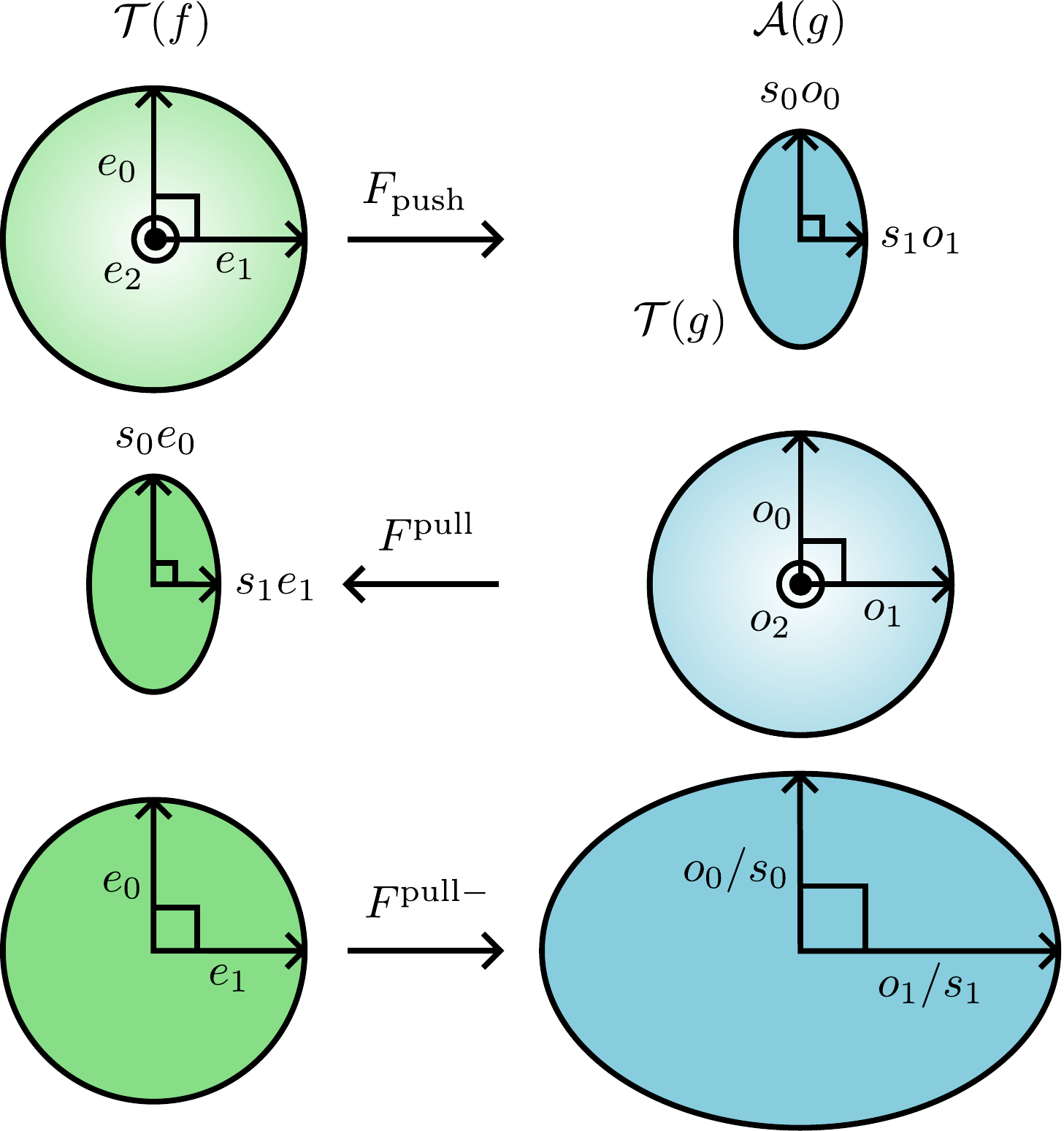}{First row: the channel pushforward $\chan_\push$ maps
  the unit sphere in the input score tangent space $\mc T(f)$ (3D in
  the picture) to an ellipsoid in the output score tangent space
  $\overline{\range\chan_\push} = \mc T(g)$ (2D in the picture), which
  is in an ambient Hilbert space $\Amb(g)$ (3D in the picture). In
  particular, each $\insing_j$ is mapped to $s_j \outsing_j$, while
  $\kernel\chan_\push$ is along $\insing_2$ in the picture.  Second
  row: the channel pullback $\chan^\pull$ maps the unit sphere in
  $\Amb(g)$ to an ellipsoid in
  $\overline{\range \chan^\pull} = (\kernel \chan_\push)^\perp
  \subseteq \mc T(f)$. $\kernel(\chan^\pull)$ coincides with
  $(\range\chan_\push)^\perp$, the subspace along which the ellipsoid
  of $\chan_\push$ appears flat (along $\outsing_2$ in the
  picture). Third row: the pullback inverse $\chan^{\pull-}$ behaves
  similarly to $\chan_\push$ but the singular values are inverted.  In
  this work, $\chan^{\pull-}$ is always applied to
  $\range \chan^\pull$ or else $\Bnd_\out = \infty$, so
  $\insing_2 \in (\range\chan^\pull)^\perp$ is not drawn in the
  bottom-left picture. }

Given a SVD of $\chan_\push$, $\chan^\pull$ and $\imap$ possess the
SVDs
\begin{align}
\chan^\pull &= \sum_j s_j \insing_j \Avg{\outsing_j,\cdot}_g,
\\
\imap &= \sum_j s_j^2 \insing_j \Avg{\insing_j,\cdot}_f,
\label{imap_svd}
\end{align}
such that $\{s_j^2\}$ are precisely the eigenvalues of $\imap$ and
$\{\insing_j\}$ are precisely the eigenvectors of $\imap$. Furthermore
\cite[Sec.~3.10.1]{villiers},
\begin{align}
\spn\BK{\insing_j} &\subseteq \range \chan^\pull,
\label{span_insing}
\\
\overline{\spn\BK{\insing_j}} &= 
\bk{\kernel \chan_\push}^\perp = \overline{\range \chan^\pull},
\\
\overline{\spn\BK{\outsing_j}} &= \bk{\kernel \chan^\pull}^\perp = 
\overline{\range \chan_\push}  = \mc T(g),
\end{align}
where $\spn$ denotes the linear span. If
\begin{align}
\effinf \in \spn\BK{\insing_j},
\label{effinf_sing}
\end{align}
then $\effinfout$ is guaranteed to exist and $\Bnd_\out < \infty$ by
Eq.~(\ref{span_insing}) and Theorem~\ref{thm_effinfout}. For example,
let $\effinf = \insing_j$. Since
$\chan^\pull \outsing_j = s_j \insing_j$ and
$\outsing_j \in \mc T(g)$, we obtain
\begin{align}
\chan^{\pull-} \insing_j &= \frac{\outsing_j}{s_j},
&
\Bnd_\out &= \norm{\chan^{\pull-} \insing_j}_g^2  = \frac{1}{s_j^2}.
\end{align}
The behavior of the pullback inverse $\chan^{\pull-}$ can thus be
pictured as an ellipsoid that has inverse axis lengths to those of
$\chan_\push$, as shown in Fig.~\ref{SVD0}. Formally, we can write
\begin{align}
\chan^{\pull-} &= \sum_j \frac{1}{s_j} \outsing_j \Avg{\insing_j,\cdot}_f,
&
\imap^- &= \sum_j \frac{1}{s_j^2} \insing_j \Avg{\insing_j,\cdot}_f,
&
\Bnd_\out &= \sum_j \frac{1}{s_j^2} \abs{\Avg{\insing_j,\effinf}_f}^2,
\label{pinv_svd}
\end{align}
although the arguments of the pseudo-inverses should be checked
against the conditions in Theorem~\ref{thm_effinfout},
Lemma~\ref{lem_effinf}, and Corollary~\ref{cor_info} before we use
these formulas in rigorous analysis. For example, if
Eq.~(\ref{effinf_sing}) holds, then the sums in Eqs.~(\ref{pinv_svd})
become finite sums and the formulas are safe to use.

\subsection{\label{sec_mono}Monotonic channels}
A fundamental property of physical channels with no access to $\theta$
is that they cannot increase the information about $\theta$, as proved
for many versions of Fisher information and Markov maps in classical
and quantum statistics \cite{ibragimov,petz}. This property is called
monotonicity and can be expressed as
\begin{align}
\norm{\chan_\push u}_{g} &\le \norm{u}_{f} \quad
\forall u \in \range f_\push.
\end{align}
This type of inequality is also called a data-processing inequality.
The inequality holds if and only if the operator norm of $\chan_\push$
obeys
\begin{align}
\norm{\chan_\push} \le 1.
\label{mono}
\end{align}
We take Eq.~(\ref{mono}) to be the definition of a monotonic
channel.
If $\chan_\push$ is compact, Eq.~(\ref{mono})
holds if and only if the maximum singular value obeys
\begin{align}
\norm{\chan_\push}  = \max_j s_j = s_0 &\le 1.
\end{align}
The efficiency bound and the information map are also monotonic in the
following sense.
\begin{theorem}
\label{thm_mono}
  If a channel is monotonic, then
\begin{align}
\Bnd_\out &\ge \Bnd,
&
\norm{\imap} &\le 1.
\end{align}
\end{theorem}
Note that a channel need not be monotonic if $\norm{u}_{f}$ at the
input does not have the meaning of an information quantity, such as
the Euclidean norm assumed in Sec.~\ref{sec_finite}, or if the channel
has access to $\theta$, such as the broadcast channel to be discussed
in Sec.~\ref{sec_broadcast}.

\section{Composite channels}
This section studies two ways of combining multiple channels and the
rules for combining the local channel maps.

\subsection{\label{sec_cascade}Cascaded channels}
When the output of one channel is fed as the input to the next
channel, the cascaded channel map is the composition
\begin{align}
\chan(N,1) &= \chan_N \circ \dots \circ \chan_2 \circ \chan_1,
\end{align}
where $\chan_j$ is the map for the $j$th channel and $\chan(N,1)$
denotes the map for the cascaded channel from channel $1$ to channel
$N$. By the chain rule of $\dot\phi$, the total channel pushforward
satisfies the covariant chain rule
\begin{align}
\chan(N,1)_\push &= \chan_{N \push} \dots \chan_{2 \push} \chan_{1 \push},
\label{push_chain}
\end{align}
while the channel pullback follows the contravariant chain rule
\begin{align}
\chan(N,1)^\pull &= \chan_1{}^{\pull} \chan_2{}^{\pull} \dots \chan_N{}^{\pull}.
\label{pull_chain}
\end{align}
The information map, on the other hand, follows the backpropagation
rule
\begin{align}
\imap(N,j) &= \chan_j{}^{\pull} \imap(N,j+1) \chan_{j\push},
&
\imap(N,N+1) &= I,
\label{imap_chain}
\end{align}
where $\imap(N,j)$ is the information map for the channel map
$\chan(N,j) = \chan_N \circ \dots \circ \chan_j$. 

\subsection{\label{sec_broadcast}Broadcast channels}
A broadcast channel takes one input to produce multiple outputs. 
We study here two types of broadcast channels:
\begin{enumerate}
\item A tensor product of channel maps, following
  Examples~\ref{exa_cc} and \ref{exa_qq}: When the broadcast channel
  produces $N$ independent classical or quantum objects, the output
  probability density or density operator can be expressed as the
  tensor product
\begin{align}
g(\theta) &= \chan\circ f(\theta) = \bigotimes_j \chan_j f(\theta),
\label{prod_maps}
\end{align}
where each $\chan_j$ is a Markov map and $f(\theta)$ is the input
probability density or density operator.

\item A direct sum of channel maps: Pick one example from
  Example~\ref{exa_pp}, \ref{exa_gg}, \ref{exa_qpp}, or \ref{exa_qgg},
  and assume a broadcast channel that produces $N$ independent outputs
  of the same type in the example. In all cases, the output
  parametrization function can be expressed as the direct sum
\begin{align}
g(\theta) &= \bigoplus_j \chan_j f(\theta).
\label{sum_maps}
\end{align}
To be specific:
\begin{enumerate}
\item For Example~\ref{exa_pp}, two independent Poisson fields with
  mean densities $\chan_1 f$ and $\chan_2 f$ form a Poisson field with
  mean density $\chan_1 f \oplus \chan_2 f$, as can be proved from the
  Laplace-functional characterization theorem
  \cite[Theorem~2.9]{cinlar}.
\item For Example~\ref{exa_gg}, two independent Gaussian fields with
  mean vectors $\chan_1 f \in \mc V_1$ and $\chan_2 f \in \mc V_2$
  form a Gaussian field with mean vector
  $\chan_1 f \oplus \chan_2 f \in \mc V_1 \oplus \mc V_2$
  \cite[Example~2.3.8]{bogachev_gauss}.

\item For Example~\ref{exa_qpp}, the tensor product of two Poisson
  states with intensity operators $\chan_1 f$ and $\chan_2f$ is a
  Poisson state with intensity operator $\chan_1 f \oplus \chan_2 f$
  \cite[Example~5]{poisson_quantum}.

\item For Example~\ref{exa_qgg}, the tensor product of two Gaussian
  states with mean vectors $\chan_1 f \in \mc V_1$ and
  $\chan_2 f \in \mc V_2$ is a Gaussian state with mean vector
  $\chan_1 f \oplus \chan_2 f \in \mc V_1 \oplus \mc V_2$
  \cite{holevo_info}.
\end{enumerate}
\end{enumerate}
For every example, it turns out that we can set the output ambient
vector space $\mc V_\out$ as a direct sum of the output ambient vector
spaces $\{\mc V_j\}$ of the individual outputs $\{\chan_jf\}$. The
concept of direct sum is reviewed in Appendix~\ref{app_direct}. We
first consider a tensor product of quantum channel maps; the classical
case is simply a special case.
\begin{lemma}
\label{lem_prod_maps}
Consider a tensor product of channel maps given by
Eq.~(\ref{prod_maps}), where each individual channel $\chan_j$ is a
Markov map modeling a quantum channel in Example~\ref{exa_qq}.  Let
$\mc V_j$ be the set of zero-mean observables with respect to each
$\chan_j f$, following Sec.~\ref{sec_quantum}. Then the output ambient
vector space $\mc V_\out$ for Eq.~(\ref{prod_maps}) can be set
as the direct sum
\begin{align}
\mc V_\out &= \bigoplus_j \mc V_j,
\label{V_direct}
\end{align}
where each output score $\score_\out(\dot\phi)$ obeys
\begin{align}
\score_\out(\dot\phi) &= \bigoplus_j \score_j(\dot\phi)
\label{score_direct}
\end{align}
in terms of the output score $\score_j(\dot\phi) \in \mc V_j$ for
each individual output $\chan_j f$. It follows from
Eq.~(\ref{V_direct}) that the output ambient Hilbert space can be set
as
\begin{align}
\Amb(g) &= \bigoplus_j \Amb(\chan_j f).
\label{Amb_direct}
\end{align}
\end{lemma}
\begin{lemma}
\label{lem_sum_maps}
Consider a direct sum of channel maps given by Eq.~(\ref{sum_maps}),
where all the individual channels are of the same type in
Example~\ref{exa_pp}, \ref{exa_gg}, \ref{exa_qpp}, or \ref{exa_qgg}.
Then Eqs.~(\ref{V_direct})--(\ref{Amb_direct}) also hold.
\end{lemma}

With the direct-sum structure, $\imap$ can be simplified.
\begin{theorem}
\label{thm_add}
Assume the direct-sum structure given by
Eqs.~(\ref{V_direct})--(\ref{Amb_direct}). The local channel maps obey
\begin{align}
\chan_\push &= \sum_j \inj_j \chan_{j\push},
\label{push_sum}
\\
\chan^\pull &= \sum_j \chan_j{}^{\pull}\pi_j,
\label{pull_sum}
\\
\imap &= \sum_j \chan_j{}^{\pull}\chan_{j\push},
\label{imap_sum}
\end{align}
where $\inj_j:\Amb(\chan_j f) \to \Amb(g)$ is the injection map and
$\pi_j = \inj_j^\adj:\Amb(g) \to \Amb(\chan_j f)$ is the projection
map, as reviewed in Appendix~\ref{app_direct}.
\end{theorem}
While $\chan_\push:\mc T(f) \to \Amb(g)$ and
$\chan^\pull:\Amb(g) \to \mc T(f)$ involve a high-dimensional Hilbert
space $\Amb(g) = \oplus_j \Amb(\chan_j f)$ that grows with the number
of outputs, notice that $\imap$ is simply a sum of the individual
information maps $\chan_j{}^\pull\chan_{j\push}:\mc T(f)\to\mc T(f)$,
each confined to the same domain and codomain $\mc T(f)$.  This
additivity generalizes the well known additivity of the Fisher or
Helstrom information matrices for independent objects. It allows the
computation of $\imap$ to avoid the high-dimensional $\Amb(g)$, thus
simplifying the computation of $\Bnd_\out$ via
Corollary~\ref{cor_info}. It also simplifies the computation of the
singular values $\{s_j\}$ and the input singular vectors
$\{\insing_j\}\subset \mc T(f)$ of $\chan_\push$, since they can be
obtained from Eq.~(\ref{imap_svd}) via the SVD of $\imap$. Once
$\{s_j\}$ and $\{\insing_j\}$ are obtained, they can be plugged into
Eq.~(\ref{push_sum}) to give the output singular vectors
$\{\outsing_j\}$ of $\chan_\push$ by
\begin{align}
\outsing_j &= \frac{1}{s_j} \chan_\push \insing_j = \frac{1}{s_j} \sum_k \inj_k \chan_{k\push} \insing_j
= \frac{1}{s_j} \bigoplus_k \chan_{k\push} \insing_j.
\end{align}
This convenient property is the main reason we introduced
the information map in Sec.~\ref{sec_imap}.

\section{\label{sec_chan_exa2}Examples of local channel maps}
We now study the channels in Examples~\ref{exa_cc}--\ref{exa_qgg} in
more detail and derive formulas for the local channel maps
$\chan_\push$, $\chan^\pull$, and $\imap$.

\subsection{\label{sec_cc}Classical channels, following Examples~\ref{exa_cc} 
and \ref{exa_pp}}
Consider the classical channel in Example~\ref{exa_cc}. We begin by
introducing a necessary concept called the dual $\chan^\dual$ of the
Markov map $\chan$, defined by
\begin{align}
\int A \chan f \dd\nu &= \int A(\elout) \chan(\elout|\el) f(\theta,\el) \dd\mu(\el) 
\dd\nu(\elout)
= \int \bk{\chan^\dual A} f \dd\mu
\label{chan_dual}
\end{align}
for any random variable $A$ at the output. An explicit formula is
\begin{align}
\bk{\chan^\dual A}(\el) &= \int A(\elout) \chan(\elout|\el) \dd\nu(\elout)
\quad
(\textrm{almost everywhere $P_\theta$}).
\end{align}
Observe that this is the predictive conditional expectation. From the
definition, we obtain
\begin{align}
\expect\Bk{\bk{\chan^\dual A}(X)} &= \expect[A(Y)],
\label{adam}
\end{align}
so the mean is conserved. Moreover, it is well known that the second
moment cannot increase under a conditional expectation
\cite{parth_pm}, in the sense that
\begin{align}
\norm{A}_{g}^2 &\ge \norm{\chan^\dual A}_f^2.
\label{eve}
\end{align}
Let $\mc V$ be the vector space of zero-mean random variables at the
input and $\mc V_\out$ be that of zero-mean random variables at the
output.  Eq.~(\ref{adam}) implies that, if we restrict the domain of
$\chan^\dual$ to $\mc V_\out$, then its codomain can be taken as
$\mc V$.  Note also that, for any zero-seminorm element $A$ of
$\mc N(g)$, $0 = \norm{A}_{g} \ge \norm{\chan^\dual A}_f$ by
Eq.~(\ref{eve}) and thus $\chan^\dual A \in \mc N(f)$, implying
\begin{align}
\chan^\dual\Bk{A + \mc N(g)}  &= \chan^\dual A + \chan^\dual \mc N(g)
\subseteq \chan^\dual A + \mc N(f) \in \mc V/\mc N(f),
\end{align}
and $\chan^\dual$ remains well defined as a map
$\chan^\dual:\mc V_\out/\mc N(g) \to \mc V/\mc N(f)$ (in the sense
that it maps each equivalence class $A + \mc N(g)$ in the domain to
one and only one equivalence class $\chan^\dual A + \mc N(f)$ in the
codomain). Eq.~(\ref{eve}) further implies that the operator norm satisfies
$\norm{\chan^\dual} \le 1$, so $\chan^\dual$ can be uniquely
extended to become a bounded map $\chan^\dual:\Amb(g) \to \Amb(f)$
with the same operator norm.

We now find the channel pushforward $\chan_\push$ in terms of
$\chan^\dual$.  Define the adjoint
$\chan^{\dual\adj}:\Amb(f)\to\Amb(g)$ of $\chan^\dual$ with respect to
the inner products of $\Amb(f)$ and $\Amb(g)$ by
\begin{align}
\Avg{\chan^{\dual\adj} A,B}_g &= \Avg{A,\chan^\dual B}_f \quad
\forall A \in \Amb(f) ,B \in \Amb (g).
\label{chan_dual_adj}
\end{align}
Following the notation of Ref.~\cite{gce2}, we abbreviate
$\chan^{\dual\adj}$ as
\begin{align}
\chan^{\dual\adj} = \chan_\dual.
\label{chan_dual_push}
\end{align}
Unpacking Eq.~(\ref{chan_dual_adj}) in terms of functions in $\mc V$
and $\mc V_\out$, we find
\begin{align}
\int \bk{\chan_\dual A} B g \dd\nu
&= \int \Bk{\chan\bk{f A }} B \dd\nu,
\end{align}
which implies 
\begin{align}
\bk{\chan_\dual A}(\elout) &= 
\frac{1}{g(\theta,\elout)}\int \chan(\elout|\el) f(\theta,\el) A(\el) \dd\mu(\el)
\quad
(\textrm{almost everywhere $Q_\theta$}).
\label{retro}
\end{align}
Eq.~(\ref{retro}) is identical to Eq.~(\ref{score_out}) that relates
the input and output score functions and coincides with the
retrodictive conditional expectation. $\chan_\dual$ is almost
our channel pushforward $\chan_\push$, except that the domain of the
former is $\Amb(f)$ while that of the latter is $\mc T(f)$. This
subtle issue can be fixed by writing
\begin{align}
\chan_\push &= \chan_\dual \inj_{\mc T(f)}
\label{chan_push_classical}
\end{align}
in terms of the injection map $\inj_{\mc T(f)}:\mc T(f) \to
\Amb(f)$. Eq.~(\ref{chan_push_classical}) is the correct expression
for the channel pushforward, giving the right formula in
Eq.~(\ref{score_out}) to relate the input and output scores and
possessing the right domain $\mc T(f)$ and codomain
$\Amb(g)$. Eq.~(\ref{chan_push_classical}) and the monotonicity of the
channel will be stated as lemmas in Sec.~\ref{sec_qq} to be proved
under the more general quantum case.

Given Eq.~(\ref{chan_push_classical}), the channel pullback
$\chan^\pull:\Amb(g) \to \mc T(f)$ becomes
\begin{align}
\chan^\pull &= \inj_{\mc T(f)}^\adj \chan^{\dual}
= \pi_{\mc T(f)}\chan^{\dual},
\label{chan_pull_classical}
\end{align}
where $\pi_{\mc T(f)}:\Amb(f) \to \mc T(f)$ is the projection map (see
Appendix~\ref{app_direct}).  The information map given by
Eq.~(\ref{imap}) becomes
\begin{align}
\imap &= \chan^\pull \chan_\push = \pi_{\mc T(f)} 
\chan^\dual \chan_\dual \inj_{\mc T(f)}
= \pi_{\mc T(f)} \IMAP \inj_{\mc T(f)},
\label{imap_classical}
\\
\IMAP &= \chan^\dual \chan_\dual:\Amb(f) \to \Amb(f).
\label{IMAP_classical}
\end{align}
We call $\IMAP$ the maximal information map.  $\IMAP$ is the
information map when the input score tangent space is maximal
($\mc T(f) = \Amb(f)$) and the $\inj_{\mc T(f)}$ and $\pi_{\mc T(f)}$
maps become redundant.

Let us now take the nonparametric model in Sec.~\ref{sec_classical}
as the input to the classical channel in Example~\ref{exa_cc}. The
parameter $\theta(\cdot) = f(\theta,\cdot)$ is an arbitrary
probability density. The parameter of interest is the expected value
$\beta(\theta) = \expect(b) = \int b \theta \dd\mu = \Avg{b,1}_\theta$ of a random
variable $b$. Sec.~\ref{sec_classical} has shown that the input
$\mc T(f)$ is maximal and the input efficient influence is
$\effinf = b-\avg{b,1}_\theta \in \Amb(f) = \mc T(f)$. By
Theorem~\ref{thm_effinfout}, $\Bnd_\out < \infty$ if and only if
Eq.~(\ref{effinf_range}) holds. Eq.~(\ref{effinf_range}) means that
\begin{align}
\exists u \in \Amb(g): 
\effinf &= b-\Avg{b,1}_\theta = \chan^\pull u = \chan^\dual u.
\label{pullback_B}
\end{align}
If such a $u$ exists, then Eq.~(\ref{pullback_B}) can be solved for
the output efficient influence $u = \effinfout \in \mc T(g)$ by
applying the pseudo-inverse $\chan^{\dual-}:\mc T(f) \to \Amb(g)$ to
$\effinf$. Further progress can be made analytically for special
channel maps, such as the exponential family \cite[Sec.~25.5]{vaart}.

Since $\effinfout$ may depend on the unknown $\theta$, it is
nontrivial to find an efficient estimator.  If the sample space
$\sspace$ is finite, the parameter spaces $\Theta$ for the models in
Secs.~\ref{sec_classical} and \ref{sec_poisson} are
finite-dimensional, and one can resort to the maximum-likelihood
estimator to achieve asymptotic efficiency \cite{vaart}. For a simpler
strategy, suppose that we have computed a preliminary estimate
$\check\theta$ of $\theta$ from some independent pilot observations
and derived $\effinfout$ at $\check\theta$; denote it as
$(\effinfout)(\check\theta)$. Then Eq.~(\ref{pullback_B}) at
$\check\theta$ leads to
\begin{align}
\chan^\dual\Bk{(\effinfout)(\check\theta)} 
&= b -\beta(\check\theta) + \mc N(\check\theta),
&
\mc N(\check\theta) &= \BK{A \in \mc V: \int A^2 \check\theta \dd\mu = 0}.
\end{align}
At the true $\theta$, if we can make the assumption 
\begin{align}
\int A \theta \dd\mu = 0 \quad
\forall A \in \mc N(\check\theta),
\label{N0}
\end{align}
which is satisfied if $P_{\check\theta}$ dominates $P_\theta$, then
\begin{align}
\expect\Bk{\chan^\dual(\effinfout)(\check\theta)}
&= \expect\Bk{\bk{\effinfout}(\check\theta)}
= \expect(b) - \beta(\check\theta) = \beta(\theta)-\beta(\check\theta).
\end{align}
It follows that the estimator
\begin{align}
\check\beta(Y) &= \beta(\check\theta) + (\effinfout)(\check\theta,Y)
\label{onestep}
\end{align}
in terms of the output random element $Y$ is unbiased.  This type of
estimator is called a one-step estimator in classical statistics
\cite{kennedy24}.  Its variance is the variance of
$(\effinfout)(\check\theta)$ at $g(\theta)$, which is close to the
efficiency bound $\norm{(\effinfout)(\theta)}_{g(\theta)}^2$ if the
preliminary estimate $\check\theta$ is sufficiently close to the true
$\theta$.  We leave the question of how to find a satisfactory
$\check\theta$, the rigorous efficiency of the one-step estimator, and
the possibility of better estimators outside the scope of this work.

The preceding results can be translated for Poisson problems.  In
Sec.~\ref{sec_poisson}, the parameter $\theta$ is an arbitrary
Poisson-field mean density and the parameter of interest is a linear
functional $\beta = \int b\theta \dd\mu$. Take the setting of
Sec.~\ref{sec_poisson} as the input to the Poisson channel in
Example~\ref{exa_pp}. Now $\Amb(f)$ and $\Amb(g)$ are no longer
restricted to have zero mean, but the input $\mc T(f)$ is still
maximal, all the maps still have the same formulas as
Eqs.~(\ref{chan_dual})--(\ref{IMAP_classical}), the channel remains
monotonic, although Eq.~(\ref{pullback_B}) should be changed to
\begin{align}
\exists u \in \Amb(g): 
\effinf &= b = \chan^\pull u = \chan^\dual u.
\label{pullback_B_poisson}
\end{align}
If such a $u$ exists, $\effinfout = \chan^{\dual-} b$ can be
solved. Given the output efficient influence
$(\effinfout)(\check\theta)$ evaluated at a preliminary estimate
$\check\theta$ and assuming Eq.~(\ref{N0}), the linear estimator
\begin{align}
\check\beta(Y) &= \int (\effinfout)(\check\theta,\elout) \dd Y(\elout)
\label{onestep_poisson}
\end{align}
in terms of the output Poisson field $Y$ is unbiased with variance
$\norm{(\effinfout)(\check\theta)}_{g(\theta)}^2$. We call
Eq.~(\ref{onestep_poisson}) a one-step estimator for the Poisson
model.

\subsection{\label{sec_qq}Quantum channels, following Examples~\ref{exa_qq} and \ref{exa_qpp}}
Similar to the classical case, we begin by introducing the dual
$\chan^\dual$ of the Markov map $\chan$, defined by
\begin{align}
\trace\Bk{\bk{\chan f} A} &= \trace\bk{f \chan^\dual A}
\label{chan_dual_q}
\end{align}
for any trace-class $f$ and bounded $A$ on $\mc H_\out$, such that
$\chan^\dual A$ is also a bounded operator on $\mc H$
\cite{holevo_structure}. Given the Kraus representation
\begin{align}
\chan f &= \sum_j K_j f K_j^\dual
\label{kraus}
\end{align}
in terms of a set of Kraus operators $K_j:\mc H \to \mc H_\out$,
$\chan^\dual$ is simply
\begin{align}
\chan^\dual A &= \sum_j K_j^\dual A K_j.
\label{kraus_dual}
\end{align}
$\chan^\dual$ is completely positive if and only if $\chan$ is
completely positive. $\chan^\dual$ is unital ($\chan^\dual I = I$) if
and only if $\chan$ is trace-preserving.  Because
$\chan^\dual\bk{A^\dual} = \bk{\chan^\dual A}^\dual$, $\chan^\dual$
maps self-adjoint operators to self-adjoint operators.  The definition
given by Eq.~(\ref{chan_dual_q}) also implies that $\chan^\dual$
conserves the mean, so we can consider $\chan^\dual$ as a map
$\chan^\dual:\mc V_\out \to \mc V$ if $\mc V_\out$ and $\mc V$
comprise zero-mean observables.  We can then extend it to become a map
on the Hilbert space $\Amb(g)$ as follows.
\begin{lemma}
\label{lem_dual}
The Markov-map dual $\chan^\dual$ can be uniquely extended to become a
bounded map $\chan^\dual:\Amb(g) \to \Amb(f)$ with operator norm 
satisfying
\begin{align}
\norm{\chan^\dual}  &\le 1
\label{mono_quantum}
\end{align}
in terms of the norms of $\Amb(f)$ and $\Amb(g)$.
\end{lemma}
Let the adjoint of $\chan^\dual:\Amb(g) \to \Amb(f)$ be
$\chan_\dual = \chan^{\dual\adj}:\Amb(f) \to \Amb(g)$; both are
generalized conditional expectations
\cite{hayashi,gce_pra,gce2}. $\chan_\dual$ has the formal expression
\begin{align}
\chan_\dual &= \mc E_{\chan f}^{-1} \chan \mc E_f,
&
\mc E_f A &= f \jordan A,
\end{align}
which generalizes Eq.~(\ref{retro}) and coincides with Eq.~(\ref{gce})
that relates the input and output score
operators. Eq.~(\ref{chan_push_classical}) continues to hold; we state
it as a lemma to be proved.
\begin{lemma}
\label{lem_score}
Given a Markov map $\chan$, the channel pushforward
$\chan_\push:\mc T(f) \to \Amb(g)$ is
\begin{align}
\chan_\push &= \chan_\dual \inj_{\mc T(f)},
\label{chan_push_quantum}
\end{align}
where $\chan_\dual = \chan^{\dual\adj}:\Amb(f) \to \Amb(g)$ is the
adjoint of $\chan^\dual$ with respect to the inner products of
$\Amb(f)$ and $\Amb(g)$, $\chan^\dual:\Amb(g) \to \Amb(f)$ is the
Markov-map dual, and $\inj_{\mc T(f)}:\mc T(f) \to \Amb(f)$ is the
injection map.
\end{lemma}
We can now prove the monotonicity of Markov maps. The theorem is well
established in the finite-dimensional case \cite{petz}, while
Ref.~\cite[Lemma~3]{qlmoment_pra2} gives a proof for arbitrary
dimensions when $\chan$ is the partial trace. We give a different
proof here for completeness.
\begin{theorem}\label{thm_markov}
Any channel modeled by a Markov map is monotonic.
\end{theorem}
\begin{proof}
  $\norm{\inj} = 1$ for any injection map $\inj$.
  Lemmas~\ref{lem_dual} and \ref{lem_score} then imply
\begin{align}
\norm{\chan_\push} \le \norm{\chan_\dual} \norm{\inj_{\mc T(f)}}
= \norm{\chan_\dual}  = \norm{\chan^{\dual}} \le 1.
\end{align}
\end{proof}
Any measurement of a quantum object can be regarded as a Markov map
from a density operator to a probability density, such that the
quantum efficiency bound is a lower bound on the classical bound for
any measurement by virtue of Theorem~\ref{thm_mono}. The quantum bound
therefore inherits the statistical significance of the classical bound
outlined in Appendix~\ref{app_error}.

For the channel pullback $\chan^\pull$ and the information map
$\imap$, Eqs.~(\ref{chan_pull_classical})--(\ref{IMAP_classical}) all
continue to hold in the quantum case.

Let us now take the nonparametric model in Sec.~\ref{sec_quantum} as
the input to the quantum channel in Example~\ref{exa_qq}.  The
parameter $\theta = f(\theta)$ is an arbitrary density operator.  The
parameter of interest is the expected value
$\beta(\theta) = \trace(b\theta)$ of an observable $b$.
Sec.~\ref{sec_quantum} has shown that the input $\mc T(f)$ is maximal
and the input efficient influence is
$\effinf = b-\avg{b,I}_\theta \in \Amb(f) = \mc T(f)$.  To check that
$\effinfout$ exists and $\Bnd_\out < \infty$ by
Theorem~\ref{thm_effinfout}, the condition is the same as
Eq.~(\ref{pullback_B}) in the classical case.  If the condition holds,
the output efficient influence $\effinfout \in \mc T(g)$ can be
obtained by applying the pseudo-inverse
$\chan^{\dual-}:\mc T(f) \to \Amb(g)$ to $\effinf$.

Since $\effinfout$ may depend on the unknown $\theta$, it is nontrivial
to find a measurement that achieves the quantum efficiency bound.  A
widely studied strategy for $N$ objects is as follows
\cite{fujiwara06,yang19,demkowicz20,tsang26}. First measure some of
the objects to obtain a preliminary estimate $\check\theta$ of
$\theta$. Then evaluate the efficient influence
$(\effinfout)(\check\theta)$ at $\check\theta$
and measure the observable
\begin{align}
\check\beta &= \beta(\check\theta) I + (\effinfout)(\check\theta)
\label{onestep_q}
\end{align}
of the rest of the objects. Similar to the classical case in
Sec.~\ref{sec_cc}, this operator-valued estimator is unbiased if we
can assume
\begin{align}
\trace(A \theta) = 0 \quad \forall A \in \mc N(\check\theta)
=\BK{A \in \mc V: \trace(A^2\check\theta)=0 },
\label{N0q}
\end{align}
and its variance is the variance of $(\effinfout)(\check\theta)$ at
$g(\theta)$. Following classical statistics and Sec.~\ref{sec_cc}, we
call Eq.~(\ref{onestep_q}) a one-step estimator. We again leave the
question of how to find $\check\theta$ and the rigorous efficiency of
the estimator outside the scope of this work.

These results can be translated for a quantum Poisson problem, where
we take the nonparametric model in Sec.~\ref{sec_poisson} as the input
to the Poisson channel in Example~\ref{exa_qpp}. Then all the maps
still have the same formulas as
Eqs.~(\ref{chan_dual_q})--(\ref{chan_push_quantum}). To check that
$\effinfout$ exists and $\Bnd_\out <\infty$ by
Theorem~\ref{thm_effinfout}, we can again use the condition given by
Eq.~(\ref{pullback_B_poisson}).  To generalize the classical one-step
estimator in Sec.~\ref{sec_cc}, let $(\effinfout)(\check\theta)$ be
the efficient influence evaluated at a preliminary
$\check\theta$. Decompose it as
\begin{align}
(\effinfout)(\check\theta) &= \int \lambda(\el) \dd E(\el)
\label{onestep_qpoisson}
\end{align}
in terms of a projection-valued measure $E$ on $\mc H_\out$. Then a
measurement of the $L$ bodies in the Poisson object based on $E$ gives
a Poisson field $Y$ with mean measure $\trace[E(\cdot)\theta]$.
Assuming Eq.~(\ref{N0q}), the linear estimator
\begin{align}
\check\beta(Y) &= \int \lambda(\elout) \dd Y(\elout)
\label{onestep_qpoisson2}
\end{align}
is unbiased and its variance is
$\norm{(\effinfout)(\check\theta)}_{g(\theta)}^2$.

\subsection{\label{sec_gg}Gaussian channels, following Examples~\ref{exa_gg}
and \ref{exa_qgg}}
In the classical case (Example~\ref{exa_gg}), define the dual
$\chan^\dual:\mc B_\out^\dual \to \mc B^\dual$ of the channel map
$\chan: \mc B \to \mc B_\out$ by
\begin{align}
\bk{\chan^\dual L} (A)  &= L(\chan A) \quad
\forall A \in \mc B, L \in \mc B_\out^\dual.
\end{align}
$\chan^\dual$ conserves the mean, in the sense that 
\begin{align}
\expect\Bk{\bk{\chan^\dual L}(X)}
&= \bk{\chan^\dual L}(f) = L(\chan f) = \expect[L(Y)].
\end{align}
Similarly, in the quantum case (Example~\ref{exa_qgg}), define
$\chan^\dual:\mc K_\out \to \mc K$ as the adjoint of
$\chan:\mc K \to \mc K_\out$ with respect to the phase spaces $\mc K$
and $\mc K_\out$. Then $\chan^\dual$ also conserves the mean in the
sense that
\begin{align}
\trace\Bk{\rho(f) X(\chan^\dual A)} = \Avg{f,\chan^\dual A}_{\mc K}
= \Avg{\chan f,A}_{\mc K_\out} = \trace\Bk{\rho_\out(\chan f) Y(A)}.
\end{align}
Following a procedure similar to those in Secs.~\ref{sec_cc} and
\ref{sec_qq}, we can treat $\chan^\dual$ as a map
$\chan^\dual:\Amb(g) \to \Amb(f)$. Defining
$\chan_\dual = \chan^{\dual\adj}:\Amb(f) \to \Amb(g)$ where $\adj$
denotes the adjoint with respect to the inner products of $\Amb(f)$
and $\Amb(g)$, we obtain
\begin{align}
\chan_\push &= \chan_\dual \inj_{\mc T(f)},
&
\chan^\pull &= \pi_{\mc T(f)} \chan^\dual,
&
\imap &= \pi_{\mc T(f)} \IMAP \inj_{\mc T(f)},
&
\IMAP &= \chan^{\dual} \chan_\dual.
\end{align}
$\chan_\dual$ has the formal expression
\begin{align}
\chan_\dual &= \Sigma_\out^{-1}\chan \Sigma,
\end{align}
so that $\chan_\push$ is consistent with the relation between the
input and output scores given by Eq.~(\ref{pushforward_gauss}).

Let us now assume that the input tangent space is maximal
($\mc T(f) = \Amb(f)$) and Eq.~(\ref{effinf_range}) in
Theorem~\ref{thm_effinfout} is satisfied, so that $\effinfout$ exists
and $\Bnd_\out < \infty$. Notice that all the inner products as well
as the local maps $\chan_\push = \Sigma_\out^{-1} \chan \Sigma$ and
$\chan^\pull = \chan^\dual$ do not depend on $\theta$. Consider the
classical case first and assume the parameter of interest
$\beta(\theta) = b(\theta)$ with $b \in \covarspace$ and
$\effinf = b$, following Sec.~\ref{sec_gauss}.  Then
\begin{align}
\effinfout &= \chan^{\pull-} \effinf = 
\chan^{\dual-} b
\label{effinfout_gauss}
\end{align}
also does not depend on $\theta$. The bound becomes
\begin{align}
\Bnd_\out &= \norm{\effinfout}_{\mc V_\out}^2
= \bk{\Sigma_\out \effinfout} \bk{\effinfout},
\end{align}
which is the variance of $(\effinfout)(Y)$.  Since
$\chan^\pull \effinfout = \chan^{\dual}\effinfout = \effinf = b$ and
$\chan^\dual$ conserves the mean, the estimator
\begin{align}
\check\beta(Y) &= \bk{\effinfout}(Y)= \bk{\chan^{\dual-} b}(Y)
\end{align}
in terms of the output Gaussian field $Y$ is unbiased as well as
efficient. Similarly, in the quantum case, let the parameter of
interest be $\beta(\theta) = \Avg{b,\theta}_{\mc K}$ with
$\effinf = b$, following Sec.~\ref{sec_gauss}. Then $\effinfout$ also
satisfies Eq.~(\ref{effinfout_gauss}) and does not depend on
$\theta$. The bound becomes
\begin{align}
\Bnd_\out &= \norm{\effinfout}_{\mc V_\out}^2 
= \Avg{\effinfout,\Sigma_\out\effinfout}_{\mc K_\out},
\end{align}
implying that the quadrature
\begin{align}
Y\bk{\effinfout}
&= Y\bk{\chan^{\dual-} b}
\label{eff_quad}
\end{align}
of the mode specified by the phase-space vector
\begin{align}
\effinfout &= 
\chan^{\dual-} b \in \mc K_\out
\label{eff_mode}
\end{align}
is an unbiased and efficient estimator. Hence, for Gaussian states,
Gaussian channels, and a scalar linear functional of the mean vector
as the parameter of interest, there is little difference between the
quantum theory and the classical one---simply use the covariance maps
of the Gaussian states in place of the classical covariance maps and a
modal quadrature in place of the linear estimator.

\section{\label{sec_coh}Classical coherent imaging}

\subsection{\label{sec_coh_spaces}Functional spaces for additive white Gaussian noise}
We now apply the abstract formalism to a concrete application:
classical coherent imaging, as illustrated by Fig.~\ref{coh_imaging}.
We follow the classical Gaussian shift model in
Examples~\ref{exa_gauss} and \ref{exa_gg} as well as
Secs.~\ref{sec_gauss} and \ref{sec_gg}. Let the input signal $X(x)$ be
a real scalar field on the object plane $x \in D \subseteq \mb R^m$.
A vector field or a complex field can be treated as direct sums of
scalar fields, although we do not consider those in this section.  For
simplicity, let $X$ be deterministic and given by
$X = f(\theta) = \theta \in \Sigma \mc V$. Let the semi-inner product
of $\mc V$ be
\begin{align}
\Avg{A,B}_{\mc V} &= \int_D A(x) B(x) \dd^m x,
\end{align}
meaning that $\Amb(f)$ is the standard Hilbert space for real
finite-energy fields on $D$, commonly denoted as $L^2(D)$.  The
tangent vectors $\dot\phi$, $\score(\dot\phi)$, and $f_\push \dot\phi$
are all representations of a perturbation of the field $\theta$ by a
function. We also express a linear functional $L \in \mc B^\dual$ in
terms of a function $L:D \to \mb R$ as
\begin{align}
L(A) &= \int_{D} L(x) A(x) \dd^m x,
\end{align}
where we have abused notation by using the same symbol to denote both
the functional and its associated function on $D$.  Following
Sec.~\ref{sec_gauss}, we assume the parameter of interest
\begin{align}
\beta(\theta) &= b(\theta) = \int_D b(x) \theta(x) \dd^m x
\end{align}
in terms of a weight function $b(x)$. For example, one may wish to
estimate a Fourier coefficient of $\theta(x)$ by setting $b(x)$ to be
a sinusoidal function, or study the behavior of $\theta(x)$ near a
point $x_0$ by setting $b(x)$ to be a narrow function concentrated at
$x_0$.  Assume further that the input score tangent space $\mc T(f)$
is maximal, so that the input efficient influence is given
by $\effinf = b$.

\fig{0.6}{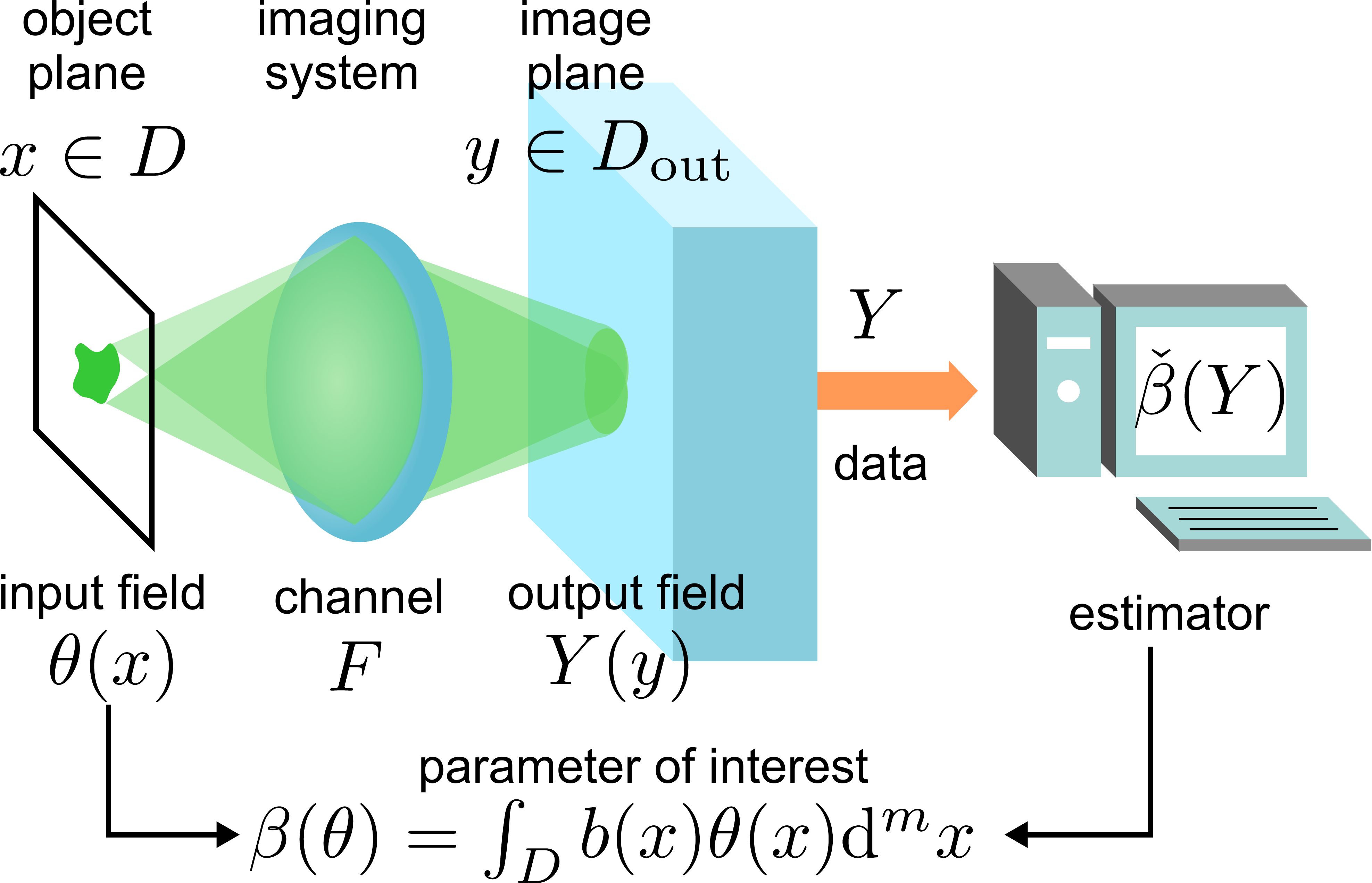}{Illustration of a classical imaging system.}

Let the output signal $Y \in \mc B_\out$ be a real scalar field on the
image plane $D_\out \subseteq \mb R^m$. We express a linear functional
$L \in \mc B_\out^\dual$ in terms of a function $L:D_\out \to \mb R$
as
\begin{align}
L(A) &= \int_{D_\out} L(y) A(y) \dd^m y.
\end{align}
The input-output relation given by Eq.~(\ref{io}) gives
\begin{align}
Y &= \chan \theta + Z, & \Sigma_\out &= \Sigma_\noise.
\end{align}
The channel map $\chan:\Sigma\mc V\to\Sigma_\out\mc V_\out$ models the
imaging system. Let the additive noise $Z$ be zero-mean and white.
Adopting the more intuitive calculus of physics and engineering, we
write its covariance function as
\begin{align}
\expect\Bk{Z(y)Z(y')} &= \epsilon\delta^m(y-y'),
\label{white}
\end{align}
where $\epsilon$ quantifies the noise level and $\delta^m$ is the
Dirac delta function in $m$ dimensions. The semi-inner product of
$\mc V_\out$ becomes
\begin{align}
\Avg{A,B}_{\mc V_\out} &= \epsilon\int_{D_\out} A(y) B(y) \dd^my.
\label{inner_Hout}
\end{align}
Apart from the prefactor $\epsilon$, we see that $\Amb(g)$ is also the
standard Hilbert space $L^2(D_\out)$. 

\subsection{Functional expressions of channel maps}
Let the channel map $\chan$ be
\begin{align}
\bk{\chan A}(y) &= \int_D \chan(y|x) A(x) \dd^m x,
\end{align}
where $\chan(y|x)$ is a point-spread function (PSF) for the
fields. Then its dual
$\chan^\dual:\mc B_\out^\dual \to \mc B^\dual$ obeys
\begin{align}
L(\chan A)
&= \int_{D_\out} L(y) \int_D \chan(y|x) A(x) \dd^m x \dd^m y
= \int_D \Bk{\int_{D_\out} L(y) \chan(y|x) \dd^m y} A(x) \dd^m x
= \bk{\chan^\dual L}(A).
\end{align}
It follows that the function $\chan^\dual L:D \to \mb R$ 
associated with the linear functional $\chan^\dual L \in \mc B^\dual$ obeys
\begin{align}
\bk{\chan^\dual L}(x) &= \int_{D_\out} L(y) \chan(y|x) \dd^m y
\end{align}
in terms of the function $L:D_\out \to \mb R$ associated with
$L \in \mc B_\out^\dual$. 

For the sake of intuition and concreteness, we express elements of
abstract Hilbert spaces as functions. The channel pushforward and
pullback become
\begin{align}
\bk{\chan_\push u}(y) &= \bk{\chan_\dual u}(y) = \frac{1}{\epsilon} \int_D \chan(y|x) u(x)\dd^m x,
\label{chan_push_coh}
\\
\bk{\chan^\pull u}(x) &= \bk{\chan^\dual u}(x) = \int_{D_\out} u(y) \chan(y|x) \dd^m y.
\label{chan_pull_coh}
\end{align}
Then the condition given by Eq.~(\ref{effinf_range}) on $\effinf$ becomes
\begin{align}
\exists u \in \Amb(g): 
\bk{\chan^\pull u}(x) = 
\int_{D_\out} u(y) \chan(y|x) \dd^m y = (\effinf)(x) = b(x).
\end{align}
Provided that $b$ satisfies the condition, we can solve for
the output efficient influence
\begin{align}
(\effinfout)(y) &= \bk{\chan^{\pull - } b}(y)= \bk{\chan^{\dual - } b}(y)
\end{align}
and compute the bound
\begin{align}
\Bnd_\out &= \norm{\effinfout}_{\mc V_\out}^2 
= \epsilon \int_{D_\out} \Bk{(\effinfout)(y)}^2 \dd^m y.
\end{align}
Following Sec.~\ref{sec_gg}, we find that
\begin{align}
\check\beta(Y) &= \int_{D_\out} \bk{\effinfout}(y) Y(y) \dd^m y
\label{eff_est_coh}
\end{align}
is an unbiased and efficient estimator.

\subsection{Functional SVD}
Assuming that $\chan_\push$ is compact, its SVD, as explained in
Sec.~\ref{sec_svd}, characterizes the behavior of the channel in
determining the efficient influence $\effinfout$ and the efficiency
bound $\Bnd_\out$. The input singular vectors $\{\insing_j\}$ and the
output singular vectors $\{\outsing_j\}$ become functions that satisfy
the orthonormal relations
\begin{align}
\int_D \insing_j(x)\insing_k(x) \dd^m x &= \delta_{jk},
&
\epsilon \int_{D_\out} \outsing_j(y)\outsing_k(y) \dd^m y &= \delta_{jk}.
\end{align}
Each pair of singular vectors satisfy $\chan_\push \insing_j = s_j \outsing_j$,
which becomes the integral equation
\begin{align}
\frac{1}{\epsilon} \int_D \chan(y|x) \insing_j(x) \dd^m x &= s_j \outsing_j(y).
\label{SVD_integral}
\end{align}
The local channel maps can then be expressed as the expansions
\begin{align}
\bk{\chan_\push u}(y)
&= \sum_j s_j \outsing_j(y) \int_D \insing_j(x) u(x) \dd^m x,
\\
\bk{\chan^\pull u}(x)
&= \sum_j s_j \insing_j(x) \epsilon \int_{D_\out} \outsing_j(y) u(y) \dd^m y,
\\
\bk{\imap u}(x) &= \sum_j s_j^2 \insing_j(x) \int_D \insing_j(x') u(x') \dd^m x'.
\end{align}
If $\effinf = b \in \spn\{\insing_j\}$, $\effinfout$ is guaranteed to exist
by Eq.~(\ref{span_insing}) and Theorem~\ref{thm_effinfout}.  We can
then express the pseudo-inverses as
\begin{align}
\bk{\effinfout}(y) &= 
\bk{\chan^{\pull-} b}(y) 
= \sum_j \frac{1}{s_j} \outsing_j(y) \int_D \insing_j(x) b(x) \dd^m x,
\label{chan_pull_coh2}
\\
\bk{\imap^- \effinf}(x) &= \sum_j \frac{1}{s_j^2} \insing_j(x) 
\int_D \insing_j(x') b(x') \dd^m x',
\label{imap_pinv_coh}
\end{align}
and the output efficiency bound as
\begin{align}
\Bnd_\out &= \sum_j \frac{1}{s_j^2} \Bk{\int_D \insing_j(x) b(x) \dd^m x}^2.
\label{Bnd_coh}
\end{align}

\subsection{\label{sec_scaling_c}Scaling}
Suppose that $\chan$ is scaled by a given constant $c > 0$ as per
\begin{align}
\chan(\cdot,c) &= c \chan(\cdot,1).
\end{align}
The constant can be used to model the strength of the input field or
the total transmission of the channel. The dependence of
$(\effinfout)(\cdot,c,\epsilon)$ and $\Bnd_\out(\cdot,c,\epsilon)$ on
$c$ and the noise level $\epsilon$ can be expressed as
\begin{align}
\bk{\effinfout}(\cdot,c,\epsilon) &= \frac{1}{c} \bk{\effinfout}(\cdot,1,1),
&
\Bnd_\out(\cdot,c,\epsilon) &= \frac{\epsilon}{c^2}\Bnd_\out(\cdot,1,1).
\label{effinf_scale}
\end{align}
It therefore suffices to compute only $\effinfout(\cdot,1,1)$ and
$\Bnd_\out(\cdot,1,1)$ for some conveniently scaled $\chan$; results
for all other values of $c$ and $\epsilon$ follow immediately.

\subsection{\label{sec_coh_c_exa}Example}
The functional SVD has been studied extensively in superresolution
research \cite{villiers}. The most well known result may be
the discovery by Slepian and coworkers \cite{slepian83} that the
singular functions $\{\insing_j\}$ as well as $\{\outsing_j\}$
satisfying Eq.~(\ref{SVD_integral}) are the so-called prolate
spheroidal wavefunctions when
\begin{align}
m &= 1,
&
D &= D_\out = \Bk{-\Delta,\Delta},
\\
\chan(y|x) &= \sinc(y-x),
&
\sinc u &=
\begin{cases} (\sin \pi u)/(\pi u), & u \neq 0,\\1, & u = 0,\end{cases}
\label{chan_sinc}
\end{align}
where all lengths such as $\Delta$ are normalized with respect to the
width of the PSF. Furthermore, for any $\Delta < \infty$, all the
singular values are positive, and
\begin{align}
\overline{\spn\BK{\insing_j}} &= \bk{\kernel \chan_\push}^\perp = 
\overline{\range \chan^\pull}  = L^2(D),
\label{slepian_input_space}
\\
\kernel \chan_\push &= \BK{0},
\end{align}
meaning that any nonzero input to Eq.~(\ref{chan_push_coh}) will
produce a nonzero output.

In practice, analytic results are rare and one must often resort to
numerical analysis. Using the numerical method detailed in
Appendix~\ref{app_num}, Fig.~\ref{svd_coh_classical0} plots the SVD
for the sinc PSF in Eqs.~(\ref{chan_sinc}) and unit noise level
$\epsilon = 1$, with three different values of $\Delta$ for the object
plane $D = [-\Delta,\Delta]$ and a fixed image plane
$D_\out = [-10,10]$. We notice that the singular functions
$\{\insing_j\}$ and $\{\outsing_j\}$ are all oscillatory, while the
singular values $\{s_j\}$, though positive, decay rapidly with $j$
when $\Delta$ is small. In fact, some $\{\insing_j\}$ can be deemed
superoscillatory, given that they have multiple zeros within a small
region with size $\Delta$. This behavior of the functional SVD for a
low-pass filter is well known \cite{villiers}.

Less well known is the role of the SVD in determining the output
efficient influence $\effinfout$ and the output efficiency bound
$\Bnd_\out$ through Eqs.~(\ref{chan_pull_coh2})--(\ref{Bnd_coh}). To
determine $\effinfout$ and $\Bnd_\out$ from the weight function
$b(x)$, Eqs.~(\ref{chan_pull_coh2}) and (\ref{Bnd_coh}) both involve
the inner product $\avg{\insing_j,b}_{\mc V}$ with oscillatory
functions $\{\insing_j\}$.  This inner product is a Fourier transform
in essence. We will hereafter call any oscillatory function a Fourier
component, which need not be a sinusoidal function exactly, and any
inner product with an oscillatory function a Fourier coefficient.  To
obtain $\effinfout$ by Eq.~(\ref{chan_pull_coh2}), each Fourier
component $\insing_j$ with Fourier coefficient
$\avg{\insing_j,b}_{\mc V}$ is mapped to an output singular function
$\outsing_j$ and scaled by $1/s_j$ before all the terms are
summed. The purpose of the $1/s_j$ factor is to compensate for the
attenuation of that component by the channel. To obtain $\Bnd_\out$ by
Eq.~(\ref{Bnd_coh}) on the other hand, each Fourier coefficient
$\avg{\insing_j,b}_{\mc V}$ is squared and scaled by $1/s_j^2$ before
the terms are summed. Each singular value $s_j$ thus determines how
$\effinfout$ and $\Bnd_\out$ weigh a Fourier component; the components
with the smallest singular values are the most detrimental to the
estimation problem, as one expects for a low-pass filter. For example,
if we consider the first column of Fig.~\ref{svd_coh_classical0} for
$\Delta = 0.1$, then the effect of a component with low singular value
$s_j$, say, $j = 4$, on $\effinfout$ and $\Bnd_\out$ would be highly
amplified.

We stress that the values of $s_j$ in Fig.~\ref{svd_coh_classical0}
are only for the normalized channel map in Eq.~(\ref{chan_sinc});
results for a physical problem also involve the noise level $\epsilon$
and the transmission coefficient $c$ through
Eqs.~(\ref{effinf_scale}). Superresolution---in the sense of
accurately estimating high-order Fourier coefficients and their linear
combinations---is possible, depending on $\epsilon$ and $c$, and the
efficiency bound serves as a benchmark for the error.

\fig{}{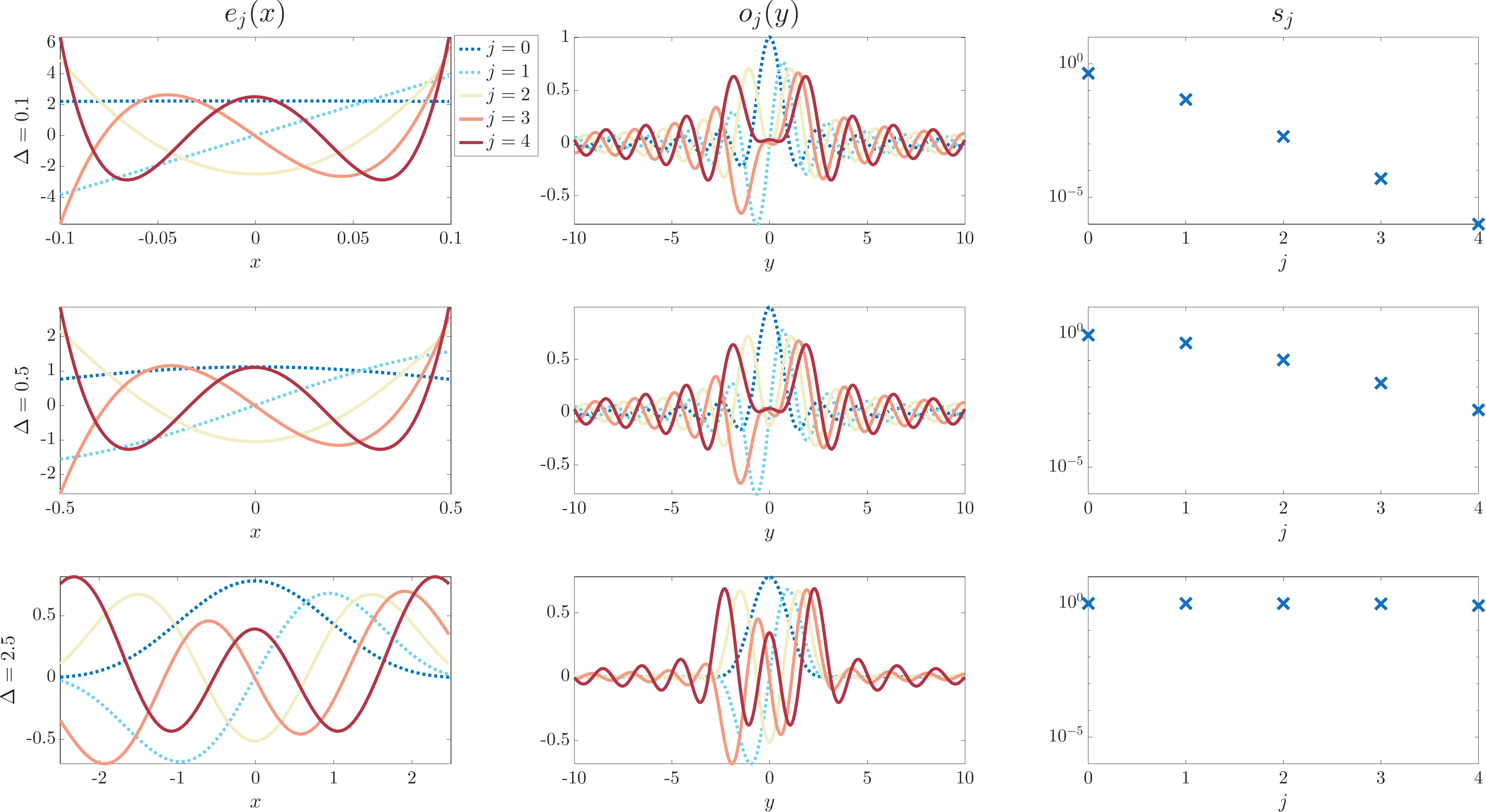}{Numerically computed singular-value
  decomposition (SVD) of the channel pushforward $\chan_\push = \chan$
  for classical coherent imaging, assuming the sinc PSF given by
  Eq.~(\ref{chan_sinc}) and additive white Gaussian noise with unit
  noise level $\epsilon = 1$. The three rows are results for three
  different object-plane sizes $\Delta = 0.1$, $0.5$, and $2.5$, while
  the image plane $D_\out = [-10,10]$ is fixed. First column: the
  first five input singular functions $\{\insing_j(x)\}$ versus the
  object-plane coordinate $x$. Second column: the first five output
  singular functions $\{\outsing_j(y)\}$ versus the image-plane
  coordinate $y$. Third column: the first five singular values
  $\{s_j\}$ in logarithmic scale as a function of $j$.  For all the
  plots in the third column, the ranges of the vertical axes are set
  to be identical to demonstrate the faster decay of the singular
  values for a smaller $\Delta$.}

The quantum theory is similar, except that complex optical fields
should be considered and the model should be treated with care to
respect quantum physics. As the analysis is more tedious, we delegate
it to Appendix~\ref{app_qcoh}. Compared with prior work on quantum
coherent imaging
\cite{yuen78,shapiro79,kolobov_fabre,beskrovnyy05,pinel12,taylor16,treps},
the semiparametric bound here does not offer any new surprise---it
shows that squeezing of an input spatial mode can improve the bound
and homodyne detection of an output mode can saturate the bound, just
as prior work suggested.  For a multidimensional $\beta$, one should
consider the Holevo--Nagaoka bound and general-dyne measurements
\cite{demkowicz20}.

\section{\label{sec_direct}Direct incoherent imaging}
As Poisson noise is fundamental in key applications of incoherent
imaging in astronomy \cite{goodman_stat,zmuidzinas03,review_cp} as
well as microscopy \cite{pawley}, the Poisson channels in
Examples~\ref{exa_pp} and \ref{exa_qpp} are natural models of the
problem. We first study the classical case of direct imaging in this
section.  Fig.~\ref{coh_imaging} illustrates the setup just as well,
except that the quantities need to be redefined.

\subsection{\label{sec_direct_funcs}Weighted functional spaces}
Follow the classical Poisson model in Example~\ref{exa_poisson} and
Sec.~\ref{sec_poisson} for the input. Let the input be a Poisson field
$X$ on the object plane $\sspace = D \subseteq \mb R^m$ with mean
measure $\mean X_\theta$ and mean density
\begin{align}
f(\theta,x) = {\dv{\mean X_\theta}{\mu}}(x) = \theta(x).
\label{input_mean}
\end{align}
Let $\dd\mu(x) = \dd^m x$; then $\theta(x)$ is the mean optical
intensity function at the input. Set the input ambient vector space
$\mc V$ as the space of real functions with the semi-inner product
\begin{align}
\Avg{A,B}_\theta &= \int_D A(x) B(x) \theta(x) \dd^m x.
\label{inner_incoh}
\end{align}
Unlike the semi-inner products for the Gaussian-noise model in
Sec.~\ref{sec_coh}, the semi-inner product here is weighted by the
input intensity $\theta$. Assume the parameter of interest
\begin{align}
\beta(\theta) &= \int_D b(x) \theta(x) \dd^m x
\label{beta_incoh}
\end{align}
in terms of a weight function $b(x)$ on the object plane and a maximal
input score tangent space $\mc T(f) = \Amb(f)$. Sec.~\ref{sec_poisson}
has shown that the input efficient influence is
\begin{align}
\bk{\effinf}(x) &= b(x).
\label{effinf_b}
\end{align}
To model the imaging system, assume the classical Poisson channel in
Example~\ref{exa_pp} and Sec.~\ref{sec_cc}.  Let the output be another
Poisson field $Y$ on the image plane
$\sspace_\out = D_\out \subseteq \mb R^m$ with mean measure
$\mean Y_\theta$ and mean density
\begin{align}
g(\theta,y) &= 
{\dv{\mean Y_\theta}{\nu}}(y) = \int_D \chan(y|x) \theta(x) \dd^m x,
\end{align}
where $\chan(y|x)$ is the PSF of the imaging system for the
intensities. Let $\dd\nu = \dd^m y$; then $g(\theta,y)$ is the mean
optical intensity on the image plane. This model is an idealization of
the measurement by an image sensor; it is called direct imaging in the
quantum literature \cite{tnl,review_cp,lvovsky26}.  Set the output
ambient vector space $\mc V_\out$ as the space of real functions with
the semi-inner product
\begin{align}
\Avg{A,B}_g &= \int_{D_\out} A(y)B(y) g(\theta,y) \dd^m y.
\end{align}
We continue to express elements of abstract Hilbert spaces as
functions for intuition. Following Sec.~\ref{sec_cc}, the channel
pushforward and pullback become
\begin{align}
\bk{\chan_\push u}(y) &= \bk{\chan_\dual u}(y) = \frac{1}{g(\theta,y)}
\int_D \chan(y|x) \theta(x) u(x) \dd^m x,
\label{push_incoh}
\\
\bk{\chan^\pull u}(x) &= \bk{\chan^\dual u}(x)  = \int_{D_\out}  u(y)\chan(y|x) \dd^m y.
\end{align}
The condition on $\effinf$ given by Eq.~(\ref{effinf_range}) for the
existence of $\effinfout$ becomes
\begin{align}
\exists u \in \Amb(g): 
\bk{\chan^\pull u}(x) &= \int_{D_\out} u(y) \chan(y|x) \dd^m y = (\effinf)(x) = b(x).
\label{effinf_range_direct}
\end{align}
Under the condition, the output efficient influence
$(\effinfout)(y) = (\chan^{\pull-}b)(y)$ exists and is a function of
the image-plane coordinate $y$.  With
$\effinfout = (\effinfout)(\check\theta)$ at a preliminary estimate
$\check\theta$, the one-step estimator in Eq.~(\ref{onestep_poisson})
becomes
\begin{align}
\check\beta(Y) &= \int_{D_\out} (\effinfout)(\check\theta,y) \dd Y(y).
\label{onestep_direct}
\end{align}

\subsection{\label{sec_incoh_svd}Functional SVD}
The functional SVD of $\chan_\push$ should be evaluated with respect
to the weighted inner products of $\Amb(f)$ and $\Amb(g)$. The
orthonormal relations for the singular vectors become
\begin{align}
\int_D \insing_j(x)\insing_k(x) \theta(x)  \dd^m x  &= \delta_{jk},
&
\int_{D_\out} \outsing_j(y)\outsing_k(y) g(\theta,y)  \dd^m y  &= \delta_{jk}.
\end{align}
They obey the integral equation
\begin{align}
\frac{1}{g(\theta,y)} \int_D \chan(y|x) \theta(x) \insing_j(x) \dd^m x
&= s_j \outsing_j(y).
\end{align}
The functional SVDs can then be expressed as
\begin{align}
\bk{\chan_\push A}(y) &= \sum_j s_j \outsing_j(y) \int_D \insing_j(x) A(x) \theta(x) \dd^m x,
\\
\bk{\chan^\pull A}(x) &= \sum_j s_j \insing_j(x) \int_{D_\out} \outsing_j(y) A(y) g(\theta,y) \dd^m y,
\\
\bk{\imap A}(x) &= \sum_j s_j^2 \insing_j(x) 
\int_D \insing_j(x') A(x') \theta(x') \dd^m x'.
\end{align}
If $\effinf = b \in \spn\{\insing_j\}$, $\effinfout$ exists by
Eq.~(\ref{span_insing}) and Theorem~\ref{thm_effinfout}. Then
\begin{align}
\bk{\effinfout}(y) &= \bk{\chan^{\pull-} b}(y) 
= \sum_j \frac{1}{s_j} \outsing_j(y) \int_D \insing_j(x) b(x) \theta(x) \dd^m x,
\label{pinv_incoh_c}
\end{align}
leading to the efficiency bound
\begin{align}
\Bnd_\out &= \sum_j \frac{1}{s_j^2} \Bk{\int_D \insing_j(x) b(x) \theta(x) \dd^m x}^2.
\label{Bnd_incoh_c}
\end{align}
The key difference between the incoherent-imaging problem here and the
coherent case studied in Sec.~\ref{sec_coh} is that the inner products
here are weighted and depend on the true input intensity $\theta$, so
$\effinfout$, $\Bnd_\out$, and the SVD all depend on $\theta$ and have
to be evaluated for each value of $\theta$. This $\theta$ dependence
is a natural consequence of the signal-dependent nature of Poisson
noise.

It is noteworthy that Bertero and coworkers have performed pioneering
studies of SVD for confocal incoherent imaging by modeling the input
and output intensities as functions in $L^2(\mb R^m)$ with unweighted
inner products \cite{bertero89a,bertero91,villiers}. Their approach
does not account for noise explicitly, however, and it is unclear how
the physical constraint of nonnegative intensities affects their
results. Here, noise, nonnegative intensities, and statistics are all
naturally integrated in our theory.

\subsection{Scaling}
The issue of scaling is similar to Secs.~\ref{sec_scaling_c} and
\ref{sec_scaling_q}, except that there is no separate noise level
$\epsilon$ in the Poisson model.  Suppose that the object field is
scaled by a constant $c > 0$, such that
\begin{align}
f(\theta,x,c) &= c\theta(x).
\end{align}
Then the inner products of $\Amb(f)$ and $\Amb(g)$ are also
proportional to $c$, leading to
\begin{align}
\bk{\effinf}(\cdot,c) &= \frac{1}{c} \bk{\effinf}(\cdot,1) ,
&
\bk{\effinfout}(\cdot,c) &= \frac{1}{c} \bk{\effinfout}(\cdot,1) ,
&
\Bnd_\out(\cdot,c) &= \frac{1}{c} \Bnd_\out(\cdot,1).
\end{align}
The same scalings for $\effinfout$ and $\Bnd_\out$ hold when the
channel map $\chan$ is scaled by a constant $c$.

\subsection{\label{sec_direct_exa}Examples}
Assume one-dimensional object and image planes ($m = 1$) with
$x \in D = \mb R$ and $y \in D_\out = \Bk{-\Delta_\out,\Delta_\out}$.
We first consider the Gaussian PSF
\begin{align}
\chan(y|x) &= \abs{\psi(y-x)}^2 = 
\frac{1}{\sqrt{2\pi}} \exp\Bk{-\frac{1}{2}\bk{y-x}^2},
&
\psi(y) &= 
\frac{1}{(2\pi)^{1/4}} \exp\bk{-\frac{y^2}{4}}.
\label{gauss_PSF}
\end{align}
Following Appendix~\ref{app_num}, a numerical computation of the SVD
of $\chan_\push$ can be performed for a given input intensity
$\theta(x)$. Fig.~\ref{svd_incoh_classical_gauss0} plots the results
for the flat-top function
\begin{align}
\theta(x) = \begin{cases}
1/(2\Delta), & |x| \le \Delta, \\ 0, & \textrm{otherwise},
\end{cases}
\label{flattop}
\end{align}
with three different object sizes $\Delta$ and image-plane sizes given
by $\Delta_\out = \Delta+5$.  A noteworthy distinction from the
coherent case in Sec.~\ref{sec_coh} is that, instead of introducing a
finite size for the object plane in the coherent case, here we have
the size of the object itself limiting the domain of $x$ in the
weighted inner product. The plots resemble those for coherent imaging
in Fig.~\ref{svd_coh_classical0}, at least qualitatively.  As
$\effinfout$ and $\Bnd_\out$ are dominated by components
$\insing_j\Avg{\insing_j,\effinf}_\theta$ of $\effinf$ with low
singular values, Fig.~\ref{svd_incoh_classical_gauss0} shows that the
high-frequency components with low singular values should contribute
more to the error, similar to the coherent-imaging case in
Sec.~\ref{sec_coh_c_exa}.

\fig{}{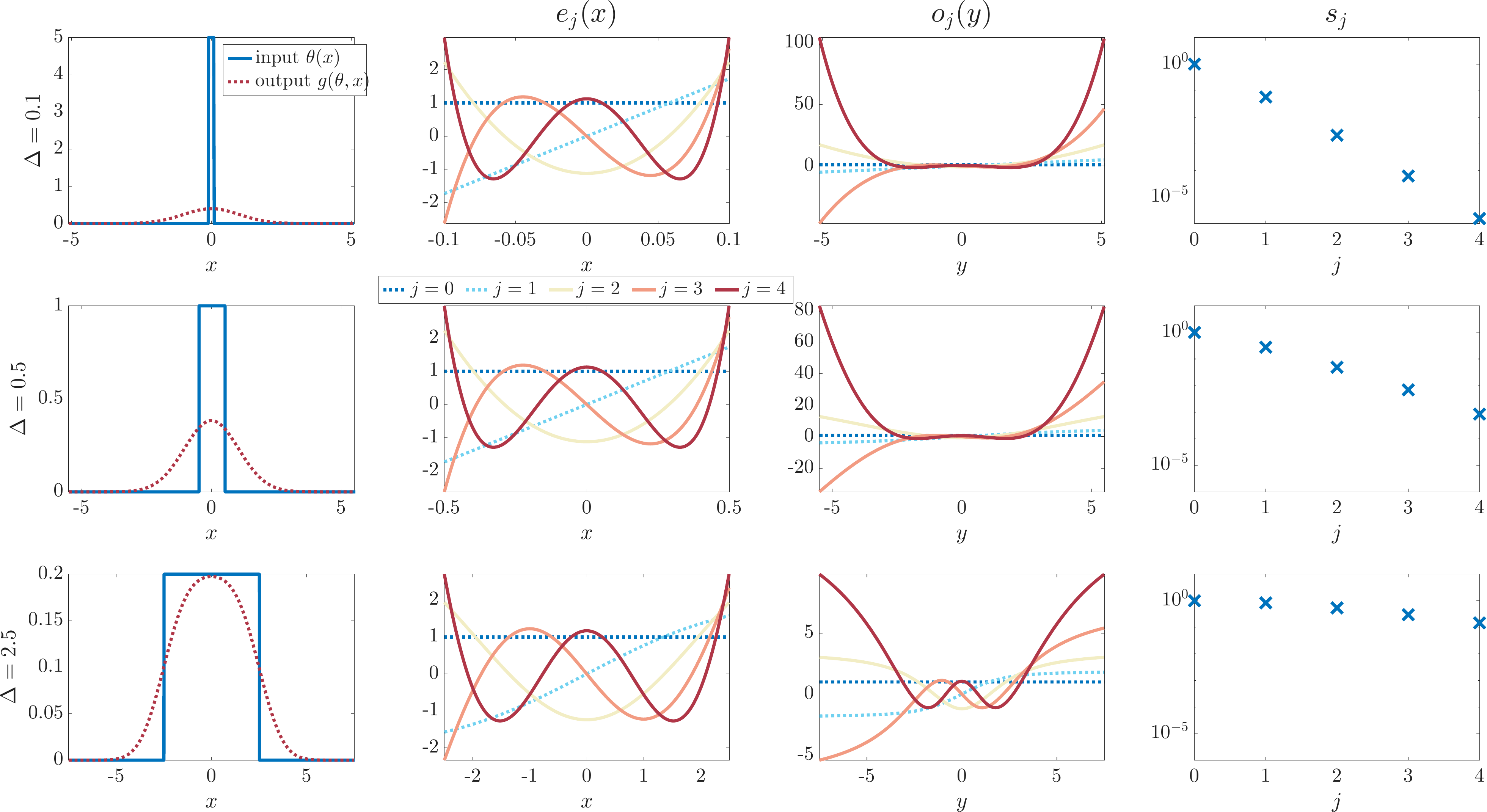}{SVD of the channel pushforward
  $\chan_\push = \chan_\dual$ for direct incoherent imaging, assuming
  the flat-top input intensity $\theta(x)$ given by
  Eq.~(\ref{flattop}), the Gaussian PSF given by
  Eq.~(\ref{gauss_PSF}), and Poisson noise. The three rows are results
  for three different object sizes $\Delta = 0.1$, $0.5$, and $2.5$,
  while the image-plane size is $\Delta_\out = \Delta + 5$. First
  column: input intensity $\theta(x)$ versus the object-plane
  coordinate $x$ and output intensity $g(\theta,x)$ versus the
  image-plane coordinate $x$. Second column: the first five input
  singular functions $\{\insing_j(x)\}$ versus object-plane coordinate
  $x$. Third column: the first five output singular functions
  $\{\outsing_j(y)\}$ versus image-plane coordinate $y$. Fourth
  column: the first five singular values $\{s_j\}$ in logarithmic
  scale as a function of $j$. }

To demonstrate that the SVD is $\theta$-dependent,
Fig.~\ref{svd_incoh_classical_gauss_tri0} plots the SVD for a triangle
input intensity
\begin{align}
\theta(x) &= 
\begin{cases}\bk{x + \Delta}/(2\Delta^2), & |x| \le \Delta,\\
0, & \textrm{otherwise}.
\end{cases}
\label{tri}
\end{align}
The singular functions are noticeably different from those in
Fig.~\ref{svd_incoh_classical_gauss0}, although the rough trends of the
singular values remain similar.

\fig{}{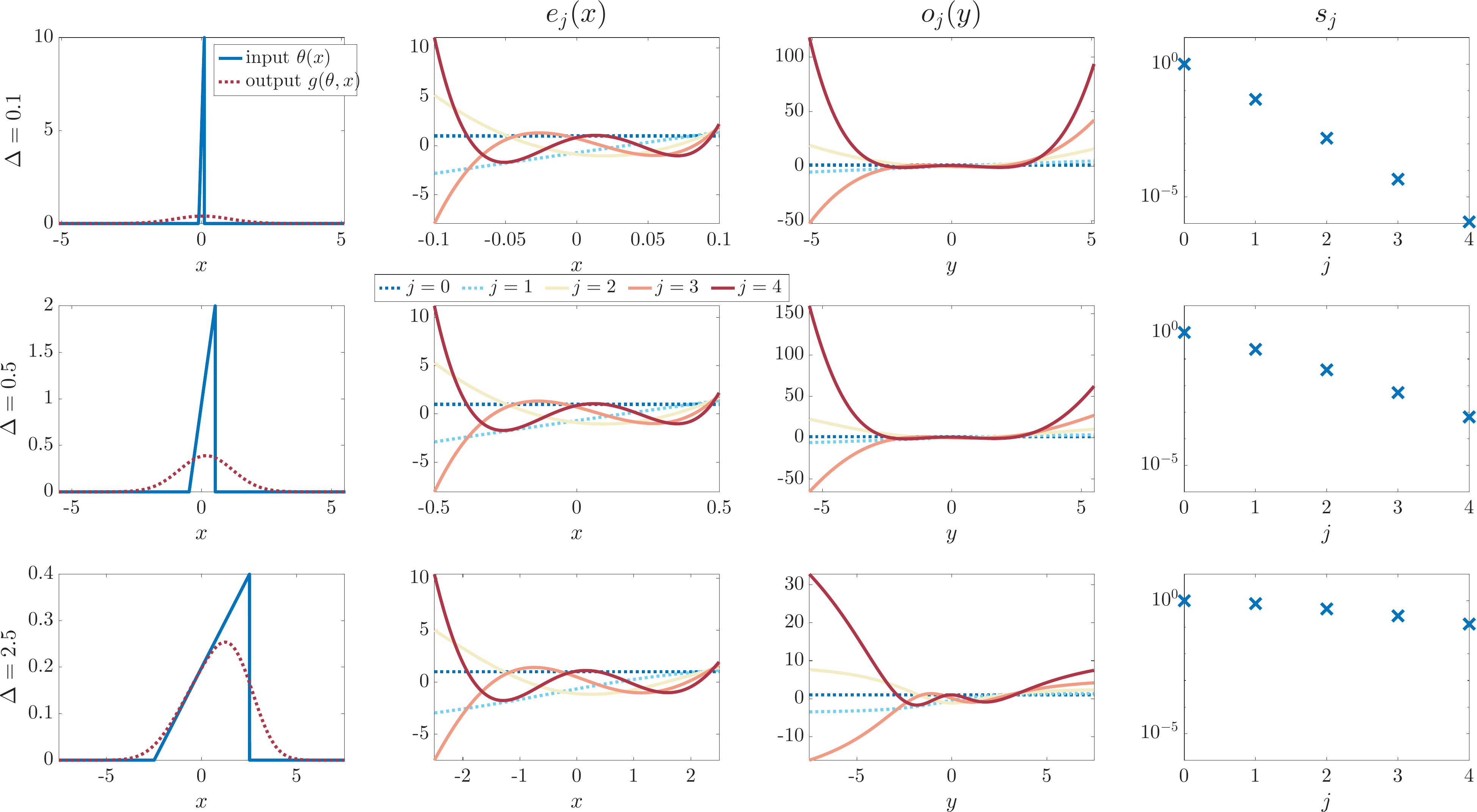}{SVD of the channel pushforward
  $\chan_\push = \chan_\dual$ for direct incoherent imaging, assuming
  the triangle input intensity given by Eq.~(\ref{tri}), the Gaussian
  PSF given by Eq.~(\ref{gauss_PSF}), and Poisson noise. The figure
  formats are the same as those of
  Fig.~\ref{svd_incoh_classical_gauss0}.}

For another example, Fig.~\ref{svd_incoh_classical_sinc0} plots the SVD
for the $\sinc^2$ PSF
\begin{align}
\chan(y|x) &= |\psi(y-x)|^2 = K \sinc^2\Bk{K(y-x)},
&
\psi(y) &= \sqrt{K} \sinc(K y),
&
K &= \frac{\sqrt{3}}{2\pi},
\label{sinc2}
\end{align}
and the flat-top input intensity given by Eq.~(\ref{flattop}).  (We
set this $K$ so that the root-mean-square bandwidth
$\pi K/\sqrt{3} = 1/2$ of $\psi(y)$ is the same as that of the
Gaussian $\psi$.)  The plots resemble those in
Figs.~\ref{svd_incoh_classical_gauss0}, except that the output
singular functions $\outsing_j(y)$ have peaks near the minima of
$g(\theta,y)$. Recall that, if the weight function is
$b(x) = \insing_j(x)$, then the efficient influence is
$(\effinfout)(y) = \outsing_j(y)/s_j$, to be used in the one-step
estimator given by Eq.~(\ref{onestep_direct}). The great weights
assigned by $\outsing_j(y)$ to the minima of $g(\theta,y)$ thus
suggests the importance of the PSF zeros in the estimation. We note
that introducing zeros to a PSF can enhance certain tasks of
incoherent imaging under Poisson noise
\cite{paur18,paur19,booth26,darekar26}, although this phenomenon is
outside the scope of this work.

\fig{}{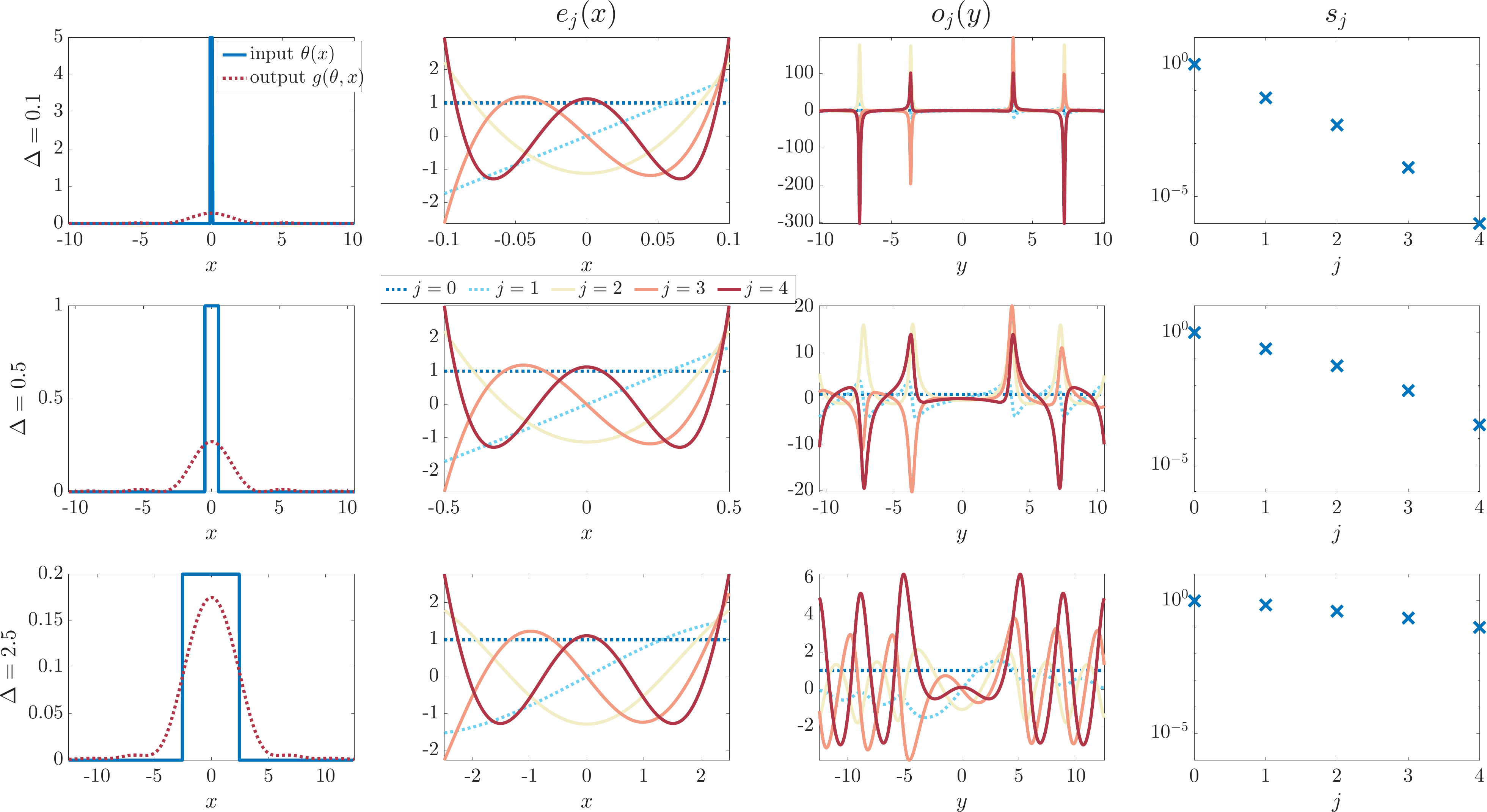}{SVD for direct incoherent imaging,
  assuming the flat-top input intensity given by Eq.~(\ref{flattop}),
  the sinc squared PSF given by Eq.~(\ref{sinc2}), and Poisson
  noise. The figure formats are the same as those of
  Figs.~\ref{svd_incoh_classical_gauss0} and
  \ref{svd_incoh_classical_gauss_tri0}, except that the image-plane
  size is set to $\Delta_\out = \Delta + 10$ to include more minima of
  the image.}

\section{\label{sec_incoh_q}Quantum model of incoherent imaging}
Helstrom pioneered the application of quantum detection and estimation
theory to incoherent imaging \cite{helstrom}. Instead of fixating on
one measurement such as direct imaging, his theory allows one to
compute fundamental quantum error bounds for any measurement from a
set of density operators that model the quantum state of light. We
build on his foundation with the following advances:
\begin{enumerate}
\item Instead of assuming the exact thermal state, we approximate it
  by the Poisson state described in Example~\ref{exa_poisson}.  The
  intensity operator of the Poisson state is determined by the mutual
  coherence matrix in statistical optics \cite{poisson_quantum}.  As
  proposed and developed in
  Refs.~\cite{stellar,tnl,review_cp,poisson_quantum}, the Poisson
  approximation is accurate at optical frequencies and allows us to
  relate the quantum model directly to the well established
  Poisson-noise model in classical incoherent imaging
  \cite{goodman_stat,zmuidzinas03,pawley}.

\item Instead of assuming a finite set of hypotheses or a
  low-dimensional parameter, we take the nonparametric model in
  Sec.~\ref{sec_poisson} as the input.

\item We observe significant gaps between the quantum bounds and the
  direct-imaging bounds for the estimation of high-order Fourier
  coefficients when the object size is subdiffraction.  The gaps
  suggest that alternative measurements can enhance incoherent imaging
  significantly.

\item In Sec.~\ref{sec_spade}, we study the performance of a simple
  version of SPADE proposed in Ref.~\cite{superosc_ieee},
  demonstrating numerically that it can come closer to the quantum
  bounds and offer superior performance to direct imaging.

\end{enumerate}
Fig.~\ref{q_imaging} illustrates the setup.

\fig{0.6}{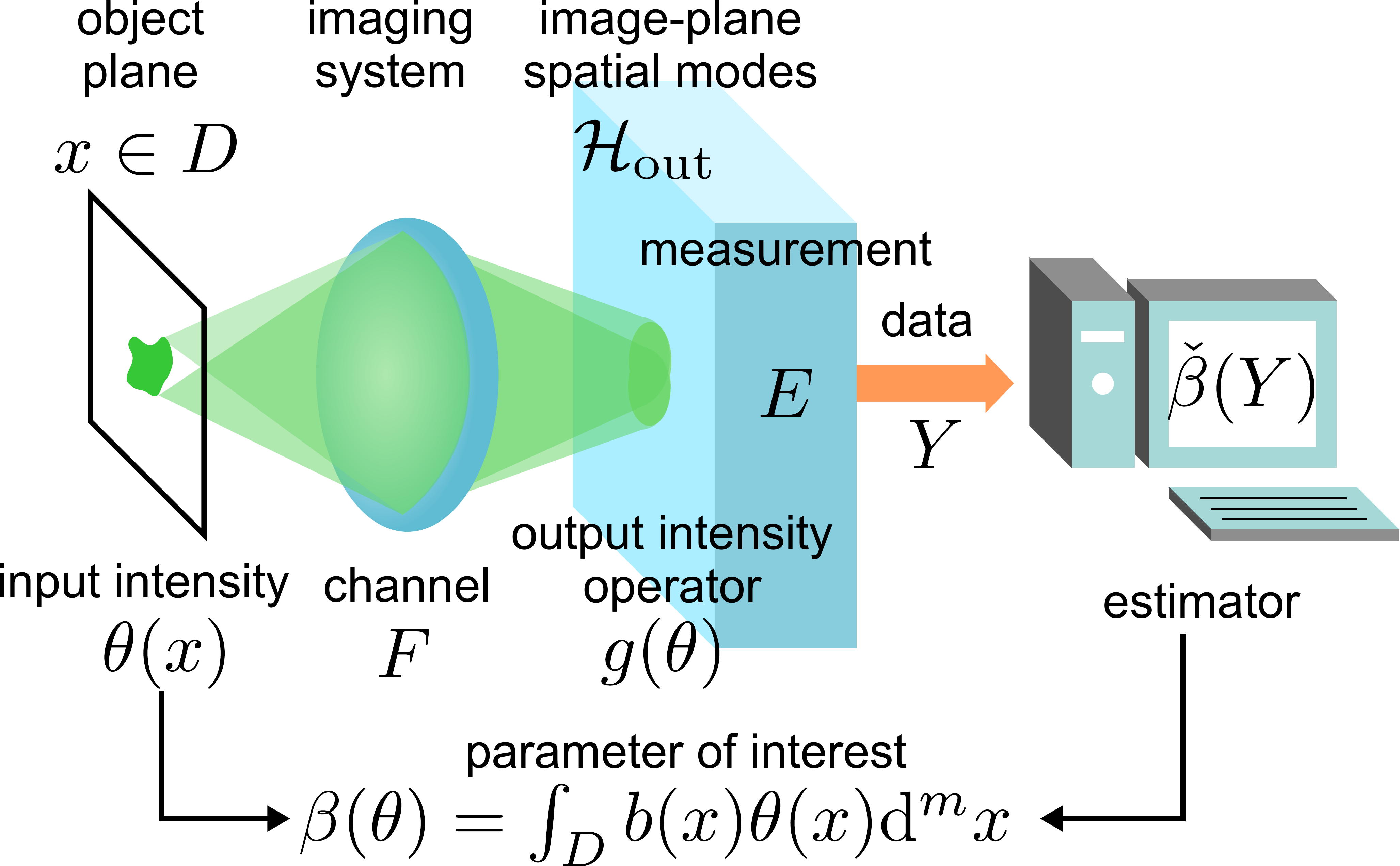}{Quantum model of incoherent imaging.}

\subsection{\label{sec_incoh_q_spaces}Operator spaces}
We follow Secs.~\ref{sec_poisson} and \ref{sec_direct_funcs} and
continue to assume Eqs.~(\ref{input_mean})--(\ref{effinf_b}) for the
object-plane quantities and the parameter of interest. To model a
diffraction-limited imaging system, follow Examples~\ref{exa_qpoisson}
and \ref{exa_qpp} and assume a classical-quantum channel $\chan$ that
produces an intensity operator $g(\theta)$ on $\mc H_\out$ for the
image-plane optical field given by \cite{poisson_quantum}
\begin{align}
g(\theta) &= \chan \theta = 
\int_D e^{-ikx} \ket{\psi}\bra{\psi} e^{ikx} \theta(x) \dd^m x,
\label{cq}
\end{align}
where $\mc H_\out$ is the Hilbert space for the spatial modes on the
image plane, $\ket{\psi} \in \mc H_\out$ models the
diffraction-limited single-photon state for a centered point source,
and $k = \mqty(k_1 & \dots & k_m)$ is the vectoral momentum operator
such that 
\begin{align}
\exp(-ikx) = \exp(-i\sum_{j=1}^m k_j x_j)
\end{align}
is a displacement operator with displacement $x \in \mb R^m$. We call
$\ket{\psi}$ the point-spread state. Let the output ambient vector
space $\mc V_\out$ be the space of self-adjoint operators with the
semi-inner product
\begin{align}
\Avg{A,B}_g &= \trace\Bk{g \bk{A \jordan B}}.
\end{align}
In the following, we again write elements of abstract Hilbert spaces
as functions or operators for intuition. Following Sec.~\ref{sec_qq},
we can express the channel pushforward and pullback as
\begin{align}
g(\theta) \jordan \bk{\chan_\push u} &= \chan\bk{u\theta}
= \int_D e^{-ikx} \ket{\psi}\bra{\psi} e^{ikx} u(x) \theta(x) \dd^mx,
\\
\bk{\chan^\pull u}(x) &= \bk{\chan^\dual u}(x) = \bra{\psi} e^{ikx} u e^{-ikx}\ket{\psi}.
\label{pull_cq}
\end{align}
The SVDs of the local channel maps are straightforward generalizations
of those in Sec.~\ref{sec_incoh_svd}.

The case of direct imaging studied in Sec.~\ref{sec_direct} can be
treated as a measurement of the photon positions on the image plane,
resulting in an output Poisson field with mean intensity
\begin{align}
\bra{y}g(\theta)\ket{y} &= 
\int_D \abs{\bra{y} e^{-ikx}\ket{\psi}}^2 \theta(x)\dd^mx,
\end{align}
where $\{\ket{y}:y \in \mb R^m\}$ are the Dirac position kets that
satisfy $\braket{y}{y'} = \delta^m(y-y')$ and
$e^{-ikx}\ket{y} = \ket{y+x}$.  The method of SPADE, on the other
hand, entails passive linear optics followed by photon counting of
each output mode \cite{tnl,review_cp}. Its most general form can be
modeled by a positive operator-valued measure (POVM) $E$ on $\mc H_\out$,
such that the output is a Poisson field $Y$ with mean measure
\cite{poisson_quantum}
\begin{align}
\mean Y_\theta(\cdot) &= \trace\Bk{ E(\cdot)g(\theta)} =  \int_D 
\bra{\psi} e^{ikx} E(\cdot) e^{-ikx} \ket{\psi} \theta(x) \dd^m x.
\label{incoh_mean}
\end{align}
For example, if the measurement is in terms of a countable
spatial-mode basis $\{\ket{\varphi_n} \in \mc H_\out: n \in \mc M\}$,
then the output is a countable set of independent Poisson random variables
with mean photon numbers given by
\begin{align}
\mean Y_\theta(n) &= 
\bra{\varphi_n}g(\theta)\ket{\varphi_n} = 
\int_D \abs{\bra{\varphi_n}e^{-ikx}\ket{\psi}}^2  \theta(x) \dd^m x.
\label{spade_mean}
\end{align}
As the measurement can be regarded as a Markov map $\chan'$ that takes an
intensity operator $g$ as input and gives a mean measure for a Poisson
field as output, Theorem~\ref{thm_markov} states that the measurement
is a monotonic channel with $\norm{\chan'_\push} \le 1$. 

For the quantum model given by Eq.~(\ref{cq}), the condition on
$\effinf$ given by Eq.~(\ref{effinf_range}) for the existence of
$\effinfout$ becomes
\begin{align}
\exists u \in \Amb(g): \bk{\chan^\pull u}(x) 
= \bra{\psi} e^{ikx} u e^{-ikx}\ket{\psi} = (\effinf)(x) = b(x).
\end{align}
Under the condition, the output efficient influence given by
$\effinfout = \chan^{\pull-}b$ is an observable on $\mc H_\out$.  Once
$(\effinfout)(\check\theta)$ is found at some preliminary estimate
$\check\theta$, the one-step estimator is given by
Eqs.~(\ref{onestep_qpoisson}) and
(\ref{onestep_qpoisson2}). Eqs.~(\ref{onestep_qpoisson}) and
(\ref{onestep_qpoisson2}) mean that we should set the
projection-valued measure of $(\effinfout)(\check\theta)$ as the $E$
in Eq.~(\ref{incoh_mean}).

We note that the quantum semiparametric efficiency bound for
incoherent imaging has been studied before in the specific context of
moment estimation \cite{zhou19,qlmoment_pra,qlmoment_pra2,tan23}, but
the SVD proposed here now allows us to analyze the general nature of
the incoherent-imaging problem.

\subsection{\label{sec_incoh_q_exa}Examples}
Assume one-dimensional imaging ($m = 1$). We first consider the
point-spread state with Gaussian wavefunction
\begin{align}
\ket{\psi} &= \int_{D_\out} \psi(y) \ket{y} \dd y,
&
\psi(y) &= \braket{y}{\psi} = \frac{1}{(2\pi)^{1/4}} \exp\bk{-\frac{y^2}{4}}
\label{gauss_PSS}
\end{align}
in terms of the Dirac position ket $\ket{y}$. The SVD for direct
imaging has been studied in Sec.~\ref{sec_direct_exa} and reported in
Fig.~\ref{svd_incoh_classical_gauss0}. The SVD for
$\chan_\push = \chan_\dual$ of the quantum model given by
Eq.~(\ref{cq}), on the other hand, is plotted in
Fig.~\ref{svd_incoh_q_gauss0}, using the numerical method in
Appendix~\ref{app_num}. The plot of each intensity matrix
$\bra{\psi_n}g(\theta)\ket{\psi_m}$ assumes the
point-spread-function-adapted (PAD) basis \cite{rehacek17,spade_pra}
\begin{align}
\ket{\psi_{n}} &= (-i)^n a_{n}(k) \ket{\psi},
\label{PAD}
\end{align}
where $\{a_n(k):n = 0,1,\dots\}$ are orthonormal polynomials in terms of
the momentum operator $k$ that satisfy
\begin{align}
\intall a_n(\kappa) a_m(\kappa) \abs{\braket{\kappa}{\psi}}^2 \dd\kappa &= \delta_{nm},
&
\ket{\kappa} = \frac{1}{\sqrt{2\pi}}\intall \exp(i\kappa y) \ket{y} \dd y.
\end{align}
In particular, for the Gaussian PSF given by Eq.~(\ref{gauss_PSS}),
$\braket{\kappa}{\psi}$ is also Gaussian, $\{a_n\}$ are
Hermite-Gaussian functions, and all the matrices happen to be real
\cite{tsang26}.

\fig{}{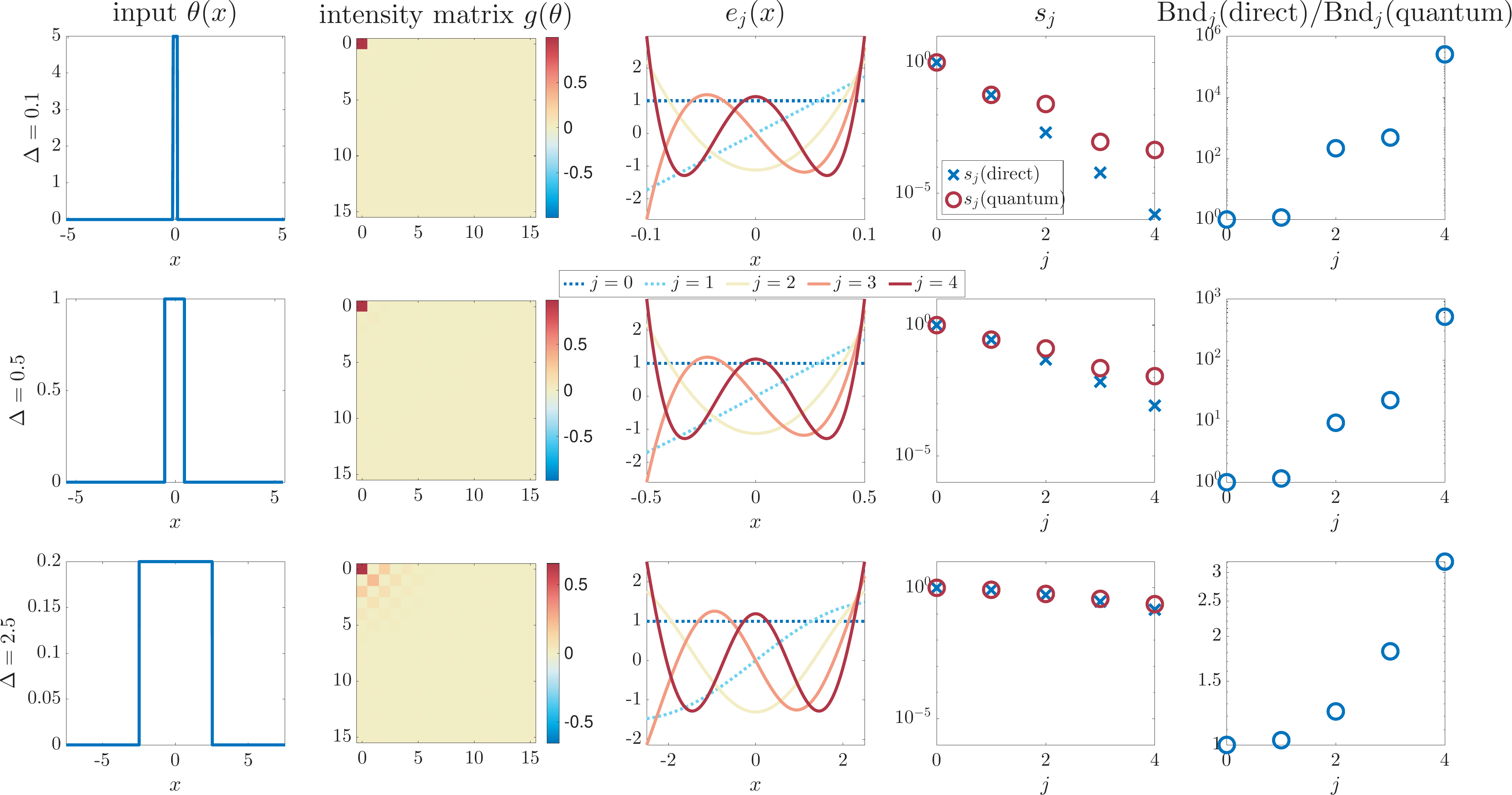}{SVD of the channel pushforward $\chan_\push$
  for the quantum model of incoherent imaging, assuming the flat-top
  input intensity $\theta(x)$ given by Eq.~(\ref{flattop}), the
  classical-quantum channel given by Eq.~(\ref{cq}), and the Gaussian
  point-spread state given by Eq.~(\ref{gauss_PSS}). Each row assumes
  a specific object size $\Delta = 0.1$, $0.5$, or $2.5$.  First
  column: input intensity $\theta(x)$ for a given object size $\Delta$
  versus the object-plane coordinate $x$.  Second column: output
  intensity matrix $\bra{\psi_n}g(\theta)\ket{\psi_m}$, where
  $\{\ket{\psi_n}\}$ is the PAD basis given by Eq.~(\ref{PAD}).  Third
  column: the first five input singular functions $\insing_j(x)$ versus the
  object-plane coordinate $x$.  Fourth column: the first five singular
  values of the quantum model ($s_j(\textrm{quantum})$, circles)
  compared with the singular values of direct imaging
  ($s_j(\textrm{direct})$, crosses) in logarithmic scale. The latter
  are the same as those in Fig.~\ref{svd_incoh_classical_gauss0}. Fifth
  column: Ratio of the efficiency bound for direct imaging
  ($\Bnd_j(\textrm{direct})$) to the quantum bound
  ($\Bnd_j(\textrm{quantum}) = 1/s_j^2(\textrm{quantum})$) in
  logarithmic scale, assuming that the weight function $b(x)$ is an
  $\insing_j(x)$ plotted in the third column. The output singular matrices
  of the quantum model are not plotted for brevity.}

We make the following observations about Fig.~\ref{svd_incoh_q_gauss0}:
\begin{enumerate}
\item When the object size $\Delta$ is small, the second column shows
  that only the intensity matrix elements
  $\bra{\psi_n}g(\theta)\ket{\psi_m}$ with low $n$ and $m$ are
  significant, meaning that only the low-order spatial modes of the
  image plane are excited significantly.

\item The input singular functions $\{\insing_j(x)\}$ in the third
  column look similar to those for direct imaging in
  Fig.~\ref{svd_incoh_classical_gauss0} and are also oscillatory.
  They in fact resemble orthogonal polynomials \cite{dunkl}.

\item The fourth column compares the first five singular values of the
  quantum model $\{s_j(\textrm{quantum}):j = 0,\dots,4\}$ with those
  of direct imaging $\{s_j(\textrm{direct}):j = 0,\dots,4\}$
  plotted earlier in Fig.~\ref{svd_incoh_classical_gauss0}. By
  monotonicity and Prop.~\ref{prop_sv}, we must have
\begin{align}
s_j(\textrm{quantum}) &\ge s_j(\textrm{direct}),
\end{align}
and indeed the two sets of singular values obey the inequality.
Remarkably, for a subdiffraction object size $\Delta = 0.1$, large
gaps between the two sets of singular values begin to develop from
$j = 2$. The quantum singular values still decay quickly, however,
imposing stringent fundamental limits to any measurement.  This
phenomenon is consistent with prior work on a multi-point-source model
\cite{bisketzi19}.

\item To demonstrate the suboptimality of direct imaging in the more
  concrete terms of efficiency bounds, the fifth column assumes the
  weight function
\begin{align}
b(x) &= \insing_j(x),
\end{align}
where $\insing_j(x)$ is an input singular function for the quantum
model, and plots the ratio of the efficiency bound for direct imaging
$\Bnd_j(\textrm{direct})$ to the quantum efficiency bound
$\Bnd_j(\textrm{quantum}) = 1/s_j(\textrm{quantum})^2$. By
monotonicity and Theorem~\ref{thm_mono}, we must have
\begin{align}
\frac{\Bnd_j(\textrm{direct})}{\Bnd_j(\textrm{quantum})} \ge 1.
\end{align}
For a subdiffraction object size $\Delta$, the ratio becomes large
starting from $j = 2$ for the quadratic-looking $\insing_2(x)$, meaning that
a substantial reduction of error becomes possible when one estimates
the Fourier coefficient
\begin{align}
\beta(\theta) &= \int \insing_j(x)\theta(x) \dd x, \quad j \ge 2.
\end{align}

\end{enumerate}

For another example, consider the point-spread state given by
\begin{align}
\ket{\psi} &= \int_{D_\out} \psi(y) \ket{y} \dd y,
&
\psi(y) &= \braket{y}{\psi} = \sqrt{K} \sinc(K y),
&
K &= \frac{\sqrt{3}}{2\pi}.
\label{sinc3}
\end{align}
The PAD basis given by Eq.~(\ref{PAD}) is now defined in terms of this
$\ket{\psi}$. The SVD of the quantum channel pushforward is plotted in
Fig.~\ref{svd_incoh_q_sinc0}.

\fig{}{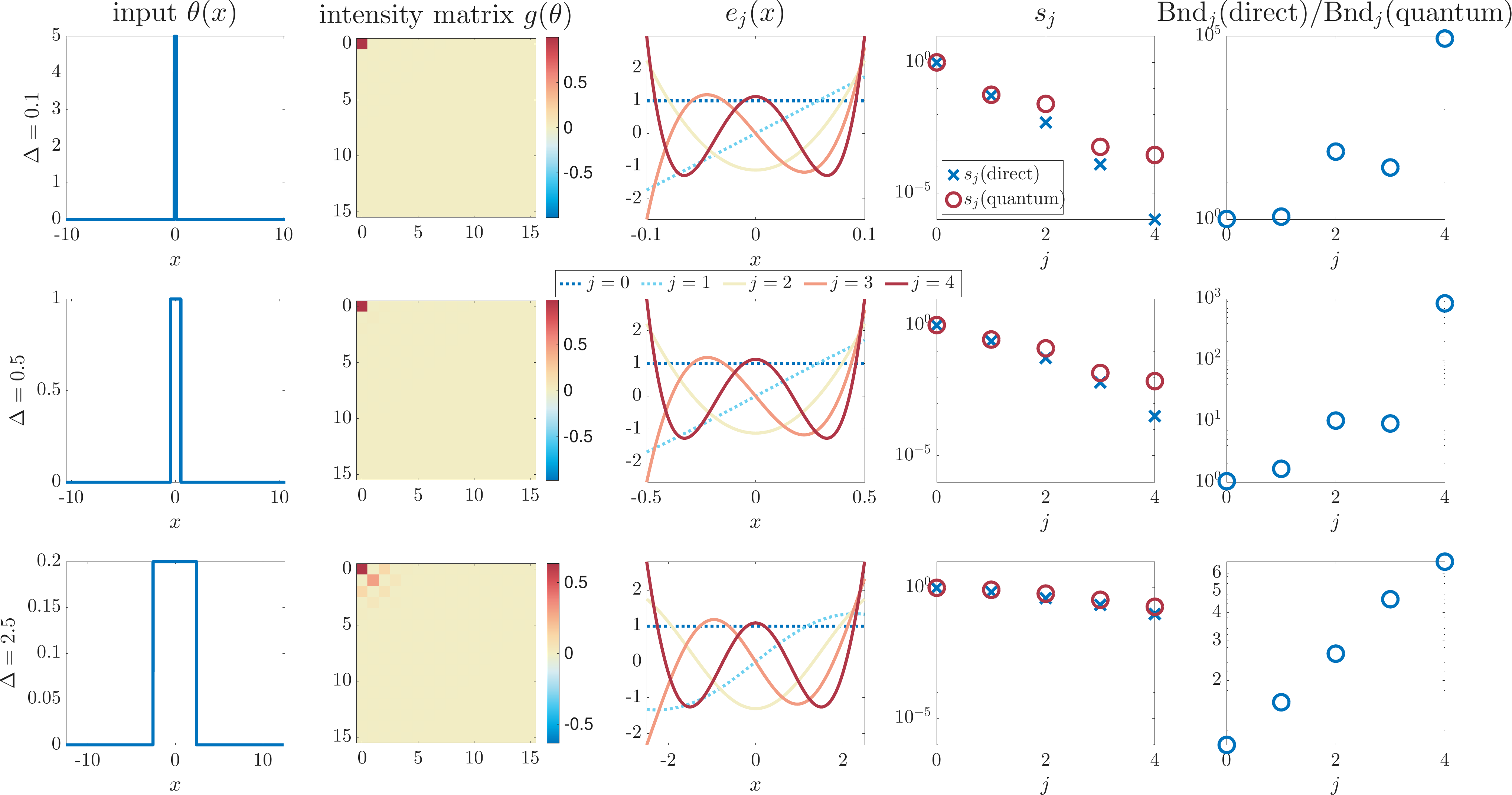}{The figure formats are the same as those of
  Fig.~\ref{svd_incoh_q_gauss0}, except that the sinc point-spread
  state given by Eq.~(\ref{sinc3}) is assumed. The singular values of
  direct imaging $\{s_j(\textrm{direct})\}$ are the same as those
  plotted in Fig.~\ref{svd_incoh_classical_sinc0}. The output singular
  matrices of the quantum model are not plotted for brevity.}

\section{\label{sec_spade}Spatial-mode demultiplexing (SPADE) for
  incoherent imaging}

The large performance gaps between direct imaging and the quantum
limit shown in Sec.~\ref{sec_incoh_q_exa} are consistent with prior
results \cite{tnl,review_cp,lvovsky26}. The study of
Fourier-coefficient estimation in Ref.~\cite{superosc_ieee} is
especially relevant. Ref.~\cite{superosc_ieee} assumes
\begin{align}
\beta(\theta) &= \int_D p_j(x) \theta(x) \dd x
= \sum_{k=0}^{j} p_{jk} \gamma_k,
&
p_j(x) &= \sum_{k=1}^j p_{jk} x^{k},
&
\gamma_k &= \int_D x^{k} \theta(x) \dd x,
\label{beta_poly}
\end{align}
where $\{p_j(x):j = 0,1,\dots\}$ are orthogonal polynomials in terms
of a fixed inner product, such as the Legendre polynomials. Since
$\beta(\theta)$ is a linear combination of the moments $\{\gamma_k\}$,
the error can be related to the errors of moment estimation, as
studied previously in great detail
\cite{spade_njp,spade_pra,zhou19,qlmoment_pra,qlmoment_pra2,tan23,tan23a}.
As the input singular functions $\{\insing_j(x)\}$ for the quantum
model resemble orthogonal polynomials, we expect that a measurement
that accurately estimates the $\beta$ given by Eqs.~(\ref{beta_poly})
should perform similarly well for the weight functions considered in
Sec.~\ref{sec_incoh_q}.

The version of SPADE proposed in Ref.~\cite{superosc_ieee} to estimate
Eqs.~(\ref{beta_poly}) is depicted in Fig.~\ref{superosc_scheme}. It
assumes a mode sorter that couples each spatial mode in the PAD basis
given by Eq.~(\ref{PAD}) to a spatially separate output mode, followed
by programmable interferometers acting on pairs of the output modes
before photon counting. The setup measures one interferometric PAD
(iPAD) basis \cite{spade_pra}
\begin{align}
\ket{\varphi_{1n}} &= 
\begin{cases}
\bk{\ket{\psi_n}+ \ket{\psi_{n+1}}}/\sqrt{2},& n = 0,2,4,\dots 
\textrm{ and } n \le n_{\textrm{max}}-1, \\
\bk{\ket{\psi_{n-1}}- \ket{\psi_n}}/\sqrt{2}, & n = 1,3,5,\dots 
\textrm{ and } n \le n_{\textrm{max}},\\
\ket{\psi_n}, & n_{\textrm{max}} \textrm{ is even} \textrm{ and } n = n_{\textrm{max}} \\
\end{cases}
\label{ipad1}
\end{align}
half of the time and another iPAD basis
\begin{align}
\ket{\varphi_{2n}} &= 
\begin{cases}
\ket{\psi_0}, & n = 0,\\
\bk{\ket{\psi_n}+ \ket{\psi_{n+1}}}/\sqrt{2},& n = 1,3,5,\dots \textrm{ and }
n \le n_{\textrm{max}}-1,\\
\bk{\ket{\psi_{n-1}}- \ket{\psi_n}}/\sqrt{2}, & n = 2,4,6,\dots \textrm{ and }
n \le n_{\textrm{max}},\\
\ket{\psi_n}, & n_{\textrm{max}} \textrm{ is odd}\textrm{ and }n = n_{\textrm{max}} 
\end{cases}
\label{ipad2}
\end{align}
half of the time, where $n_{\textrm{max}}$ denotes the maximum mode
number that one can implement in practice. The POVM becomes
\begin{align}
E(\tau,n) &= \frac{1}{2} \ket{\varphi_{\tau n}}\bra{\varphi_{\tau n}},
\quad
\tau = 1,2,
\quad
n = 0,1,\dots,n_{\textrm{max}},
\label{POVM_spade}
\end{align}
and the classical channel map becomes
\begin{align}
g(\theta,\tau,n) &= \int_D \chan(\tau,n|x) \theta(x) \dd x,
&
\chan(\tau,n|x) &= \bra{\psi} e^{ikx} E(\tau,n) e^{-ikx}\ket{\psi},
\label{chan_spade}
\end{align}
where $g(\theta,\tau,n) = \mean Y_\theta(\tau,n)$ is the mean photon
number of each output labeled by $(\tau,n)$. In terms of estimating
the $\beta$ and $\gamma$ given by Eqs.~(\ref{beta_poly}) and the
scalings of the errors with the object size $\Delta$ as
$\Delta \to 0$, this method is optimal
\cite{zhou19,qlmoment_pra,qlmoment_pra2,tan23,superosc_ieee}, but we
need numerics to determine its precise performance.

\fig{}{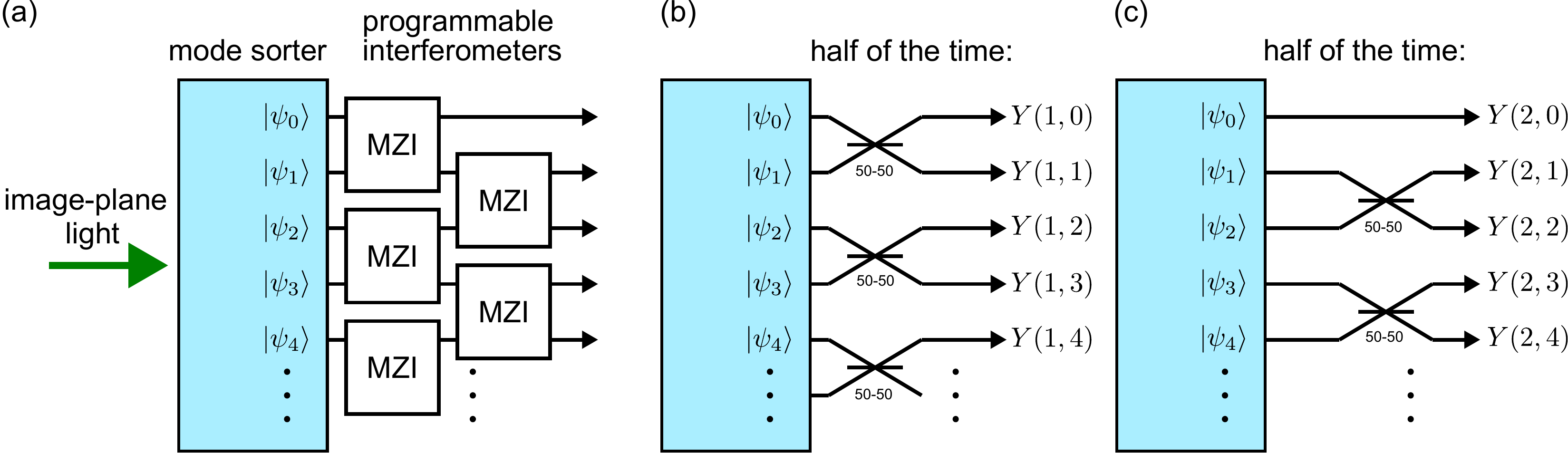}{Illustration of SPADE for the estimation of
  Fourier coefficients, as proposed in Ref.~\cite{superosc_ieee}. (a)
  The setup consists of a mode sorter that channels each spatial mode
  $\ket{\psi_n}$ in the PAD basis given by Eq.~(\ref{PAD}) to a
  spatially separate output mode, followed by programmable
  Mach-Zehnder interferometers (MZI) that act on pairs of the
  outputs. (b) With the first setting of the interferometers, the
  setup measures the basis given by Eq.~(\ref{ipad1}) half of the
  time, producing photon numbers $Y(1,n)$. (c) With the second setting
  of the interferometers, the setup measures the basis given by
  Eq.~(\ref{ipad2}) half of the time, producing photon numbers
  $Y(2,n)$.}

Following Appendix~\ref{app_num}, Fig.~\ref{svd_spade_gauss0} plots
the numerical SVD of SPADE, as modeled by
Eqs.~(\ref{ipad1})--(\ref{chan_spade}) with $n_{\textrm{max}} = 4$,
assuming the flat-top input intensity given by Eq.~(\ref{flattop}) and
the Gaussian point-spread state given by Eq.~(\ref{gauss_PSS}). We
make a few observations:
\begin{enumerate}
\item The third column plots the input singular functions for SPADE.
  Similar to those in Figs.~\ref{svd_incoh_classical_gauss0} and
  \ref{svd_incoh_q_gauss0}, they are also oscillatory and resemble
  orthogonal polynomials.

\item The fourth column plots the singular values of SPADE
  ($s_j(\textrm{SPADE})$), compared with those of direct imaging
  ($s_j(\textrm{direct})$, also plotted in
  Figs.~\ref{svd_incoh_classical_gauss0} and \ref{svd_incoh_q_gauss0})
  and those of the quantum model ($s_j(\textrm{quantum})$, also
  plotted in Fig.~\ref{svd_incoh_q_gauss0}).  By monotonicity and
  Prop.~\ref{prop_sv}, we must have
\begin{align}
s_j(\textrm{quantum}) &\ge s_j(\textrm{SPADE}),
\end{align}
and indeed the numerical results observe the inequality.  Compared
with the singular values of direct imaging, however, the singular
values of SPADE are closer to the quantum limits for small object
sizes.

\item Similar to Fig.~\ref{svd_incoh_q_gauss0}, the fifth column here
  assumes that $b(x)$ is an input singular function $\insing_j(x)$ for the
  quantum model and plots the ratio of the direct-imaging bound
  $\Bnd_j(\textrm{direct})$ to the SPADE bound
  $\Bnd_j(\textrm{SPADE})$ to demonstrate the improvement offered by
  SPADE. Similar to the ratio
  $\Bnd_j(\textrm{direct})/\Bnd_j(\textrm{quantum})$ in terms of the
  quantum limit, $\Bnd_j(\textrm{direct})/\Bnd_j(\textrm{SPADE})$
  becomes high for $\Delta = 0.1$ and $\Delta = 0.5$ starting from
  $j = 2$. This improvement is consistent with
  Ref.~\cite{superosc_ieee}. By monotonicity and
  Theorem~\ref{thm_mono},
\begin{align}
\Bnd_j(\textrm{SPADE}) &\ge \Bnd_j(\textrm{quantum}),
&
\frac{\Bnd_j(\textrm{direct})}{\Bnd_j(\textrm{SPADE})}
&\le 
\frac{\Bnd_j(\textrm{direct})}{\Bnd_j(\textrm{quantum})},
\end{align}
and indeed the numerical results observe the latter inequality.  Gaps
persist between the two ratios, especially for larger object sizes and
higher-order Fourier coefficients, demonstrating that there is still
room for improvement for this version of SPADE.

\end{enumerate}

\fig{}{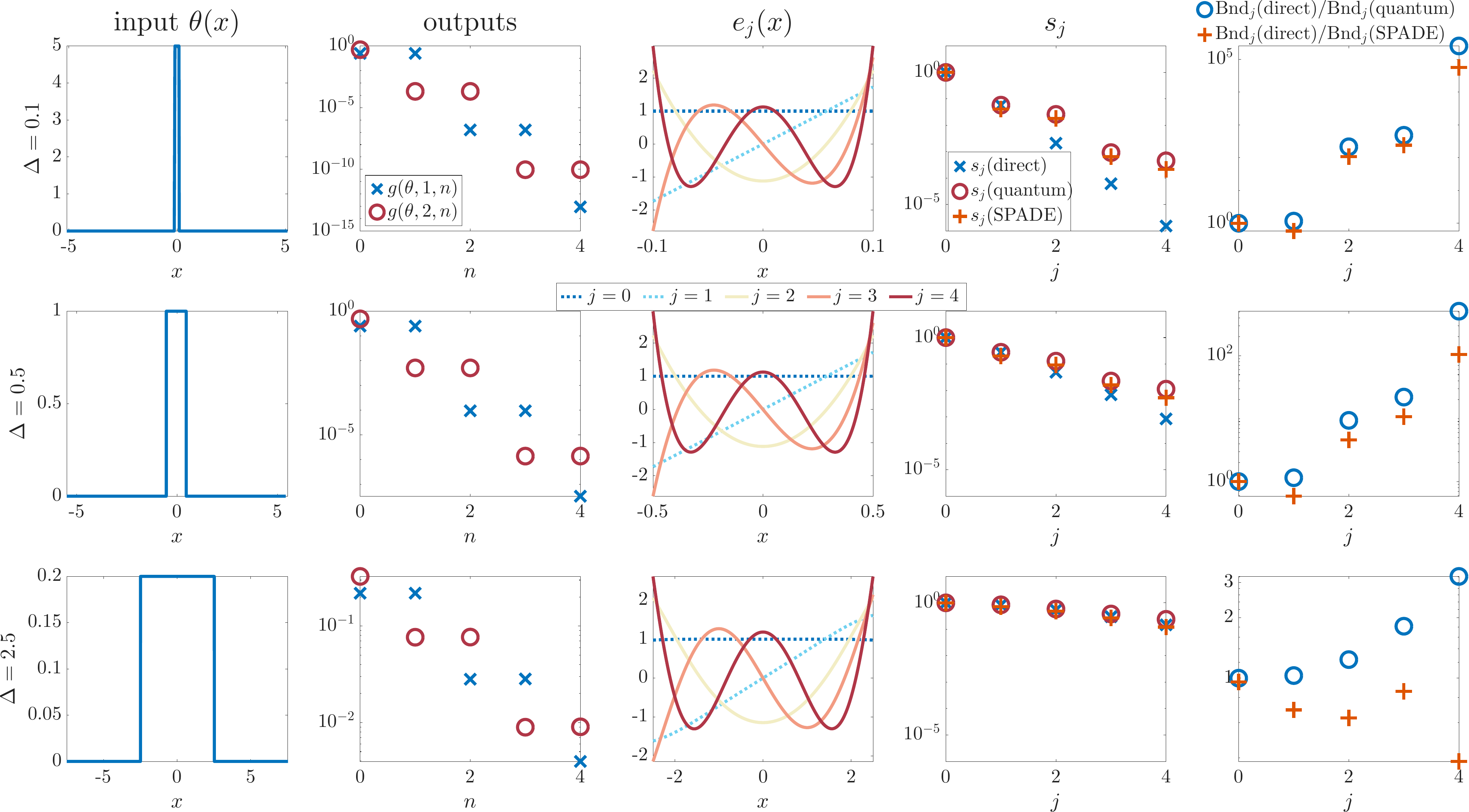}{SVD of the channel pushforward $\chan_\push$
  for SPADE, assuming the flat-top input intensity $\theta(x)$ given
  by Eq.~(\ref{flattop}), the Gaussian point-spread state given by
  Eq.~(\ref{gauss_PSS}), and the classical channel given by
  Eqs.~(\ref{ipad1})--(\ref{chan_spade}). Each row assumes a specific
  object size $\Delta = 0.1$, $0.5$, or $2.5$.  First column: input
  intensity $\theta(x)$ for a given object size $\Delta$ versus the
  object-plane coordinate $x$.  Second column: output mean photon
  numbers $g(\theta,1,n)$ and $g(\theta,2,n)$ in logarithmic scale.
  Third column: the first five input singular functions $\insing_j(x)$
  versus the object-plane coordinate $x$.  Fourth column: The first
  five singular values of direct imaging ($s_j(\textrm{direct})$,
  crosses), those of the quantum model ($s_j(\textrm{quantum})$,
  circles), and those of SPADE ($s_j(\textrm{SPADE})$, triangles) in
  logarithmic scale.  $\{s_j(\textrm{direct})\}$ are the same as those
  in Figs.~\ref{svd_incoh_classical_gauss0} and
  \ref{svd_incoh_q_gauss0}, while $\{s_j(\textrm{quantum})\}$ are the
  same as those in Fig.~\ref{svd_incoh_q_gauss0}.  Fifth column: ratio
  of the efficiency bound for direct imaging
  $\Bnd_j(\textrm{direct})$ to the SPADE bound
  $\Bnd_j(\textrm{SPADE})$ in logarithmic scale, assuming that the
  weight function $b(x)$ is an $\insing_j(x)$ for the quantum model. For
  reference, the quantum upper bounds
  $\Bnd_j(\textrm{direct})/\Bnd_j(\textrm{quantum})$, same as those in
  Fig.~\ref{svd_incoh_q_gauss0}, are also plotted. The output singular
  vectors for SPADE are not plotted for brevity.}

For completeness, we also plot the numerical results
for the sinc point-spread state given by Eq.~(\ref{sinc3})
in Fig.~\ref{svd_spade_sinc0}. The plots resemble
those in Fig.~\ref{svd_spade_gauss0}, at least qualitatively.

\fig{}{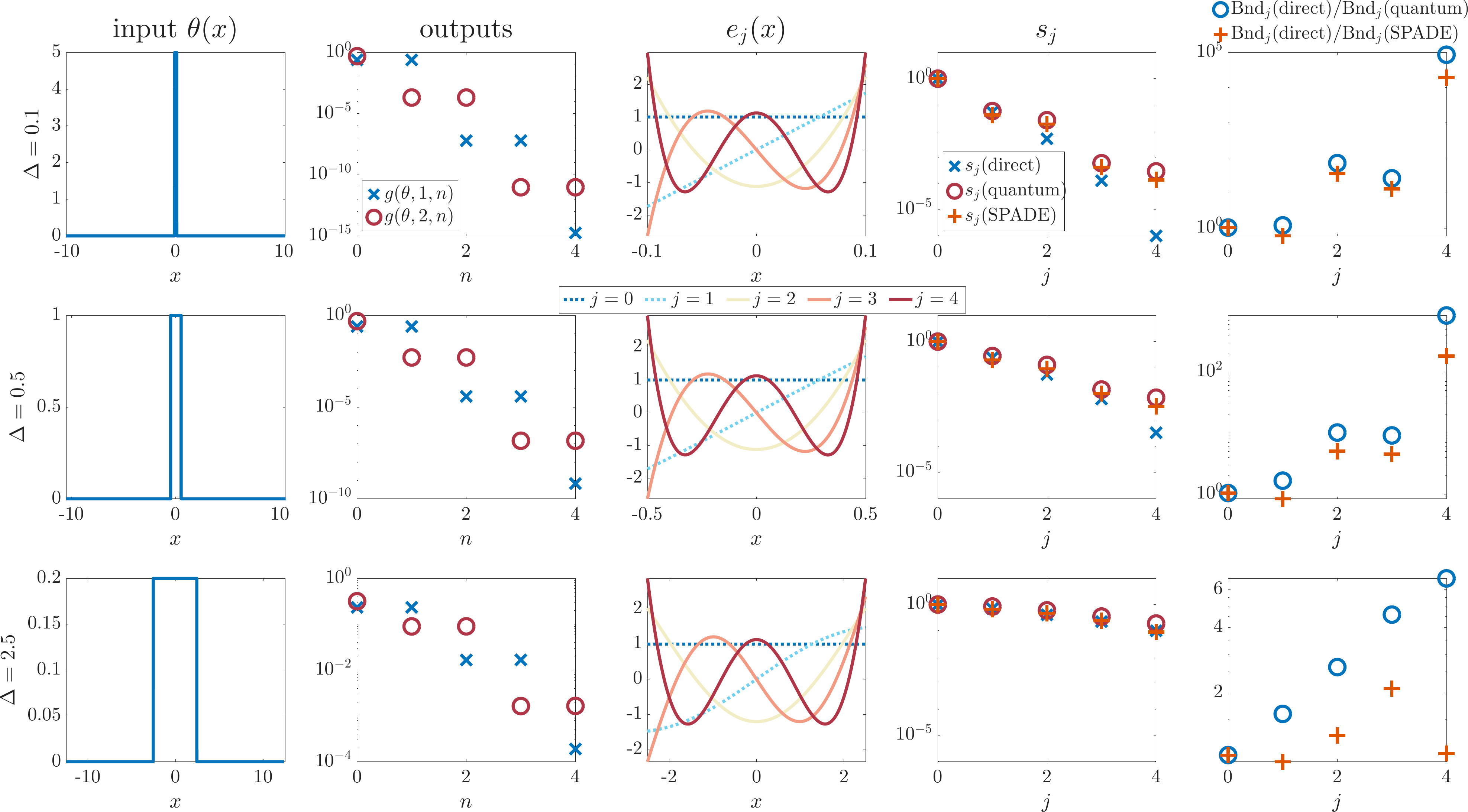}{The figure formats are the same as those of
  Fig.~\ref{svd_spade_sinc0}, except that the sinc point-spread state
  given by Eq.~(\ref{sinc3}) is assumed. In the fourth column,
  $\{s_j(\textrm{direct})\}$ are the same as those in
  Figs.~\ref{svd_incoh_classical_sinc0} and \ref{svd_incoh_q_sinc0},
  while $\{s_j(\textrm{quantum})\}$ are the same as those in
  Fig.~\ref{svd_incoh_q_sinc0}. In the fifth column, the ratios
  $\Bnd_j(\textrm{direct})/\Bnd_j(\textrm{quantum})$ are the same as
  those in Fig.~\ref{svd_incoh_q_sinc0}.}

To improve upon the version of SPADE studied here and come even closer
to the quantum limits, an adaptive protocol can be envisioned:
\begin{enumerate}
\item estimate $\theta$ based on the outcomes from the current version of SPADE, 
\item use the estimate $\check\theta$ to compute the efficient
  influence operator $(\effinfout)(\check\theta)$ of the quantum
  model, and
\item update the interferometer based on the projection-valued measure of
$(\effinfout)(\check\theta)$ to implement the one-step estimator given
by Eqs.~(\ref{onestep_qpoisson}) and (\ref{onestep_qpoisson2}).
\end{enumerate}
Similar adaptive SPADE protocols have been studied previously for
multi-point-source models \cite{lee23,bao21,choi24}. For a
multidimensional $\beta = \mqty(\beta_1 & \dots & \beta_q)$, the set
of efficient influence operators
$\mqty((\effinfout_1)(\check\theta) & \dots &
(\effinfout_q)(\check\theta))$ may be incompatible and it becomes
unclear how one should choose the basis for the adaptive SPADE. One
solution, albeit even more challenging to implement, is to use the
quantum-nondemolition measurements proposed in Ref.~\cite{tsang26},
which can in theory approach the Holevo--Nagaoka bound for a
multidimensional $\beta$.

\section{\label{sec_confocal}Incoherent imaging with structured illuminations}

Our formalism can handle structured illuminations \cite{pawley}, such
as confocal microscopy, just as well. We focus on the incoherent case.
Fig.~\ref{SIM} illustrates the setup. Suppose that the illumination
intensity pattern changes over $J$ time intervals. Let $h_j(x)$ be the
illumination intensity in the $j$th time interval. The optical field
on the image plane in the $j$th time interval can then be modeled as a
Poisson state with intensity operator
\begin{align}
\chan_j\theta &= \int_D e^{-ikx} \ket{\psi}\bra{\psi} e^{ikx} h_j(x) \theta(x) \dd^m x,
\end{align}
which modifies Eq.~(\ref{cq}) for the widefield case by inserting
$h_j(x)$. It is reasonable to assume that the optical fields in
different time intervals are independent. The total output intensity
operator can then be modeled as \cite[Example~5]{poisson_quantum}
\begin{align}
g(\theta) &= \bigoplus_j \chan_j \theta.
\end{align}
This direct sum is an example of the broadcast channel discussed in
Sec.~\ref{sec_broadcast}.  The additivity of the information map due
to Theorem~\ref{thm_add} should become useful for the study of the
quantum and classical broadcast channels that arise from the varying
illuminations.

For example, with an ideal image sensor, the mean output intensity
becomes
\begin{align}
g'(\theta,y_1,\dots,y_J) &= \bigoplus_j g_j'(\theta,y_j),
\label{g_ism}
\end{align}
where 
\begin{align}
g_j'(\theta,y_j) &= \bra{y_j} \chan_j\theta \ket{y_j}
= \int_D \abs{\bra{y_j}e^{-ikx} \ket{\psi}}^2 h_j(x) \theta(x) \dd^m x
\label{g_ism2}
\end{align}
is the mean image-plane intensity as a function of the image-plane
coordinate $y_j \in \mb R^m$ in the $j$th time interval.  If $h_j(x)$
models a focused illumination beam that moves with $j$, then
Eqs.~(\ref{g_ism}) and (\ref{g_ism2}) become a model of image-scanning
microscopy \cite{sheppard88,mueller10,gregor19}.

\fig{0.6}{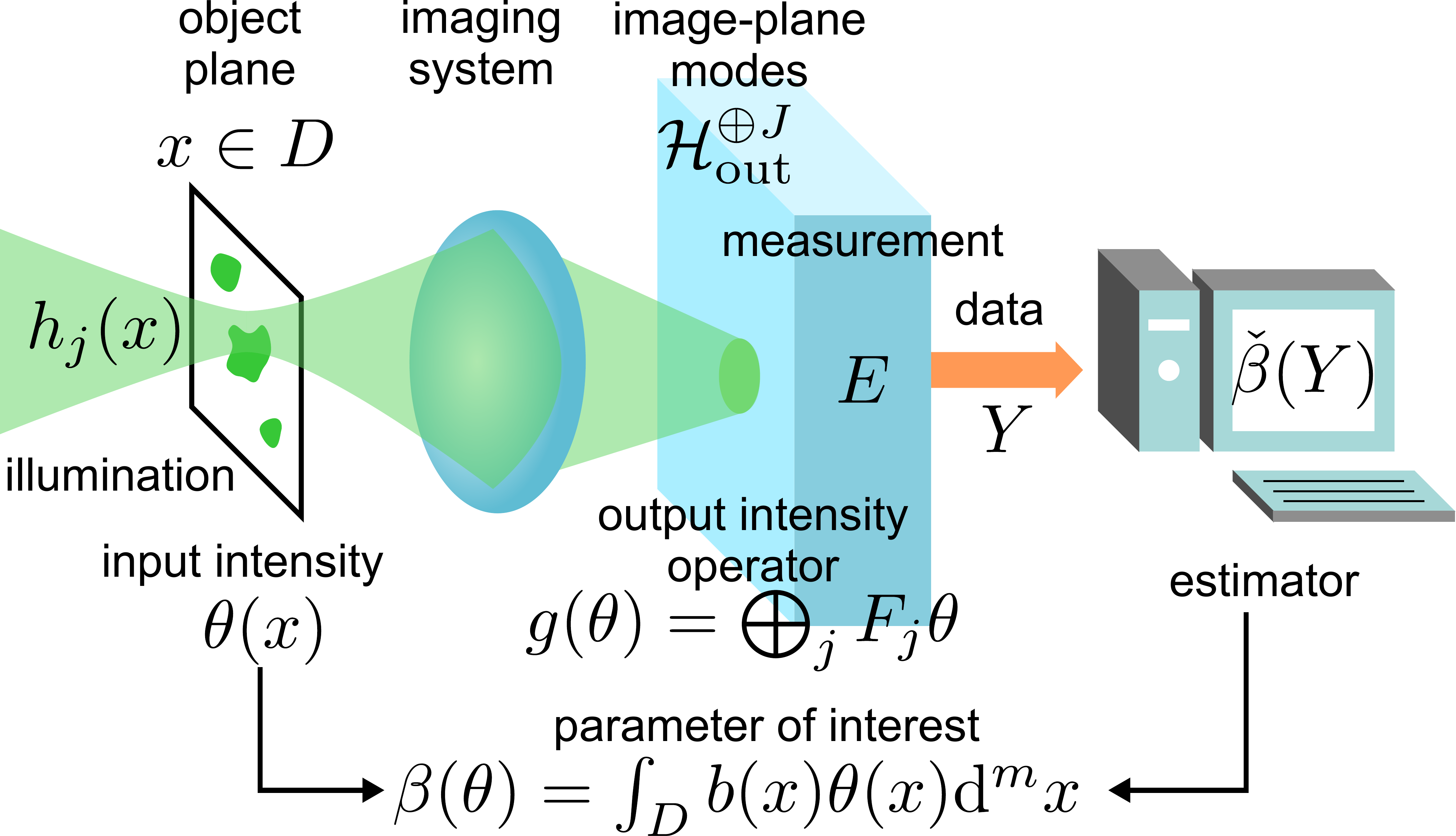}{Illustration of incoherent imaging with structured
  illuminations; confocal microscopy is an example.}

The quantum bound, the relative efficiency of direct imaging, and the
possible enhancement by SPADE or even more exotic methods are all
fascinating open questions. We leave the detailed analysis to future
work.

\section{\label{sec_ext}Potential extensions}
Despite the length of this work, it introduces only a framework of the
semiparametric efficiency theory and begets many more open questions.
We highlight some of them in the following:

\begin{enumerate}
\item Many more semiparametric problems can be found in the classical
  statistics literature
  \cite{bickel,tsiatis,vaart,laan,newey90,kennedy24} and may be
  generalizable to quantum problems.  One notable omission in this
  work is a treatment of explicit nuisance parameters, which is a vast
  subject in classical statistics and has seen some studies in the
  quantum arena as well \cite{semi_prx,suzuki20}.

\item This work focuses on parameters of interest that are linear
  functionals of the underlying parameter, but the general theory can
  also be applied to nonlinear functionals, such as the entropy of a
  probability density or density operator \cite{bickel,semi_prx}.

\item The classical and quantum Gaussian models in
  Examples~\ref{exa_gauss} and \ref{exa_qgauss} assume that the
  covariance maps are given. The models can be generalized for unknown
  covariance maps \cite{sekine95,kay,monras13} or unknown power
  spectral densities for stationary processes
  \cite{shumway_stoffer,ng16,noise_spec_pra,sui26} by generalizing the
  score vector spaces and the semi-inner products.

\item When the efficient influence is parameter-dependent, we suggest
  only one-step estimators and gloss over the details of their
  implementations. The implementations require more studies,
  especially for quantum problems where the physical processing layer
  is experimentally nontrivial.

\item The theory of broadcast channels in Sec.~\ref{sec_broadcast}
  should find many more applications in sensing and imaging beyond the
  example of structured illuminations in Sec.~\ref{sec_confocal}, any
  time scanning, multiple probes, multiple apertures, or multiple
  modalities are employed to observe the same sample
  \cite{brady,padilla26}.

\item The efficiency theory is limited by its regularity assumptions,
  its reliance on asymptotics, and its inability to incorporate prior
  information. These limitations are especially detrimental to the
  problem of reconstructing a functional parameter $\theta$, for which
  a Bayesian or minimax framework may be more appropriate
  \cite{vantrees,tsybakov,johnstone,meister,twc,chen21}. A
  transition to those frameworks does not mean that the efficiency
  theory will become obsolete---it can still offer useful solutions
  and insights via various bridging concepts such as local asymptotic
  normality \cite{vaart,demkowicz20,johnstone} and Bayesian
  generalizations of the Cram\'er--Rao bound
  \cite{schutzenberger57,vantrees,gill95,bcrb_pra}.  From the
  geometric perspective, the transition means that we should advance
  from the local theory presented here to a global theory that
  considers all points in the parameter space
  \cite{bcrb_pra,zzb_gen,amari,amari_app}.
  
\end{enumerate}

\section{\label{sec_con}Conclusion}
The essence of the semiparametric efficiency theory is geometric.
Perturbations of a parameter are represented by tangent vectors.  The
vectors are endowed with statistically motivated metrics. The infinite
dimensionality of the parameter space is handled by treating the
tangent vectors in Hilbert spaces. The parameter of interest as well
as the effects of channels are then linearized as Hilbert-space maps,
allowing their statistical properties to be analyzed by standard
mathematics, such as the SVD. Two central quantities emerge: the
efficiency bound and the efficient influence.  They establish a
fundamental notion of efficiency for semiparametric estimation with
classical data as well as quantum objects, generalizing the popular
Cram\'er--Rao and Helstrom bounds.

The examples demonstrate that similar techniques apply to many
problems, and analytic or numerical computation of the efficient
influence and the efficiency bound is often feasible. The numerical
approach is especially vital to future research, as analytic
derivations are often difficult if not impossible. For the imaging
examples, the models are admittedly simplistic, but the details are
perhaps less important than the general insight they foster: a
judicious design of the physical processing layer can enhance sensing
and imaging and is vital for reaching the fundamental quantum limits,
even for classical light and semiparametric models.

Just as classical computers have revolutionized classical statistics,
we can dream that quantum computers can similarly revolutionize
quantum sensing and imaging by performing arbitrary physical
processing to achieve the quantum limits. A universal quantum computer
with negligible noise is perhaps too far away, but specialized
hardware, such as SPADE and interferometers, can still become helpful
in the near future, in light of the recent experimental progress and
on-sky demonstrations \cite{lvovsky26,kim25,wallis26}. It will be
desirable to make the hardware programmable, so that it can better
adapt to reality. Machine learning can, of course, help implement the
physical control and data processing strategies, but it must still
respect the efficiency limits, which play the same irreplaceable role
as the laws of thermodynamics for engines. A union of the
semiparametric efficiency theory with machine learning---not only for
data processing \cite{laan,kennedy24} but also for physical
control---will be a promising, if not inevitable, trend in both
classical and quantum regimes.

\section*{Acknowledgments}
This work was inspired by insightful discussions with Alex Lvovsky's
group. The author gratefully acknowledges the hospitality of his
group, the Keble College, and the Department of Physics at the
University of Oxford, where this work was initiated in 2024.

The author acknowledges consultations with the free model of ChatGPT
in 2026 regarding minor mathematical questions, MATLAB coding,
references, and proofreading.  All suggestions were independently
checked by the author.

The MATLAB routine to generate color maps for colorblind viewers was
downloaded from Ref.~\cite{bemis21}.

\appendix

\section{\label{app_note}Notations}
\begin{enumerate}
\item We reserve the term \emph{operator} for an operator on a complex
  Hilbert space in quantum mechanics.  In general, we call an operator
  on a vector space a \emph{map} to avoid confusion with the quantum
  operators.  The operator norm of a map is still called the operator
  norm, however.

\item $I$ denotes the identity map; it may be the identity operator,
  the identity map, or the identity matrix, depending on the context.
  Similarly, $0$ may denote the zero vector, the zero operator, the
  zero map, or a zero matrix of appropriate size.

\item For a vector space $\mc V$, $\mc V^\dual$ denotes its dual, a
  space of linear functionals of vectors, also called covectors.
  $A^\dual$ of a Banach-space map $A$ denotes its dual.  $A^\dual$ of
  a Hilbert-space map $A$ denotes its adjoint. We also use $\adj$ to
  denote the adjoint when the Hilbert spaces are different from those
  for $\dual$. For a matrix $A$,
\begin{enumerate}
\item $A^\dual$ denotes the conjugate
transpose, 
\item $A^\top$ denotes the transpose, and
\item $\cconj{A}$ denotes the
entry-wise complex conjugate.
\end{enumerate}

\item If a proof is not given immediately after a numbered proposition
  and no reference is provided, the proof is delegated to
  Appendix~\ref{app_proofs}.

\item Table~\ref{tab_symbols} provides a list of symbols.
\end{enumerate}

\begin{longtable}[c]{lll}
\hline
Symbol & Definition & First mention \\
\hline
$\theta$ & parameter & Sec.~\ref{sec_tangent}\\
$\Theta$ & parameter space & Sec.~\ref{sec_tangent}\\
$\phi$ & parametric submodel & Sec.~\ref{sec_tangent}\\
$\dot\phi$ & directional derivative & Sec.~\ref{sec_tangent}\\
$\mc T_\theta$ & tangent space at $\theta$ & Sec.~\ref{sec_tangent}\\
$\beta(\theta)$ & parameter of interest & Sec.~\ref{sec_tangent}\\
$\dd\beta(\dot\phi)$ & differential of $\beta$ & Sec.~\ref{sec_tangent}\\
$\bnd_\theta(\dot\phi)$ & submodel efficiency bound & Sec.~\ref{sec_bnd}\\
$\info_\theta(\dot\phi,\dot\chi)$ & information semimetric & Sec.~\ref{sec_bnd}\\
$\score(\dot\phi)$ & score vector & Sec.~\ref{sec_bnd}\\
$f(\theta)$ & parametrization function & Sec.~\ref{sec_bnd}\\
$\cdot$ & generic argument of a function & Sec.~\ref{sec_bnd}\\
$P_\theta$ & probability measure given $\theta$ & Example~\ref{exa_classical}\\
$\expect$ & expected value given $\theta$ & Example~\ref{exa_classical}\\
$X,Y,Z$ & random elements, random variables, or observables & Example~\ref{exa_classical} \\
$\tilde X_\theta$ & mean measure of a Poisson field $X$ given $\theta$ & Example~\ref{exa_poisson}\\
$\Sigma$ & covariance map or covariance matrix & Example~\ref{exa_gauss}\\
$\mc H$ & complex Hilbert space & Example~\ref{exa_quantum}\\
$\jordan$ & Jordan product & Example~\ref{exa_quantum}\\
$\trace$ & trace of an operator & Example~\ref{exa_quantum}\\
$\mc K$ & phase space for optical modes & Example~\ref{exa_qgauss}\\
$\partial\beta$ & Euclidean gradient of $\beta$ & Example~\ref{exa_finite}\\
$\Imat$ & information matrix & Example~\ref{exa_finite}\\
$\Avg{\cdot,\cdot}_{f(\theta)}$ & semi-inner product or inner product
for scores & Sec.~\ref{sec_score}\\
$\norm{\cdot}_{f(\theta)}$ & seminorm or norm for scores & Sec.~\ref{sec_score}\\
$\mc V$ & ambient vector space & Sec.~\ref{sec_score}\\
$\mc N(f(\theta))$ & set of elements with zero seminorm 
$\norm{\cdot}_{f(\theta)} = 0$ & Sec.~\ref{sec_score}\\
$\complete$  & completion of an inner-product space & Sec.~\ref{sec_score}\\
$\overline{S}$ & closure of a set $S$ & Sec.~\ref{sec_score}\\
$\Amb(f)$ & ambient Hilbert space & Sec.~\ref{sec_score}\\
$f_\push$ & pushforward & Sec.~\ref{sec_score}\\
$\mc T(f)$ & score tangent space & Sec.~\ref{sec_score}\\
$\sdiff$ & score differential & Sec.~\ref{sec_Bnd}\\
$\Bnd$ & efficiency bound & Sec.~\ref{sec_Bnd}\\
$f_\pull$ & pullback & Sec.~\ref{sec_Bnd}\\
$\effinf$ & efficient influence & Sec.~\ref{sec_Bnd}\\
$\check\beta$ & estimator of $\beta$ & Sec.~\ref{sec_Bnd}\\
$(\effinf)(\check\theta)$ & efficient influence 
at estimate $\check\theta$ & Sec.~\ref{sec_Bnd}\\
$b(\cdot)$ & linear functional of $\theta$ or weight function & Sec.~\ref{sec_classical}\\
$E$ & projection-valued measure or positive operator-valued measure (POVM) & 
Sec.~\ref{sec_qpoisson}\\
$\chan$ & channel map & Sec.~\ref{sec_chan_exa}\\
$\circ$ & composition of two functions & Sec.~\ref{sec_chan_exa}\\
$g(\theta)$ & output parametrization function & Sec.~\ref{sec_chan_exa}\\
$\chan_\push$ & channel pushforward & Sec.~\ref{sec_chan_push}\\
$\Bnd_\out$ & efficiency bound at the channel output  & Sec.~\ref{sec_chan_pull}\\
$\Delta_\out\beta$ & score differential at the channel output  & Sec.~\ref{sec_chan_pull}\\
$\chan_\pull$ & channel pullback for score differentials & Sec.~\ref{sec_chan_pull} \\
$\effinfout$ & efficient influence at the channel output  & Sec.~\ref{sec_chan_pull}\\
$\adj$ & adjoint with respect to $\Amb(f)$ and $\Amb(g)$ & Sec.~\ref{sec_chan_pull}\\
$\chan^\pull$ & channel pullback for efficient influences,
$\chan_\push^\adj$  & Sec.~\ref{sec_chan_pull}\\
$T^-$ & pseudo-inverse of a map $T$ & Theorem~\ref{thm_effinfout}\\
$\imap$ & information map & Sec.~\ref{sec_imap}\\
$T^+$ & Moore--Penrose inverse of a matrix $T$ & Corollary~\ref{cor_finite}\\
$s_j$ & singular value & Sec.~\ref{sec_svd}\\
$\insing_j$ & input singular vector & Sec.~\ref{sec_svd}\\
$\outsing_j$ & output singular vector & Sec.~\ref{sec_svd}\\
$\otimes$ & tensor product & Sec.~\ref{sec_broadcast}\\
$\oplus$ & direct sum & Sec.~\ref{sec_broadcast}\\
$\inj$ & injection map & Sec.~\ref{sec_broadcast}\\
$\pi$ & projection map & Sec.~\ref{sec_broadcast}\\
$\chan^\dual$ & dual of channel map $\chan$ & Sec.~\ref{sec_cc}\\
$\chan_\dual$ & $\chan^{\dual\adj}$ & Sec.~\ref{sec_cc}\\
$\IMAP$ & maximal information map, $\chan^\dual\chan_\dual$ & Sec.~\ref{sec_cc}\\
$D$, $D_\out$ & object plane, image plane & Sec.~\ref{sec_coh_spaces} \\
$\epsilon$ & Gaussian noise level & Sec.~\ref{sec_coh_spaces}\\
$\Delta$ & object-plane size or object size & Sec.~\ref{sec_coh_c_exa} \\
$\Delta_\out$ & image-plane size & Sec.~\ref{sec_direct_exa} \\
$\psi(y)$ & point-spread function for the fields & Sec.~\ref{sec_direct_exa} \\
$\ket{\psi}$ & point-spread state & Sec.~\ref{sec_incoh_q_spaces}\\
$k$ & momentum operator & Sec.~\ref{sec_incoh_q_spaces}\\
$\ket{y}$ & Dirac position ket & Sec.~\ref{sec_incoh_q_spaces}\\
$\ket{\psi_n}$ & point-spread-function-adapted (PAD) state & 
Sec.~\ref{sec_incoh_q_exa}\\
\hline
\caption{\label{tab_symbols}List of symbols.}
\end{longtable}

\section{\label{app_fock}Bosonic Fock space and Gaussian states}
\subsection{Bosonic Fock space}
Let $\fock(\mc H)$ be the bosonic Fock space constructed from a
complex mode Hilbert space $\mc H$ \cite{parth_qsc}. $\fock(\mc H)$ is
a complex Hilbert space itself; we use the bra-ket notation with it.
An important concept is a pure quantum state
$\ket{\alpha} \in \fock(\mc H)$ called the coherent state, with
$\alpha \in \mc H$ modeling the mean field. $\ket{0}$ is the vacuum
state. The linear span of $\{\ket{\alpha}: \alpha \in \mc H\}$ is a
dense subset of $\fock(\mc H)$.  The inner product between two
coherent states is
\begin{align}
\braket{\alpha}{\gamma} &= 
\exp\Bk{-\frac{1}{2} \norm{\alpha-\gamma}_{\mc H}^2 + i\Im\Avg{\alpha,\gamma}_{\mc H}}.
\end{align}
There exists a unitary operator $D(\alpha)$ on $\fock(\mc H)$ called
the displacement operator that gives
\begin{align}
D(\alpha) \ket{0} &= \ket{\alpha} \quad\forall \alpha \in \mc H.
\end{align}
It satisfies the fundamental relations
\begin{align}
D(\alpha)D(\gamma) &= e^{-i \Im\Avg{\alpha,\gamma}_{\mc H}} D(\alpha+\gamma)
= e^{-2i \Im\Avg{\alpha,\gamma}_{\mc H}} D(\gamma)D(\alpha) .
\end{align}
Define a quadrature observable $p(\alpha)$ of the $\alpha$ mode in
terms of the Stone generator of $\{D(t\alpha):t \in \mb R\}$ by
\begin{align}
D(t\alpha) &= \exp\Bk{-i t \sqrt{2} p(\alpha)}.
\end{align}
It follows that
\begin{align}
p(t\alpha) &= t p(\alpha) \quad \forall t \in \mb R,
&
p(\alpha) + p(\gamma) &= p(\alpha+\gamma).
\label{p_lin}
\end{align}
Define also the operators
\begin{align}
q(\alpha) &= p(-i\alpha),
&
a(\alpha) &= \frac{q(\alpha) + ip(\alpha)}{\sqrt{2}},
&
a(\alpha)^\dual &= \frac{q(\alpha) - ip(\alpha)}{\sqrt{2}}
\end{align}
in terms of the $\alpha$ mode.  The canonical commutation relation
becomes
\begin{align}
\Bk{q(\alpha),q(\gamma)} \ket{\eta}
&= \Bk{p(\alpha),p(\gamma)} \ket{\eta}
= i \Im\Avg{\alpha,\gamma}_{\mc H} \ket{\eta},
&
\Bk{q(\alpha),p(\gamma)} \ket{\eta} &= i \Re\Avg{\alpha,\gamma}_{\mc H} \ket{\eta},
\\
\Bk{a(\alpha),a(\gamma)^\dual}\ket{\eta} &= \Avg{\alpha,\gamma}_{\mc H} \ket{\eta}
\end{align}
for any $\alpha,\gamma,\eta \in \mc H$.  The mode space $\mc H$
generalizes the role of linear functionals $\mc B^\dual$ in
Example~\ref{exa_gauss}, while the modal quadrature operators
$q(\alpha)$ and $p(\alpha)$ with $\alpha \in \mc H$ generalize the
role of random variables $L(X)$ with $L \in \mc B^\dual$.

An orthonormal basis $\{e_j\}$ of $\mc H$ defines a discrete set of
modes, such that $\{q(e_j)\}$, $\{p(e_j)\}$, and $\{a(e_j)\}$ are
precisely the standard quadrature and annihilation operators for the
discrete modes, satisfying the standard multimode commutation
relations.

In optics, we often write \cite{mandel}
\begin{align}
\Avg{\alpha,\gamma}_{\mc H} &= \sum_s \int \cconj{\alpha(k,s)} \gamma(k,s)  \dd^3 k,
&
a(\alpha) &= \sum_s \int \cconj{\alpha(k,s)} C(k,s) \dd^3k,
\label{inner_EM}
\end{align}
where $\alpha(k,s)$ is the mode function of the wavevector
$k \in \mb R^3$ and the polarization $s \in \{1,2\}$, while the
annihilation operator $C(k,s)$ satisfies
$[C(k,s),C(k',s')^\dual] = \delta^3(k-k')\delta_{ss'}$ with respect to
the Dirac delta function $\delta^3$ and the Kronecker delta
$\delta_{ss'}$.

\subsection{Phase space}
To conform with the classical treatment in Example~\ref{exa_gauss},
where all spaces are real, we now transform the complex $\mc H$ to a
real Hilbert space $\mc K$ using a procedure called realification
\cite{fomenko}. $\mc K$ is called the phase space in physics. Let
$\conj:\mc H \to \mc H$ be an antiunitary conjugation map satisfying
\begin{align}
\Avg{\alpha,\gamma}_{\mc H} &= \Avg{\conj \gamma,\conj\alpha}_{\mc H}
\quad
\forall \alpha,\gamma \in \mc H,
&
\conj^2 &= I.
\label{conjugation}
\end{align}
Such a map always exists---we can take an orthonormal basis $\{e_j\}$
of $\mc H$ and define $\conj$ by
\begin{align}
\conj \alpha &= \sum_j e_j \cconj{\Avg{e_j,\alpha}_{\mc H}}.
\end{align}
Now define the real and imaginary parts of $\alpha$ as, respectively,
\begin{align}
\Real\alpha &= \frac{\alpha + \conj \alpha}{2},
&
\imag\alpha &= \frac{\alpha - \conj \alpha}{2i}.
\label{Re_Im}
\end{align}
Then $\conj \Real = \Real$ and $\conj\imag = \imag$, leading to 
\begin{align}
\Avg{\Real\alpha,\Real\gamma}_{\mc H} &\in \mb R,
&
\Avg{\imag\alpha,\imag\gamma}_{\mc H} &  \in \mb R.
\label{real_inner}
\end{align}
These are real inner products. We also have
\begin{align}
\Re \Avg{\alpha,\gamma}_{\mc H} &= \Avg{\Real\alpha,\Real\gamma}_{\mc H} + 
\Avg{\imag\alpha,\imag\gamma}_{\mc H},
&
\Im \Avg{\alpha,\gamma}_{\mc H} &= \Avg{\Real\alpha,\imag\gamma}_{\mc H}
-\Avg{\imag\alpha,\Real\gamma}_{\mc H}.
\end{align}
Beware that $\conj$, $\Real$, and $\imag$ are linear only with respect
to real numbers. With imaginary numbers, the rules are
\begin{align}
\conj i &= -i \conj, & \Real i &= -\imag, 
&
\imag i &= \Real.
\end{align}
Eqs.~(\ref{real_inner}) together with
\begin{align}
\norm{\alpha}_{\mc H}^2 &= \norm{\Real\alpha}_{\mc H}^2 + \norm{\imag\alpha}_{\mc H}^2
\end{align}
imply that $\Real\mc H$ and $\imag \mc H$ are complete real
inner-product spaces, or real Hilbert spaces by definition.  We can
now define the following.
\begin{enumerate}
\item Let the phase space be the real Hilbert space
\begin{align}
\mc K &= \Real\mc H \oplus \imag \mc H.
\end{align}
The inner product is defined as
\begin{align}
\Avg{A'\oplus A'', B'\oplus B''}_{\mc K}
&= \Avg{A',B'}_{\mc H} + \Avg{A'',B''}_{\mc H}.
\end{align}
It is often convenient to write each element of $\mc K$ in the block
form
\begin{align}
A &= A' \oplus A'' = \mqty(A'\\ A'').
\end{align}

\item Define the map $U:\mc H \to \mc K$ as
\begin{align}
U\alpha &= \Real \alpha \oplus \imag \alpha = \mqty(\Real\\ \imag) \alpha,
\label{realification}
\end{align}
such that $\mc K = U \mc H$ and $U$ is surjective.  $U$ is called the
realification map. $U$ is isometric in the sense of
\begin{align}
\Re\Avg{\alpha,\gamma}_{\mc H} &= \Avg{U\alpha,U\gamma}_{\mc K},
\\
\norm{\alpha}_{\mc H} &= \norm{U\alpha}_{\mc K},
\label{real_iso}
\end{align}
but $U$ is not a unitary map in the conventional sense, because
$\mc H$ is complex while $\mc K$ is real.

The inverse of $U$ can be expressed as
\begin{align}
U^{-1} A&= 
U^{-1} \mqty(A'\\ A'') = \mqty(1 & i)\mqty(A'\\ A'') 
= A' + i A''.
\label{realification2}
\end{align}

\item Let the complex-structure map $J:\mc K \to \mc K$ be
\begin{align}
J(A'\oplus A'') &= (-A'') \oplus A',
\end{align}
such that
\begin{align}
\Im \Avg{\alpha,\gamma}_{\mc H} &= -\Avg{U\alpha,JU\gamma}_{\mc K},
&
J^\dual &= -J, & J^2 &= -I,
\end{align}
where $\dual$ of a matrix denotes the conjugate transpose. In block
form,
\begin{align}
J &= \mqty(0 & -I\\ I & 0).
\end{align}
\item Let the modal quadrature observables on $\fock(\mc H)$ in terms
  of a phase-space vector $A \in \mc K$ be
\begin{align}
P(A) &= p(U^{-1} A),
&
X(A) &= q(U^{-1} A) = P(-JA).
\end{align}
$P(A) = X(JA)$ and $X(A)$ are both linear with respect to $A$ by
virtue of Eqs.~(\ref{p_lin}). The canonical commutation relations
become
\begin{align}
\Bk{X(A),P(B)} \ket{\eta} &= i \Avg{A,B}_{\mc K} \ket{\eta},
&
\Bk{X(A),X(B)} \ket{\eta} &= -i \Avg{A,JB}_{\mc K} \ket{\eta}.
\end{align}
\item Let the phase-space displacement operator $W(A)$ on
  $\fock(\mc H)$ in terms of $A \in \mc K$ be
\begin{align}
W(A) &= \exp[-i P(A)] = D(U^{-1} A/\sqrt{2}).
\end{align}
A convenient formula is
\begin{align}
W(B)^\dual e^{iX(A)} W(B) &= e^{i\Avg{A,B}_{\mc K}}  e^{iX(A)} .
\end{align}
\end{enumerate}
Now the quadrature observable $X(A)$ of the mode specified by
$A \in \mc K$ is a more precise generalization of the random variable
$L(X)$ with $L \in \mc B^\dual$ in Example~\ref{exa_gauss}.

\subsection{\label{app_gauss_state}Gaussian states}
When the mode space $\mc H$ is infinite-dimensional, the literature
does not seem to offer any general definition of the covariance map
$\Sigma$ on the phase space $\mc K$, Gaussian states on the bosonic
Fock space $\fock(\mc H)$, or Gaussian channels, so we write down some
formal expressions that are correct at least in the finite-dimensional
case \cite{holevo_info}, while adopting notations that conform with
the rest of this work.

The means and covariances of quadratures with respect to a quantum
state $\rho$ on $\fock(\mc H)$ should define a mean vector
$f \in \mc K$ and covariance map $\Sigma:\mc K \to \mc K$ by
\begin{align}
\trace\Bk{\rho X(A)} &= \Avg{f,A}_{\mc K},
&
\trace\BK{\rho\Bk{X(A) \jordan X(B)}} &= \Avg{A,\Sigma B}_{\mc K}
+ \Avg{A,f}_{\mc K}\Avg{f,B}_{\mc K}.
\end{align}
$\Sigma$ should be self-adjoint and positive-semidefinite.  It is also
reasonable to assume that $\Sigma$ is bounded, so that the variance of
any quadrature $X(A)$ is finite.  Let $\mu$ be a zero-mean state
($f = 0$).  With the complex covariance given by
\begin{align}
\trace\Bk{\mu q(\alpha) q(\gamma)} &= \Avg{\alpha,\Upsilon \gamma}_{\mc H},
\end{align}
where $\Upsilon$ should also be self-adjoint and positive-semidefinite,
we find
\begin{align}
\Upsilon &= U^{-1} \Sigma U - \frac{i}{2} U^{-1} J U \ge 0.
\end{align}
We may abbreviate the inequality as
\begin{align}
\Sigma - \frac{i}{2}J &\ge 0,
\label{heisenberg}
\end{align}
which is an abstract version of the Heisenberg uncertainty relation
and agrees with Ref.~\cite[Eq.~(12.82)]{holevo_info}.

Define a Gaussian state $\rho$ on $\fock(\mc H)$ as a state that
possesses the Gaussian characteristic function
\begin{align}
\trace\Bk{\rho e^{i X(A)}} &= \exp(i \Avg{f,A}_{\mc K} - \frac{1}{2}
\Avg{A,\Sigma A}_{\mc K}),
\end{align}
which coincides with that of a classical Gaussian measure on $\mc K$
\cite{bogachev_gauss}, except that $\Sigma$ here need not be
trace-class. For example, when $\rho$ is a coherent state with
\begin{align}
\rho &= \ket{\alpha}\bra{\alpha},
&
\ket{\alpha} &= D(\alpha) \ket{0} = W(f) \ket{0},
&
f &= \sqrt{2} U \alpha,
\end{align}
we find
\begin{align}
\trace\Bk{\rho e^{i X(A)}} &= \bra{\alpha} e^{iX(A)}\ket{\alpha}
= \exp\bk{i \Avg{f,A}_{\mc K}-\frac{1}{4} \norm{A}_{\mc K}^2},
\end{align}
implying that $\Sigma = I/2$, which is bounded but not even compact
when $\mc K$ is infinite-dimensional. In general, we expect that, for
any physical state, the minimum variance of any quadrature is strictly
above $0$. In other words, there should exist a $c > 0$ such that
\begin{align}
\Avg{A,\Sigma A}_{\mc K} \ge c \norm{A}^2_{\mc K} \quad \forall A \in \mc K,
\end{align}
implying that $\Sigma$ is invertible.  Then $\Sigma^{-1}$ is also
self-adjoint, positive-definite, bounded, and invertible.

When $\mc H$ is finite-dimensional, we can assume $\mc H = \mb C^q$
and $\mc K = \mb R^{2q}$. Then we can simply write Eqs.~(\ref{fj}) and
(\ref{Sjk}) and also
\begin{align}
X(A) &= X^\top A,
&
\Avg{f,A}_{\mc K} &= f^\top A,
&
\Avg{A,\Sigma B}_{\mc K} &= A^\top \Sigma B.
\end{align}
Notice that the Fock space $\fock(\mc H)$ is infinite-dimensional
even if $\mc H$ is finite-dimensional. Gaussian states are convenient
mainly because one can work with the smaller mode space $\mc H$
instead of $\fock(\mc H)$ directly.

Let $\mc F$ be a Markov map that transforms a density operator on
$\fock(\mc H)$ to another density operator on $\fock(\mc H_\out)$, as
described in Example~\ref{exa_qq}.  Define a Gaussian channel as a
channel that gives, for any input Gaussian state $\rho$,
\begin{align}
\trace\Bk{\bk{\mc F \rho} e^{i Y(A)}}
&= 
\exp\bk{i \Avg{g,A}_{\mc K_\out} -\frac{1}{2} \Avg{A,\Sigma_\out A}_{\mc K_\out}},
\end{align}
where $\mc K_\out$ is the output phase space constructed from
$\mc H_\out$, $A \in \mc K_\out$, $Y(A)$ is an output modal quadrature
observable, $g \in \mc K_\out$ is the output mean vector, and
$\Sigma_\out$ is the output covariance map on $\mc K_\out$.  The
input-output relations should be given by Eqs.~(\ref{chan_gauss}) and
(\ref{chan_gauss_cov}) in terms of a channel map
$\chan:\mc K \to \mc K_\out$.  The uncertainty relation imposes a
condition on $\chan$ and $\Sigma_\noise$ as
\begin{align}
\Sigma_\noise - \frac{i}{2} \bk{J_\out - \chan J \chan^\dual} \ge 0,
\end{align}
where $J_\out$ is the complex-structure map on $\mc K_\out$.  This
inequality agrees with Ref.~\cite[Eq.~(12.132)]{holevo_info}.

\section{\label{app_error}Efficiency bounds for unbiased and regular estimators}
\subsection{Unbiased estimators}
Let $\{f(\theta,\cdot):\theta \in \Theta\}$ be a set of probability
densities, following Example~\ref{exa_classical}. An estimator
$\check\beta:\sspace \to \mb R$ is called unbiased if
\begin{align}
\expect\Bk{\check\beta - \beta(\theta)}
= \int \Bk{\check\beta(\el) - \beta(\theta)} f(\theta,\el)\dd\mu(\el) &= 0
\quad
\forall \theta \in \Theta.
\label{unbiased}
\end{align}
A mathematically weaker condition called local unbiasedness follows if
Eq.~(\ref{unbiased}) holds at the true parameter value $\theta$ and
also
\begin{align}
\dot\phi \expect\Bk{\check\beta - \beta(\theta)}&= 0
\quad
\forall \dot\phi \in \mc T_\theta
\end{align}
at the true $\theta$. The local condition can be expressed as
\begin{align}
\Avg{\delta,1}_{f(\theta)} &= 0,
&
\Avg{\delta,\score(\dot\phi)}_{f(\theta)} &= \dd\beta(\dot\phi)
\quad
\forall \dot\phi \in \mc T_\theta
\label{lub}
\end{align}
for $\delta = \check\beta-\beta$. The mean-square error, on the other hand, is
\begin{align}
\error(\theta) &= 
\expect\BK{\Bk{\check\beta - \beta(\theta)}^2 }
= \norm{\check\beta-\beta}^2_{f(\theta)}.
\end{align}
As long as $\error(\theta) < \infty$, $\check\beta-\beta$ is a
zero-mean finite-variance random variable.  Any finite-variance random
variable $\delta$ that satisfies Eqs.~(\ref{lub}) is called an
influence function; $\check\beta-\beta$ is one. Following
Sec.~\ref{sec_classical}, we can regard influence functions as
elements of the ambient Hilbert space $\Amb\bk{f(\theta)}$. An
influence function $\delta$ is similar to the efficient influence
$\effinf$ in that both are elements of $\Amb\bk{f(\theta)}$ satisfying
Eqs.~(\ref{lub}), with the crucial difference being that $\effinf$ is
in the score tangent space $\mc T\bk{f(\theta)}$ while $\delta$ need
not be. It is straightforward to prove the following theorem to obtain
a lower error bound, which in fact holds in the general setting of
Sec.~\ref{sec_Bnd}.
\begin{theorem}
\label{thm_inf}
Let $\delta$ be any element of $\Amb\bk{f(\theta)}$ that satisfies
\begin{align}
\Avg{\delta,f_\push\dot\phi}_{f(\theta)} &= \dd\beta(\dot\phi) \quad
\forall \dot\phi \in \mc T_\theta.
\label{inf}
\end{align}
We call $\delta$ an influence and denote the set of all influences as
$\mc D\bk{f(\theta)} \subseteq \Amb\bk{f(\theta)}$.  Assume that
$\mc D\bk{f(\theta)}$ is not empty.  Then the efficient influence
$\effinf$ is in $\mc D\bk{f(\theta)}$ and satisfies
\begin{align}
\effinf &= \Pi_{\mc T(f(\theta))} \delta \quad
\forall \delta \in \mc D(f(\theta)),
&
\inf_{\delta \in \mc D(f(\theta))}\norm{\delta}^2_{f(\theta)} &=
 \norm{\effinf}_{f(\theta)}^2 = \Bnd(\theta).
\label{inf_bnd}
\end{align}
\end{theorem}
\begin{proof}
  The projection $\Pi_{\mc T(f(\theta))}\delta$ satisfies the
  definition of the efficient influence $\effinf$, which is the unique
  element in $\mc T(f(\theta)) = \overline{\range f_\push}$ satisfying
  Eq.~(\ref{inf}) by the Riesz representation theorem. Thus
  $\norm{\delta}_{f(\theta)} \ge \norm{\Pi_{\mc
      T(f)}\delta}_{f(\theta)} = \norm{\effinf}_{f(\theta)}$.
\end{proof}
By Theorem~\ref{thm_inf}, we obtain the classical efficiency bound
\begin{align}
\error(\theta) &\ge \Bnd(\theta)
\end{align}
for any locally unbiased estimator. The quantum case is similar
\cite{semi_prx} or can be treated as a consequence of monotonicity
after a measurement, as explained in Sec.~\ref{sec_qq}.

\subsection{Regular estimators}
Since the estimator is not supposed to know the true parameter, local
unbiasedness is hardly a relaxation of the global condition in
practice and many estimators do not satisfy the conditions; even the
maximum-likelihood estimator may be biased---locally and
globally---and violate the efficiency bound
\cite{vaart,minimax_jmo}. To truly relax the conditions, one way is to
consider a sequence of observations and appeal to asymptotics
\cite{vaart}.  For each $N \in \mb N$, let $X^{(N)} \in \sspace^{(N)}$
be the observation with probability measure $P_\theta^{(N)}$, such as
a set of $N$ IID observations with
$X^{(N)} = (X_1,\dots,X_N) \in \sspace^{(N)} = \sspace^N$ and
$P_\theta^{(N)} = P_\theta^{\otimes N}$. Let
$\check\beta^{(N)}:\sspace^{(N)} \to \mb R$ be an estimator and
$\error_N(\theta)$ be its error. An estimator sequence is called
regular at a reference parameter value $\bar\theta \in \Theta$ if, for
any submodel $\phi \in \Phi(\bar\theta)$ and any fixed
$t \in [t_1,t_2]$,
$\sqrt{N}\Bk{\check\beta^{(N)}(X^{(N)}) - \beta(\phi(t/\sqrt{N}))}$
converges in distribution under $P_{\phi(t/\sqrt{N})}^{(N)}$ to a
random variable $Z$ with the same probability distribution
$Q(\bar\theta,\cdot)$ that does not depend on $\phi$ or $t$. This
definition roughly means that, for any $\theta$ close to the reference
$\bar\theta$, a regular estimator for a large enough $N$ can be
linearized as
\begin{align}
\check\beta^{(N)} &\approx \beta(\theta) + \frac{Z}{\sqrt{N}}
\end{align}
in terms of a noise random variable $Z$ that is insensitive to
$\theta$. For IID observations with some technical conditions on the
statistical model and any regular estimator, it can then be proved
that
\begin{align}
\liminf_{N\to \infty} N \error_N(\theta) &\ge \Bnd(\theta),
\label{nerror}
\end{align}
where $\Bnd$ is the efficiency bound for $N = 1$.  The
regular-estimator assumption can be further relaxed via the so-called
local asymptotic minimax theorem \cite{vaart} or the Bayesian
Cram\'er--Rao bound \cite{gill95}, although those bounds are too
complicated to explain here.

At least for a finite-dimensional parameter space, the
maximum-likelihood estimator is known to be regular and asymptotically
efficient under general conditions, in the sense that its error
saturates the bound in Eq.~(\ref{nerror}). The rough idea is that,
under $P_\theta^{(N)} = P_\theta^{\otimes N}$, the maximum-likelihood
estimator can be shown to obey
\begin{align}
\sqrt{N}\Bk{\check\beta^{(N)} - \beta(\theta)}
&\approx \frac{1}{\sqrt{N}} \sum_{n=1}^N \bk{\effinf}(X_n),
\end{align}
and the right-hand side converges by the central limit theorem to a
zero-mean normal random variable with variance
$\norm{\effinf}_{f(\theta)}^2 = \Bnd(\theta)$.

There remain two caveats with the asymptotic view. First, Eq.~(\ref{nerror})
means
\begin{align}
\error_N(\theta) \ge \frac{\Bnd(\theta)}{N} + o\bk{\frac{1}{N}}
\end{align}
with some term $o(1/N)$ that is asymptotically negligible relative to
$1/N$, and the efficiency bound offers no absolute limit on the term
for any finite $N$. Second, when $\Bnd(\theta) = \infty$, it merely
implies $\liminf_{N\to\infty}N\error_N(\theta) = \infty$; there may
exist estimators with errors that decrease at a slower rate, such as
$\error_N(\theta) = O(1/N^r)$ with $r < 1$. If one deems these caveats
important, they should consult Bayesian and minimax approaches
\cite{gill95,bcrb_pra,tsybakov}, albeit outside the scope of this
work.

While it is true that the efficiency bound is not perfect, it is often
easier to compute than Bayesian or minimax alternatives and remains
tractable for infinite-dimensional problems, as the literature as well
as this work have shown.

\section{\label{app_pinv}Pseudo-inverse for Hilbert-space maps}
Generalizations of the Moore--Penrose inverse for Hilbert-space maps
are called Tseng inverses \cite{benisrael,tseng49}. We focus on one
version called the maximal Tseng inverse and call it the
pseudo-inverse; all other Tseng inverses are restrictions of this
version to a smaller domain. We summarize its properties in the
following proposition.

\begin{proposition}[{Ref.~\cite[Sec.~9.3]{benisrael}}]
\label{prop_pinv}
Let $T$ be a densely defined and closed linear map with
$\dom T \subseteq H_1$ and codomain $H_2$, where $H_1$ and $H_2$ are
Hilbert spaces. Then there exists a unique pseudo-inverse linear map
\begin{align}
T^-: \range T \oplus \bk{\range T}^\perp \to H_1
\label{pinv_dom}
\end{align}
that satisfies the following properties:
\begin{align}
\kernel T^- &= \bk{\range T}^\perp,
&
\range T^- &= \bk{\dom T} \cap \bk{\kernel T}^\perp,
\label{pinv_props}
\\
T^- T u &= \Pi_{\overline{\range T^-}} u \quad \forall u \in \dom T,
&
T T^- u &= \Pi_{\overline{\range T}} u \quad \forall u \in \dom T^-,
\label{pinv_props2}
\end{align}
where $\perp$ denotes the orthocomplement relative to $H_1$ or $H_2$
and $\Pi_{K}:H \to H$ is the orthogonal projection onto the closed
subspace $K$ of $H = H_1$ or $H_2$.  Furthermore,
\begin{align}
\kernel T^-  &= \kernel T^\dual,
&
\kernel T^{\dual-} &= \kernel T,
&
T^{\dual-} &= T^{-\dual},
\\
T^{--} &= T,
&
\bk{T^\dual T}^- &= T^- T^{\dual-}, 
&
\bk{T T^\dual}^- &= T^{\dual -} T^-,
\end{align}
on the domains where the maps are defined.
\end{proposition}
In this work, we consider only the pseudo-inverse of bounded maps
with $\dom T = H_1$, and a bounded map is always closed and
densely defined. When $\dom T = H_1$, we also have
\begin{align}
(\dom T) \cap (\kernel T)^\perp = (\kernel T)^\perp.
\end{align}
We note, however, that $T^-$ need not be bounded when the dimensions
of $H_1$ and $H_2$ are both infinite.

In essence, $T^-y$ for a bounded $T$ on $H_1$ and
$y \in \range T \oplus (\range T)^\perp$ comprises three steps:
\begin{enumerate}
\item Null the component of $y$ in $(\range T)^\perp$. The result $y'$
  is in $\range T$.

\item Solve $Tx = y'$. Since $y' \in \range T$, a solution
  $x \in \dom T = H_1$ always exists. Given a solution $x$, any other
  solution is in the form $x + z$ for some $z \in \kernel T$.

\item Output the solution $x'$ with minimum norm
 \begin{align}
\norm{x'} = \inf_{x \in H_1:Tx=y'}\norm{x}.
 \end{align}
 This solution $x'$ is unique and in $(\kernel T)^\perp$.
\end{enumerate}
In this work, all the arguments of $T^-$ are in $\range T$ (or else we
set the output to be $\infty$ essentially), so the first step is never
involved. The second and final steps are the essential operations in
our use of the pseudo-inverse. It follows that
\begin{align}
T^- y &= \argmin_{x\in H_1:Tx = y} \norm{x}
\quad
\forall y \in \range T.
\label{pinv_min}
\end{align}
For finite-dimensional $H_1$ and $H_2$ with the Euclidean inner
product, the pseudo-inverse becomes the Moore--Penrose inverse
\cite{benisrael}.

\section{\label{app_sv}Some facts about singular values}
\begin{proposition}[{\cite[Chap.~1]{simon10}}]
\label{prop_sv}
Let $T:H_1 \to H_2$ be a compact map where $H_{j}$ are Hilbert spaces.
Let the SVD be $T = \sum_j s_j(T) \outsing_j \Avg{\insing_j,\cdot}$. Assume that the
singular values are sorted in descending order
$s_0(T) \ge s_1(T) \ge \dots$.
\begin{enumerate}
\item Min-max characterization:
\begin{align}
s_j(T) &= \min_{K:\dim K = j}\max_{x \in K^\perp, \norm{x} = 1} \norm{T x},
\label{minmax}
\end{align}
In other words, for any $j$th-dimensional subspace $K$ of $H_1$,
$s_j \le \max_{x \in K^\perp, \norm{x} = 1} \norm{T x}$ for the
orthocomplement $K^\perp$. The minimization is achieved when
$K = \{0\}$ and $K^\perp = H_1$ for $j = 0$ and
$K = \spn\{\insing_0,\dots,\insing_{j-1}\}$ for $j \ge 1$.

\item Let $T'$ be another bounded Hilbert-space map
  ($\norm{T'} < \infty$) with appropriate domain and codomain.  For
  any compact $T$ and any bounded $T'$,
\begin{align}
s_j\bk{T'T} &\le \Norm{T'} s_j\bk{T},
&
s_j\bk{TT'} &\le s_j(T) \Norm{T'}.
\end{align}
\item If both maps are compact,
\begin{align}
s_{j+k}\bk{T'T} &\le 
s_{j}\bk{T'} s_{k}\bk{T},
&
s_{j+k}(T+T') &\le s_j(T) + s_k(T').
\end{align}
\end{enumerate}
\end{proposition}

\section{\label{app_direct}Direct sum of vector spaces}
Let $\{\mc V_j:j = 1,2,\dots,N\}$ be a set of vector spaces. The
direct sum $\mc W = \oplus_j \mc V_j$ is the vector space where each
element is written as $\oplus_j A_j$ with $A_j \in \mc
V_j$. $\oplus_j$ is assumed to be linear, that is,
\begin{align}
c \bigoplus_j A_j &= \bigoplus_j c A_j, \quad \forall c \in \mb R \textrm{ or } \mb C,
&
\bigoplus_j A_j + \bigoplus_j B_j &= \bigoplus_j \bk{A_j + B_j}.
\end{align}
If each $\mc V_j$ is equipped with a semi-inner product,
then the semi-inner product of $\mc W$ is defined as
\begin{align}
\Avg{\bigoplus_j A_j,\bigoplus_j B_j} &= \sum_j \Avg{A_j,B_j},
\label{ds_inner}
\end{align}
so that the seminorm is 
\begin{align}
\Norm{\bigoplus_j A_j}^2 &= \sum_j \norm{A_j}^2.
\end{align}
If a product between two elements of $\mc V_j$ to produce another
element is defined, then the product in $\mc W$ is defined
component-wise:
\begin{align}
\bk{\bigoplus_j A_j} \bk{\bigoplus_j B_j} 
&= \bigoplus_j \bk{A_j B_j}.
\end{align}
For example, the direct sum $\mc V_1\oplus\mc V_2$ of two matrix
spaces with the Hilbert-Schmidt inner product can be expressed as
\begin{align}
A_1 \oplus A_2 &= \mqty(A_1 & 0\\ 0 & A_2),
\end{align}
where $0$ denotes a zero matrix of appropriate size. Then all the
above rules are satisfied.

The quotient and Hilbert spaces constructed from these vector spaces
also follow the direct-sum relation.
\begin{proposition}
\label{prop_direct}
  Let $\{\mc V_j\}$ be a finite set of vector spaces with semi-inner
  products, $\mc N$ be the set of zero-seminorm elements in
  $\mc W = \oplus_j \mc V_j$, and $\mc N_j$ be those in $\mc V_j$.
  Define the Hilbert spaces
\begin{align}
H &= \complete\bk{\mc W/\mc N},
&
H_j &= \complete\bk{\mc V_j/\mc N_j}.
\end{align}
Then
\begin{align}
\mc W/\mc N &= \bigoplus_j \mc V_j/\mc N_j,
\label{quo_direct}
\\
H &= \bigoplus_j H_j.
\label{H_direct}
\end{align}
\end{proposition}
To deal with a direct sum of Hilbert spaces, it is convenient to
define the injection maps and the projection maps as follows.
\begin{proposition}
\label{prop_inj}
Given Eq.~(\ref{H_direct}), define the injection map $\inj_j : H_j \to H$ as
\begin{align}
\inj_j u &= 0^{\oplus(j-1)} \oplus u \oplus 0^{\oplus(N-j)},
\end{align}
which injects the argument into the $j$th component of a vector in $H$
and sets all the other components zero.  Define also the projection
map $\pi_j:H \to H_j$ as
\begin{align}
\pi_j \bigoplus_k u_k &= u_j.
\end{align}
Then 
\begin{align}
\bigoplus_j u_j &= \sum_j \inj_j u_j,
&
\inj_j^\dual &= \pi_j,
&
\inj_j^\dual \inj_k = \pi_j \inj_k &= 
\begin{cases} I, & j = k,\\ 0, & j\neq k,\end{cases}
&
\inj_j \pi_j &= \Pi_j,
\label{inj_orth}
\end{align}
where $\Pi_j:H \to H$ is the projection operator onto $H_j$. 
\end{proposition}
We omit the proof of Prop.~\ref{prop_inj} since it is trivial.

Now let $K$ be a closed vector subspace of $H$. The orthocomplement
$K^\perp$ is defined as
\begin{align}
K^\perp &= \BK{u \in H: \Avg{u,v} = 0 \ \forall v \in K}.
\end{align}
It turns out that $K^\perp$ is also a closed vector subspace and $H$
becomes isomorphic to $K\oplus K^\perp$, in the sense that 
each $u \in H$ can be uniquely decomposed as
\begin{align}
u &= u_1 + u_2, \quad
u_1 \in K, \quad u_2 \in K^\perp,
\end{align}
and the inner product becomes
\begin{align}
\Avg{u,v} &= \Avg{u_1,v_1} + \Avg{u_2,v_2}
\end{align}
in terms of the decomposition. Uniqueness means that no other pair
$w_1 \in K$ and $w_2 \in K^\perp$ give the same $u$.  The uniqueness
of $u_1 \in K$ for a $u \in H$ defines a projection map from each $u$
to $u_1$; it is commonly defined as $\Pi_K: H \to H$ with $H$ as the
codomain, but the unique decomposition can just as well define a
projection map $\pi_K:H \to K$ with the subspace $K$ as the
codomain. The adjoint of $\pi_K$ is the injection map $\inj_K:K \to H$
by Prop.~\ref{prop_inj}.

\section{\label{app_num}Numerical analysis}
Numerical software packages assume the Euclidean inner product for
finite-dimensional vectors, that is, $\Avg{u,v} = \sum_j u_j v_j$.  To
compute the pseudo-inverse and the SVD proposed in this work
numerically, one must convert between the abstract vectors and the
Euclidean vectors properly. We adopt the following procedure:
\begin{enumerate}
\item Approximate each abstract vector in $\Amb(f)$ by a concrete
  representation in a finite-dimensional vector space
  $H_1$. Similarly, approximate each vector in $\Amb(g)$ by a
  concrete representation in a finite-dimensional $H_2$.

\item Define inner products for $H_{1,2}$ to approximate the inner
  products for $\Amb(f)$ and $\Amb(g)$.

\item Convert the representations in $H_{1,2}$ to Euclidean vectors in
  $H_{I,II}$ via operators $U:H_1\to H_I$ and $V:H_2 \to H_{II}$, such
  that the inner products for $H_{1,2}$ coincide with the Euclidean
  inner products for $H_{I,II}$.

\item Given $U:H_1 \to H_I$ and $V:H_2 \to H_{II}$, abstract maps on
  $\Amb(f)$ and $\Amb(g)$ can be converted to matrices on $H_{I,II}$.

\item Given the Euclidean vectors in $H_{I,II}$ and the matrices on
  $H_{I,II}$, the linear algebra can be computed numerically using
  standard routines. In particular, the Euclidean pseudo-inverse is
  the Moore--Penrose inverse.

\item Convert the Euclidean vectors $V \effinfout$, $\{U \insing_j\}$,
  and $\{V \outsing_j\}$ in $H_{I,II}$ back to the original
  representations in $H_{1,2}$ via $U^{-1}$ and $V^{-1}$ if necessary,
  for plotting for example.
\end{enumerate}

The following proposition summarizes the basic facts about the
conversions.

\begin{proposition}
\label{prop_iso}
Let $H_j$ be a Hilbert space with inner product
$\avg{\cdot,\cdot}_{j}$. $H_1$ and $H_I$ are isomorphic if and only if
there exists a unitary map $U:H_1\to H_I$ such that
\begin{align}
\Avg{u,v}_1 &= \Avg{U u,Uv}_I
\quad
\forall u,v \in H_1.
\end{align}
It follows that $U$ commutes with the orthocomplement, that is,
\begin{align}
\bk{U K}^\perp &=  U\bk{K^\perp} \quad \forall K \subseteq H_1.
\end{align}
Let $T:H_1 \to H_1$ be a linear map.  The map
\begin{align}
\mc U T &=  U TU^{-1}:H_I \to H_I
\end{align}
is an equivalent version of $T$ in the sense that
\begin{align}
\Avg{Uu,\bk{\mc U T} Uv}_{I} &= \Avg{u,T v}_1
\quad
\forall u,v \in H_1.
\end{align}
Let $H_2$ and $H_{II}$ be another pair of isomorphic Hilbert spaces
with unitary map $V:H_2 \to H_{II}$.  Define $\mc U$ for the following
maps as
\begin{align}
T:H_2 &\to H_2, & \mc U T &= V T V^{-1}:H_{II} \to H_{II},
\\
T:H_1 &\to H_2, & \mc U T &= V T U^{-1}:H_{I} \to H_{II},
\\
T:H_2 &\to H_1, & \mc U T &= U T V^{-1}:H_{II} \to H_I.
\end{align}
Then
\begin{align}
\mc U\bk{T R} &= \bk{\mc U T}\bk{\mc U R},
\\
\range\bk{\mc U T} &= U \range T
\quad
(\codom T = H_1), 
\\
\kernel\bk{\mc U T} &= U \kernel T
\quad
(\dom T = H_1),
\\
\mc U\bk{T^\dual} &= \bk{\mc UT}^\dual,
\\
\mc U\bk{T^-} &= \bk{\mc U T}^-.
\label{pinv_iso}
\end{align}
Moreover, the SVDs of $T:H_1 \to H_2$ 
and $\mc U T: H_I \to H_{II}$  are related by
\begin{align}
T &= \sum_j s_j \outsing_j \Avg{\insing_j,\cdot}_1,
&
\mc U T &= \sum_j s_j \bk{V \outsing_j} \Avg{U \insing_j,\cdot}_I.
\end{align}
\end{proposition}
\begin{proof}
  Eq.~(\ref{pinv_iso}) follows from the fact that $\mc U(T^-)$
  satisfies all the properties of the unique pseudo-inverse of
  $\mc U T$ given by Eqs.~(\ref{pinv_props})--(\ref{pinv_props2}) in
  Prop.~\ref{prop_pinv}. The rest of the proposition is basic.
\end{proof}

Proposition~\ref{prop_iso} means that, if $H_{I,II}$
are Euclidean spaces, we can compute
\begin{align}
\mc U(\chan^\pull) &= (\mc U \chan_\push)^\dual,
&
\mc U(\chan^{\pull-}) &= \Bk{\mc U\bk{\chan^\pull}}^{-} = 
\Bk{\mc U\bk{\chan_\push}}^{\dual-},
\label{num1}
\\
V \effinfout &= \Bk{\mc U \bk{\chan^{\pull-}}} U \effinf,
&
\Bnd_\out &= \norm{\effinfout}_g^2 = \norm{V\effinfout}^2,
\label{num2}
\\
\mc U \chan_\push &= \sum_j s_j \bk{V \outsing_j}\Avg{U\insing_j,\cdot}_I
\label{num3}
\end{align}
using standard numerical routines, once $\effinf$ and $\chan_\push$ or
$\chan^\pull$ are given.  $V \effinfout$, $\{U \insing_j\}$, and
$\{V \outsing_j\}$ can be converted back to the original representations in
$H_{1,2}$ for plotting if needed.

Next, we assume a maximal input score tangent space and derive
$\mc U\chan_\push = \mc U\chan_\dual$ or
$\mc U\chan^\pull = \mc U \chan^\dual$ for the examples in the main
text, so that Eqs.~(\ref{num1})--(\ref{num3}) can be computed
numerically.

\subsection{Classical coherent imaging}
In Sec.~\ref{sec_coh_c_exa}, each vector in the input $\mc V$ is a
function $A(x)$ of $x \in D = [-\Delta,\Delta]$. An obvious and
intuitive finite-dimensional representation of $A(x)$ is then
\begin{align}
A &= \mqty(A(x_1) \\ A(x_2) \\ \vdots ) \in H_1,
&
x_j &= x_0 + j\Delta x.
\label{A_imaging}
\end{align}
The inner product can be approximated as
\begin{align}
\Avg{A,B}_f &= \int_D A(x)B(x) \dd x \approx \sum_j A(x_j) B(x_j)\Delta x.
\end{align}
To convert $A \in H_1$ to a Euclidean $UA \in H_I$ through
$U:H_1\to H_I$, define
\begin{align}
\bk{UA}_j &= \sqrt{\Delta x} A(x_j).
\end{align}
Similarly, for each function $A(y)$ in $\mc V_\out$ with
$y \in D_\out$, let
\begin{align}
A &= \mqty(A(y_1) \\ A(y_2) \\ \vdots ) \in H_2,
&
y_j &= y_0 + j\Delta y, 
&
\bk{V A}_j &= \sqrt{\Delta y} A(y_j),
\label{A_imaging2}
\end{align}
where $V:H_2 \to H_{II}$. $\chan_\push$ should be discretized as the
matrix $\chan_\push:H_1 \to H_2$ with
\begin{align}
\bk{\chan_\push}_{jk} &= \bk{\chan_\dual}_{jk} = \chan(y_j|x_k) \Delta x.
\end{align}
The Euclidean version $\mc U \chan_\push: H_I \to H_{II}$ becomes
\begin{align}
\bk{\mc U \chan_\push}_{jk} &= \bk{V \chan_\push U^{-1}}_{jk} = 
\sqrt{\Delta y} \chan(y_j|x_k) \sqrt{\Delta x}.
\end{align}
The numerical results in Fig.~\ref{svd_coh_classical0} assume
$\Delta x = \Delta/100$ and $\Delta y = \Delta_\out/1000$ to ensure
that the discretization has minimal effects.

\subsection{Direct incoherent imaging}
In Sec.~\ref{sec_direct_exa}, each vector in $\mc V$ is a function
$A(x)$ of $x \in D = [-\Delta,\Delta]$. Assume the representation given by
Eqs.~(\ref{A_imaging}). The inner product in Eqs.~(\ref{inner_coh})
can be approximated as
\begin{align}
\Avg{A,B}_\theta &\approx \sum_j A(x_j) B(x_j) \theta(x_j) \Delta x.
\label{inner_direct}
\end{align}
The Euclidean version $U A \in H_I$ should then be defined as
\begin{align}
\bk{UA}_j &= \sqrt{\theta(x_j)\Delta x}  A(x_j).
\label{U_direct}
\end{align}
Similarly, for each function $A(y)$ in $\mc V_\out$, assume
\begin{align}
A &= \mqty(A(y_1) \\ A(y_2) \\ \vdots ) \in H_2,
&
y_j &= y_0 + j\Delta y, 
&
\bk{V A}_j &= \sqrt{ g(\theta,y_j)\Delta y} A(y_j).
\end{align}
The $\chan_\push$ given by Eq.~(\ref{push_incoh}) should then be
discretized as the matrix $\chan_\push:H_1 \to H_2$ with
\begin{align}
\bk{\chan_\push}_{jk} &= \bk{\chan_\dual}_{jk} = \frac{1}{g(\theta,y_j)} 
\chan(y_j|x_k) \theta(x_k) \Delta x.
\end{align}
The Euclidean version $\mc U\chan_\push:H_I\to H_{II}$ becomes
\begin{align}
\bk{\mc U \chan_\push}_{jk} &= \bk{V \chan_\push U^{-1} }_{jk} 
= \sqrt{\frac{\Delta y}{g(\theta,y_j)}} \chan(y_j|x_k) \sqrt{\theta(x_k) \Delta x}.
\end{align}
The numerical results in Figs.~\ref{svd_incoh_classical_gauss0},
\ref{svd_incoh_classical_gauss_tri0}, and
\ref{svd_incoh_classical_sinc0} assume $\Delta x = \Delta/500$ and
$\Delta y = \Delta_\out/500$ to ensure that the discretization has
minimal effects.

\subsection{\label{app_num_q}Quantum channel}
Consider the quantum channel in Example~\ref{exa_qq}. 
Diagonalizing the input density operator as
\begin{align}
f &= \sum_{j} \lambda_j \ket{j}\bra{j},
\end{align}
and taking the complex matrix representation of an operator $A$ as
\begin{align}
A_{jk} &= \bra{j} A \ket{k}
\end{align}
in terms of the eigenvectors of $f$, the complex operator inner
product can be approximated as the finite sum \cite{holevo_aspect}
\begin{align}
\trace\Bk{f \bk{A^\dual\jordan B}}
&\approx \sum_{jk} \frac{\lambda_j + \lambda_k}{2} \cconj{A}_{jk} B_{jk}.
\end{align}
To map each $J\times J$ self-adjoint matrix $A$ in $H_1$ to a
$J^2$-dimensional Euclidean vector in $H_I$, let
\begin{align}
(UA)_\mu &= \sum_{jk}U_{\mu, jk} A_{jk},
&
U_{\mu, jk} &= W_{\mu, jk}\sqrt{\frac{\lambda_j+\lambda_k}{2}},
\end{align}
where we require $W$ to be unitary in the sense that
\begin{align}
\sum_\mu \cconj{W}_{\mu,jk} W_{\mu,j'k'} &= \delta_{jk,j'k'},
\end{align}
so that the weighted inner product is converted to the standard
version given by
\begin{align}
\sum_{jk} \frac{\lambda_j + \lambda_k}{2} \cconj{A}_{jk} B_{jk}
&= \sum_{\mu} (\cconj{UA})_\mu (UB)_\mu.
\end{align}
To map self-adjoint matrices to real vectors, we also require
\begin{align}
W_{\mu,jk} &= \cconj{W}_{\mu,kj} = \sum_{j'k'} \cconj{W}_{\mu,j'k'} T_{j'k',jk},
\end{align}
where
\begin{align}
T_{j'k',jk} &= \delta_{j'k}\delta_{k'j}
\end{align}
is called the commutation matrix \cite{horn2}. These requirements mean
that, if we treat $jk$ as one index of a matrix, $W$ should satisfy
$W^\dual W = W W^\dual = I$ and $W = \cconj{W} T$.  The latter
condition is then equivalent to $W^\top W = T$. To find $W$, observe
that a set of $J(J+1)/2$ eigenvectors $\{u^a\} \cup \{v^{ab}\}$ of $T$
with eigenvalue $1$ is given by
\begin{align}
u_{jk}^a &= \delta_{ja} \delta_{ka}, & a &= 0,\dots,J-1,
&
v_{jk}^{ab} &= \frac{1}{\sqrt{2}} 
\bk{\delta_{ja}\delta_{kb} + \delta_{jb}\delta_{ka}}, & 0&\le a < b \le J-1,
\end{align}
while a set of $J(J-1)/2$ eigenvectors $\{w^{ab}\}$ of $T$ with
eigenvalue $-1$ is given by
\begin{align}
w_{jk}^{ab} &= \frac{1}{\sqrt{2}} 
\bk{\delta_{ja}\delta_{kb} - \delta_{jb}\delta_{ka}}, & 0 &\le a < b \le J-1.
\end{align}
These eigenvectors form an orthonormal basis of the Euclidean
$J^2$-dimensional $H_I$. We can then set
\begin{align}
T &= Q D Q^\top, \quad D = \mqty(I_{J(J+1)/2} & 0\\ 0 & -I_{J(J-1)/2})
= \mqty(I_{J(J+1)/2} & 0\\ 0 & -iI_{J(J-1)/2})^2,
\\
Q &= \mqty(u^{0} & \dots & v^{01} &\dots & w^{01} & \dots ),
\\
W &= \mqty(I_{J(J+1)/2} & 0\\ 0 & -iI_{J(J-1)/2})Q^\top = 
\mqty(u^{0} & \dots & v^{01} &\dots & -i w^{01} & \dots )^\top,
\end{align}
where $I_n$ is the $n\times n$ identity matrix. In summary, for a
matrix $A$, $\sum_{jk} W_{\mu,jk} A_{jk}$ maps the diagonal entries
$\{A_{aa} = \sum_{jk} u_{jk}^aA_{jk}\}$, the symmetric off-diagonal
entries
$\{(A_{ab} + A_{ba})/\sqrt{2} = \sum_{jk} v_{jk}^{ab}A_{jk}: a < b\}$,
and the antisymmetric off-diagonal entries
$\{(A_{ab} - A_{ba})/(\sqrt{2}i) = -i\sum_{jk} w_{jk}^{ab} A_{jk}: a <
b\}$ to a vector.  For a self-adjoint matrix,
$A_{ba} = \cconj{A}_{ab}$, and the resulting vector is real, as
desired.

Let the diagonal form of the output density operator be
\begin{align}
g &= \sum_{j} \lambda_j'\ket{j'}\bra{j'}.
\end{align}
The conversion $V:H_2 \to H_{II}$ of the output self-adjoint matrices
in $H_2$ to Euclidean vectors in $H_{II}$ can be defined similarly:
\begin{align}
A_{jk} &= \bra{j'}A\ket{k'},
&\bk{V A}_\mu &= \sum_{jk} V_{\mu,jk} A_{jk},
&
V_{\mu,jk} &= W_{\mu,jk} \sqrt{\frac{\lambda_{j}'+\lambda_{k}'}{2}}.
\label{V_q}
\end{align}
For a Markov-map dual in the Kraus form given by
Eq.~(\ref{kraus_dual}), we can write $\chan^\dual:H_2 \to H_1$ as
\begin{align}
\bk{\chan^\dual A}_{jk} &= \bk{\sum_a K_a^\dual A K_a}_{jk}
= \sum_{j'k'} \bk{\chan^\dual}_{jk,j'k'} A_{j'k'},
\label{dual_q}
\\
\bk{\chan^\dual}_{jk,j'k'} &= 
\sum_a \bra{j} K_a^\dual\ket{j'} \bra{k'} K_a\ket{k},
\end{align}
and $\mc U(\chan^\dual):H_{II} \to H_I$ as
\begin{align}
\Bk{\mc U\bk{\chan^\dual}}_{\mu\nu} &= \bk{U \chan^\dual V^{-1}}_{\mu\nu}
= \sum_{jk,j'k'}U_{\mu,jk} \bk{\chan^\dual}_{jk,j'k'} \bk{V^{-1}}_{j'k',\nu}.
\label{dual_q2}
\end{align}
As $\chan^\dual$ maps self-adjoint matrices to self-adjoint matrices,
our choice of $U$ and $V$ implies that $\mc U(\chan^\dual)$ maps real
vectors to real vectors, or in other words, $\mc U(\chan^\dual)$ must
be a real matrix.

Assuming the maximal input score tangent space $\mc T(f) = \Amb(f)$,
$\chan^\pull = \chan^\dual$ by Lemma~\ref{lem_score}.  Then
Eqs.~(\ref{num1})--(\ref{num3}) can be computed numerically.  It is
remarkable that the formulas do not require any explicit
parametrization or solving the Lyapunov equation for the score
operators.

\subsection{Quantum model of incoherent imaging}
In Sec.~\ref{sec_incoh_q_exa}, the classical-quantum channel given by
Eq.~(\ref{cq}) is studied. Assuming Eqs.~(\ref{inner_direct}) and
(\ref{U_direct}) for the input spaces and Eqs.~(\ref{V_q}) for the
output spaces, the Markov-map dual given by Eq.~(\ref{pull_cq}) can be
expressed as
\begin{align}
\bk{\chan^\dual A}_\mu &= 
\sum_{j'k'} \bk{\chan^\dual}_{\mu,j'k'}A_{j'k'} ,
&
\bk{\chan^\dual}_{\mu,j'k'} &= \bra{\psi} e^{ikx_\mu}\ket{j'}\bra{k'} e^{-ikx_\mu}\ket{\psi},
\end{align}
which is a special case of Eqs.~(\ref{dual_q}). $\mc U(\chan^\dual):H_{II} \to H_I$
becomes
\begin{align}
\Bk{\mc U\bk{\chan^\dual}}_{\mu\nu} &= \bk{U \chan^\dual V^{-1}}_{\mu\nu}
= \sum_{j'k'}\sqrt{\theta(x_\mu)\Delta x} \bk{\chan^\dual}_{\mu,j'k'} 
\bk{V^{-1}}_{j'k',\nu},
\end{align}
which is a special case of Eq.~(\ref{dual_q2}).  The numerical results
in Figs.~\ref{svd_incoh_q_gauss0} and \ref{svd_incoh_q_sinc0} assume
$\Delta x = \Delta/500$ and $16$ dimensions for the mode Hilbert space
$\mc H_\out$.

To verify the pseudo-inverse method proposed here, we apply it to the
moment estimation problem studied in
Ref.~\cite{tan23}. Ref.~\cite{tan23} used a parametric-submodel method
to compute the quantum efficiency bounds when the parameter of
interest is given by Eq.~(\ref{beta_incoh}) and the weight function
$b_\mu(x)$ is $b_\mu(x) = x^\mu + o(\Delta^\mu)$ in the leading
order. Here, we simplify the setting slightly by assuming
$b_\mu(x) = x^\mu$ and the PSF given by Eq.~(\ref{sinc2}) with
$K = 1$. Fig.~\ref{moment_QCRB} compares the quantum efficiency bounds
computed by the submodel method in Ref.~\cite{tan23} with the same
bounds computed by the pseudo-inverse method. The results are
indistinguishable, offering support that both methods are valid.

\fig{}{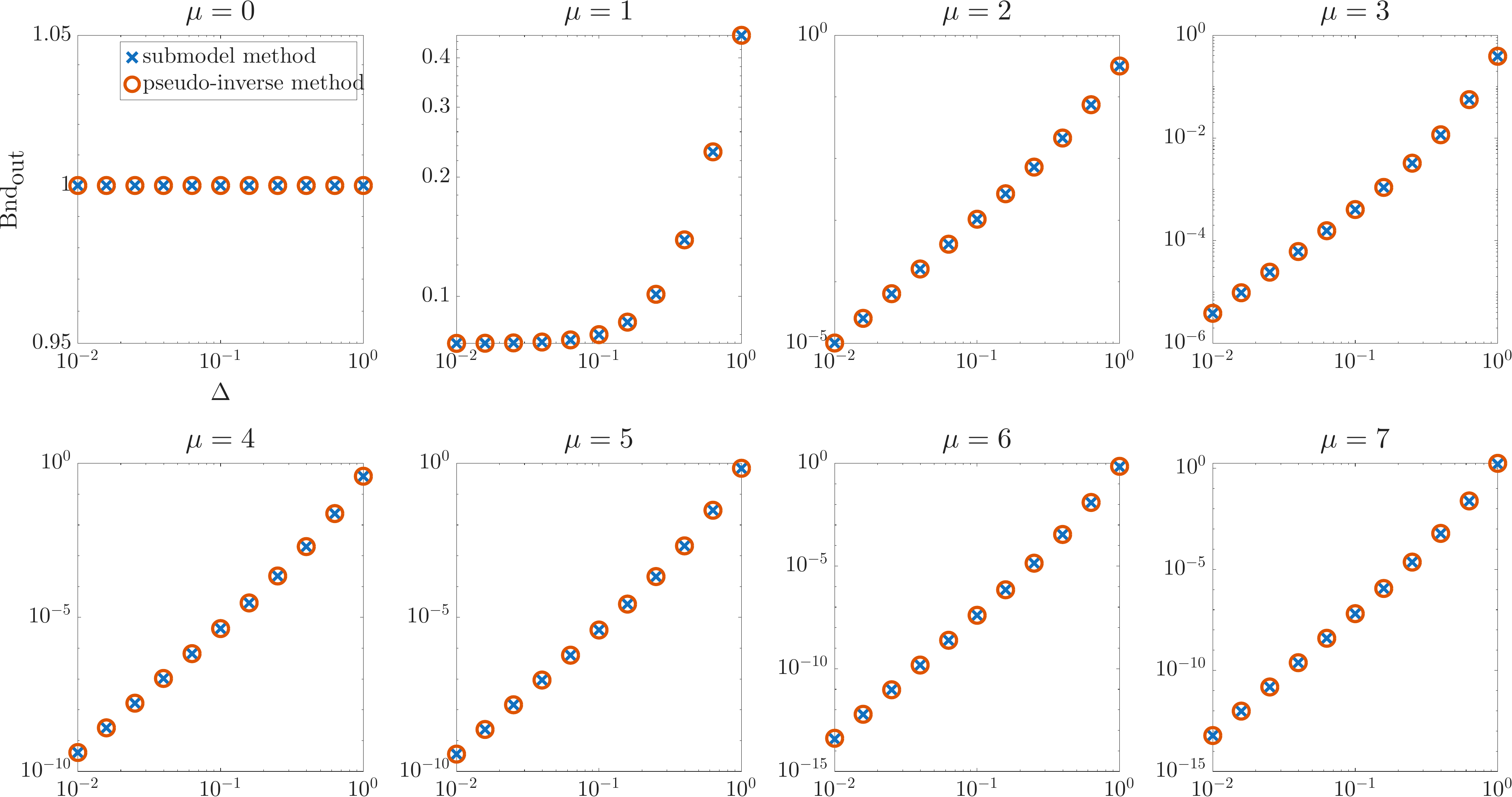}{Quantum efficiency bounds $\Bnd_\out$ for moment
  estimation versus the object size $\Delta$ in log-log scale.  The
  crosses are computed by the submodel method in Ref.~\cite{tan23}
  while the circles are computed by the pseudo-inverse method proposed
  here. The true input density $\theta(x)$ is assumed to be the
  flat-top function given by Eq.~(\ref{flattop}). The title of each
  plot denotes the order of the moment $\mu$, such that the weight
  function is $b_\mu(x) = x^\mu$.  Both methods assume the
  Hilbert-space dimension $\dim\mc H = 16$ for the quantum operators.
  The submodel method further assumes $10$ parameters for the
  submodel, as per Ref.~\cite{tan23}.}

\subsection{Spatial-mode demultiplexing (SPADE) for incoherent imaging}
In Sec.~\ref{sec_spade}, SPADE for incoherent imaging
as per Eqs.~(\ref{ipad1})--(\ref{chan_spade}) is studied.
Assume Eqs.~(\ref{inner_direct}) and (\ref{U_direct}) for the input spaces.
For the output spaces, let
\begin{align}
A &= \mqty(A(1,0)\\ A(1,1) \\ \vdots \\ A(2,0) \\ A(2,1)\\ \vdots) \in H_2,
&
\bk{VA}_{\tau n} &= \sqrt{g(\theta,\tau,n)} A(\tau,n),
\end{align}
treating $(\tau,n)$ as one index for a vector or a matrix. It follows
that
\begin{align}
\bk{\mc U \chan_\push}_{\tau n,k} &= \bk{\mc U \chan_\dual}_{\tau n,k} 
= \frac{1}{\sqrt{g(\theta,\tau,n)}}
\chan(\tau,n|x_k) \sqrt{\theta(x_k)\Delta x}.
\end{align}

\section{\label{app_qcoh}Quantum model of coherent imaging}

\subsection{Hilbert spaces}
We follow the quantum Gaussian models in Examples~\ref{exa_qgauss} and
\ref{exa_qgg} as well as Secs.~\ref{sec_gauss} and
Sec.~\ref{sec_gg}. Let the mode space $\mc H$ for a paraxial
quasimonochromatic field of one polarization be the space of complex
functions on the object plane $D \subseteq \mb R^m$ with inner product
\begin{align}
\Avg{\alpha,\gamma}_{\mc H} &= \int_D \cconj{\alpha(x)} \gamma(x) \dd^m x.
\end{align}
Each $\alpha \in \mc H$ specifies a spatiotemporal mode on the object
plane. Following Yuen and Shapiro \cite{yuen78}, define a point
annihilation operator $C(x)$ on the Fock space $\fock(\mc H)$ that
obeys
\begin{align}
\Bk{C(x),C(x')^\dual} &= \delta^m(x-x'),
&
C(x) \ket{\alpha} &= \alpha(x) \ket{\alpha},
\end{align}
where $\ket{\alpha}$ is a coherent state with $\alpha \in \mc H$ and
we have used the bra-ket notation for $\fock(\mc H)$, following
Appendix~\ref{app_fock}. The modal annihilation operator $a(\alpha)$
defined in Appendix~\ref{app_fock} becomes
\begin{align}
a(\alpha) &= \int_D \cconj{\alpha(x)} C(x) \dd^m x.
\end{align}
With the photon-number operator
\begin{align}
M &= \int_D C(x)^\dual C(x) \dd^m x,
\end{align}
the average photon number for a coherent state is
\begin{align}
\bra{\alpha}M\ket{\alpha}
&= \int_D \abs{\alpha(x)}^2 \dd^m x
= \norm{\alpha}^2_{\mc H}.
\end{align}
Similarly, let $\mc H_\out$ be the space of complex functions on the
image plane $D_\out \subseteq \mb R^m$ with inner product
\begin{align}
\Avg{\alpha,\gamma}_{\mc H_\out} &=  \int_{D_\out} \cconj{\alpha(y)} \gamma(y) \dd^m y.
\end{align}
The point annihilation operator $C_\out(y)$ and the modal annihilation
operator $a_\out(\alpha)$ on $\fock(\mc H_\out)$ are defined in the
same manner as $C(x)$ and $a(\alpha)$, respectively.  A full
nonparaxial vectoral theory is possible via Eqs.~(\ref{inner_EM}), but
we focus on the paraxial theory for simplicity.

Let us follow Appendix~\ref{app_fock} to transform the complex mode
spaces $\mc H$ and $\mc H_\out$ to real phase spaces $\mc K$ and
$\mc K_\out$. Define the $\conj$, $\Real$, and $\imag$ maps on $\mc H$
as
\begin{align}
\bk{\conj \alpha}(x) &= \cconj{\alpha(x)},
&
\bk{\Real \alpha}(x) &= \Re\Bk{\alpha(x)},
&
\bk{\imag \alpha}(x) &= \Im\Bk{\alpha(x)}.
\end{align}
The $\conj$, $\Real$, and $\imag$ maps on $\mc H_\out$ are defined
likewise.  Let $U:\mc H \to \mc K$ and $V:\mc H_\out \to \mc K_\out$
be the realification maps defined by Eq.~(\ref{realification}). Then the
complex mode functions in $\mc H$ and $\mc H_\out$ are related to the
phase-space mode functions in $\mc K$ and $\mc K_\out$ by
\begin{align}
U\alpha &= \mqty(\Re[\alpha(x)] \\ \Im[\alpha(x)]),
&
U^{-1} A &= A'(x) + i A''(x),
&
V\alpha &= \mqty(\Re[\alpha(y)] \\ \Im[\alpha(y)]),
&
V^{-1} A &= A'(y) + i A''(y).
\end{align}
The inner products of $\mc K$ and $\mc K_\out$ become
\begin{align}
\Avg{A,B}_{\mc K}&= \Re \Avg{U^{-1}A,U^{-1}B}_{\mc H}
= \int_D \mqty(A'(x) & A''(x))\mqty(B'(x)\\ B''(x)) \dd^m x,
\label{inner_qcoh}
\\
\Avg{A,B}_{\mc K_\out}
&= \int_{D_\out} \mqty(A'(y) & A''(y))\mqty(B'(y)\\ B''(y)) \dd^m y.
\label{inner_qcoh2}
\end{align}
The modal quadrature observable $X(A)$ on $\fock(\mc H)$ at the input
is given by
\begin{align}
X(A) &= \frac{a(U^{-1}A) + a(U^{-1}A)^\dual}{\sqrt{2}}
= \frac{1}{\sqrt{2}}\int_D \Bk{A'(x) - i A''(x)} C(x) \dd^m x + \textrm{H.c.}
\\
&= \int_D \mqty(A'(x) & A''(x)) \mqty(Q'(x)\\ Q''(x)) \dd^m x,
\end{align}
where H.c.\ denotes the Hermitian conjugate (i.e., adjoint) of the
first term and
\begin{align}
Q'(x) &= \frac{C(x)+C(x)^\dual}{\sqrt{2}},
&
Q''(x) &= \frac{C(x)-C(x)^\dual}{\sqrt{2}i}
\end{align}
are the point quadrature operators for the object plane. Similarly,
the modal quadrature observable $Y(A)$ on $\fock(\mc H_\out)$ at the
output is given by
\begin{align}
Y(A) &= \frac{a_\out(V^{-1}A) + a_\out(V^{-1}A)^\dual}{\sqrt{2}}
= \frac{1}{\sqrt{2}}\int_{D_\out} \bk{A' - i A''} C_\out \dd^m y + \textrm{H.c.}
\\
&= \int_{D_\out} \mqty(A' & A'')  \mqty(Q_\out'\\ Q_\out'') \dd^m y,
\label{Y_coh}
\end{align}
where
\begin{align}
Q_\out'(y) &= \frac{C_\out(y)+C_\out(y)^\dual}{\sqrt{2}},
&
Q_\out''(y) &= \frac{C_\out(y)-C_\out(y)^\dual}{\sqrt{2}i}
\end{align}
are the point quadrature operators for the image plane.

\subsection{Coherent states}
Assume that the imaging system comprises passive linear optics and a
zero-temperature environment. It is a Gaussian channel
\cite{holevo_info}. Let us first consider a coherent state
$\ket{\alpha} \in \fock(\mc H)$ for the input. The output state turns
out to be a coherent state $\ket{\gamma} \in \fock(\mc H_\out)$ as
well \cite{tsang_intro}.  A linear map $S:\mc H\to \mc H_\out$ relates
the input mean field $\alpha \in \mc H$ to the output mean field
$\gamma \in \mc H_\out$ by
\begin{align}
\gamma &= \scatter\alpha,
&
\gamma(y) &= \int_D \scatter(y|x)\alpha(x) \dd^m x,
\label{imaging}
\end{align}
where $\scatter(y|x)$ is the PSF of the imaging system.  As the system
is lossy, we must have
\begin{align}
\norm{\scatter\alpha}_{\mc H_\out}^2 &\le \norm{\alpha}_{\mc H}^2
\quad
\forall \alpha \in \mc H,
\end{align}
meaning that the operator norm of $\scatter$ obeys
\begin{align}
\norm{\scatter}_{\mc H} \le 1
\label{S_norm}
\end{align}
in terms of the norms of $\mc H$ and $\mc H_\out$. We can convert
Eq.~(\ref{imaging}) to an input-output relation for the phase-space
mean fields $f = U\alpha \in \mc K$ and $g = V\gamma \in \mc K_\out$
given by
\begin{align}
g &= \chan f,
&
\chan &= V \scatter U^{-1}.
\end{align}
To be more explicit, we can write $\chan$ in block form as
\begin{align}
\chan A &= \mqty(\Real \\ \imag) \scatter \mqty(1 & i) \mqty(A'\\ A'')
=\mqty((\scatter+\conj \scatter)/2 & -(\scatter - \conj \scatter)/2i\\ 
(\scatter-\conj\scatter)/2i & (\scatter+\conj \scatter)/2)  \mqty(A'\\ A'')
\\
&= \mqty((\scatter+\conj\scatter\conj)/2 & -(\scatter-\conj\scatter\conj)/2i\\
(\scatter-\conj\scatter\conj)/2i & (\scatter+\conj\scatter\conj)/2) \mqty(A' \\ A''),
\label{chan_rep}
\end{align}
where Eq.~(\ref{chan_rep}) has introduced a $\conj$ after each
$\conj\scatter$ to ensure that $\conj\scatter\conj$ remains a linear
map on $\mc H$.  This matrix of maps is a representation of $\scatter$
in the same way $\mqty(\Re z & -\Im z\\ \Im z & \Re z)$ is a
representation of the complex number $z$ in the multiplicative group
\cite{fomenko}. In terms of functions, we can write
\begin{align}
\bk{\chan A}(y) &= \int_D \mqty( \Re \scatter & -\Im \scatter\\
\Im \scatter & \Re \scatter)\mqty(A'(x) \\ A''(x))
 \dd^m x,
\label{chan_qgauss}
\\
\Re \scatter &= \Re\Bk{\scatter(y|x)}, \quad \Im \scatter = \Im\Bk{\scatter(y|x)}.
\end{align}
The phase-space adjoint $\chan^\dual:\mc K_\out \to \mc K$ becomes
\begin{align}
\bk{\chan^\dual A}(x) &= 
\int_{D_\out} \mqty( \Re \scatter & \Im \scatter\\-\Im \scatter & \Re \scatter) \mqty(A'(y) \\ A''(y))\dd^m y.
\label{chan_dual_qgauss}
\end{align}
These integrals can be discretized as matrix relations in numerical
analysis.

With 
\begin{align}
\Sigma &= \frac{I}{2}, & \Sigma_\out &= \frac{I}{2}
\label{covar_coh}
\end{align}
for coherent states, we find
\begin{align}
\Avg{\cdot,\cdot}_{\mc V} &= \frac{1}{2} \Avg{\cdot,\cdot}_{\mc K},
&
\Avg{\cdot,\cdot}_{\mc V_\out} &= \frac{1}{2} \Avg{\cdot,\cdot}_{\mc K_\out}.
\label{inner_coh}
\end{align}
The inner products of the score tangent spaces $\mc V$ and
$\mc V_\out$ thus become proportional to Eqs.~(\ref{inner_qcoh}) and
(\ref{inner_qcoh2}). They resemble the inner product given by
Eq.~(\ref{inner_Hout}) for a classical real field and white noise in
Sec.~\ref{sec_coh}, except that here we have a direct sum of real
functions $A'$ and $A''$ for each phase-space mode function
$A$. Following Sec.~\ref{sec_gg}, we also find
\begin{align}
\chan_\push &= \chan_\dual = V \scatter U^{-1},
\label{push_coh}
\\
\chan^\pull &= \chan^\dual = U \scatter^\dual V^{-1},
\label{pull_coh}
\\
\IMAP &= U \scatter^\dual \scatter U^{-1},
\label{IMAP_coh}
\end{align}
where $\scatter^\dual:\mc H_\out \to \mc H$ is the adjoint of
$\scatter$ with respect to $\mc H$ and $\mc H_\out$.
Eqs.~(\ref{push_coh}) and (\ref{pull_coh}) imply that the integrals in
Eqs.~(\ref{chan_qgauss}) and (\ref{chan_dual_qgauss}) serve as
explicit formulas for $\chan_\push$ and $\chan^\pull$, similar to
Eqs.~(\ref{chan_push_coh}) and (\ref{chan_pull_coh}) in the classical
case.

Using Eqs.~(\ref{S_norm}), (\ref{inner_coh}), and (\ref{push_coh}) and
the fact that $U$ and $V$ are isometric as per Eq.~(\ref{real_iso}),
we obtain
\begin{align}
\norm{\chan_\push}_{\mc V} = \norm{\chan}_{\mc V} = \norm{\chan}_{\mc
  K} = \norm{\scatter}_{\mc H} \le 1,
\label{gauss_mono}
\end{align}
where each subscript emphasizes the input Hilbert space on which the
operator norm is defined. The channel is hence monotonic.

\subsection{Gaussian states}
Now consider an arbitrary Gaussian input state specified by a certain
covariance map $\Sigma$, and suppose that the channel remains the
same. Then the channel maps $\scatter$ and $\chan$ for the mean fields as
well as the noise covariance map $\Sigma_\noise$ should all remain the
same. To derive $\Sigma_\noise$, we combine
Eqs.~(\ref{chan_gauss_cov}) and (\ref{covar_coh}) for coherent states
to obtain
\begin{align}
\Sigma_\noise &= \frac{1}{2}\bk{I-\chan \chan^\dual}.
\end{align}
$\Sigma_\out$ obeys Eq.~(\ref{chan_gauss_cov}) in terms of this
$\Sigma_\noise$ in general, leading to
\begin{align}
\Sigma_\out &= \chan \Sigma \chan^\dual + \Sigma_\noise 
= \chan \Sigma \chan^\dual + \frac{1}{2}\bk{I-\chan \chan^\dual},
\label{covar_gauss}
\end{align}
which is a dissipation-fluctuation relation. Eq.~(\ref{pull_coh})
remains the same, but Eqs.~(\ref{push_coh}) and (\ref{IMAP_coh})
should be changed to
\begin{align}
\chan_\push &= \Sigma_\out^{-1}\chan \Sigma 
= \Sigma_\out^{-1} V S U^{-1} \Sigma,
&
\IMAP &= \chan^\pull \chan_\push = 
U \scatter^\dual V^{-1} \Sigma_\out^{-1} V \scatter U^{-1} \Sigma.
\end{align}
The SVD of $\chan_\push$ with respect to the inner products of $\mc V$
and $\mc V_\out$ can be evaluated, at least numerically, to discover
the singular functions and singular values that characterize the
imaging system. The principles remain the same as the classical case
in Sec.~\ref{sec_coh}, but because the phase spaces have a more
complicated direct-sum structure and the noise may no longer be white,
the workings become much more tedious and we do not compute the SVD
explicitly here.

Following Sec.~\ref{sec_gauss}, suppose that the unknown parameter
$\theta$ is the input mean field in phase space. Let the parameter of
interest be
\begin{align}
\beta(\theta) &= \Avg{b,\theta}_{\mc K} = \int_D \mqty(b'(x) & b''(x))
\mqty(\theta'(x)\\ \theta''(x)) \dd^m x
\end{align}
in terms of a weight function $b \in \mc K$.  The condition given by
Eq.~(\ref{effinf_range}) for the existence of $\effinfout$ in
Theorem~\ref{thm_effinfout} becomes
\begin{align}
\exists u \in \Amb(g): \bk{\chan^\dual u}(x) = 
\int_{D_\out} \mqty( \Re \scatter & \Im \scatter\\-\Im \scatter 
& \Re \scatter) \mqty(u'(y) \\ u''(y))\dd^m y
= \mqty(b'(x)\\ b''(x)).
\end{align}
Provided that the condition is satisfied, we can find
\begin{align}
(\effinfout)(y) &= \mqty(
(\effinfout)'(y)\\ (\effinfout)''(y)) = \bk{\chan^{\dual-} b}(y).
\end{align}
$\effinfout$ specifies a spatiotemporal mode on the image plane; its
quadrature $Y(\effinfout)$ given by Eq.~(\ref{Y_coh}) is an unbiased
and efficient estimator, as per Sec.~\ref{sec_gg}. The quadrature can
be measured by homodyne detection with a local oscillator that is
shaped according to $\effinfout$.  If optical squeezing is used
\cite{yuen78,shapiro79,kolobov_fabre,beskrovnyy05,pinel12,taylor16,treps},
clearly one should aim to squeeze the quadrature being measured. While
it may seem difficult to introduce squeezing that can survive a lossy
imaging system, some groups have met the experimental challenge
\cite{taylor16,treps}.

\subsection{\label{sec_scaling_q}Scaling}
Unlike Sec.~\ref{sec_scaling_c} in the classical case, scaling
coefficients in the quantum case must respect physics and thus require
more care to model. Three types of scaling at different stages of the
system are possible:
\begin{enumerate}
\item If the input mean field $f(\theta)$ is scaled by a given
  constant $c_1 > 0$, it can be modeled by assuming
\begin{align}
f(\theta,c_1) &= c_1 \theta,
\end{align}
while the unknown parameters $\theta$ and $\beta(\theta)$ remain fixed
for any $c_1$. Then
\begin{align}
\score(\dot\phi,c_1) &= c_1 \dot\phi(\theta) = \score(\dot\phi,1),
&
\bk{\effinf}(\cdot,c_1) &= \frac{1}{c_1} \bk{\effinf}(\cdot,1),
\\
\bk{\effinfout}(\cdot,c_1) &= \frac{1}{c_1} \bk{\effinfout}(\cdot,1),
&
\Bnd_\out(\cdot,c_1) &= \frac{1}{c_1^2} \Bnd_\out(\cdot,1).
\end{align}

\item If the channel map is scaled by a given transmission coefficient
  $c_2 > 0$ as per
\begin{align}
\chan(c_2) = c_2\chan(1),
\end{align}
$c_2$ must be within a range $\mc C_2$ such that the channel remains
monotonic for all $c_2 \in \mc C_2$.

Suppose that $\chan(c_2)$ has been confirmed to be monotonic for all
$c_2 \in \mc C_2$. For coherent states, the covariance maps are given
by Eqs.~(\ref{covar_coh}) and not affected by $\chan$, so we can still
write
\begin{align}
\bk{\effinfout}(\cdot,c_2) &= \frac{1}{c_2} \bk{\effinfout}(\cdot,1),
&
\Bnd_\out(\cdot,c_2) &= \frac{1}{c_2^2} \Bnd_\out(\cdot,1),
\end{align}
similar to the classical case given by Eqs.~(\ref{effinf_scale}). For
other Gaussian states, however, the output covariance $\Sigma_\out$ is
affected by $c_2$ owing to the dissipation-fluctuation relation given
by Eq.~(\ref{covar_gauss}), so there is no simple scaling law for
varying $c_2$ any more.

\item If the input covariance is scaled by a given constant $c_3$ as
  per
\begin{align}
\Sigma(c_3) = c_3 \Sigma(1),
\end{align}
$c_3$ must also be restricted to a certain range so that $\Sigma(c_3)$
respects Eq.~(\ref{heisenberg}) and the Gaussian quantum state remains
physical. $\Sigma_\out$ depends nontrivially on $\Sigma(c_3)$ and thus
$c_3$, such that no simple scaling law is available.
\end{enumerate}

\section{\label{app_proofs}Proofs}
\begin{proof}[Proof of Lemma~\ref{lem_h}]
$\chan_\push \range f_\push \subseteq \chan_\push \overline{\range f_\push}$ implies
\begin{align}
\overline{\chan_\push \range f_\push} \subseteq \overline{\chan_\push \overline{\range f_\push} }
= \overline{\chan_\push \mc T(f)}.
\end{align}
The left-hand side is $\mc T(g)$ by Eq.~(\ref{Tout}).  To prove the
opposite inclusion, we note that, because $\chan_\push$ is continuous,
\begin{align}
\chan_\push \overline{\range f_\push} &\subseteq \overline{\chan_\push \range f_\push},
&
\overline{\chan_\push \overline{\range f_\push}} &\subseteq \overline{\chan_\push \range f_\push},
\end{align}
which implies $\overline{\chan_\push \mc T(f)} \subseteq \mc T(g)$.
\end{proof}

\begin{proof}[Proof of Theorem~\ref{thm_effinfout}]
  If $\effinf$ is not in the range of $\chan^\pull$, then
  $\effinfout$ cannot exist, and $\Bnd_\out = \infty$ by
  Theorem~\ref{thm_effinf}.  
To prove Eq.~(\ref{effinf_pinv}), 
first note that, by basic Hilbert-space theory,
\begin{align}
 \bk{\kernel \chan^\pull}^\perp &= \overline{\range \chan_\push}
 = \mc T(g),
 \label{null_adj}
 \end{align}
 where $\kernel$ denotes the null space of the map, $\perp$ denotes the
 orthocomplement relative to the ambient Hilbert space, and the last
 step follows from Lemma~\ref{lem_h}. Then use
 Eq.~(\ref{pullback_effinf}) and Prop.~\ref{prop_pinv} to find
\begin{align}
\chan^{\pull -} \effinf
= \chan^{\pull -} \chan^\pull \effinfout
= \Pi_{(\kernel \chan^\pull)^\perp} \effinfout
= \Pi_{\mc T(g)}\effinfout = \effinfout,
\end{align}
where the last step follows from the fact $\effinfout \in \mc T(g)$.
Eq.~(\ref{Cout}) then follows from Theorem~\ref{thm_effinf}.
\end{proof}

\begin{proof}[Proof of Lemma~\ref{lem_effinf}]
Eq.~(\ref{effinf_J}) implies Eq.~(\ref{effinf_range2}) because
\begin{align}
\effinfout &= \chan^{\pull -} \effinf \in 
\chan^{\pull -} \chan^\pull \range \chan_\push 
= \Pi_{\overline{\range \chan^{\pull -}}} \range \chan_\push \\
&= \Pi_{(\kernel \chan^\pull)^\perp}\range \chan_\push 
= \Pi_{\mc T(g)}\range \chan_\push = \range \chan_\push,
\end{align}
where we have used Eq.~(\ref{pullback_effinf}), Prop.~\ref{prop_pinv},
Eq.~(\ref{null_adj}), and Lemma~\ref{lem_h}. Conversely,
Eqs.~(\ref{pullback_effinf}) and (\ref{effinf_range2}) imply
\begin{align}
\effinf &= \chan^{\pull}  \effinfout \in \chan^\pull \range \chan_\push,
\end{align}
which is Eq.~(\ref{effinf_J}).  Eq.~(\ref{effinf_range2}) implies
Eq.~(\ref{effinf_T}) because
$\range \chan_\push \subseteq \overline{\range \chan_\push}$.
Theorem~\ref{thm_effinfout} has proved that Eqs.~(\ref{effinf_T}) and
(\ref{effinf_range3}) are equivalent.
\end{proof}

\begin{proof}[Proof of Corollary~\ref{cor_info}]
Given Eq.~(\ref{effinf_range2}), Eq.~(\ref{Cout}) becomes
\begin{align}
\Bnd_\out &= \Avg{\chan^{\pull -} \effinf,\effinfout}_g
= \Avg{\effinf,\chan_\push^{-} \effinfout}_f
\\
&= \Avg{\effinf,\chan_\push^{-} \chan^{\pull -}\effinf}_f
= \Avg{\effinf, \bk{\chan^\pull \chan_\push}^- \effinf}_f,
\end{align}
where the second equality holds because
$\effinfout \in \range \chan_\push \subseteq \dom \chan_\push^-$ and
$(\chan^{\pull -})^\adj = \chan_\push^-$, while the final equality comes
from $\chan_\push^{-} \chan^{\pull -} = \bk{\chan^\pull \chan_\push}^-$, following
Prop.~\ref{prop_pinv}. Eq.~(\ref{effinf_J}) implies
\begin{align}
\effinf \in \range \imap \subseteq \dom \imap^-,
\end{align}
so $\imap^-\effinf$ is well defined.
\end{proof}
\begin{proof}[Proof of Corollary~\ref{cor_finite}]
  Eq.~(\ref{imap_matrix}) shows that $\imap = \Imat$.  Since
\begin{align}
\range \chan_\push = \BK{\chan_\push u:u \in \mc T(f) = \mb R^p},
\end{align}
its dimension is finite,
$\range \chan_\push = \overline{\range \chan_\push} = \mc T(g)$, and
all the conditions in Lemma~\ref{lem_effinf} are equivalent. Then
$\effinf =\partial\beta\notin \range \Imat$ implies the failure of
Eq.~(\ref{effinf_range}), leading to $\Bnd_\out = \infty$ by
Theorem~\ref{thm_effinfout}. Otherwise, Eq.~(\ref{effinfout_finite})
satisfies the definition of the output efficient influence, since
\begin{align}
\effinfout \in \mc T(g) = \range \chan_\push = \spn\BK{\score}
\end{align}
by Eq.~(\ref{push_finite}) and
\begin{align}
\chan^\pull \effinfout  &= \Imat \Imat^+ \partial\beta = \Pi_{\range \Imat} \partial\beta
= \partial\beta = \effinf
\end{align}
by Eq.~(\ref{pull_finite}), Prop.~\ref{prop_pinv}, and the fact that
$\Imat^-$ for the Euclidean inner product $\avg{\cdot,\cdot}_f$ is the
Moore--Penrose inverse $\Imat^+$ \cite{benisrael}.
Corollary~\ref{cor_info} leads to Eq.~(\ref{Bnd_finite}) via
Eqs.~(\ref{effinf_kid}) and (\ref{imap_matrix}).
\end{proof}
\begin{proof}[Proof of Theorem~\ref{thm_mono}]
  By basic Hilbert-space theory, $\chan^\pull$ defined by
  Eq.~(\ref{adj}) has the same operator norm as that of $\chan_\push$, so
\begin{align}
\norm{\chan^\pull} = \norm{\chan_\push} \le 1.
\end{align}
Theorem~\ref{thm_effinfout} and Eq.~(\ref{pullback_effinf}) then imply
\begin{align}
\Bnd &= \norm{\effinf}_f^2 = \norm{\chan^\pull\effinfout}_f^2 
\le \norm{\chan^\pull}^2 \norm{\effinfout}_g^2 
\le \norm{\effinfout}_g^2 = \Bnd_\out.
\end{align}
For the information map, on the other hand, we have
\begin{align}
\norm{\imap} &= \norm{\chan^\pull \chan_\push} \le \norm{\chan^\pull}
\norm{\chan_\push} \le 1.
\end{align}
\end{proof}
\begin{proof}[Proof of Lemma~\ref{lem_prod_maps}]
Eq.~(\ref{gce}) implies
\begin{align}
\score_\out(\dot\phi) &= \sum_j \Amp_{j} \score_j(\dot\phi),
\label{score_amp}
\end{align}
where 
\begin{align}
\Amp_j A &= I^{\otimes(j-1)} \otimes A\otimes I^{\otimes(N-j)}
\end{align}
is called an ampliation \cite{parth_qsc}. Since all output scores are
in the form of Eq.~(\ref{score_amp}), we can set the output ambient
vector space as
\begin{align}
\mc V_\out &= \sum_j \Amp_j \mc V_j.
\end{align}
In other words, $\mc V_\out$ comprises the sums of zero-mean one-body
observables.
  Observe that, for any pair of zero-mean observables $A \in \mc V_j$
  and $B \in \mc V_k$,
\begin{align}
\Avg{\Amp_j A,\Amp_k B}_{g}
&= \begin{cases}
\Avg{A,B}_{\chan_j f}, & j = k,\\
0, & j \neq k.
\end{cases}
\label{AjAk}
\end{align}
Thus, for any $A, B\in \mc V_\out$,
\begin{align}
\Avg{A,B}_g &= \Avg{\sum_j \Amp_j A_j,\sum_k \Amp_k B_k}_g 
= \sum_j\Avg{A_j,B_j}_{\chan_j f},
\end{align}
meaning that $\mc V_\out$ satisfies the definition of the direct sum
given by Eq.~(\ref{V_direct}).  Eq.~(\ref{score_amp}), in particular,
becomes Eq.~(\ref{score_direct}). To prove Eq.~(\ref{Amb_direct}) from
Eq.~(\ref{V_direct}), use Prop.~\ref{prop_direct}.
\end{proof}
\begin{proof}[Proof of Lemma~\ref{lem_sum_maps}]
  We prove it only for the quantum Poisson channel in
  Example~\ref{exa_qpp}; the rest follow the same
  argument. Eq.~(\ref{score_direct}) can be shown to hold for the
  tensor product of two Poisson states by writing
\begin{align}
\dot\phi\bk{g}
&= \bk{\chan_1 f \oplus \chan_2 f}  \jordan \score_\out(\dot\phi)
\label{poisson_score1}
\\
&= \dot\phi\bk{\chan_1 f \oplus \chan_2 f}
= \dot\phi(\chan_1 f) \oplus \dot\phi(\chan_2 f)
= \Bk{\chan_1 f \jordan \score_1(\dot\phi)} \oplus 
\Bk{\chan_2 f \jordan \score_2(\dot\phi)}
\\
&= 
\bk{\chan_1 f \oplus \chan_2 f}  \jordan 
\Bk{\score_1(\dot\phi) \oplus \score_2(\dot\phi)}
\label{poisson_score2}
\end{align}
and comparing Eqs.~(\ref{poisson_score1}) and (\ref{poisson_score2}).
It follows that, for multiple Poisson channels, any output score is a
direct sum as per Eq.~(\ref{score_direct}). We can then set $\mc V_\out$
as the direct sum in Eq.~(\ref{V_direct}).
\end{proof}

\begin{proof}[Proof of Theorem~\ref{thm_add}]
Eq.~(\ref{score_direct}) implies
\begin{align}
\Avg{A,\score_\out(\dot\phi)}_g &= \Avg{A,\bigoplus_j \score_j(\dot\phi)}_g
\quad
\forall A \in \mc V_\out.
\end{align}
This equality in turn implies that, given
$g_\push\dot\phi = \score_\out(\dot\phi) + \mc N(g)$ and
$\chan_{j\push} f_\push \dot\phi = \score_j(\dot\phi) + \mc
N(\chan_j f)$,
\begin{align}
\Avg{u,g_\push\dot\phi}_g &= \Avg{u,\bigoplus_j\chan_{j\push} f_\push\dot\phi}_g
\quad \forall u \in \mc V_\out/\mc N(g).
\end{align}
Since $\mc V_\out/\mc N(g)$ is dense in $\Amb(g)$ and the inner
product is continuous, we obtain
\begin{align}
\Avg{u,g_\push\dot\phi}_g &= \Avg{u,\bigoplus_j\chan_{j\push} f_\push\dot\phi}_g
\quad
\forall u \in \Amb(g),
\end{align}
leading to
\begin{align}
g_\push \dot\phi &= \bigoplus_j \chan_{j\push} f_\push\dot\phi
= \sum_j \inj_j \chan_{j\push} f_\push\dot\phi.
\end{align}
By the definition of $\chan_\push$ in Eq.~(\ref{cov}), we obtain
Eq.~(\ref{push_sum}). Prop.~\ref{prop_inj} then leads to
Eqs.~(\ref{pull_sum}) and (\ref{imap_sum}).
\end{proof}

\begin{proof}[Proof of Lemma~\ref{lem_dual}]
  Since $\chan^\dual$ is positive and unital, the Schwarz inequality
  \cite[Proposition~3.2.4]{bratteli} implies that, for any
  $A \in \mc V_\out$,
\begin{align}
\bk{\chan^\dual A}^2 &\le \chan^\dual\bk{A^2},
\\
\norm{\chan^\dual A}^2_f = \trace\Bk{f\bk{\chan^\dual A}^2 } 
&\le \trace\Bk{f \chan^\dual\bk{A^2}} = \norm{A}^2_g.
\label{schwarz}
\end{align}
For any $A \in \mc N(g)$, Eq.~(\ref{schwarz}) implies that
$0 = \norm{A}_g \ge \norm{\chan^\dual A}_f$, so
$\chan^\dual A \in \mc N(f)$,
$\chan^\dual \mc N(g) \subseteq \mc N(f)$, and $\chan^\dual$ remains
well defined as a map from $\mc V_\out/\mc N(g)$ to $\mc V/\mc
N(f)$. Eq.~(\ref{schwarz}) further implies that $\chan^\dual$ is
bounded with operator norm satisfying Eq.~(\ref{mono_quantum}).
$\chan^\dual$ can therefore be uniquely extended to become a map with
domain $\overline{\mc V_\out/\mc N(g)} = \Amb(g)$ and the same
operator norm.
\end{proof}

\begin{proof}[Proof of Lemma~\ref{lem_score}]
  Denote the set of bounded self-adjoint operators on $\mc H$ as
  $\mc O(\mc H)$.  The rigorous definition of a score operator is
  \cite[Eq.~(2.23)]{holevo_structure}
\begin{align}
\trace\Bk{A \bk{\dot\phi f}} &= \Avg{A,\score(\dot\phi)}_f
\quad
\forall A \in \mc O(\mc H).
\end{align}
Similarly, for the output of a quantum channel,
\begin{align}
\trace\Bk{A \bk{\dot\phi g}} &= \Avg{A,\score_\out(\dot\phi)}_g
\quad
\forall A \in \mc O(\mc H_\out).
\end{align}
Then we can write
\begin{align}
\Avg{A,\score_\out(\dot\phi)}_g
&= \trace\bk{A \dot\phi g}  = \trace\bk{A \dot\phi \chan f} 
= \trace\bk{A \chan \dot\phi f}
= \trace\Bk{\bk{\chan^\dual A} \dot\phi f} 
= \Avg{\chan^\dual A,\score(\dot\phi)}_f.
\end{align}
This equality in terms of semi-inner products continues to hold for
$f_\push\dot\phi = \score(\dot\phi) + \mc N(f)$ and
$g_\push\dot\phi = \score_\out(\dot\phi) + \mc N(g)$, implying that,
for any $u \in \mc V_\out/\mc N(g)$,
\begin{align}
\Avg{u,g_\push\dot\phi}_g &= \Avg{\chan^\dual u,f_\push\dot\phi}_f
= \Avg{u,\chan_\dual f_\push\dot\phi}_g
= \Avg{u,\chan_\dual \inj_{\mc T(f)} f_\push\dot\phi}_g.
\end{align}
This equality holds also for all $u \in \Amb(g)$ because
$\mc V_\out/\mc N(g)$ is dense in $\Amb(g)$ and the inner product is
continuous. Hence,
\begin{align}
g_\push\dot\phi &= \chan_\dual \inj_{\mc T(f)} f_\push\dot\phi,
\end{align}
leading to Eq.~(\ref{chan_push_quantum}) via the definition of
$\chan_\push$ in Eq.~(\ref{cov}).
\end{proof}
\begin{proof}[Proof of Proposition~\ref{prop_direct}]
$\oplus_j A_j \in \mc N$ if and only if
$A_j \in \mc N_j$ for all $j$ because
\begin{align}
\Norm{\bigoplus_j A_j}^2 &= \sum_j \Norm{A_j}^2 = 0.
\end{align}
Hence,
\begin{align}
\mc N &= \bigoplus_j \mc N_j.
\end{align}
It follows that, for any $\oplus_j A_j \in \mc W$,
\begin{align}
\bigoplus_j A_j + \mc N = \bigoplus_j A_j + \bigoplus_j \mc N_j
= \bigoplus_j\bk{A_j + \mc N_j},
\end{align}
and Eq.~(\ref{quo_direct}) holds. $\mc W/\mc N$ and $\mc V_j/\mc N_j$
are now inner-product spaces.

Consider an arbitrary $u \in H$. Since $\mc W/\mc N$ is dense in $H$,
there exists a Cauchy sequence $\{\oplus_j u_{nj}: n = 1,2,\dots\}$ in
$\mc W/\mc N$ that converges to $u$, where each $u_{nj}$ is in
$\mc V_j/\mc N_j$. By the definition of Cauchy sequences, for any
$\epsilon > 0$, there exists an $M$ such that, for all $n,m > M$,
\begin{align}
\Norm{\bigoplus_j u_{nj} - \bigoplus_j u_{mj}}^2
&= \sum_j \Norm{u_{nj} - u_{mj}}^2 < \epsilon.
\end{align}
It follows that $\{u_{nj}:n = 1,2,\dots\}$ is also a Cauchy sequence
in $\mc V_j/\mc N_j$ for each $j$ and converges to an element
$u_j \in H_j$. This procedure of generating a $u_j \in H_j$ from
$u \in H$ does not depend on the Cauchy sequence $\{\oplus_j u_{nj}\}$
we choose.  To prove so, consider another Cauchy sequence
$\{\oplus_j v_{nj}\}$ that converges to the same $u$. Then
\begin{align}
\sum_j \Norm{v_{nj}-u_{nj}}^2 &= \Norm{\bigoplus_j v_{nj} - 
\bigoplus_j u_{nj}}^2 
\to 0,
\end{align}
meaning that $\norm{v_{nj}-u_{nj}} \to 0$ and $\{v_{nj}\}$ must
converge to the same $u_j \in H_j$. We can therefore define a map
$U:H \to \oplus_j H_j$ as
\begin{align}
U u &= \bigoplus_j u_j.
\label{iso_direct}
\end{align}
By the continuity of the norm,
\begin{align}
\norm{u}^2 &= \lim_{n\to\infty} \Norm{\bigoplus_j u_{nj}}^2
= \lim_{n\to\infty} \sum_j \norm{u_{nj}}^2
= \sum_j \lim_{n\to\infty} \norm{u_{nj}}^2 = \sum_j \norm{u_j}^2
= \norm{\bigoplus_j u_j}^2 = \norm{U u}^2.
\end{align}
This equality proves that $U$ is isometric. To prove that $U$ is also
surjective, consider an arbitrary $v_j \in H_j$ for each $j$. Because
$H_j$ is complete and $\mc V_j/\mc N_j$ is dense in it, a Cauchy
sequence $\{v_{nj}\}$ in $\mc V_j/\mc N_j$ converging to $v_j$ exists.
For any $\epsilon > 0$, there exists an $M$ such that, for all
$n,m > M$ and for all $j$, $\norm{v_{nj}-v_{mj}}^2 < \epsilon/N$,
leading to
\begin{align}
\Norm{\bigoplus_j v_{nj} - \bigoplus_j v_{mj}}^2 &= \sum_j \norm{v_{nj}-v_{mj}}^2 < \epsilon.
\end{align}
It follows that the sequence $\{\bigoplus_j v_{nj}\}$ in $\mc W/\mc N$
is also Cauchy and the limit point $v \in H$ exists. If we apply the
$U$ in Eq.~(\ref{iso_direct}) to $v$, we find that it gives back
$\oplus_j v_j$, which is an arbitrary element of $\oplus_j H_j$.
Hence, $U$ is surjective as well as isometric, meaning that $H$ and
$\oplus_j H_j$ are isomorphic.
\end{proof}

\section{\label{app_foot}Footnotes}
\begin{enumerate}
\item In Example~\ref{exa_classical}, a rigorous condition for a
  classical parametric model to have a well defined
  $\bnd_\theta(\dot\phi)$ is called differentiability in quadratic
  mean \cite{vaart}. It is an open question how this condition should
  be extended for the other examples.

\item  Existing results concerning the quantum Poisson and
 Gaussian states in Examples~\ref{exa_qpoisson} and \ref{exa_qgauss}
 may be less rigorous mathematically than their classical counterparts
 or may assume that $\mc H$ is finite-dimensional, although there is
 no strong reason to doubt that they hold in general. 

\item In classical statistics, the channel pushforward $\chan_\push$
  defined by Eq.~(\ref{cov}) is called the score operator, the channel
  pullback $\chan^{\pull}$ defined by Eq.~(\ref{chan_pull}) is called
  the efficient influence operator, and the information map $\imap$
  defined by Eq.~(\ref{imap}) is called the information operator
  \cite{bickel,vaart}, but we refrain from those names to avoid
  confusion with quantum operators.

\item Following Theorems~\ref{thm_effinfout} and \ref{thm_inf}, an
  alternative method to compute $\effinfout$ is to solve for any
  $\delta_\out \in \Amb(g)$ that satisfies
\begin{align}
\chan^\pull \delta_\out &= \effinf.
\end{align}
By Theorem~\ref{thm_inf}, this $\delta_\out$ is an influence
for the output satisfying 
\begin{align}
\Avg{\delta_\out,g_\push\dot\phi}_{g} &= 
\Avg{\delta_\out,\chan_\push f_\push\dot\phi}_{g} = 
\Avg{\chan^\pull \delta_\out,f_\push\dot\phi}_f
= \Avg{\effinf,f_\push\dot\phi}_f = \dd\beta(\dot\phi).
\end{align}
The output efficient influence is then the projection
$\effinfout = \Pi_{\mc T(g)}\delta_\out$ into the output score tangent
space $\mc T(g)$. The pseudo-inverse in Theorem~\ref{thm_effinfout}
implements the same procedure.

\item The chain rules given by
  Eqs.~(\ref{push_chain})--(\ref{imap_chain}) suggest that, in a
  continuous-time limit, where each channel is applied for only an
  infinitesimal duration, linear differential equations may be derived
  for $\chan(k,j)_\push$, $\chan(k,j)^\pull$, and $\imap(N,j)$.

\item A map is compact if and only if it has a SVD that converges in
  operator norm.  We have the following conditions \cite{reed_simon}:
\begin{enumerate}
\item  All compact maps are bounded ($\norm{T} <\infty$).
\item If $\range T$ is finite-dimensional, then $T$ is finite-rank and
  thus compact.

\item If $\range T$ is infinite-dimensional, then some bounded maps,
  such as the identity map and any unitary map, are not compact.

\item All finite-rank maps (for which $\{s_j\}$ is a finite set),
  trace-class maps (for which $\sum_j s_j < \infty$), Hilbert-Schmidt
  maps (for which $\sum_j s_j^2 < \infty$), and in fact all
  $p$-Schatten-class maps (for which $\sum_j s_j^p < \infty$) with any
  $1 \le p < \infty$ are compact.
\end{enumerate}
We have not proved that $\chan_\push$ for any of the channels studied
in Secs.~\ref{sec_coh}--\ref{sec_spade} is compact, so the rigorous
existence of its SVD is questionable. However, the numerical analysis
must approximate each $\chan_\push$ by a finite-dimensional matrix, which
always has a SVD. The open question then becomes how well the
finite-dimensional matrix and its SVD approximate the physical map.

\item In our study of incoherent imaging in
  Secs.~\ref{sec_direct}--\ref{sec_confocal}, we consider only
  measurements that produce Poisson fields. It is possible to model
  more general measurements that produce non-Poisson fields, such as
  homodyne or heterodyne detection, by returning to the quantum state
  of optical field from which the Poisson state is derived. However,
  such measurements must still obey the quantum bounds and usually
  perform much worse than photon-counting methods for thermal sources
  \cite{stellar,yang17}.

\end{enumerate}

\bibliography{research5}

\begin{thebibliography}{121}%
\makeatletter
\providecommand \@ifxundefined [1]{%
 \@ifx{#1\undefined}
}%
\providecommand \@ifnum [1]{%
 \ifnum #1\expandafter \@firstoftwo
 \else \expandafter \@secondoftwo
 \fi
}%
\providecommand \@ifx [1]{%
 \ifx #1\expandafter \@firstoftwo
 \else \expandafter \@secondoftwo
 \fi
}%
\providecommand \natexlab [1]{#1}%
\providecommand \enquote  [1]{``#1''}%
\providecommand \bibnamefont  [1]{#1}%
\providecommand \bibfnamefont [1]{#1}%
\providecommand \citenamefont [1]{#1}%
\providecommand \href@noop [0]{\@secondoftwo}%
\providecommand \href [0]{\begingroup \@sanitize@url \@href}%
\providecommand \@href[1]{\@@startlink{#1}\@@href}%
\providecommand \@@href[1]{\endgroup#1\@@endlink}%
\providecommand \@sanitize@url [0]{\catcode `\\12\catcode `\$12\catcode
  `\&12\catcode `\#12\catcode `\^12\catcode `\_12\catcode `\%12\relax}%
\providecommand \@@startlink[1]{}%
\providecommand \@@endlink[0]{}%
\providecommand \url  [0]{\begingroup\@sanitize@url \@url }%
\providecommand \@url [1]{\endgroup\@href {#1}{\urlprefix }}%
\providecommand \urlprefix  [0]{URL }%
\providecommand \Eprint [0]{\href }%
\providecommand \doibase [0]{http://dx.doi.org/}%
\providecommand \selectlanguage [0]{\@gobble}%
\providecommand \bibinfo  [0]{\@secondoftwo}%
\providecommand \bibfield  [0]{\@secondoftwo}%
\providecommand \translation [1]{[#1]}%
\providecommand \BibitemOpen [0]{}%
\providecommand \bibitemStop [0]{}%
\providecommand \bibitemNoStop [0]{.\EOS\space}%
\providecommand \EOS [0]{\spacefactor3000\relax}%
\providecommand \BibitemShut  [1]{\csname bibitem#1\endcsname}%
\let\auto@bib@innerbib\@empty
\bibitem [{\citenamefont {{van der Vaart}}(1998)}]{vaart}%
  \BibitemOpen
  \bibfield  {author} {\bibinfo {author} {\bibfnamefont {A.~W.}\ \bibnamefont
  {{van der Vaart}}},\ }\href {\doibase 10.1017/CBO9780511802256} {\emph
  {\bibinfo {title} {Asymptotic {{Statistics}}}}}\ (\bibinfo  {publisher}
  {Cambridge University Press},\ \bibinfo {address} {Cambridge, UK},\ \bibinfo
  {year} {1998})\BibitemShut {NoStop}%
\bibitem [{\citenamefont {Bickel}\ \emph {et~al.}(1993)\citenamefont {Bickel},
  \citenamefont {Klaassen}, \citenamefont {Ritov},\ and\ \citenamefont
  {Wellner}}]{bickel}%
  \BibitemOpen
  \bibfield  {author} {\bibinfo {author} {\bibfnamefont {Peter~J.}\
  \bibnamefont {Bickel}}, \bibinfo {author} {\bibfnamefont {Chris A.~J.}\
  \bibnamefont {Klaassen}}, \bibinfo {author} {\bibfnamefont {Ya'acov}\
  \bibnamefont {Ritov}}, \ and\ \bibinfo {author} {\bibfnamefont {John~A.}\
  \bibnamefont {Wellner}},\ }\href@noop {} {\emph {\bibinfo {title} {Efficient
  and Adaptive Estimation for Semiparametric Models}}}\ (\bibinfo  {publisher}
  {Springer},\ \bibinfo {address} {New York},\ \bibinfo {year}
  {1993})\BibitemShut {NoStop}%
\bibitem [{\citenamefont {Tsiatis}(2006)}]{tsiatis}%
  \BibitemOpen
  \bibfield  {author} {\bibinfo {author} {\bibfnamefont {Anastasios~A.}\
  \bibnamefont {Tsiatis}},\ }\href {\doibase 10.1007/0-387-37345-4} {\emph
  {\bibinfo {title} {Semiparametric {{Theory}} and {{Missing Data}}}}}\
  (\bibinfo  {publisher} {Springer},\ \bibinfo {address} {New York},\ \bibinfo
  {year} {2006})\BibitemShut {NoStop}%
\bibitem [{\citenamefont {Van Der~Laan}\ and\ \citenamefont
  {Rose}(2011)}]{laan}%
  \BibitemOpen
  \bibfield  {author} {\bibinfo {author} {\bibfnamefont {Mark~J.}\ \bibnamefont
  {Van Der~Laan}}\ and\ \bibinfo {author} {\bibfnamefont {Sherri}\ \bibnamefont
  {Rose}},\ }\href {\doibase 10.1007/978-1-4419-9782-1} {\emph {\bibinfo
  {title} {Targeted {{Learning}}: {{Causal Inference}} for {{Observational}}
  and {{Experimental Data}}}}},\ Springer {{Series}} in {{Statistics}}\
  (\bibinfo  {publisher} {Springer},\ \bibinfo {address} {New York, NY},\
  \bibinfo {year} {2011})\BibitemShut {NoStop}%
\bibitem [{\citenamefont {Kennedy}(2024)}]{kennedy24}%
  \BibitemOpen
  \bibfield  {author} {\bibinfo {author} {\bibfnamefont {Edward~H.}\
  \bibnamefont {Kennedy}},\ }\bibfield  {title} {\enquote {\bibinfo {title}
  {Semiparametric {{Doubly Robust Targeted Double Machine Learning}}: {{A
  Review}}},}\ }in\ \href {\doibase 10.1201/9781003216223-10} {\emph {\bibinfo
  {booktitle} {Handbook of {{Statistical Methods}} for {{Precision
  Medicine}}}}}\ (\bibinfo  {publisher} {{Chapman and Hall/CRC}},\ \bibinfo
  {address} {Boca Raton},\ \bibinfo {year} {2024})\ Chap.~\bibinfo {chapter}
  {10}, pp.\ \bibinfo {pages} {207--236}\BibitemShut {NoStop}%
\bibitem [{\citenamefont {Newey}(1990)}]{newey90}%
  \BibitemOpen
  \bibfield  {author} {\bibinfo {author} {\bibfnamefont {Whitney~K.}\
  \bibnamefont {Newey}},\ }\bibfield  {title} {\enquote {\bibinfo {title}
  {Semiparametric efficiency bounds},}\ }\href {\doibase
  10.1002/jae.3950050202} {\bibfield  {journal} {\bibinfo  {journal} {Journal
  of Applied Econometrics}\ }\textbf {\bibinfo {volume} {5}},\ \bibinfo {pages}
  {99--135} (\bibinfo {year} {1990})}\BibitemShut {NoStop}%
\bibitem [{\citenamefont {Tsang}(2019{\natexlab{a}})}]{spade_prr}%
  \BibitemOpen
  \bibfield  {author} {\bibinfo {author} {\bibfnamefont {Mankei}\ \bibnamefont
  {Tsang}},\ }\bibfield  {title} {\enquote {\bibinfo {title} {Semiparametric
  estimation for incoherent optical imaging},}\ }\href {\doibase
  10.1103/PhysRevResearch.1.033006} {\bibfield  {journal} {\bibinfo  {journal}
  {Physical Review Research}\ }\textbf {\bibinfo {volume} {1}},\ \bibinfo
  {pages} {033006} (\bibinfo {year} {2019}{\natexlab{a}})}\BibitemShut
  {NoStop}%
\bibitem [{\citenamefont {{de Villiers}}\ and\ \citenamefont
  {Pike}(2016)}]{villiers}%
  \BibitemOpen
  \bibfield  {author} {\bibinfo {author} {\bibfnamefont {Geoffrey}\
  \bibnamefont {{de Villiers}}}\ and\ \bibinfo {author} {\bibfnamefont
  {E.~Roy}\ \bibnamefont {Pike}},\ }\href {\doibase 10.1201/9781315366708}
  {\emph {\bibinfo {title} {The Limits of Resolution}}}\ (\bibinfo  {publisher}
  {CRC Press},\ \bibinfo {address} {Boca Raton},\ \bibinfo {year}
  {2016})\BibitemShut {NoStop}%
\bibitem [{\citenamefont {Tsang}\ \emph {et~al.}(2020)\citenamefont {Tsang},
  \citenamefont {Albarelli},\ and\ \citenamefont {Datta}}]{semi_prx}%
  \BibitemOpen
  \bibfield  {author} {\bibinfo {author} {\bibfnamefont {Mankei}\ \bibnamefont
  {Tsang}}, \bibinfo {author} {\bibfnamefont {Francesco}\ \bibnamefont
  {Albarelli}}, \ and\ \bibinfo {author} {\bibfnamefont {Animesh}\ \bibnamefont
  {Datta}},\ }\bibfield  {title} {\enquote {\bibinfo {title} {Quantum
  {{Semiparametric Estimation}}},}\ }\href {\doibase
  10.1103/PhysRevX.10.031023} {\bibfield  {journal} {\bibinfo  {journal}
  {Physical Review X}\ }\textbf {\bibinfo {volume} {10}},\ \bibinfo {pages}
  {031023} (\bibinfo {year} {2020})}\BibitemShut {NoStop}%
\bibitem [{\citenamefont {Helstrom}(1976)}]{helstrom}%
  \BibitemOpen
  \bibfield  {author} {\bibinfo {author} {\bibfnamefont {Carl~W.}\ \bibnamefont
  {Helstrom}},\ }\href
  {http://www.sciencedirect.com/science/bookseries/00765392/123} {\emph
  {\bibinfo {title} {Quantum Detection and Estimation Theory}}}\ (\bibinfo
  {publisher} {Academic Press},\ \bibinfo {address} {New York},\ \bibinfo
  {year} {1976})\BibitemShut {NoStop}%
\bibitem [{\citenamefont {Tsang}(2021{\natexlab{a}})}]{qlmoment_pra2}%
  \BibitemOpen
  \bibfield  {author} {\bibinfo {author} {\bibfnamefont {Mankei}\ \bibnamefont
  {Tsang}},\ }\bibfield  {title} {\enquote {\bibinfo {title} {Quantum limit to
  subdiffraction incoherent optical imaging. {{II}}. {{A}} parametric-submodel
  approach},}\ }\href {\doibase 10.1103/PhysRevA.104.052411} {\bibfield
  {journal} {\bibinfo  {journal} {Physical Review A}\ }\textbf {\bibinfo
  {volume} {104}},\ \bibinfo {pages} {052411} (\bibinfo {year}
  {2021}{\natexlab{a}})}\BibitemShut {NoStop}%
\bibitem [{\citenamefont {Tan}\ and\ \citenamefont {Tsang}(2023)}]{tan23}%
  \BibitemOpen
  \bibfield  {author} {\bibinfo {author} {\bibfnamefont {Xiao-Jie}\
  \bibnamefont {Tan}}\ and\ \bibinfo {author} {\bibfnamefont {Mankei}\
  \bibnamefont {Tsang}},\ }\bibfield  {title} {\enquote {\bibinfo {title}
  {Quantum limit to subdiffraction incoherent optical imaging. {{III}}.
  {{Numerical}} analysis},}\ }\href {\doibase 10.1103/PhysRevA.108.052416}
  {\bibfield  {journal} {\bibinfo  {journal} {Physical Review A}\ }\textbf
  {\bibinfo {volume} {108}},\ \bibinfo {pages} {052416} (\bibinfo {year}
  {2023})}\BibitemShut {NoStop}%
\bibitem [{\citenamefont {Tsang}\ \emph {et~al.}(2016)\citenamefont {Tsang},
  \citenamefont {Nair},\ and\ \citenamefont {Lu}}]{tnl}%
  \BibitemOpen
  \bibfield  {author} {\bibinfo {author} {\bibfnamefont {Mankei}\ \bibnamefont
  {Tsang}}, \bibinfo {author} {\bibfnamefont {Ranjith}\ \bibnamefont {Nair}}, \
  and\ \bibinfo {author} {\bibfnamefont {Xiao-Ming}\ \bibnamefont {Lu}},\
  }\bibfield  {title} {\enquote {\bibinfo {title} {Quantum theory of
  superresolution for two incoherent optical point sources},}\ }\href {\doibase
  10.1103/PhysRevX.6.031033} {\bibfield  {journal} {\bibinfo  {journal}
  {Physical Review X}\ }\textbf {\bibinfo {volume} {6}},\ \bibinfo {pages}
  {031033} (\bibinfo {year} {2016})}\BibitemShut {NoStop}%
\bibitem [{\citenamefont {Tsang}(2019{\natexlab{b}})}]{review_cp}%
  \BibitemOpen
  \bibfield  {author} {\bibinfo {author} {\bibfnamefont {Mankei}\ \bibnamefont
  {Tsang}},\ }\bibfield  {title} {\enquote {\bibinfo {title} {Resolving
  starlight: A quantum perspective},}\ }\href {\doibase
  10.1080/00107514.2020.1736375} {\bibfield  {journal} {\bibinfo  {journal}
  {Contemporary Physics}\ }\textbf {\bibinfo {volume} {60}},\ \bibinfo {pages}
  {279--298} (\bibinfo {year} {2019}{\natexlab{b}})}\BibitemShut {NoStop}%
\bibitem [{\citenamefont {Lvovsky}\ \emph {et~al.}(2026)\citenamefont
  {Lvovsky}, \citenamefont {Grace}, \citenamefont {Guha}, \citenamefont
  {Tsang}, \citenamefont {Adesso},\ and\ \citenamefont {Treps}}]{lvovsky26}%
  \BibitemOpen
  \bibfield  {author} {\bibinfo {author} {\bibfnamefont {A.~I.}\ \bibnamefont
  {Lvovsky}}, \bibinfo {author} {\bibfnamefont {Michael~R.}\ \bibnamefont
  {Grace}}, \bibinfo {author} {\bibfnamefont {Saikat}\ \bibnamefont {Guha}},
  \bibinfo {author} {\bibfnamefont {Mankei}\ \bibnamefont {Tsang}}, \bibinfo
  {author} {\bibfnamefont {Gerardo}\ \bibnamefont {Adesso}}, \ and\ \bibinfo
  {author} {\bibfnamefont {Nicolas}\ \bibnamefont {Treps}},\ }\href {\doibase
  10.48550/arXiv.2605.10767} {\enquote {\bibinfo {title} {Passive optical
  superresolution at the quantum limit},}\ } (\bibinfo {year} {2026}),\ \Eprint
  {http://arxiv.org/abs/2605.10767} {arXiv:2605.10767 [quant-ph]} \BibitemShut
  {NoStop}%
\bibitem [{\citenamefont {Tsang}(2017)}]{spade_njp}%
  \BibitemOpen
  \bibfield  {author} {\bibinfo {author} {\bibfnamefont {Mankei}\ \bibnamefont
  {Tsang}},\ }\bibfield  {title} {\enquote {\bibinfo {title} {Subdiffraction
  incoherent optical imaging via spatial-mode demultiplexing},}\ }\href
  {\doibase 10.1088/1367-2630/aa60ee} {\bibfield  {journal} {\bibinfo
  {journal} {New Journal of Physics}\ }\textbf {\bibinfo {volume} {19}},\
  \bibinfo {pages} {023054} (\bibinfo {year} {2017})}\BibitemShut {NoStop}%
\bibitem [{\citenamefont {Tsang}(2018{\natexlab{a}})}]{spade_pra}%
  \BibitemOpen
  \bibfield  {author} {\bibinfo {author} {\bibfnamefont {Mankei}\ \bibnamefont
  {Tsang}},\ }\bibfield  {title} {\enquote {\bibinfo {title} {Subdiffraction
  incoherent optical imaging via spatial-mode demultiplexing: {{Semiclassical}}
  treatment},}\ }\href {\doibase 10.1103/PhysRevA.97.023830} {\bibfield
  {journal} {\bibinfo  {journal} {Physical Review A}\ }\textbf {\bibinfo
  {volume} {97}},\ \bibinfo {pages} {023830} (\bibinfo {year}
  {2018}{\natexlab{a}})}\BibitemShut {NoStop}%
\bibitem [{\citenamefont {Zhou}\ and\ \citenamefont {Jiang}(2019)}]{zhou19}%
  \BibitemOpen
  \bibfield  {author} {\bibinfo {author} {\bibfnamefont {Sisi}\ \bibnamefont
  {Zhou}}\ and\ \bibinfo {author} {\bibfnamefont {Liang}\ \bibnamefont
  {Jiang}},\ }\bibfield  {title} {\enquote {\bibinfo {title} {Modern
  description of {{Rayleigh}}'s criterion},}\ }\href {\doibase
  10.1103/PhysRevA.99.013808} {\bibfield  {journal} {\bibinfo  {journal}
  {Physical Review A}\ }\textbf {\bibinfo {volume} {99}},\ \bibinfo {pages}
  {013808} (\bibinfo {year} {2019})}\BibitemShut {NoStop}%
\bibitem [{\citenamefont {Tsang}(2019{\natexlab{c}})}]{qlmoment_pra}%
  \BibitemOpen
  \bibfield  {author} {\bibinfo {author} {\bibfnamefont {Mankei}\ \bibnamefont
  {Tsang}},\ }\bibfield  {title} {\enquote {\bibinfo {title} {Quantum limit to
  subdiffraction incoherent optical imaging},}\ }\href {\doibase
  10.1103/PhysRevA.99.012305} {\bibfield  {journal} {\bibinfo  {journal}
  {Physical Review A}\ }\textbf {\bibinfo {volume} {99}},\ \bibinfo {pages}
  {012305} (\bibinfo {year} {2019}{\natexlab{c}})}\BibitemShut {NoStop}%
\bibitem [{\citenamefont {Tan}\ \emph {et~al.}(2023)\citenamefont {Tan},
  \citenamefont {Qi}, \citenamefont {Chen}, \citenamefont {Danner},
  \citenamefont {Kanchanawong},\ and\ \citenamefont {Tsang}}]{tan23a}%
  \BibitemOpen
  \bibfield  {author} {\bibinfo {author} {\bibfnamefont {Xiao-Jie}\
  \bibnamefont {Tan}}, \bibinfo {author} {\bibfnamefont {Luo}\ \bibnamefont
  {Qi}}, \bibinfo {author} {\bibfnamefont {Lianwei}\ \bibnamefont {Chen}},
  \bibinfo {author} {\bibfnamefont {Aaron~J.}\ \bibnamefont {Danner}}, \bibinfo
  {author} {\bibfnamefont {Pakorn}\ \bibnamefont {Kanchanawong}}, \ and\
  \bibinfo {author} {\bibfnamefont {Mankei}\ \bibnamefont {Tsang}},\ }\bibfield
   {title} {\enquote {\bibinfo {title} {Quantum-inspired superresolution for
  incoherent imaging},}\ }\href {\doibase 10.1364/OPTICA.493227} {\bibfield
  {journal} {\bibinfo  {journal} {Optica}\ }\textbf {\bibinfo {volume} {10}},\
  \bibinfo {pages} {1189--1194} (\bibinfo {year} {2023})}\BibitemShut {NoStop}%
\bibitem [{\citenamefont {Rao}(1945)}]{rao45}%
  \BibitemOpen
  \bibfield  {author} {\bibinfo {author} {\bibfnamefont {C.~Radhakrishna}\
  \bibnamefont {Rao}},\ }\bibfield  {title} {\enquote {\bibinfo {title}
  {Information and accuracy attainable in the estimation of statistical
  parameters},}\ }\href@noop {} {\bibfield  {journal} {\bibinfo  {journal}
  {Bulletin of the Calcutta Mathematical Society}\ }\textbf {\bibinfo {volume}
  {37}},\ \bibinfo {pages} {81--91} (\bibinfo {year} {1945})}\BibitemShut
  {NoStop}%
\bibitem [{\citenamefont {Amari}\ and\ \citenamefont {Nagaoka}(2000)}]{amari}%
  \BibitemOpen
  \bibfield  {author} {\bibinfo {author} {\bibfnamefont {Shun-ichi}\
  \bibnamefont {Amari}}\ and\ \bibinfo {author} {\bibfnamefont {Hiroshi}\
  \bibnamefont {Nagaoka}},\ }\href@noop {} {\emph {\bibinfo {title} {Methods of
  Information Geometry}}}\ (\bibinfo  {publisher} {American Mathematical
  Society/Oxford University Press},\ \bibinfo {address} {Providence},\ \bibinfo
  {year} {2000})\BibitemShut {NoStop}%
\bibitem [{\citenamefont {Amari}(2016)}]{amari_app}%
  \BibitemOpen
  \bibfield  {author} {\bibinfo {author} {\bibfnamefont {Shun-ichi}\
  \bibnamefont {Amari}},\ }\href {\doibase 10.1007/978-4-431-55978-8} {\emph
  {\bibinfo {title} {Information {{Geometry}} and {{Its Applications}}}}}\
  (\bibinfo  {publisher} {Springer Japan},\ \bibinfo {address} {Tokyo},\
  \bibinfo {year} {2016})\BibitemShut {NoStop}%
\bibitem [{\citenamefont {Trabs}(2015)}]{trabs15}%
  \BibitemOpen
  \bibfield  {author} {\bibinfo {author} {\bibfnamefont {Mathias}\ \bibnamefont
  {Trabs}},\ }\bibfield  {title} {\enquote {\bibinfo {title} {Information
  bounds for inverse problems with application to deconvolution and {{L\'evy}}
  models},}\ }\href {\doibase 10.1214/14-AIHP627} {\bibfield  {journal}
  {\bibinfo  {journal} {Annales de l'Institut Henri Poincar\'e, Probabilit\'es
  et Statistiques}\ }\textbf {\bibinfo {volume} {51}},\ \bibinfo {pages}
  {1620--1650} (\bibinfo {year} {2015})}\BibitemShut {NoStop}%
\bibitem [{\citenamefont {Brady}(2009)}]{brady}%
  \BibitemOpen
  \bibfield  {author} {\bibinfo {author} {\bibfnamefont {David~J.}\
  \bibnamefont {Brady}},\ }\href {\doibase 10.1002/9780470443736} {\emph
  {\bibinfo {title} {Optical Imaging and Spectroscopy}}}\ (\bibinfo
  {publisher} {Wiley},\ \bibinfo {address} {Hoboken},\ \bibinfo {year}
  {2009})\BibitemShut {NoStop}%
\bibitem [{\citenamefont {Chen}(2018)}]{xudong}%
  \BibitemOpen
  \bibfield  {author} {\bibinfo {author} {\bibfnamefont {Xudong}\ \bibnamefont
  {Chen}},\ }\href {\doibase 10.1002/9781119311997} {\emph {\bibinfo {title}
  {Computational {{Methods}} for {{Electromagnetic Inverse Scattering}}}}}\
  (\bibinfo  {publisher} {John Wiley \& Sons, Ltd},\ \bibinfo {address}
  {Singapore},\ \bibinfo {year} {2018})\BibitemShut {NoStop}%
\bibitem [{\citenamefont {Bogachev}(2007)}]{bogachev}%
  \BibitemOpen
  \bibfield  {author} {\bibinfo {author} {\bibfnamefont {V.}~\bibnamefont
  {Bogachev}},\ }\href {\doibase 10.1007/978-3-540-34514-5} {\emph {\bibinfo
  {title} {Measure {{Theory}}}}}\ (\bibinfo  {publisher} {Springer-Verlag},\
  \bibinfo {address} {Berlin Heidelberg},\ \bibinfo {year} {2007})\BibitemShut
  {NoStop}%
\bibitem [{\citenamefont {Kingman}(1992)}]{kingman}%
  \BibitemOpen
  \bibfield  {author} {\bibinfo {author} {\bibfnamefont {J.~F.~C.}\
  \bibnamefont {Kingman}},\ }\href
  {https://global.oup.com/academic/product/poisson-processes-9780198536932}
  {\emph {\bibinfo {title} {Poisson Processes}}}\ (\bibinfo  {publisher}
  {Oxford University Press},\ \bibinfo {address} {Oxford, England, UK},\
  \bibinfo {year} {1992})\BibitemShut {NoStop}%
\bibitem [{\citenamefont {Holevo}(2019)}]{holevo_info}%
  \BibitemOpen
  \bibfield  {author} {\bibinfo {author} {\bibfnamefont {Alexander~S.}\
  \bibnamefont {Holevo}},\ }\href {\doibase 10.1515/9783110642490} {\emph
  {\bibinfo {title} {Quantum Systems, Channels, Information}}},\ \bibinfo
  {edition} {2nd}\ ed.\ (\bibinfo  {publisher} {De Gruyter},\ \bibinfo
  {address} {Berlin},\ \bibinfo {year} {2019})\BibitemShut {NoStop}%
\bibitem [{\citenamefont {Tsang}(2021{\natexlab{b}})}]{poisson_quantum}%
  \BibitemOpen
  \bibfield  {author} {\bibinfo {author} {\bibfnamefont {Mankei}\ \bibnamefont
  {Tsang}},\ }\bibfield  {title} {\enquote {\bibinfo {title} {Poisson {{Quantum
  Information}}},}\ }\href {\doibase 10.22331/q-2021-08-19-527} {\bibfield
  {journal} {\bibinfo  {journal} {Quantum}\ }\textbf {\bibinfo {volume} {5}},\
  \bibinfo {pages} {527} (\bibinfo {year} {2021}{\natexlab{b}})}\BibitemShut
  {NoStop}%
\bibitem [{\citenamefont {Yuen}\ and\ \citenamefont {Shapiro}(1978)}]{yuen78}%
  \BibitemOpen
  \bibfield  {author} {\bibinfo {author} {\bibfnamefont {Horace~P.}\
  \bibnamefont {Yuen}}\ and\ \bibinfo {author} {\bibfnamefont {Jeffrey~H.}\
  \bibnamefont {Shapiro}},\ }\bibfield  {title} {\enquote {\bibinfo {title}
  {Optical communication with two-photon coherent states--{{Part I}}:
  {{Quantum-state}} propagation and quantum-noise},}\ }\href {\doibase
  10.1109/TIT.1978.1055958} {\bibfield  {journal} {\bibinfo  {journal} {IEEE
  Transactions on Information Theory}\ }\textbf {\bibinfo {volume} {24}},\
  \bibinfo {pages} {657--668} (\bibinfo {year} {1978})}\BibitemShut {NoStop}%
\bibitem [{\citenamefont {Shapiro}\ \emph {et~al.}(1979)\citenamefont
  {Shapiro}, \citenamefont {Yuen},\ and\ \citenamefont {Mata}}]{shapiro79}%
  \BibitemOpen
  \bibfield  {author} {\bibinfo {author} {\bibfnamefont {J.}~\bibnamefont
  {Shapiro}}, \bibinfo {author} {\bibfnamefont {H.}~\bibnamefont {Yuen}}, \
  and\ \bibinfo {author} {\bibfnamefont {A.}~\bibnamefont {Mata}},\ }\bibfield
  {title} {\enquote {\bibinfo {title} {Optical communication with two-photon
  coherent states--{{Part II}}: {{Photoemissive}} detection and structured
  receiver performance},}\ }\href {\doibase 10.1109/TIT.1979.1056033}
  {\bibfield  {journal} {\bibinfo  {journal} {IEEE Transactions on Information
  Theory}\ }\textbf {\bibinfo {volume} {25}},\ \bibinfo {pages} {179--192}
  (\bibinfo {year} {1979})}\BibitemShut {NoStop}%
\bibitem [{\citenamefont {Kolobov}\ and\ \citenamefont
  {Fabre}(2000)}]{kolobov_fabre}%
  \BibitemOpen
  \bibfield  {author} {\bibinfo {author} {\bibfnamefont {Mikhail~I.}\
  \bibnamefont {Kolobov}}\ and\ \bibinfo {author} {\bibfnamefont {Claude}\
  \bibnamefont {Fabre}},\ }\bibfield  {title} {\enquote {\bibinfo {title}
  {Quantum limits on optical resolution},}\ }\href {\doibase
  10.1103/PhysRevLett.85.3789} {\bibfield  {journal} {\bibinfo  {journal}
  {Physical Review Letters}\ }\textbf {\bibinfo {volume} {85}},\ \bibinfo
  {pages} {3789--3792} (\bibinfo {year} {2000})}\BibitemShut {NoStop}%
\bibitem [{\citenamefont {Beskrovnyy}\ and\ \citenamefont
  {Kolobov}(2005)}]{beskrovnyy05}%
  \BibitemOpen
  \bibfield  {author} {\bibinfo {author} {\bibfnamefont {Vladislav~N.}\
  \bibnamefont {Beskrovnyy}}\ and\ \bibinfo {author} {\bibfnamefont
  {Mikhail~I.}\ \bibnamefont {Kolobov}},\ }\bibfield  {title} {\enquote
  {\bibinfo {title} {Quantum limits of super-resolution in reconstruction of
  optical objects},}\ }\href {\doibase 10.1103/PhysRevA.71.043802} {\bibfield
  {journal} {\bibinfo  {journal} {Physical Review A}\ }\textbf {\bibinfo
  {volume} {71}},\ \bibinfo {pages} {043802} (\bibinfo {year}
  {2005})}\BibitemShut {NoStop}%
\bibitem [{\citenamefont {Pinel}\ \emph {et~al.}(2012)\citenamefont {Pinel},
  \citenamefont {Fade}, \citenamefont {Braun}, \citenamefont {Jian},
  \citenamefont {Treps},\ and\ \citenamefont {Fabre}}]{pinel12}%
  \BibitemOpen
  \bibfield  {author} {\bibinfo {author} {\bibfnamefont {Olivier}\ \bibnamefont
  {Pinel}}, \bibinfo {author} {\bibfnamefont {Julien}\ \bibnamefont {Fade}},
  \bibinfo {author} {\bibfnamefont {Daniel}\ \bibnamefont {Braun}}, \bibinfo
  {author} {\bibfnamefont {Pu}~\bibnamefont {Jian}}, \bibinfo {author}
  {\bibfnamefont {Nicolas}\ \bibnamefont {Treps}}, \ and\ \bibinfo {author}
  {\bibfnamefont {Claude}\ \bibnamefont {Fabre}},\ }\bibfield  {title}
  {\enquote {\bibinfo {title} {Ultimate sensitivity of precision measurements
  with intense {{Gaussian}} quantum light: {{A}} multimodal approach},}\ }\href
  {\doibase 10.1103/PhysRevA.85.010101} {\bibfield  {journal} {\bibinfo
  {journal} {Physical Review A}\ }\textbf {\bibinfo {volume} {85}},\ \bibinfo
  {pages} {010101} (\bibinfo {year} {2012})}\BibitemShut {NoStop}%
\bibitem [{\citenamefont {Taylor}\ and\ \citenamefont
  {Bowen}(2016)}]{taylor16}%
  \BibitemOpen
  \bibfield  {author} {\bibinfo {author} {\bibfnamefont {Michael~A.}\
  \bibnamefont {Taylor}}\ and\ \bibinfo {author} {\bibfnamefont {Warwick~P.}\
  \bibnamefont {Bowen}},\ }\bibfield  {title} {\enquote {\bibinfo {title}
  {Quantum metrology and its application in biology},}\ }\href {\doibase
  10.1016/j.physrep.2015.12.002} {\bibfield  {journal} {\bibinfo  {journal}
  {Physics Reports}\ }\textbf {\bibinfo {volume} {615}},\ \bibinfo {pages}
  {1--59} (\bibinfo {year} {2016})}\BibitemShut {NoStop}%
\bibitem [{\citenamefont {Treps}\ \emph {et~al.}(2003)\citenamefont {Treps},
  \citenamefont {Grosse}, \citenamefont {Bowen}, \citenamefont {Fabre},
  \citenamefont {Bachor},\ and\ \citenamefont {Lam}}]{treps}%
  \BibitemOpen
  \bibfield  {author} {\bibinfo {author} {\bibfnamefont {Nicolas}\ \bibnamefont
  {Treps}}, \bibinfo {author} {\bibfnamefont {Nicolai}\ \bibnamefont {Grosse}},
  \bibinfo {author} {\bibfnamefont {Warwick~P.}\ \bibnamefont {Bowen}},
  \bibinfo {author} {\bibfnamefont {Claude}\ \bibnamefont {Fabre}}, \bibinfo
  {author} {\bibfnamefont {Hans-A.}\ \bibnamefont {Bachor}}, \ and\ \bibinfo
  {author} {\bibfnamefont {Ping~Koy}\ \bibnamefont {Lam}},\ }\bibfield  {title}
  {\enquote {\bibinfo {title} {A quantum laser pointer},}\ }\href {\doibase
  10.1126/science.1086489} {\bibfield  {journal} {\bibinfo  {journal}
  {Science}\ }\textbf {\bibinfo {volume} {301}},\ \bibinfo {pages} {940--943}
  (\bibinfo {year} {2003})}\BibitemShut {NoStop}%
\bibitem [{\citenamefont {Goodman}(1985)}]{goodman_stat}%
  \BibitemOpen
  \bibfield  {author} {\bibinfo {author} {\bibfnamefont {Joseph~W.}\
  \bibnamefont {Goodman}},\ }\href@noop {} {\emph {\bibinfo {title}
  {Statistical Optics}}}\ (\bibinfo  {publisher} {Wiley},\ \bibinfo {address}
  {New York},\ \bibinfo {year} {1985})\BibitemShut {NoStop}%
\bibitem [{\citenamefont {Zmuidzinas}(2003)}]{zmuidzinas03}%
  \BibitemOpen
  \bibfield  {author} {\bibinfo {author} {\bibfnamefont {Jonas}\ \bibnamefont
  {Zmuidzinas}},\ }\bibfield  {title} {\enquote {\bibinfo {title}
  {Cram\'er--{{Rao}} sensitivity limits for astronomical instruments:
  Implications for interferometer design},}\ }\href {\doibase
  10.1364/JOSAA.20.000218} {\bibfield  {journal} {\bibinfo  {journal} {Journal
  of the Optical Society of America A}\ }\textbf {\bibinfo {volume} {20}},\
  \bibinfo {pages} {218--233} (\bibinfo {year} {2003})}\BibitemShut {NoStop}%
\bibitem [{\citenamefont {Pawley}(2006)}]{pawley}%
  \BibitemOpen
  \bibinfo {editor} {\bibfnamefont {James~B.}\ \bibnamefont {Pawley}},\ ed.,\
  \href {\doibase 10.1007/978-0-387-45524-2} {\emph {\bibinfo {title} {Handbook
  of Biological Confocal Microscopy}}}\ (\bibinfo  {publisher} {Springer},\
  \bibinfo {address} {New York},\ \bibinfo {year} {2006})\BibitemShut {NoStop}%
\bibitem [{\citenamefont {Deschout}\ \emph {et~al.}(2014)\citenamefont
  {Deschout}, \citenamefont {Zanacchi}, \citenamefont {Mlodzianoski},
  \citenamefont {Diaspro}, \citenamefont {Bewersdorf}, \citenamefont {Hess},\
  and\ \citenamefont {Braeckmans}}]{deschout}%
  \BibitemOpen
  \bibfield  {author} {\bibinfo {author} {\bibfnamefont {Hendrik}\ \bibnamefont
  {Deschout}}, \bibinfo {author} {\bibfnamefont {Francesca~Cella}\ \bibnamefont
  {Zanacchi}}, \bibinfo {author} {\bibfnamefont {Michael}\ \bibnamefont
  {Mlodzianoski}}, \bibinfo {author} {\bibfnamefont {Alberto}\ \bibnamefont
  {Diaspro}}, \bibinfo {author} {\bibfnamefont {Joerg}\ \bibnamefont
  {Bewersdorf}}, \bibinfo {author} {\bibfnamefont {Samuel~T.}\ \bibnamefont
  {Hess}}, \ and\ \bibinfo {author} {\bibfnamefont {Kevin}\ \bibnamefont
  {Braeckmans}},\ }\bibfield  {title} {\enquote {\bibinfo {title} {Precisely
  and accurately localizing single emitters in fluorescence microscopy},}\
  }\href {\doibase 10.1038/nmeth.2843} {\bibfield  {journal} {\bibinfo
  {journal} {Nature Methods}\ }\textbf {\bibinfo {volume} {11}},\ \bibinfo
  {pages} {253--266} (\bibinfo {year} {2014})}\BibitemShut {NoStop}%
\bibitem [{\citenamefont {Chao}\ \emph {et~al.}(2016)\citenamefont {Chao},
  \citenamefont {Sally~Ward},\ and\ \citenamefont {Ober}}]{chao16}%
  \BibitemOpen
  \bibfield  {author} {\bibinfo {author} {\bibfnamefont {Jerry}\ \bibnamefont
  {Chao}}, \bibinfo {author} {\bibfnamefont {E.}~\bibnamefont {Sally~Ward}}, \
  and\ \bibinfo {author} {\bibfnamefont {Raimund~J.}\ \bibnamefont {Ober}},\
  }\bibfield  {title} {\enquote {\bibinfo {title} {Fisher information theory
  for parameter estimation in single molecule microscopy: Tutorial},}\ }\href
  {\doibase 10.1364/JOSAA.33.000B36} {\bibfield  {journal} {\bibinfo  {journal}
  {Journal of the Optical Society of America A}\ }\textbf {\bibinfo {volume}
  {33}},\ \bibinfo {pages} {B36} (\bibinfo {year} {2016})}\BibitemShut
  {NoStop}%
\bibitem [{\citenamefont {{von Diezmann}}\ \emph {et~al.}(2017)\citenamefont
  {{von Diezmann}}, \citenamefont {Shechtman},\ and\ \citenamefont
  {Moerner}}]{diezmann17}%
  \BibitemOpen
  \bibfield  {author} {\bibinfo {author} {\bibfnamefont {Alex}\ \bibnamefont
  {{von Diezmann}}}, \bibinfo {author} {\bibfnamefont {Yoav}\ \bibnamefont
  {Shechtman}}, \ and\ \bibinfo {author} {\bibfnamefont {W.~E.}\ \bibnamefont
  {Moerner}},\ }\bibfield  {title} {\enquote {\bibinfo {title}
  {Three-{{Dimensional Localization}} of {{Single Molecules}} for
  {{Super-Resolution Imaging}} and {{Single-Particle Tracking}}},}\ }\href
  {\doibase 10.1021/acs.chemrev.6b00629} {\bibfield  {journal} {\bibinfo
  {journal} {Chemical Reviews}\ }\textbf {\bibinfo {volume} {117}},\ \bibinfo
  {pages} {7244--7275} (\bibinfo {year} {2017})}\BibitemShut {NoStop}%
\bibitem [{\citenamefont {Tsang}(2023{\natexlab{a}})}]{superosc_ieee}%
  \BibitemOpen
  \bibfield  {author} {\bibinfo {author} {\bibfnamefont {Mankei}\ \bibnamefont
  {Tsang}},\ }\bibfield  {title} {\enquote {\bibinfo {title} {Efficient
  superoscillation measurement for incoherent optical imaging},}\ }\href
  {\doibase 10.1109/JSTSP.2022.3212173} {\bibfield  {journal} {\bibinfo
  {journal} {IEEE Journal of Selected Topics in Signal Processing}\ }\textbf
  {\bibinfo {volume} {17}},\ \bibinfo {pages} {513--524} (\bibinfo {year}
  {2023}{\natexlab{a}})}\BibitemShut {NoStop}%
\bibitem [{\citenamefont {Kim}\ \emph {et~al.}(2025)\citenamefont {Kim},
  \citenamefont {Fitzgerald}, \citenamefont {Vievard}, \citenamefont {Lin},
  \citenamefont {Xin}, \citenamefont {Lucas}, \citenamefont {Guyon},
  \citenamefont {Lozi}, \citenamefont {Deo}, \citenamefont {Huby},
  \citenamefont {Lacour}, \citenamefont {Lallement}, \citenamefont
  {{Amezcua-Correa}}, \citenamefont {{Leon-Saval}}, \citenamefont {Norris},
  \citenamefont {Nowak}, \citenamefont {Sallum}, \citenamefont {Sarrazin},
  \citenamefont {Taras}, \citenamefont {Yerolatsitis},\ and\ \citenamefont
  {Jovanovic}}]{kim25}%
  \BibitemOpen
  \bibfield  {author} {\bibinfo {author} {\bibfnamefont {Yoo~Jung}\
  \bibnamefont {Kim}}, \bibinfo {author} {\bibfnamefont {Michael~P.}\
  \bibnamefont {Fitzgerald}}, \bibinfo {author} {\bibfnamefont {S{\'e}bastien}\
  \bibnamefont {Vievard}}, \bibinfo {author} {\bibfnamefont {Jonathan}\
  \bibnamefont {Lin}}, \bibinfo {author} {\bibfnamefont {Yinzi}\ \bibnamefont
  {Xin}}, \bibinfo {author} {\bibfnamefont {Miles}\ \bibnamefont {Lucas}},
  \bibinfo {author} {\bibfnamefont {Olivier}\ \bibnamefont {Guyon}}, \bibinfo
  {author} {\bibfnamefont {Julien}\ \bibnamefont {Lozi}}, \bibinfo {author}
  {\bibfnamefont {Vincent}\ \bibnamefont {Deo}}, \bibinfo {author}
  {\bibfnamefont {Elsa}\ \bibnamefont {Huby}}, \bibinfo {author} {\bibfnamefont
  {Sylvestre}\ \bibnamefont {Lacour}}, \bibinfo {author} {\bibfnamefont
  {Manon}\ \bibnamefont {Lallement}}, \bibinfo {author} {\bibfnamefont
  {Rodrigo}\ \bibnamefont {{Amezcua-Correa}}}, \bibinfo {author} {\bibfnamefont
  {Sergio}\ \bibnamefont {{Leon-Saval}}}, \bibinfo {author} {\bibfnamefont
  {Barnaby}\ \bibnamefont {Norris}}, \bibinfo {author} {\bibfnamefont
  {Mathias}\ \bibnamefont {Nowak}}, \bibinfo {author} {\bibfnamefont {Steph}\
  \bibnamefont {Sallum}}, \bibinfo {author} {\bibfnamefont {Jehanne}\
  \bibnamefont {Sarrazin}}, \bibinfo {author} {\bibfnamefont {Adam}\
  \bibnamefont {Taras}}, \bibinfo {author} {\bibfnamefont {Stephanos}\
  \bibnamefont {Yerolatsitis}}, \ and\ \bibinfo {author} {\bibfnamefont
  {Nemanja}\ \bibnamefont {Jovanovic}},\ }\bibfield  {title} {\enquote
  {\bibinfo {title} {On-sky {{Demonstration}} of {{Subdiffraction-limited
  Astronomical Measurement Using}} a {{Photonic Lantern}}},}\ }\href {\doibase
  10.3847/2041-8213/ae0739} {\bibfield  {journal} {\bibinfo  {journal} {The
  Astrophysical Journal Letters}\ }\textbf {\bibinfo {volume} {993}},\ \bibinfo
  {pages} {L3} (\bibinfo {year} {2025})}\BibitemShut {NoStop}%
\bibitem [{\citenamefont {Wallis}\ \emph {et~al.}(2026)\citenamefont {Wallis},
  \citenamefont {McCann}, \citenamefont {Collier}, \citenamefont
  {{Toms-Hardman}}, \citenamefont {Frost}, \citenamefont {{Dix-Matthews}},\
  and\ \citenamefont {Gozzard}}]{wallis26}%
  \BibitemOpen
  \bibfield  {author} {\bibinfo {author} {\bibfnamefont {John~S.}\ \bibnamefont
  {Wallis}}, \bibinfo {author} {\bibfnamefont {Ayden~S.}\ \bibnamefont
  {McCann}}, \bibinfo {author} {\bibfnamefont {Joshua~J.}\ \bibnamefont
  {Collier}}, \bibinfo {author} {\bibfnamefont {Lilani~D.}\ \bibnamefont
  {{Toms-Hardman}}}, \bibinfo {author} {\bibfnamefont {Alex~M.}\ \bibnamefont
  {Frost}}, \bibinfo {author} {\bibfnamefont {Benjamin~P.}\ \bibnamefont
  {{Dix-Matthews}}}, \ and\ \bibinfo {author} {\bibfnamefont {David~R.}\
  \bibnamefont {Gozzard}},\ }\href {\doibase 10.48550/arXiv.2606.18025}
  {\enquote {\bibinfo {title} {On-sky binary source hypothesis testing beyond
  the diffraction limit using spatial mode demultiplexing based detection},}\ }
  (\bibinfo {year} {2026}),\ \Eprint {http://arxiv.org/abs/2606.18025}
  {arXiv:2606.18025 [astro-ph.IM]} \BibitemShut {NoStop}%
\bibitem [{\citenamefont {Lee}(2003)}]{lee}%
  \BibitemOpen
  \bibfield  {author} {\bibinfo {author} {\bibfnamefont {John~M.}\ \bibnamefont
  {Lee}},\ }\href {\doibase 10.1007/978-0-387-21752-9} {\emph {\bibinfo {title}
  {Introduction to {{Smooth Manifolds}}}}}\ (\bibinfo  {publisher}
  {Springer-Verlag},\ \bibinfo {address} {New York},\ \bibinfo {year}
  {2003})\BibitemShut {NoStop}%
\bibitem [{\citenamefont {Lehmann}\ and\ \citenamefont
  {Casella}(1998)}]{lehmann_casella}%
  \BibitemOpen
  \bibfield  {author} {\bibinfo {author} {\bibfnamefont {E.~L.}\ \bibnamefont
  {Lehmann}}\ and\ \bibinfo {author} {\bibfnamefont {George}\ \bibnamefont
  {Casella}},\ }\href {\doibase 10.1007/b98854} {\emph {\bibinfo {title}
  {Theory of {{Point Estimation}}}}},\ \bibinfo {edition} {2nd}\ ed.\ (\bibinfo
   {publisher} {Springer},\ \bibinfo {address} {New York},\ \bibinfo {year}
  {1998})\BibitemShut {NoStop}%
\bibitem [{\citenamefont {{\c C}{\i}nlar}(2011)}]{cinlar}%
  \BibitemOpen
  \bibfield  {author} {\bibinfo {author} {\bibfnamefont {Erhan}\ \bibnamefont
  {{\c C}{\i}nlar}},\ }\href {\doibase 10.1007/978-0-387-87859-1} {\emph
  {\bibinfo {title} {Probability and Stochastics}}}\ (\bibinfo  {publisher}
  {Springer},\ \bibinfo {address} {New York},\ \bibinfo {year}
  {2011})\BibitemShut {NoStop}%
\bibitem [{\citenamefont {Takahashi}(1990)}]{takahashi90}%
  \BibitemOpen
  \bibfield  {author} {\bibinfo {author} {\bibfnamefont {Yoichiro}\
  \bibnamefont {Takahashi}},\ }\bibfield  {title} {\enquote {\bibinfo {title}
  {Absolute continuity of poisson random fields},}\ }\href {\doibase
  10.2977/prims/1195170849} {\bibfield  {journal} {\bibinfo  {journal}
  {Publications of the Research Institute for Mathematical Sciences}\ }\textbf
  {\bibinfo {volume} {26}},\ \bibinfo {pages} {629--647} (\bibinfo {year}
  {1990})}\BibitemShut {NoStop}%
\bibitem [{\citenamefont {Bogachev}(1998)}]{bogachev_gauss}%
  \BibitemOpen
  \bibfield  {author} {\bibinfo {author} {\bibfnamefont {Vladimir~I.}\
  \bibnamefont {Bogachev}},\ }\href@noop {} {\emph {\bibinfo {title} {Gaussian
  Measures}}}\ (\bibinfo  {publisher} {American Mathematical Society},\
  \bibinfo {address} {Providence, Rhode Island},\ \bibinfo {year}
  {1998})\BibitemShut {NoStop}%
\bibitem [{\citenamefont {Sekine}(1995)}]{sekine95}%
  \BibitemOpen
  \bibfield  {author} {\bibinfo {author} {\bibfnamefont {Jun}\ \bibnamefont
  {Sekine}},\ }\bibfield  {title} {\enquote {\bibinfo {title} {The {{Hilbert
  Riemannian}} structure of equivalent {{Gaussian}} measures associated with
  the {{Fisher}} information},}\ }\href
  {https://projecteuclid.org/journals/osaka-journal-of-mathematics/volume-32/issue-1/The-Hilbert-Riemannian-structure-of-equivalent-Gaussian-measures-associated-with/ojm/1200785865.full}
  {\bibfield  {journal} {\bibinfo  {journal} {Osaka Journal of Mathematics}\
  }\textbf {\bibinfo {volume} {32}},\ \bibinfo {pages} {71--95} (\bibinfo
  {year} {1995})}\BibitemShut {NoStop}%
\bibitem [{\citenamefont {Kay}(1993)}]{kay}%
  \BibitemOpen
  \bibfield  {author} {\bibinfo {author} {\bibfnamefont {Steven~M.}\
  \bibnamefont {Kay}},\ }\href@noop {} {\emph {\bibinfo {title} {Fundamentals
  of Statistical Signal Processing: {{Estimation}} Theory}}}\ (\bibinfo
  {publisher} {Prentice Hall},\ \bibinfo {address} {Upper Saddle River},\
  \bibinfo {year} {1993})\BibitemShut {NoStop}%
\bibitem [{\citenamefont {Petz}(2010)}]{petz}%
  \BibitemOpen
  \bibfield  {author} {\bibinfo {author} {\bibfnamefont {D{\'e}nes}\
  \bibnamefont {Petz}},\ }\href {\doibase 10.1007/978-3-540-74636-2} {\emph
  {\bibinfo {title} {Quantum Information Theory and Quantum Statistics}}}\
  (\bibinfo  {publisher} {Springer},\ \bibinfo {address} {Berlin, Germany},\
  \bibinfo {year} {2010})\BibitemShut {NoStop}%
\bibitem [{\citenamefont {Hayashi}(2017)}]{hayashi}%
  \BibitemOpen
  \bibfield  {author} {\bibinfo {author} {\bibfnamefont {Masahito}\
  \bibnamefont {Hayashi}},\ }\href {\doibase 10.1007/978-3-662-49725-8} {\emph
  {\bibinfo {title} {Quantum {{Information Theory}}: {{Mathematical
  Foundation}}}}},\ \bibinfo {edition} {2nd}\ ed.\ (\bibinfo  {publisher}
  {Springer},\ \bibinfo {address} {Berlin},\ \bibinfo {year}
  {2017})\BibitemShut {NoStop}%
\bibitem [{\citenamefont {Parthasarathy}(1992)}]{parth_qsc}%
  \BibitemOpen
  \bibfield  {author} {\bibinfo {author} {\bibfnamefont {K.~R.}\ \bibnamefont
  {Parthasarathy}},\ }\href {\doibase 10.1007/978-3-0348-0566-7} {\emph
  {\bibinfo {title} {An Introduction to Quantum Stochastic Calculus}}}\
  (\bibinfo  {publisher} {Birkh\"auser},\ \bibinfo {address} {Basel},\ \bibinfo
  {year} {1992})\BibitemShut {NoStop}%
\bibitem [{\citenamefont {Holevo}(2011)}]{holevo_aspect}%
  \BibitemOpen
  \bibfield  {author} {\bibinfo {author} {\bibfnamefont {Alexander~S.}\
  \bibnamefont {Holevo}},\ }\href {\doibase 10.1007/978-88-7642-378-9} {\emph
  {\bibinfo {title} {Probabilistic and Statistical Aspects of Quantum
  Theory}}}\ (\bibinfo  {publisher} {Scuola Normale Superiore Pisa},\ \bibinfo
  {address} {Pisa, Italy},\ \bibinfo {year} {2011})\BibitemShut {NoStop}%
\bibitem [{\citenamefont {Monras}(2013)}]{monras13}%
  \BibitemOpen
  \bibfield  {author} {\bibinfo {author} {\bibfnamefont {Alex}\ \bibnamefont
  {Monras}},\ }\bibfield  {title} {\enquote {\bibinfo {title} {Phase space
  formalism for quantum estimation of {{Gaussian}} states},}\ }\href
  {http://arxiv.org/abs/1303.3682} {\bibfield  {journal} {\bibinfo  {journal}
  {arXiv:1303.3682 [quant-ph]}\ } (\bibinfo {year} {2013})},\ \Eprint
  {http://arxiv.org/abs/1303.3682} {arXiv:1303.3682 [quant-ph]} \BibitemShut
  {NoStop}%
\bibitem [{\citenamefont {Conway}(2007)}]{conway}%
  \BibitemOpen
  \bibfield  {author} {\bibinfo {author} {\bibfnamefont {John~B.}\ \bibnamefont
  {Conway}},\ }\href {https://link.springer.com/book/10.1007/978-1-4757-4383-8}
  {\emph {\bibinfo {title} {A Course in Functional Analysis}}}\ (\bibinfo
  {publisher} {Springer},\ \bibinfo {address} {New York, NY, USA},\ \bibinfo
  {year} {2007})\BibitemShut {NoStop}%
\bibitem [{\citenamefont {Reed}\ and\ \citenamefont
  {Simon}(1980)}]{reed_simon}%
  \BibitemOpen
  \bibfield  {author} {\bibinfo {author} {\bibfnamefont {Michael}\ \bibnamefont
  {Reed}}\ and\ \bibinfo {author} {\bibfnamefont {Barry}\ \bibnamefont
  {Simon}},\ }\href@noop {} {\emph {\bibinfo {title} {Methods of Modern
  Mathematical Physics. {{I}}: {{Functional}} Analysis}}}\ (\bibinfo
  {publisher} {Academic Press},\ \bibinfo {address} {San Diego},\ \bibinfo
  {year} {1980})\BibitemShut {NoStop}%
\bibitem [{\citenamefont {Debnath}\ and\ \citenamefont
  {Mikusi{\'n}ski}(2005)}]{debnath}%
  \BibitemOpen
  \bibfield  {author} {\bibinfo {author} {\bibfnamefont {Lokenath}\
  \bibnamefont {Debnath}}\ and\ \bibinfo {author} {\bibfnamefont {Piotr}\
  \bibnamefont {Mikusi{\'n}ski}},\ }\href@noop {} {\emph {\bibinfo {title}
  {Introduction to Hilbert Spaces with Applications}}}\ (\bibinfo  {publisher}
  {Elsevier},\ \bibinfo {address} {Amsterdam},\ \bibinfo {year}
  {2005})\BibitemShut {NoStop}%
\bibitem [{\citenamefont {Gross}\ and\ \citenamefont {Caves}(2020)}]{gross20}%
  \BibitemOpen
  \bibfield  {author} {\bibinfo {author} {\bibfnamefont {Jonathan~Arthur}\
  \bibnamefont {Gross}}\ and\ \bibinfo {author} {\bibfnamefont {Carlton~M.}\
  \bibnamefont {Caves}},\ }\bibfield  {title} {\enquote {\bibinfo {title} {One
  from many: {{Estimating}} a function of many parameters},}\ }\href {\doibase
  10.1088/1751-8121/abb9ed} {\bibfield  {journal} {\bibinfo  {journal} {Journal
  of Physics A: Mathematical and Theoretical}\ }\textbf {\bibinfo {volume}
  {54}},\ \bibinfo {pages} {014001} (\bibinfo {year} {2020})}\BibitemShut
  {NoStop}%
\bibitem [{\citenamefont {Nagaoka}(1989)}]{nagaoka89}%
  \BibitemOpen
  \bibfield  {author} {\bibinfo {author} {\bibfnamefont {Hiroshi}\ \bibnamefont
  {Nagaoka}},\ }\bibfield  {title} {\enquote {\bibinfo {title} {A new approach
  to {{Cram\'er-Rao}} bounds for quantum state estimation},}\ }\href@noop {}
  {\bibfield  {journal} {\bibinfo  {journal} {IEICE Technical Report}\ }\textbf
  {\bibinfo {volume} {IT 89-42}},\ \bibinfo {pages} {9--14} (\bibinfo {year}
  {1989})},\ \bibinfo {note} {reprinted in
  \href{http://dx.doi.org/10.1142/9789812563071_0009}{\emph{Asymptotic Theory
  of Quantum Statistical Inference: Selected Papers}, edited by Masahito
  Hayashi (World Scientific, Singapore, 2005) Chap. 8,
  pp.~100--112}.}\BibitemShut {Stop}%
\bibitem [{\citenamefont {Hayashi}(2024)}]{hayashi24}%
  \BibitemOpen
  \bibfield  {author} {\bibinfo {author} {\bibfnamefont {Masahito}\
  \bibnamefont {Hayashi}},\ }\bibfield  {title} {\enquote {\bibinfo {title}
  {Alexander {{S}}. {{Holevo}}'s researches in quantum information theory in
  20th century},}\ }\href {\doibase 10.1142/S0219749924400069} {\bibfield
  {journal} {\bibinfo  {journal} {International Journal of Quantum
  Information}\ }\textbf {\bibinfo {volume} {22}},\ \bibinfo {pages} {2440006}
  (\bibinfo {year} {2024})}\BibitemShut {NoStop}%
\bibitem [{\citenamefont {Yang}\ \emph {et~al.}(2019)\citenamefont {Yang},
  \citenamefont {Chiribella},\ and\ \citenamefont {Hayashi}}]{yang19}%
  \BibitemOpen
  \bibfield  {author} {\bibinfo {author} {\bibfnamefont {Yuxiang}\ \bibnamefont
  {Yang}}, \bibinfo {author} {\bibfnamefont {Giulio}\ \bibnamefont
  {Chiribella}}, \ and\ \bibinfo {author} {\bibfnamefont {Masahito}\
  \bibnamefont {Hayashi}},\ }\bibfield  {title} {\enquote {\bibinfo {title}
  {Attaining the {{Ultimate Precision Limit}} in {{Quantum State
  Estimation}}},}\ }\href {\doibase 10.1007/s00220-019-03433-4} {\bibfield
  {journal} {\bibinfo  {journal} {Communications in Mathematical Physics}\
  }\textbf {\bibinfo {volume} {368}},\ \bibinfo {pages} {223--293} (\bibinfo
  {year} {2019})}\BibitemShut {NoStop}%
\bibitem [{\citenamefont {{Demkowicz-Dobrza{\'n}ski}}\ \emph
  {et~al.}(2020)\citenamefont {{Demkowicz-Dobrza{\'n}ski}}, \citenamefont
  {G{\'o}recki},\ and\ \citenamefont {Gu{\c t}{\u a}}}]{demkowicz20}%
  \BibitemOpen
  \bibfield  {author} {\bibinfo {author} {\bibfnamefont {Rafa{\l}}\
  \bibnamefont {{Demkowicz-Dobrza{\'n}ski}}}, \bibinfo {author} {\bibfnamefont
  {Wojciech}\ \bibnamefont {G{\'o}recki}}, \ and\ \bibinfo {author}
  {\bibfnamefont {M{\u a}d{\u a}lin}\ \bibnamefont {Gu{\c t}{\u a}}},\
  }\bibfield  {title} {\enquote {\bibinfo {title} {Multi-parameter estimation
  beyond quantum {{Fisher}} information},}\ }\href {\doibase
  10.1088/1751-8121/ab8ef3} {\bibfield  {journal} {\bibinfo  {journal} {Journal
  of Physics A: Mathematical and Theoretical}\ }\textbf {\bibinfo {volume}
  {53}},\ \bibinfo {pages} {363001} (\bibinfo {year} {2020})}\BibitemShut
  {NoStop}%
\bibitem [{\citenamefont {Tsang}(2026)}]{tsang26}%
  \BibitemOpen
  \bibfield  {author} {\bibinfo {author} {\bibfnamefont {Mankei}\ \bibnamefont
  {Tsang}},\ }\bibfield  {title} {\enquote {\bibinfo {title} {Approaching the
  ultimate limit of quantum multiparameter estimation by many-body physics},}\
  }\href {\doibase 10.48550/arXiv.2603.17955} {\bibfield  {journal} {\bibinfo
  {journal} {ArXiv e-prints}\ } (\bibinfo {year} {2026}),\
  10.48550/arXiv.2603.17955},\ \Eprint {http://arxiv.org/abs/2603.17955}
  {2603.17955} \BibitemShut {NoStop}%
\bibitem [{\citenamefont {Fujiwara}(2006)}]{fujiwara06}%
  \BibitemOpen
  \bibfield  {author} {\bibinfo {author} {\bibfnamefont {Akio}\ \bibnamefont
  {Fujiwara}},\ }\bibfield  {title} {\enquote {\bibinfo {title} {Strong
  consistency and asymptotic efficiency for adaptive quantum estimation
  problems},}\ }\href {\doibase 10.1088/0305-4470/39/40/014} {\bibfield
  {journal} {\bibinfo  {journal} {Journal of Physics A: Mathematical and
  General}\ }\textbf {\bibinfo {volume} {39}},\ \bibinfo {pages} {12489--12504}
  (\bibinfo {year} {2006})}\BibitemShut {NoStop}%
\bibitem [{\citenamefont {Parthasarathy}(2005{\natexlab{a}})}]{parth_qm}%
  \BibitemOpen
  \bibfield  {author} {\bibinfo {author} {\bibfnamefont {K.~R.}\ \bibnamefont
  {Parthasarathy}},\ }\href
  {https://link.springer.com/book/10.1007/978-93-86279-28-6} {\emph {\bibinfo
  {title} {Mathematical Foundation of Quantum Mechanics}}}\ (\bibinfo
  {publisher} {Hindustan Book Agency},\ \bibinfo {address} {New Delhi, India},\
  \bibinfo {year} {2005})\BibitemShut {NoStop}%
\bibitem [{\citenamefont {Bratteli}\ and\ \citenamefont
  {Robinson}(1987)}]{bratteli}%
  \BibitemOpen
  \bibfield  {author} {\bibinfo {author} {\bibfnamefont {Ola}\ \bibnamefont
  {Bratteli}}\ and\ \bibinfo {author} {\bibfnamefont {Derek~W.}\ \bibnamefont
  {Robinson}},\ }\href {\doibase 10.1007/978-3-662-02520-8} {\emph {\bibinfo
  {title} {Operator Algebras and Quantum Statistical Mechanics 1}}}\ (\bibinfo
  {publisher} {Springer},\ \bibinfo {address} {Berlin, Germany},\ \bibinfo
  {year} {1987})\BibitemShut {NoStop}%
\bibitem [{\citenamefont {Tsang}(2022)}]{gce_pra}%
  \BibitemOpen
  \bibfield  {author} {\bibinfo {author} {\bibfnamefont {Mankei}\ \bibnamefont
  {Tsang}},\ }\bibfield  {title} {\enquote {\bibinfo {title} {Generalized
  conditional expectations for quantum retrodiction and smoothing},}\ }\href
  {\doibase 10.1103/PhysRevA.105.042213} {\bibfield  {journal} {\bibinfo
  {journal} {Physical Review A}\ }\textbf {\bibinfo {volume} {105}},\ \bibinfo
  {pages} {042213} (\bibinfo {year} {2022})}\BibitemShut {NoStop}%
\bibitem [{\citenamefont {Tsang}(2023{\natexlab{b}})}]{gce2}%
  \BibitemOpen
  \bibfield  {author} {\bibinfo {author} {\bibfnamefont {Mankei}\ \bibnamefont
  {Tsang}},\ }\bibfield  {title} {\enquote {\bibinfo {title} {Operational
  meanings of a generalized conditional expectation in quantum metrology},}\
  }\href {\doibase 10.22331/q-2023-11-03-1162} {\bibfield  {journal} {\bibinfo
  {journal} {Quantum}\ }\textbf {\bibinfo {volume} {7}},\ \bibinfo {pages}
  {1162} (\bibinfo {year} {2023}{\natexlab{b}})},\ \Eprint
  {http://arxiv.org/abs/2212.13162v6} {2212.13162v6} \BibitemShut {NoStop}%
\bibitem [{\citenamefont {Stoica}\ and\ \citenamefont
  {Marzetta}(2001)}]{stoica01}%
  \BibitemOpen
  \bibfield  {author} {\bibinfo {author} {\bibfnamefont {P.}~\bibnamefont
  {Stoica}}\ and\ \bibinfo {author} {\bibfnamefont {T.~L.}\ \bibnamefont
  {Marzetta}},\ }\bibfield  {title} {\enquote {\bibinfo {title} {Parameter
  estimation problems with singular information matrices},}\ }\href {\doibase
  10.1109/78.890346} {\bibfield  {journal} {\bibinfo  {journal} {IEEE
  Transactions on Signal Processing}\ }\textbf {\bibinfo {volume} {49}},\
  \bibinfo {pages} {87--90} (\bibinfo {year} {2001})}\BibitemShut {NoStop}%
\bibitem [{\citenamefont {Goldberg}\ \emph {et~al.}(2021)\citenamefont
  {Goldberg}, \citenamefont {Romero}, \citenamefont {Sanz},\ and\ \citenamefont
  {{S{\'a}nchez-Soto}}}]{goldberg21}%
  \BibitemOpen
  \bibfield  {author} {\bibinfo {author} {\bibfnamefont {Aaron~Z.}\
  \bibnamefont {Goldberg}}, \bibinfo {author} {\bibfnamefont {Jos{\'e}~L.}\
  \bibnamefont {Romero}}, \bibinfo {author} {\bibfnamefont {{\'A}ngel~S.}\
  \bibnamefont {Sanz}}, \ and\ \bibinfo {author} {\bibfnamefont {Luis~L.}\
  \bibnamefont {{S{\'a}nchez-Soto}}},\ }\bibfield  {title} {\enquote {\bibinfo
  {title} {Taming singularities of the quantum {{Fisher}} information},}\
  }\href {\doibase 10.1142/S0219749921400049} {\bibfield  {journal} {\bibinfo
  {journal} {International Journal of Quantum Information}\ }\textbf {\bibinfo
  {volume} {19}},\ \bibinfo {pages} {2140004} (\bibinfo {year}
  {2021})}\BibitemShut {NoStop}%
\bibitem [{\citenamefont {Kwon}\ \emph {et~al.}(2025)\citenamefont {Kwon},
  \citenamefont {Tsubouchi}, \citenamefont {Chu},\ and\ \citenamefont
  {Jiang}}]{kwon25}%
  \BibitemOpen
  \bibfield  {author} {\bibinfo {author} {\bibfnamefont {Hyukgun}\ \bibnamefont
  {Kwon}}, \bibinfo {author} {\bibfnamefont {Kento}\ \bibnamefont {Tsubouchi}},
  \bibinfo {author} {\bibfnamefont {Chia-Tung}\ \bibnamefont {Chu}}, \ and\
  \bibinfo {author} {\bibfnamefont {Liang}\ \bibnamefont {Jiang}},\ }\bibfield
  {title} {\enquote {\bibinfo {title} {Criteria for unbiased estimation:
  Applications to noise-agnostic sensing and learnability of quantum
  channel},}\ }\href {\doibase 10.48550/arXiv.2503.17362} {\bibfield  {journal}
  {\bibinfo  {journal} {ArXiv e-prints}\ } (\bibinfo {year} {2025}),\
  10.48550/arXiv.2503.17362},\ \Eprint {http://arxiv.org/abs/2503.17362}
  {2503.17362} \BibitemShut {NoStop}%
\bibitem [{\citenamefont {{Ben-Israel}}\ and\ \citenamefont
  {Greville}(2003)}]{benisrael}%
  \BibitemOpen
  \bibfield  {author} {\bibinfo {author} {\bibfnamefont {Adi}\ \bibnamefont
  {{Ben-Israel}}}\ and\ \bibinfo {author} {\bibfnamefont {Thomas N.~E.}\
  \bibnamefont {Greville}},\ }\href {\doibase 10.1007/b97366} {\emph {\bibinfo
  {title} {Generalized Inverses}}},\ \bibinfo {edition} {2nd}\ ed.\ (\bibinfo
  {publisher} {Springer},\ \bibinfo {address} {New York, NY, USA},\ \bibinfo
  {year} {2003})\BibitemShut {NoStop}%
\bibitem [{\citenamefont {Ibragimov}\ and\ \citenamefont
  {Has'minskii}(1981)}]{ibragimov}%
  \BibitemOpen
  \bibfield  {author} {\bibinfo {author} {\bibfnamefont {I.~A.}\ \bibnamefont
  {Ibragimov}}\ and\ \bibinfo {author} {\bibfnamefont {R.~Z.}\ \bibnamefont
  {Has'minskii}},\ }\href {\doibase 10.1007/978-1-4899-0027-2} {\emph {\bibinfo
  {title} {Statistical {{Estimation}}: {{Asymptotic Theory}}}}}\ (\bibinfo
  {publisher} {Springer},\ \bibinfo {address} {New York},\ \bibinfo {year}
  {1981})\BibitemShut {NoStop}%
\bibitem [{\citenamefont {Parthasarathy}(2005{\natexlab{b}})}]{parth_pm}%
  \BibitemOpen
  \bibfield  {author} {\bibinfo {author} {\bibfnamefont {K.~R.}\ \bibnamefont
  {Parthasarathy}},\ }\href {\doibase 10.1007/978-93-86279-27-9} {\emph
  {\bibinfo {title} {Introduction to {{Probability}} and {{Measure}}}}}\
  (\bibinfo  {publisher} {Hindustan Book Agency},\ \bibinfo {address} {New
  Delhi},\ \bibinfo {year} {2005})\BibitemShut {NoStop}%
\bibitem [{\citenamefont {Holevo}(2001)}]{holevo_structure}%
  \BibitemOpen
  \bibfield  {author} {\bibinfo {author} {\bibfnamefont {Alexander~S.}\
  \bibnamefont {Holevo}},\ }\href {\doibase 10.1007/3-540-44998-1} {\emph
  {\bibinfo {title} {Statistical Structure of Quantum Theory}}}\ (\bibinfo
  {publisher} {Springer-Verlag},\ \bibinfo {address} {Berlin},\ \bibinfo {year}
  {2001})\BibitemShut {NoStop}%
\bibitem [{\citenamefont {Slepian}(1983)}]{slepian83}%
  \BibitemOpen
  \bibfield  {author} {\bibinfo {author} {\bibfnamefont {David}\ \bibnamefont
  {Slepian}},\ }\bibfield  {title} {\enquote {\bibinfo {title} {Some
  {{Comments}} on {{Fourier Analysis}}, {{Uncertainty}} and {{Modeling}}},}\
  }\href {https://www.jstor.org/stable/2029386} {\bibfield  {journal} {\bibinfo
   {journal} {SIAM Review}\ }\textbf {\bibinfo {volume} {25}},\ \bibinfo
  {pages} {379--393} (\bibinfo {year} {1983})},\ \Eprint
  {http://arxiv.org/abs/2029386} {2029386} \BibitemShut {NoStop}%
\bibitem [{\citenamefont {Bertero}\ \emph {et~al.}(1989)\citenamefont
  {Bertero}, \citenamefont {Boccacci}, \citenamefont {Defrise}, \citenamefont
  {Mol},\ and\ \citenamefont {Pike}}]{bertero89a}%
  \BibitemOpen
  \bibfield  {author} {\bibinfo {author} {\bibfnamefont {M.}~\bibnamefont
  {Bertero}}, \bibinfo {author} {\bibfnamefont {P.}~\bibnamefont {Boccacci}},
  \bibinfo {author} {\bibfnamefont {M.}~\bibnamefont {Defrise}}, \bibinfo
  {author} {\bibfnamefont {C.~De}\ \bibnamefont {Mol}}, \ and\ \bibinfo
  {author} {\bibfnamefont {E.~R.}\ \bibnamefont {Pike}},\ }\bibfield  {title}
  {\enquote {\bibinfo {title} {Super-resolution in confocal scanning
  microscopy: {{II}}. {{The}} incoherent case},}\ }\href {\doibase
  10.1088/0266-5611/5/4/003} {\bibfield  {journal} {\bibinfo  {journal}
  {Inverse Problems}\ }\textbf {\bibinfo {volume} {5}},\ \bibinfo {pages} {441}
  (\bibinfo {year} {1989})}\BibitemShut {NoStop}%
\bibitem [{\citenamefont {Bertero}\ \emph {et~al.}(1991)\citenamefont
  {Bertero}, \citenamefont {Boccacci}, \citenamefont {Davies},\ and\
  \citenamefont {Pike}}]{bertero91}%
  \BibitemOpen
  \bibfield  {author} {\bibinfo {author} {\bibfnamefont {M.}~\bibnamefont
  {Bertero}}, \bibinfo {author} {\bibfnamefont {P.}~\bibnamefont {Boccacci}},
  \bibinfo {author} {\bibfnamefont {R.~E.}\ \bibnamefont {Davies}}, \ and\
  \bibinfo {author} {\bibfnamefont {E.~R.}\ \bibnamefont {Pike}},\ }\bibfield
  {title} {\enquote {\bibinfo {title} {Super-resolution in confocal scanning
  microscopy: {{III}}. {{The}} case of circular pupils},}\ }\href {\doibase
  10.1088/0266-5611/7/5/002} {\bibfield  {journal} {\bibinfo  {journal}
  {Inverse Problems}\ }\textbf {\bibinfo {volume} {7}},\ \bibinfo {pages} {655}
  (\bibinfo {year} {1991})}\BibitemShut {NoStop}%
\bibitem [{\citenamefont {Pa{\'u}r}\ \emph {et~al.}(2018)\citenamefont
  {Pa{\'u}r}, \citenamefont {Stoklasa}, \citenamefont {Grover}, \citenamefont
  {Krzic}, \citenamefont {{S{\'a}nchez-Soto}}, \citenamefont {Hradil},\ and\
  \citenamefont {{\v R}eh{\'a}{\v c}ek}}]{paur18}%
  \BibitemOpen
  \bibfield  {author} {\bibinfo {author} {\bibfnamefont {Martin}\ \bibnamefont
  {Pa{\'u}r}}, \bibinfo {author} {\bibfnamefont {Bohumil}\ \bibnamefont
  {Stoklasa}}, \bibinfo {author} {\bibfnamefont {Jai}\ \bibnamefont {Grover}},
  \bibinfo {author} {\bibfnamefont {Andrej}\ \bibnamefont {Krzic}}, \bibinfo
  {author} {\bibfnamefont {Luis~L.}\ \bibnamefont {{S{\'a}nchez-Soto}}},
  \bibinfo {author} {\bibfnamefont {Zden{\v e}k}\ \bibnamefont {Hradil}}, \
  and\ \bibinfo {author} {\bibfnamefont {Jaroslav}\ \bibnamefont {{\v
  R}eh{\'a}{\v c}ek}},\ }\bibfield  {title} {\enquote {\bibinfo {title}
  {Tempering {{Rayleigh}}'s curse with {{PSF}} shaping},}\ }\href {\doibase
  10.1364/OPTICA.5.001177} {\bibfield  {journal} {\bibinfo  {journal} {Optica}\
  }\textbf {\bibinfo {volume} {5}},\ \bibinfo {pages} {1177--1180} (\bibinfo
  {year} {2018})}\BibitemShut {NoStop}%
\bibitem [{\citenamefont {Pa{\'u}r}\ \emph {et~al.}(2019)\citenamefont
  {Pa{\'u}r}, \citenamefont {Stoklasa}, \citenamefont {Koutn{\'y}},
  \citenamefont {{\v R}eh{\'a}{\v c}ek}, \citenamefont {Hradil}, \citenamefont
  {Grover}, \citenamefont {Krzic},\ and\ \citenamefont
  {{S{\'a}nchez-Soto}}}]{paur19}%
  \BibitemOpen
  \bibfield  {author} {\bibinfo {author} {\bibfnamefont {M.}~\bibnamefont
  {Pa{\'u}r}}, \bibinfo {author} {\bibfnamefont {B.}~\bibnamefont {Stoklasa}},
  \bibinfo {author} {\bibfnamefont {D.}~\bibnamefont {Koutn{\'y}}}, \bibinfo
  {author} {\bibfnamefont {J.}~\bibnamefont {{\v R}eh{\'a}{\v c}ek}}, \bibinfo
  {author} {\bibfnamefont {Z.}~\bibnamefont {Hradil}}, \bibinfo {author}
  {\bibfnamefont {J.}~\bibnamefont {Grover}}, \bibinfo {author} {\bibfnamefont
  {A.}~\bibnamefont {Krzic}}, \ and\ \bibinfo {author} {\bibfnamefont {L.~L.}\
  \bibnamefont {{S{\'a}nchez-Soto}}},\ }\bibfield  {title} {\enquote {\bibinfo
  {title} {Reading out {{Fisher}} information from the zeros of the point
  spread function},}\ }\href {\doibase 10.1364/OL.44.003114} {\bibfield
  {journal} {\bibinfo  {journal} {Optics Letters}\ }\textbf {\bibinfo {volume}
  {44}},\ \bibinfo {pages} {3114--3117} (\bibinfo {year} {2019})}\BibitemShut
  {NoStop}%
\bibitem [{\citenamefont {Booth}\ \emph {et~al.}(2026)\citenamefont {Booth},
  \citenamefont {{Clunies-Ross}}, \citenamefont {Amor}, \citenamefont
  {Mauranyapin}, \citenamefont {Huang}, \citenamefont {Taylor},\ and\
  \citenamefont {Bowen}}]{booth26}%
  \BibitemOpen
  \bibfield  {author} {\bibinfo {author} {\bibfnamefont {Larnii}\ \bibnamefont
  {Booth}}, \bibinfo {author} {\bibfnamefont {Kyle}\ \bibnamefont
  {{Clunies-Ross}}}, \bibinfo {author} {\bibfnamefont {Rumelo}\ \bibnamefont
  {Amor}}, \bibinfo {author} {\bibfnamefont {Nicolas}\ \bibnamefont
  {Mauranyapin}}, \bibinfo {author} {\bibfnamefont {Zixin}\ \bibnamefont
  {Huang}}, \bibinfo {author} {\bibfnamefont {Michael~A.}\ \bibnamefont
  {Taylor}}, \ and\ \bibinfo {author} {\bibfnamefont {Warwick~P.}\ \bibnamefont
  {Bowen}},\ }\href {\doibase 10.48550/arXiv.2604.00413} {\enquote {\bibinfo
  {title} {Structured detection microscopy},}\ } (\bibinfo {year} {2026}),\
  \Eprint {http://arxiv.org/abs/2604.00413} {arXiv:2604.00413 [physics.optics]}
  \BibitemShut {NoStop}%
\bibitem [{\citenamefont {Darekar}\ \emph {et~al.}(2026)\citenamefont
  {Darekar}, \citenamefont {Jha}, \citenamefont {Grace}, \citenamefont
  {Sajjad},\ and\ \citenamefont {Guha}}]{darekar26}%
  \BibitemOpen
  \bibfield  {author} {\bibinfo {author} {\bibfnamefont {Parth~Hemant}\
  \bibnamefont {Darekar}}, \bibinfo {author} {\bibfnamefont {Amit~Kumar}\
  \bibnamefont {Jha}}, \bibinfo {author} {\bibfnamefont {Michael~R.}\
  \bibnamefont {Grace}}, \bibinfo {author} {\bibfnamefont {Aqil}\ \bibnamefont
  {Sajjad}}, \ and\ \bibinfo {author} {\bibfnamefont {Saikat}\ \bibnamefont
  {Guha}},\ }\href {\doibase 10.48550/arXiv.2606.00968} {\enquote {\bibinfo
  {title} {Fundamental {{Limit}} for {{One}} versus {{Two Point Sources
  Detection}} using {{Direct Imaging}}},}\ } (\bibinfo {year} {2026}),\ \Eprint
  {http://arxiv.org/abs/2606.00968} {arXiv:2606.00968 [physics.optics]}
  \BibitemShut {NoStop}%
\bibitem [{\citenamefont {Tsang}(2011)}]{stellar}%
  \BibitemOpen
  \bibfield  {author} {\bibinfo {author} {\bibfnamefont {Mankei}\ \bibnamefont
  {Tsang}},\ }\bibfield  {title} {\enquote {\bibinfo {title} {Quantum
  nonlocality in weak-thermal-light interferometry},}\ }\href {\doibase
  10.1103/PhysRevLett.107.270402} {\bibfield  {journal} {\bibinfo  {journal}
  {Physical Review Letters}\ }\textbf {\bibinfo {volume} {107}},\ \bibinfo
  {pages} {270402} (\bibinfo {year} {2011})}\BibitemShut {NoStop}%
\bibitem [{\citenamefont {{\v R}eh{\'a}{\v c}ek}\ \emph
  {et~al.}(2017)\citenamefont {{\v R}eh{\'a}{\v c}ek}, \citenamefont
  {Pa{\'u}r}, \citenamefont {Stoklasa}, \citenamefont {Hradil},\ and\
  \citenamefont {{S{\'a}nchez-Soto}}}]{rehacek17}%
  \BibitemOpen
  \bibfield  {author} {\bibinfo {author} {\bibfnamefont {J.}~\bibnamefont {{\v
  R}eh{\'a}{\v c}ek}}, \bibinfo {author} {\bibfnamefont {M.}~\bibnamefont
  {Pa{\'u}r}}, \bibinfo {author} {\bibfnamefont {B.}~\bibnamefont {Stoklasa}},
  \bibinfo {author} {\bibfnamefont {Z.}~\bibnamefont {Hradil}}, \ and\ \bibinfo
  {author} {\bibfnamefont {L.~L.}\ \bibnamefont {{S{\'a}nchez-Soto}}},\
  }\bibfield  {title} {\enquote {\bibinfo {title} {Optimal measurements for
  resolution beyond the {{Rayleigh}} limit},}\ }\href {\doibase
  10.1364/OL.42.000231} {\bibfield  {journal} {\bibinfo  {journal} {Optics
  Letters}\ }\textbf {\bibinfo {volume} {42}},\ \bibinfo {pages} {231--234}
  (\bibinfo {year} {2017})}\BibitemShut {NoStop}%
\bibitem [{\citenamefont {Dunkl}\ and\ \citenamefont {Xu}(2014)}]{dunkl}%
  \BibitemOpen
  \bibfield  {author} {\bibinfo {author} {\bibfnamefont {Charles~F.}\
  \bibnamefont {Dunkl}}\ and\ \bibinfo {author} {\bibfnamefont {Yuan}\
  \bibnamefont {Xu}},\ }\href {\doibase 10.1017/CBO9781107786134} {\emph
  {\bibinfo {title} {Orthogonal Polynomials of Several Variables}}}\ (\bibinfo
  {publisher} {Cambridge University Press},\ \bibinfo {address} {Cambridge},\
  \bibinfo {year} {2014})\BibitemShut {NoStop}%
\bibitem [{\citenamefont {Bisketzi}\ \emph {et~al.}(2019)\citenamefont
  {Bisketzi}, \citenamefont {Branford},\ and\ \citenamefont
  {Datta}}]{bisketzi19}%
  \BibitemOpen
  \bibfield  {author} {\bibinfo {author} {\bibfnamefont {Evangelia}\
  \bibnamefont {Bisketzi}}, \bibinfo {author} {\bibfnamefont {Dominic}\
  \bibnamefont {Branford}}, \ and\ \bibinfo {author} {\bibfnamefont {Animesh}\
  \bibnamefont {Datta}},\ }\bibfield  {title} {\enquote {\bibinfo {title}
  {Quantum limits of localisation microscopy},}\ }\href {\doibase
  10.1088/1367-2630/ab58a0} {\bibfield  {journal} {\bibinfo  {journal} {New
  Journal of Physics}\ }\textbf {\bibinfo {volume} {21}},\ \bibinfo {pages}
  {123032} (\bibinfo {year} {2019})}\BibitemShut {NoStop}%
\bibitem [{\citenamefont {Lee}\ \emph {et~al.}(2023)\citenamefont {Lee},
  \citenamefont {Gagatsos}, \citenamefont {Guha},\ and\ \citenamefont
  {Ashok}}]{lee23}%
  \BibitemOpen
  \bibfield  {author} {\bibinfo {author} {\bibfnamefont {Kwan~Kit}\
  \bibnamefont {Lee}}, \bibinfo {author} {\bibfnamefont {Christos~N.}\
  \bibnamefont {Gagatsos}}, \bibinfo {author} {\bibfnamefont {Saikat}\
  \bibnamefont {Guha}}, \ and\ \bibinfo {author} {\bibfnamefont {Amit}\
  \bibnamefont {Ashok}},\ }\bibfield  {title} {\enquote {\bibinfo {title}
  {Quantum-inspired multi-parameter adaptive bayesian estimation for sensing
  and imaging},}\ }\href {\doibase 10.1109/JSTSP.2022.3214774} {\bibfield
  {journal} {\bibinfo  {journal} {IEEE Journal of Selected Topics in Signal
  Processing}\ }\textbf {\bibinfo {volume} {17}},\ \bibinfo {pages} {491--501}
  (\bibinfo {year} {2023})}\BibitemShut {NoStop}%
\bibitem [{\citenamefont {Bao}\ \emph {et~al.}(2021)\citenamefont {Bao},
  \citenamefont {Choi}, \citenamefont {Aggarwal},\ and\ \citenamefont
  {Jacob}}]{bao21}%
  \BibitemOpen
  \bibfield  {author} {\bibinfo {author} {\bibfnamefont {Fanglin}\ \bibnamefont
  {Bao}}, \bibinfo {author} {\bibfnamefont {Hyunsoo}\ \bibnamefont {Choi}},
  \bibinfo {author} {\bibfnamefont {Vaneet}\ \bibnamefont {Aggarwal}}, \ and\
  \bibinfo {author} {\bibfnamefont {Zubin}\ \bibnamefont {Jacob}},\ }\bibfield
  {title} {\enquote {\bibinfo {title} {Quantum-accelerated imaging of {{N}}
  stars},}\ }\href {\doibase 10.1364/OL.430404} {\bibfield  {journal} {\bibinfo
   {journal} {Optics Letters}\ }\textbf {\bibinfo {volume} {46}},\ \bibinfo
  {pages} {3045--3048} (\bibinfo {year} {2021})}\BibitemShut {NoStop}%
\bibitem [{\citenamefont {Choi}\ \emph {et~al.}(2024)\citenamefont {Choi},
  \citenamefont {Bao},\ and\ \citenamefont {Jacob}}]{choi24}%
  \BibitemOpen
  \bibfield  {author} {\bibinfo {author} {\bibfnamefont {Hyunsoo}\ \bibnamefont
  {Choi}}, \bibinfo {author} {\bibfnamefont {Fanglin}\ \bibnamefont {Bao}}, \
  and\ \bibinfo {author} {\bibfnamefont {Zubin}\ \bibnamefont {Jacob}},\
  }\bibfield  {title} {\enquote {\bibinfo {title} {Adaptive quantum accelerated
  imaging for space domain awareness},}\ }\href {\doibase
  10.1088/1367-2630/ad668c} {\bibfield  {journal} {\bibinfo  {journal} {New
  Journal of Physics}\ }\textbf {\bibinfo {volume} {26}},\ \bibinfo {pages}
  {073050} (\bibinfo {year} {2024})}\BibitemShut {NoStop}%
\bibitem [{\citenamefont {Sheppard}(1988)}]{sheppard88}%
  \BibitemOpen
  \bibfield  {author} {\bibinfo {author} {\bibfnamefont {C.~J.~R.}\
  \bibnamefont {Sheppard}},\ }\bibfield  {title} {\enquote {\bibinfo {title}
  {Super-resolution in confocal imaging},}\ }\href@noop {} {\bibfield
  {journal} {\bibinfo  {journal} {Optik}\ }\textbf {\bibinfo {volume} {80}},\
  \bibinfo {pages} {53} (\bibinfo {year} {1988})}\BibitemShut {NoStop}%
\bibitem [{\citenamefont {M{\"u}ller}\ and\ \citenamefont
  {Enderlein}(2010)}]{mueller10}%
  \BibitemOpen
  \bibfield  {author} {\bibinfo {author} {\bibfnamefont {Claus~B.}\
  \bibnamefont {M{\"u}ller}}\ and\ \bibinfo {author} {\bibfnamefont {J{\"o}rg}\
  \bibnamefont {Enderlein}},\ }\bibfield  {title} {\enquote {\bibinfo {title}
  {Image {{Scanning Microscopy}}},}\ }\href {\doibase
  10.1103/PhysRevLett.104.198101} {\bibfield  {journal} {\bibinfo  {journal}
  {Physical Review Letters}\ }\textbf {\bibinfo {volume} {104}},\ \bibinfo
  {pages} {198101} (\bibinfo {year} {2010})}\BibitemShut {NoStop}%
\bibitem [{\citenamefont {Gregor}\ and\ \citenamefont
  {Enderlein}(2019)}]{gregor19}%
  \BibitemOpen
  \bibfield  {author} {\bibinfo {author} {\bibfnamefont {Ingo}\ \bibnamefont
  {Gregor}}\ and\ \bibinfo {author} {\bibfnamefont {J{\"o}{\"o}rg}\
  \bibnamefont {Enderlein}},\ }\bibfield  {title} {\enquote {\bibinfo {title}
  {Image scanning microscopy},}\ }\href {\doibase 10.1016/j.cbpa.2019.05.011}
  {\bibfield  {journal} {\bibinfo  {journal} {Current Opinion in Chemical
  Biology}\ }\textbf {\bibinfo {volume} {51}},\ \bibinfo {pages} {74--83}
  (\bibinfo {year} {2019})}\BibitemShut {NoStop}%
\bibitem [{\citenamefont {Suzuki}\ \emph {et~al.}(2020)\citenamefont {Suzuki},
  \citenamefont {Yang},\ and\ \citenamefont {Hayashi}}]{suzuki20}%
  \BibitemOpen
  \bibfield  {author} {\bibinfo {author} {\bibfnamefont {Jun}\ \bibnamefont
  {Suzuki}}, \bibinfo {author} {\bibfnamefont {Yuxiang}\ \bibnamefont {Yang}},
  \ and\ \bibinfo {author} {\bibfnamefont {Masahito}\ \bibnamefont {Hayashi}},\
  }\bibfield  {title} {\enquote {\bibinfo {title} {Quantum state estimation
  with nuisance parameters},}\ }\href {\doibase 10.1088/1751-8121/ab8b78}
  {\bibfield  {journal} {\bibinfo  {journal} {Journal of Physics A:
  Mathematical and Theoretical}\ }\textbf {\bibinfo {volume} {53}},\ \bibinfo
  {pages} {453001} (\bibinfo {year} {2020})}\BibitemShut {NoStop}%
\bibitem [{\citenamefont {Shumway}\ and\ \citenamefont
  {Stoffer}(2017)}]{shumway_stoffer}%
  \BibitemOpen
  \bibfield  {author} {\bibinfo {author} {\bibfnamefont {Robert~H.}\
  \bibnamefont {Shumway}}\ and\ \bibinfo {author} {\bibfnamefont {David~S.}\
  \bibnamefont {Stoffer}},\ }\href {\doibase 10.1007/978-3-319-52452-8} {\emph
  {\bibinfo {title} {Time Series Analysis and Its Applications}}},\ \bibinfo
  {edition} {4th}\ ed.\ (\bibinfo  {publisher} {Springer},\ \bibinfo {address}
  {Cham, Switzerland},\ \bibinfo {year} {2017})\BibitemShut {NoStop}%
\bibitem [{\citenamefont {Ng}\ \emph {et~al.}(2016)\citenamefont {Ng},
  \citenamefont {Ang}, \citenamefont {Wheatley}, \citenamefont {Yonezawa},
  \citenamefont {Furusawa}, \citenamefont {Huntington},\ and\ \citenamefont
  {Tsang}}]{ng16}%
  \BibitemOpen
  \bibfield  {author} {\bibinfo {author} {\bibfnamefont {Shilin}\ \bibnamefont
  {Ng}}, \bibinfo {author} {\bibfnamefont {Shan~Zheng}\ \bibnamefont {Ang}},
  \bibinfo {author} {\bibfnamefont {Trevor~A.}\ \bibnamefont {Wheatley}},
  \bibinfo {author} {\bibfnamefont {Hidehiro}\ \bibnamefont {Yonezawa}},
  \bibinfo {author} {\bibfnamefont {Akira}\ \bibnamefont {Furusawa}}, \bibinfo
  {author} {\bibfnamefont {Elanor~H.}\ \bibnamefont {Huntington}}, \ and\
  \bibinfo {author} {\bibfnamefont {Mankei}\ \bibnamefont {Tsang}},\ }\bibfield
   {title} {\enquote {\bibinfo {title} {Spectrum analysis with quantum
  dynamical systems},}\ }\href {\doibase 10.1103/PhysRevA.93.042121} {\bibfield
   {journal} {\bibinfo  {journal} {Physical Review A}\ }\textbf {\bibinfo
  {volume} {93}},\ \bibinfo {pages} {042121} (\bibinfo {year}
  {2016})}\BibitemShut {NoStop}%
\bibitem [{\citenamefont {Tsang}(2023{\natexlab{c}})}]{noise_spec_pra}%
  \BibitemOpen
  \bibfield  {author} {\bibinfo {author} {\bibfnamefont {Mankei}\ \bibnamefont
  {Tsang}},\ }\bibfield  {title} {\enquote {\bibinfo {title} {Quantum noise
  spectroscopy as an incoherent imaging problem},}\ }\href {\doibase
  10.1103/PhysRevA.107.012611} {\bibfield  {journal} {\bibinfo  {journal}
  {Physical Review A}\ }\textbf {\bibinfo {volume} {107}},\ \bibinfo {pages}
  {012611} (\bibinfo {year} {2023}{\natexlab{c}})}\BibitemShut {NoStop}%
\bibitem [{\citenamefont {Sui}\ and\ \citenamefont {Tsang}(2026)}]{sui26}%
  \BibitemOpen
  \bibfield  {author} {\bibinfo {author} {\bibfnamefont {Xinyi}\ \bibnamefont
  {Sui}}\ and\ \bibinfo {author} {\bibfnamefont {Mankei}\ \bibnamefont
  {Tsang}},\ }\href {\doibase 10.48550/arXiv.2604.11614} {\enquote {\bibinfo
  {title} {Spectrum analysis with quantum dynamical systems. {{II}}.
  {{Finite-time}} analysis},}\ } (\bibinfo {year} {2026}),\ \Eprint
  {http://arxiv.org/abs/2604.11614} {arXiv:2604.11614 [quant-ph]} \BibitemShut
  {NoStop}%
\bibitem [{\citenamefont {Padilla}\ \emph {et~al.}(2026)\citenamefont
  {Padilla}, \citenamefont {Sajjad}, \citenamefont {Saif},\ and\ \citenamefont
  {Guha}}]{padilla26}%
  \BibitemOpen
  \bibfield  {author} {\bibinfo {author} {\bibfnamefont {Isack}\ \bibnamefont
  {Padilla}}, \bibinfo {author} {\bibfnamefont {Aqil}\ \bibnamefont {Sajjad}},
  \bibinfo {author} {\bibfnamefont {Babak~N.}\ \bibnamefont {Saif}}, \ and\
  \bibinfo {author} {\bibfnamefont {Saikat}\ \bibnamefont {Guha}},\ }\bibfield
  {title} {\enquote {\bibinfo {title} {Superresolution {{Imaging}} with
  {{Entanglement-Enhanced Telescopy}}},}\ }\href {\doibase 10.1103/354q-ch63}
  {\bibfield  {journal} {\bibinfo  {journal} {Physical Review Letters}\
  }\textbf {\bibinfo {volume} {136}},\ \bibinfo {pages} {010803} (\bibinfo
  {year} {2026})}\BibitemShut {NoStop}%
\bibitem [{\citenamefont {Van~Trees}(2001)}]{vantrees}%
  \BibitemOpen
  \bibfield  {author} {\bibinfo {author} {\bibfnamefont {Harry~L.}\
  \bibnamefont {Van~Trees}},\ }\href@noop {} {\emph {\bibinfo {title}
  {Detection, Estimation, and Modulation Theory, Part {{I}}.}}}\ (\bibinfo
  {publisher} {John Wiley \& Sons},\ \bibinfo {address} {New York},\ \bibinfo
  {year} {2001})\BibitemShut {NoStop}%
\bibitem [{\citenamefont {Tsybakov}(2009)}]{tsybakov}%
  \BibitemOpen
  \bibfield  {author} {\bibinfo {author} {\bibfnamefont {Alexandre~B.}\
  \bibnamefont {Tsybakov}},\ }\href {\doibase 10.1007/b13794} {\emph {\bibinfo
  {title} {Introduction to Nonparametric Estimation}}}\ (\bibinfo  {publisher}
  {Springer},\ \bibinfo {address} {New York},\ \bibinfo {year}
  {2009})\BibitemShut {NoStop}%
\bibitem [{\citenamefont {Johnstone}(2019)}]{johnstone}%
  \BibitemOpen
  \bibfield  {author} {\bibinfo {author} {\bibfnamefont {Iain~M.}\ \bibnamefont
  {Johnstone}},\ }\href {https://imjohnstone.su.domains/GE_09_16_19.pdf} {\emph
  {\bibinfo {title} {Gaussian Estimation: {{Sequence}} and Wavelet Models}}},\
  \bibinfo {edition} {draft}\ ed.\ (\bibinfo  {publisher} {online},\ \bibinfo
  {year} {2019})\ \bibinfo {note} {{Retrieved 18 July 2026.
  \url{https://imjohnstone.su.domains/GE_09_16_19.pdf}}}\BibitemShut {NoStop}%
\bibitem [{\citenamefont {Meister}(2009)}]{meister}%
  \BibitemOpen
  \bibfield  {author} {\bibinfo {author} {\bibfnamefont {Alexander}\
  \bibnamefont {Meister}},\ }\href {\doibase 10.1007/978-3-540-87557-4} {\emph
  {\bibinfo {title} {Deconvolution Problems in Nonparametric Statistics}}}\
  (\bibinfo  {publisher} {Springer},\ \bibinfo {address} {Berlin},\ \bibinfo
  {year} {2009})\BibitemShut {NoStop}%
\bibitem [{\citenamefont {Tsang}\ \emph {et~al.}(2011)\citenamefont {Tsang},
  \citenamefont {Wiseman},\ and\ \citenamefont {Caves}}]{twc}%
  \BibitemOpen
  \bibfield  {author} {\bibinfo {author} {\bibfnamefont {Mankei}\ \bibnamefont
  {Tsang}}, \bibinfo {author} {\bibfnamefont {Howard~M.}\ \bibnamefont
  {Wiseman}}, \ and\ \bibinfo {author} {\bibfnamefont {Carlton~M.}\
  \bibnamefont {Caves}},\ }\bibfield  {title} {\enquote {\bibinfo {title}
  {Fundamental {{Quantum Limit}} to {{Waveform Estimation}}},}\ }\href
  {\doibase 10.1103/PhysRevLett.106.090401} {\bibfield  {journal} {\bibinfo
  {journal} {Physical Review Letters}\ }\textbf {\bibinfo {volume} {106}},\
  \bibinfo {pages} {090401} (\bibinfo {year} {2011})}\BibitemShut {NoStop}%
\bibitem [{\citenamefont {Chen}\ and\ \citenamefont {Moitra}(2021)}]{chen21}%
  \BibitemOpen
  \bibfield  {author} {\bibinfo {author} {\bibfnamefont {Sitan}\ \bibnamefont
  {Chen}}\ and\ \bibinfo {author} {\bibfnamefont {Ankur}\ \bibnamefont
  {Moitra}},\ }\bibfield  {title} {\enquote {\bibinfo {title} {Algorithmic
  foundations for the diffraction limit},}\ }in\ \href {\doibase
  10.1145/3406325.3451078} {\emph {\bibinfo {booktitle} {Proceedings of the
  53rd Annual {{ACM SIGACT}} Symposium on Theory of Computing}}},\ \bibinfo
  {series and number} {Stoc 2021}\ (\bibinfo  {publisher} {Association for
  Computing Machinery},\ \bibinfo {address} {New York, NY, USA},\ \bibinfo
  {year} {2021})\ pp.\ \bibinfo {pages} {490--503}\BibitemShut {NoStop}%
\bibitem [{\citenamefont {Sch{\"u}tzenberger}(1957)}]{schutzenberger57}%
  \BibitemOpen
  \bibfield  {author} {\bibinfo {author} {\bibfnamefont {M.~P.}\ \bibnamefont
  {Sch{\"u}tzenberger}},\ }\bibfield  {title} {\enquote {\bibinfo {title} {A
  generalization of the {{Fr\'echet-Cram\'er}} inequality to the case of
  {{Bayes}} estimation},}\ }\href {\doibase 10.1090/S0002-9904-1957-10102-5}
  {\bibfield  {journal} {\bibinfo  {journal} {Bull. Amer. Math. Soc.}\ }\textbf
  {\bibinfo {volume} {63}},\ \bibinfo {pages} {142} (\bibinfo {year}
  {1957})}\BibitemShut {NoStop}%
\bibitem [{\citenamefont {Gill}\ and\ \citenamefont {Levit}(1995)}]{gill95}%
  \BibitemOpen
  \bibfield  {author} {\bibinfo {author} {\bibfnamefont {Richard~D.}\
  \bibnamefont {Gill}}\ and\ \bibinfo {author} {\bibfnamefont {Boris~Y.}\
  \bibnamefont {Levit}},\ }\bibfield  {title} {\enquote {\bibinfo {title}
  {Applications of the {{Van Trees}} inequality: A {{Bayesian Cram\'er-Rao}}
  bound},}\ }\href {http://www.jstor.org/stable/3318681} {\bibfield  {journal}
  {\bibinfo  {journal} {Bernoulli. Official Journal of the Bernoulli Society
  for Mathematical Statistics and Probability}\ }\textbf {\bibinfo {volume}
  {1}},\ \bibinfo {pages} {59--79} (\bibinfo {year} {1995})},\ \Eprint
  {http://arxiv.org/abs/3318681} {3318681} \BibitemShut {NoStop}%
\bibitem [{\citenamefont {Tsang}(2020)}]{bcrb_pra}%
  \BibitemOpen
  \bibfield  {author} {\bibinfo {author} {\bibfnamefont {Mankei}\ \bibnamefont
  {Tsang}},\ }\bibfield  {title} {\enquote {\bibinfo {title} {Physics-inspired
  forms of the {{Bayesian Cram\'er-Rao}} bound},}\ }\href {\doibase
  10.1103/PhysRevA.102.062217} {\bibfield  {journal} {\bibinfo  {journal}
  {Physical Review A}\ }\textbf {\bibinfo {volume} {102}},\ \bibinfo {pages}
  {062217} (\bibinfo {year} {2020})}\BibitemShut {NoStop}%
\bibitem [{\citenamefont {Tsang}(2023{\natexlab{d}})}]{zzb_gen}%
  \BibitemOpen
  \bibfield  {author} {\bibinfo {author} {\bibfnamefont {Mankei}\ \bibnamefont
  {Tsang}},\ }\href {\doibase 10.48550/arXiv.2306.08660} {\enquote {\bibinfo
  {title} {Ziv-{{Zakai-type}} error bounds for general statistical models},}\ }
  (\bibinfo {year} {2023}{\natexlab{d}}),\ \Eprint
  {http://arxiv.org/abs/2306.08660} {arXiv:2306.08660 [math.ST]} \BibitemShut
  {NoStop}%
\bibitem [{\citenamefont {Bemis}(2016)}]{bemis21}%
  \BibitemOpen
  \bibfield  {author} {\bibinfo {author} {\bibfnamefont {Robert}\ \bibnamefont
  {Bemis}},\ }\href
  {https://www.mathworks.com/matlabcentral/fileexchange/17555-light-bartlein-color-maps}
  {\enquote {\bibinfo {title} {Light {{Bartlein Color Maps}}},}\ }\bibinfo
  {howpublished} {MATLAB Central File Exchange (online)} (\bibinfo {year}
  {2016}),\ \bibinfo {note}
  {\url{https://www.mathworks.com/matlabcentral/fileexchange/17555-light-bartlein-color-maps}.
  Retrieved May 30, 2021}\BibitemShut {NoStop}%
\bibitem [{\citenamefont {Mandel}\ and\ \citenamefont {Wolf}(1995)}]{mandel}%
  \BibitemOpen
  \bibfield  {author} {\bibinfo {author} {\bibfnamefont {Leonard}\ \bibnamefont
  {Mandel}}\ and\ \bibinfo {author} {\bibfnamefont {Emil}\ \bibnamefont
  {Wolf}},\ }\href {\doibase 10.1017/CBO9781139644105} {\emph {\bibinfo {title}
  {Optical Coherence and Quantum Optics}}}\ (\bibinfo  {publisher} {Cambridge
  University Press},\ \bibinfo {address} {Cambridge},\ \bibinfo {year}
  {1995})\BibitemShut {NoStop}%
\bibitem [{\citenamefont {Fomenko}(1995)}]{fomenko}%
  \BibitemOpen
  \bibfield  {author} {\bibinfo {author} {\bibfnamefont {A.~T.}\ \bibnamefont
  {Fomenko}},\ }\href@noop {} {\emph {\bibinfo {title} {Symplectic
  Geometry}}},\ \bibinfo {edition} {2nd}\ ed.\ (\bibinfo  {publisher} {{Gordon
  and Breach}},\ \bibinfo {address} {Luxembourg},\ \bibinfo {year}
  {1995})\BibitemShut {NoStop}%
\bibitem [{\citenamefont {Tsang}(2018{\natexlab{b}})}]{minimax_jmo}%
  \BibitemOpen
  \bibfield  {author} {\bibinfo {author} {\bibfnamefont {Mankei}\ \bibnamefont
  {Tsang}},\ }\bibfield  {title} {\enquote {\bibinfo {title} {Conservative
  classical and quantum resolution limits for incoherent imaging},}\ }\href
  {\doibase 10.1080/09500340.2017.1377306} {\bibfield  {journal} {\bibinfo
  {journal} {Journal of Modern Optics}\ }\textbf {\bibinfo {volume} {65}},\
  \bibinfo {pages} {1385--1391} (\bibinfo {year}
  {2018}{\natexlab{b}})}\BibitemShut {NoStop}%
\bibitem [{\citenamefont {Tseng}(1949)}]{tseng49}%
  \BibitemOpen
  \bibfield  {author} {\bibinfo {author} {\bibfnamefont {Yuan-Yung}\
  \bibnamefont {Tseng}},\ }\bibfield  {title} {\enquote {\bibinfo {title}
  {Generalized inverses of unbounded operators between two unitary spaces},}\
  }\href@noop {} {\bibfield  {journal} {\bibinfo  {journal} {Doklady Akad. Nauk
  SSSR (N.S.)}\ }\textbf {\bibinfo {volume} {67}},\ \bibinfo {pages} {431--434}
  (\bibinfo {year} {1949})}\BibitemShut {NoStop}%
\bibitem [{\citenamefont {Simon}(2005)}]{simon10}%
  \BibitemOpen
  \bibfield  {author} {\bibinfo {author} {\bibfnamefont {Barry}\ \bibnamefont
  {Simon}},\ }\href {\doibase 10.1090/surv/120} {\emph {\bibinfo {title} {Trace
  {{Ideals}} and {{Their Applications}}}}},\ \bibinfo {edition} {2nd}\ ed.,\
  \bibinfo {series} {Mathematical {{Surveys}} and {{Monographs}}}, Vol.\
  \bibinfo {volume} {120}\ (\bibinfo  {publisher} {American Mathematical
  Society},\ \bibinfo {address} {Providence, Rhode Island},\ \bibinfo {year}
  {2005})\BibitemShut {NoStop}%
\bibitem [{\citenamefont {Horn}\ and\ \citenamefont {Johnson}(1991)}]{horn2}%
  \BibitemOpen
  \bibfield  {author} {\bibinfo {author} {\bibfnamefont {Roger~A.}\
  \bibnamefont {Horn}}\ and\ \bibinfo {author} {\bibfnamefont {Charles~R.}\
  \bibnamefont {Johnson}},\ }\href {\doibase 10.1017/CBO9780511840371} {\emph
  {\bibinfo {title} {Topics in Matrix Analysis}}}\ (\bibinfo  {publisher}
  {Cambridge University Press},\ \bibinfo {address} {Cambridge, England, UK},\
  \bibinfo {year} {1991})\BibitemShut {NoStop}%
\bibitem [{\citenamefont {Tsang}(2025)}]{tsang_intro}%
  \BibitemOpen
  \bibfield  {author} {\bibinfo {author} {\bibfnamefont {Mankei}\ \bibnamefont
  {Tsang}},\ }\href@noop {} {\emph {\bibinfo {title} {Introduction to Quantum
  Optics}}}\ (\bibinfo  {publisher} {Public},\ \bibinfo {address} {Online},\
  \bibinfo {year} {2025})\ \bibinfo {note}
  {\href{https://www.ece.nus.edu.sg/stfpage/tmk/Quantum_Optics_v0.5.pdf}{https://www.ece.nus.edu.sg/stfpage/tmk/Quantum\_Optics\_v0.5.pdf}}\BibitemShut
  {NoStop}%
\bibitem [{\citenamefont {Yang}\ \emph {et~al.}(2017)\citenamefont {Yang},
  \citenamefont {Nair}, \citenamefont {Tsang}, \citenamefont {Simon},\ and\
  \citenamefont {Lvovsky}}]{yang17}%
  \BibitemOpen
  \bibfield  {author} {\bibinfo {author} {\bibfnamefont {Fan}\ \bibnamefont
  {Yang}}, \bibinfo {author} {\bibfnamefont {Ranjith}\ \bibnamefont {Nair}},
  \bibinfo {author} {\bibfnamefont {Mankei}\ \bibnamefont {Tsang}}, \bibinfo
  {author} {\bibfnamefont {Christoph}\ \bibnamefont {Simon}}, \ and\ \bibinfo
  {author} {\bibfnamefont {Alexander~I.}\ \bibnamefont {Lvovsky}},\ }\bibfield
  {title} {\enquote {\bibinfo {title} {Fisher information for far-field linear
  optical superresolution via homodyne or heterodyne detection in a
  higher-order local oscillator mode},}\ }\href {\doibase
  10.1103/PhysRevA.96.063829} {\bibfield  {journal} {\bibinfo  {journal}
  {Physical Review A}\ }\textbf {\bibinfo {volume} {96}},\ \bibinfo {pages}
  {063829} (\bibinfo {year} {2017})}\BibitemShut {NoStop}%
\end{thebibliography}%

\end{document}